\documentclass[a4paper,10pt]{amsart}
\usepackage{amssymb}
\usepackage{amsmath}
\usepackage[mathscr]{eucal}
\usepackage{mathrsfs}
\usepackage{xspace}
\usepackage[latin1]{inputenc}
\usepackage[english]{babel}
\usepackage[all,2cell]{xy}
\usepackage{xcolor}
\usepackage{graphicx}
\usepackage{wasysym} % music

\usepackage{stmaryrd} % symbols computer sciences
\UseAllTwocells
\newcommand{\at}[1]%
            {\ensuremath{\protect\underline{\mathbf{#1}}}} % algebraic theory
\newcommand{\op}[1]{\ensuremath{\operatorname{#1}}}        % Operador
\newcommand{\h}[1][]                                       % \h-conjunto $\Omega^{\h}$
 {\ifthenelse{\boolean{mmode}}%
  {$\mathrm{h}$}%
  {h\nobreakdash#1\hspace{0pt}}}

\newcommand{\id}{\mathrm{id}}        % id(entidad)
\newcommand{\comp}{\circ}          % composicion
\newcommand{\adcomp}%
  {\overset{\operatorname{ad}}{\comp}} % Ad-composicion cuadrados adjuntos
\newcommand{\funcomp}%
  {\overset{\operatorname{fn}}{\comp}}% Fun-composiciÛn cuadrados adjuntos
\newcommand{\sccat}
{\mathbin{\kern-1pt\raisebox{6pt}{.}\kern-5pt
\downarrow\kern-5pt\raisebox{6pt}{.}\kern-1pt}}

\newcommand{\parrow}[1]%           % Projective arrow
   {\underset{{\displaystyle \raisebox{5pt}%
   {$\longleftarrow$}}}{\op{#1}}{\,}}
\newcommand{\iarrow}[1]%           % Inductive arrow
   {\underset{{\displaystyle \raisebox{5pt}%
   {$\longrightarrow$}}}{\op{#1}}{\,}}

\newcommand{\uadj}{\top}           % uadj T
\newcommand{\rest}%
{\mathnormal{\restriction}}        % restriction

\newcommand{\iso}{\cong}           % isomorfo
\DeclareMathOperator{\card}{card}  % cardinal card(A)
\newcommand{\function}[4]{
            \begin{array}{@{\:}c@{\:}c@{\:}l}
                   #1 &\mor& #2 \\
                   #3 &\longmapsto& #4
            \end{array} }
\newcommand{\nfunction}[4]
    {\left\{
     \function{#1}{#2}{#3}{#4}
     \right. }

\newcommand{\bb}[1]{\ensuremath
 {\lvert #1 \rvert}}
\DeclareMathOperator{\Cgr}{Cgr}    % Congruencia
\newcommand{\pr}{\mathrm{pr}}        % proyecciÛn canÛnica de A en A/Phi
\newcommand{\vs}[1]{\mathbin{\downarrow}#1}
\DeclareMathOperator{\G}{G}        % Ground functor olvido
\DeclareMathOperator{\bconcat}
            {\curlywedge}

\newcommand{\concat}
  {\ensuremath{\text
  {\Large $\curlywedge$}}}
\newcommand{\ext}[1]
  {\ensuremath{#1^{\sharp}}}

\newcommand{\ol}{\overline}

\newcommand{\brel}{\ensuremath{\xymatrix{{}\arity@{{*}{-}{*}}[r] & {}}}}

\newcommand{\nseq}[3]{\xymatrix@1@C=16pt{#1 \arity@{>}[r]_-{\scriptscriptstyle{#2}} & #3 }}

\NoCompileMatrices              % compila los diagramas \NoCompileMatrices
\xymatrixrowsep = {8ex}         % o 40pt
\xymatrixcolsep = {10ex}        % o 50pt

\newdir{<= }{:a(+25)\dir{-}*:a(-25)\dir{-}*!/:a(85)+1.8pt/:a(-25)\dir{-}}
\newdir{ = }{*@{=}}
\newdir{+>}{@{|}*@{-}!/-6.5pt/@{>}*!/-13pt/{}}
\newdir{ +}{{}*!/-5pt/@{+}}
\newdir{ >}{{}*!/-10pt/@{>}}
\newdir{.-}{{}*:a(-90)h!/2.1pt/{.}}
\newdir{.->}{:a(-90)h!/2.5pt/{.}*!/4pt/\dir{-}*!/0pt/\dir{-}*%
!/-2pt/\dir{-}*!/-8pt/\dir{>}}
\newdir{-.->}%
{:a(-90)h!/2pt/{.}*!/8pt/\dir{-}*!/4pt/\dir{-}*!/0pt/\dir{-}*%
!/-5pt/\dir{-}*!/-11pt/\dir{>}}
\newdir{~>}{!/4.6pt/\dir{~}*!/1pt/\dir{-}*!/-5pt/\dir{>}}
\newdir{~~>}{!/7pt/\dir{~}*!/0pt/\dir{~}*!/-5pt/\dir{>}}
\newdir{=~>}{*!/1pt/@2{~}*!/6pt/{}*!/-6pt/@2{>}}
\newdir{<~~>}{!/21pt/\dir{<}*!/21pt/\dir{-}!/12pt/\dir{~}!/5pt/\dir{~}*!/1pt/\dir{-}*!/-3pt/\dir{>}}
\newdir{=>}{!/5pt/\dir{=}!/2.5pt/\dir{=}*!/-5pt/\dir2{>}*!/7pt/\dir{ }}
\newdir{==>}{!/-2pt/\dir{=}!/-6pt/\dir{=}%
 !/-10pt/\dir{=}*!/-18pt/\dir2{>}}
\newdir{<==>}{!\dir2{<}!/-2pt/\dir{=}!/-6pt/\dir{=}%
 !/-10pt/\dir{=}*!/-17pt/\dir2{>}}
\newdir{< }{{}*!/12pt/\dir{<}}
\newdir{-o}{!/1.3pt/{}*!/0pt/{}@{o}*!/-5pt/{}}

\newsavebox{\xymor}  % --->
\newsavebox{\xymon}  % +-->
\newsavebox{\xyepi}  % --+>
\newsavebox{\xytn}   % -.->
\newsavebox{\xyrel}  % ---o
\newsavebox{\xycel}  % ===>
\newsavebox{\xymdf}  % ~~~>
\newsavebox{\xyumor} % ---`
\newsavebox{\xydmor} % ---,
\newsavebox{\xyomor} % ---<
\newsavebox{\xyemor} % >--->

\newcommand{\xynode}{\makebox[0ex]{}}
\savebox{\xymor}{\ensuremath{%
\xymatrix@1@C=19pt{\xynode \ar@{>}[r] & \xynode }}}
\savebox{\xymon}{\ensuremath{%
\xymatrix@1@C=19pt{\xynode \ar@{{ +}{-}{>}}[r] & \xynode }}}
\savebox{\xyepi}{\ensuremath{%
\xymatrix@1@C=19pt{\xynode \ar@{{}{-}{+>}}[r] & \xynode }}}
\savebox{\xytn}{\ensuremath{%
\xymatrix@1@C=19pt{\xynode \ar[r]|(.44){\object@{.-}} & \xynode
}}}
\savebox{\xyrel}{\ensuremath{%
\xymatrix@1@C=19pt{\xynode \ar@{{}{-}{-o}}[r] & \xynode }}}
\savebox{\xycel}{\ensuremath{%
\xymatrix@1@C=19pt{\xynode \ar@{=>}[r] & \xynode }}}
\savebox{\xymdf}{\ensuremath{%
\xymatrix@1@C=16pt{\xynode \ar@{}[r]|{\dir{~>}} & \xynode}}}
\savebox{\xyumor}{\ensuremath{%
\xymatrix@1@C=19pt{\xynode \ar@{{}{-}^{>}}[r] & \xynode }}}
\savebox{\xydmor}{\ensuremath{%
\xymatrix@1@C=19pt{\xynode \ar@{{}{-}_{>}}[r] & \xynode }}}
\savebox{\xyomor}{\ensuremath{%
\xymatrix@1@C=19pt{\xynode \ar@{{}{-}^{< }}[r] & \xynode }}}
\savebox{\xyemor}{\ensuremath{%
\xymatrix@1@C=19pt{\xynode \ar@{{ >}{-}{>}}[r] & \xynode }}}

\newcommand{\mor}{\usebox{\xymor}}    % mor  --->
\newcommand{\functor}[9]{
 \xymatrix{
    #4 \save[]+<0ex,5ex>*+{#1}="1"  \restore
      \arity[d]_{#6}  \arity@{}[rd]|{\longmapsto}
  & #5 \save[]+<0ex,5ex>*+{#3}="3"  \restore
      \arity[d]^{#7}
  \\
   #8 & #9 \arity "1";"3"^-{#2} } }
\newcommand{\functornd}[9]{
 \xymatrix{
    #4 \save[]+<0ex,5ex>*+{#1}="1"  \restore
      \arity[d]_{#6}  \arity@{}[rd]|{\longmapsto}
  & #5 \save[]+<0ex,5ex>*+{#3}="3"  \restore
  \\
   #8 & #9 \arity[u]_{#7} \arity "1";"3"^-{#2} } }
\newcommand{\functordn}[9]{
 \xymatrix{
    #4 \save[]+<0ex,5ex>*+{#1}="1"  \restore
       \arity@{}[rd]|{\longmapsto}
  & #5 \save[]+<0ex,5ex>*+{#3}="3"  \restore
      \arity[d]^{#7}
  \\
   #8  \arity[u]^{#6}  & #9 \arity "1";"3"^-{#2} } }

\newcommand{\larr}{->}
\newcommand{\rarr}{->}

\newcommand{\xfunctor}[9]{
 \xymatrix{
    #4 \save[]+<0ex,5ex>*+{#1}="1"  \restore
      \ifthenelse{\equal{\larr}{->}}{\arity[d]_{#6}}{}
      \ifthenelse{\equal{\larr}{<-}}{\arity[d];[]^{#6}}{}
      \ifthenelse{\equal{\larr}{-<}}{\arity@{< }[d]_{#6}}{}
      \arity@{}[rd]|{\longmapsto}
  & #5 \save[]+<0ex,5ex>*+{#3}="3"  \restore
      \ifthenelse{\equal{\rarr}{->}}{\arity[d]^{#7}}{}
      \ifthenelse{\equal{\rarr}{<-}}{\arity[d];[]_{#7}}{}
      \ifthenelse{\equal{\rarr}{-<}}{\arity@{< }[d]^{#7}}{}
  \\
   #8 & #9 \arity "1";"3"^-{#2} } }

\theoremstyle{plain}
\newtheorem{theorem}{Theorem}[section]
\newtheorem{proposition}[theorem]{Proposition}
\newtheorem{corollary}[theorem]{Corollary}
\newtheorem{lemma}[theorem]{Lemma}
\theoremstyle{definition}
\newtheorem{definition}[theorem]{Definition}

\newtheorem*{remark}{Remark}
\newtheorem*{assumption}{\bf Assumption}

\newcommand{\arity}{\mathsf{ar}}

\numberwithin{equation}{section}
\begin{document}
%% Two authors
\title[Congruence based proofs of the recognizability theorems]{Congruence based proofs of the recognizability theorems for free many-sorted algebras}

\author[Climent]{J. Climent Vidal}
\address{Universitat de Val\`{e}ncia\\
         Departament de L\`{o}gica i Filosofia de la Ci\`{e}ncia\\
         Av. Blasco Ib\'{a}\~{n}ez, 30-$7^{\mathrm{a}}$, 46010 Val\`{e}ncia, Spain}
\email{Juan.B.Climent@uv.es}
%\thanks{The research of the first author}
\author[Cosme]{E. Cosme Ll\'{o}pez}
\address{Universitat de Val\`{e}ncia\\
         Departament de Matem\`{a}tiques\\
         Dr. Moliner, 50, 46100 Burjassot, Val\`{e}ncia, Spain}
\email{Enric.Cosme@uv.es}
%\thanks{The research of the second author}

 \keywords{}
\date{March 5th, 2018}

\begin{abstract}
We generalize several recognizability theorems for free single-sorted algebras to the field of many-sorted algebras and provide, in a uniform way and without using neither regular tree grammars nor tree automata, purely algebraic proofs of them based on the concept of congruence.
\end{abstract}

\maketitle

\section{Introduction}

The definition of recognizable language set down by Rabin and Scott in~\cite{RS59}, Definition 2, on p.~116, originated in the characterizations of regular languages given by Myhill~\cite{Myh57}, in terms of congruence relations of finite index, and Nerode~\cite{Ner58}, in terms of right invariant equivalence relations of finite index, and was formulated for sets of elements of an arbitrary single-sorted abs\-tract algebra by Mezei and Wright~\cite{MW67}. In this setting, a subset of an algebra is recognizable if there exists an homomorphism from it into a finite one for which the language coincides with the inverse image by the homomorphism of a subset of the finite algebra. Such a recognizability notion, as is well known, is equivalent to providing a congruence of finite index on the algebra for which the language in question is saturated.

%The terms of the free algebra have a graphical representation as trees with the root symbol at the top and leaves or frontier symbols, at the bottom~\cite{Cou89,GS84}. Different devices were introduced to recognize tree languages. A frontier-to-root, or bottom-up, recognizer~\cite{Don70, MM69, TW68} evaluates an expression for given assignments on the variables and it decides then, on the basis of this value, whether the expression belongs to a given set or not. On the contrary, a root-to-frontier, or top down, recognizer~\cite{MM69, Rab69} accepts a language if it is possible to choose an initial state for the root of the terms in the language leading to frontier acceptance. These devices can be deterministic or not.

%In~\cite{GS84} it is shown that, among all this four devices, deterministic root-to-frontier recognizers are essentially weaker and they recognize a proper subset, whilst the notion of recognizability on the other three devices is shown to be equivalent and coincides with the notion of recognizability introduced by Mezei and Wright~\cite{MW67} for arbitrary algebras. As in the case of formal languages, the recognition of a language of a free algebra can be done by using either an algebra or by using a term grammar. The grammar approach started in~\cite{Bra69}. In~\cite[Theorem~3.6]{GS84} it is shown that both approaches are equivalent for free algebras.

This paper is devoted to the study of several recognizability results for free many-sorted algebras. Concretely, we generalize to the field of many-sorted algebras the results presented by G\'ecseg and Steinby in~Section~2.4 of~\cite{GS84}, in the field of single-sorted algebras, on the preservation and, when appropriate, reflection of the recognizability under the action of the operators of substitution, iteration, quotient, inverse image by a tree homomorphism, and direct image by linear tree homomorphism on terms ($\equiv$ trees) of free single-sorted algebras.
%On this point it is suitable to quote what G\'ecseg and Steinby, in Chapter~$\mathrm{II}$ of~\cite{GS84}, on p.~121, wrote: ``Many of the basic results presented in this chapter were obtained in various forms by several authors, and often it would be hard to establish any priority.''

The device used by G\'ecseg and Steinby in~\cite{GS84} for proving the aforementioned results concerning recognizability was that of regular tree grammars.
%In fact, a tree language can be generated by a regular tree grammar if and only if it is recognizable (for a proof see~\cite{GS84}, Theorem~3.6, on p.~71).
As in the case of formal languages, there exists another natural device intended for the same purpose: Tree automata. We recall that tree automata were defined, among others, by Doner~\cite{Don65, Don70} and Thatcher and Wright~\cite{TW65, TW68} and that the primary goal of the theory of tree automata was to apply it to decision problems of second order logic. On the other hand, regular tree  grammars were defined by Brainerd~\cite{Bra69} and the main result of \cite{Bra69}, Theorem 4.9, on p.~230, is: ``The sets of trees generated by regular systems [$\equiv$ regular tree grammars, \emph{we add}] are exactly those accepted by tree automata.''

%``Tree homomorphisms are not really homomorphisms in the sense of algebra. The concept is the result of the dual nature of words'', see~\cite[p.~70]{GS84}. General tree morphisms can be seen as special cases of finite state transductions, see~\cite{Tha73}. The greater generality of terms compared with languages provides an entirely new phenomena such as subtree copying, with become problematic on some instances. General tree morphisms do not preserve recognizability and it is necessary to restrict the concept to get the desired result. Thus, imposing linearity in the tree morphisms leads to good recognizability properties. An special case of this notion are alphabetic tree morphisms, also called projections, which are precisely derivors in the sense of. It was proved in~\cite{TW68} that the direct and inverse image of a recognizable language under a derivor is again recognizable.

As an alternative to the above devices, regular tree grammars and tree automata, in this paper we supply congruence based proofs of the different recognizability theorems for the strictly more general setting of free many-sorted algebras. All results stated in this paper follow the same methodology: Starting from a congruence of finite index saturating input languages, which is ultimately based on the different syntactic congruences determined by the input languages, we provide a finite index refinement for it recognizing the transformed language which is obtained by means of the aforementioned operators.
%In this way uniformity in the proofs of the theorems is achieved.
This proof strategy, which, ultimately, leads to a proof schema, has proven very effective in providing a uniform approach for each of the cases in question, and, in addition, it has served to confirm, once again, the fundamental role played by the notion of congruence cogenerated by an equivalence relation on the underlying many-sorted set of a many-sorted algebra in the area of theoretical computer science.

In this regard it is worthwhile quoting (1) a part of Atiyah's answer in \cite{rs04}, p.~24, to the question, formulated by the interviewers %M. Raussen and Chr. Skau,
about the underlying motivation for providing different proofs with different strategies for the Atiyah-Singer Index Theorem: ``Any good theorem should have several proofs, the more the better. For two reasons: usually, different proofs have different strengths and weaknesses, and they generalise in different directions---they are not just repetitions of each other. [$\ldots$] the more perspectives, the better!'', and (2) what Lang, referring to his own review of the historical development of class field theory, wrote in the preface of \cite{sl96}: ``As I stated in the preface to my \emph{Algebraic Number Theory}, there are several approaches to class field theory. None of them makes any other obsolete, and each gives a different insight from the others.''

%Besides, we have been keen to stress the importance of the categorial rendering of the operators introduced in the paper.  We want to highlight the work of formalization, notation and exposition of all the different operators used throughout the text, in particular in the subsection devoted to tree homomorphisms, where we have simplified most of the proofs proposed by GÈcseg and Steinby in~\cite{GS84}. In this regard, we have decided to include a last subsection on derivors, a particular case of the notion of hyperderivor defining a tree homomorphism, for its relevance as transformations,  for their good compositional properties, and for the mathematical naturalness of the concept, in contrast with other, more general transformations, like tree homomorphisms.  We hope that this systematic effort could help to enhance the  understanding of the different objects and operators at work.

We next proceed to summarize the contents of the subsequent sections of this paper.

%We next proceed to succinctly summarize the contents of the subsequent sections of this paper. The reader will find a more detailed explanation at the beginning of the succeeding sections.

In Section 2, for the convenience of the reader, we recall those notions and facts, these latter mostly without proofs, on many-sorted sets, many-sorted algebras, and recognizability for many-sorted algebras, that we will need.

In Section 3 we provide congruence based proofs of the recognizability theorems for free many-sorted algebras.  As a matter of fact, in order to deal with the different cases of recognizability, classified according to the type of operator under consideration, we have divided  this section into several subsections.
In Subsection 1, entitled \emph{Basic terms}, we prove that the final sets containing a variable, a constant or an operation symbol applied to a suitable family of variables are recognizable.
In Subsection 2, entitled \emph{Substitutions}, after defining several substitution op\-era\-tors associated to a free many-sorted algebra and investigating the relationships be\-tween them, we state the main result of this subsection: If all input languages for a given substitution are recognizable, then the output language is recognizable as well.
In Subsection 3, entitled \emph{Iterations}, we introduce the notion of iteration of a language with respect to a variable and we prove that, if the input language is recognizable then its iteration with respect to a variable is also recognizable. In Subsection 4, entitled \emph{Quotients}, we introduce the notion of quotient of a language by another, not necessarily recognizable, language with respect to a variable and we prove that if the input language is recognizable then its quotient by another language with respect to a variable is also recognizable.
In Subsection 5, entitled \emph{Tree Homomorphisms}, we define the notion of hyperderivor from a many-sorted sig\-na\-ture labeled with a suitable many-sorted set (the domain of the hyperderivor) to another (the codomain of the hyperderivor) as consisting of two components: One transforming many-sorted operation symbols from the underlying many-sorted sig\-na\-ture of the domain into terms (with variables in the coproduct of the underlying many-sorted set of the codomain and an ``initial segment'', of a fixed standard many-sorted set of variables, which depends on the many-sorted operation symbol) relative to the many-sorted signature of the codomain, and the other associating terms (with variables in the underlying many-sorted set of the codomain) relative to the many-sorted signature of the codomain, to variables from the underlying many-sorted set of the domain. In the particular case that, for each many-sorted operation symbol, it is fulfilled that, for every variable, the number of its oc\-cur\-rences in the term associated to the many-sorted operation symbol is at most one, the hyperderivor is called linear. Then we show that each hyperderivor determines, in a canonical way, a tree homomorphism, which is, in fact, a homomorphism from a free many-sorted algebra obtained from the domain of the hyperderivor to an\-oth\-er many-sorted algebra of the same many-sorted signature, itself derived from a many-sorted algebra obtained from the codomain of the hyperderivor. After this  we prove that the inverse image of a recognizable language under a tree homomorphism is recognizable and that the direct image of a recognizable language under a linear tree homomorphism, which is a tree homomorphism determined by a linear hyperderivor, is also recognizable.
Finally, in Subsection~6, entitled \emph{Derivors and recognizability}, after defining the many-sorted finitary variety of Hall algebras, we define, by means of the homomorphisms between Hall algebras, the category of many-sorted signatures and derivors. Then we construct a category whose objects are ordered pairs consisting of a many-sorted signature and a many-sorted algebra of such a signature, and morphisms between them ordered pairs consisting of a derivor between the underlying many-sorted signatures and a homomorphism from the underlying many-sorted algebra of the source to a derived many-sorted algebra of the sink. Following this, after showing that every derivor, together with some additional data, gives rise to a hyperderivor, we state that the inverse image under a convenient morphism, of the just mentioned category, labeled by a derivor, and the direct image under a morphism, labeled by a linear derivor, of a recognizable language is also recognizable.

Our underlying set theory is $\mathbf{ZFSk}$, Zermelo-Fraenkel-Skolem set theory (also known as $\mathbf{ZFC}$, i.e.,  Zermelo-Fraenkel set theory with the axiom of choice) plus the existence of a Grothendieck universe $\boldsymbol{\mathcal{U}}$, fixed once and for all (see~\cite{sM98}, pp.~21--24). We recall that the elements of $\boldsymbol{\mathcal{U}}$ are called $\boldsymbol{\mathcal{U}}$-small sets and the subsets of $\boldsymbol{\mathcal{U}}$ are called $\boldsymbol{\mathcal{U}}$-large sets or classes. Moreover, from now on $\mathbf{Set}$ stands for the category of sets, i.e., the category whose set of objects is $\boldsymbol{\mathcal{U}}$ and whose set of morphisms is the set of all mappings between $\boldsymbol{\mathcal{U}}$-small sets.

In all that follows we use standard concepts and constructions from category theory, see e.g., \cite{jwG66}, \cite{hs73} and \cite{sM98}, universal algebra, see e.g., \cite{gb15}, \cite{bs81}, \cite{gG08}, and \cite{w92}, and set theory, see e.g., \cite{nB70}. Nevertheless, regarding set theory, we have adopted the following conventions. An \emph{ordinal} $\alpha$ is a transitive set that is well-ordered by $\in$, thus $\alpha = \{\,\beta\mid \beta\in \alpha\,\}$. The first transfinite ordinal $\omega_{0}$ will be denoted by  $\mathbb{N}$, which is the set of all \emph{natural numbers}, and, from what we have just said about the ordinals, for every $n\in \mathbb{N}$, $n = \{0, \ldots,n-1\}$. If $\Phi$ and $\Psi$ are (binary) relations in a set $A$, then we will say that $\Psi$ is a \emph{refinement} of $\Phi$ if $\Psi\subseteq \Phi$.
We will denote by $\mathrm{Fnc}(A,B)$ the set of all functions from $A$ to $B$. We recall that a function from $A$ to $B$ is a subset $F$ of $A\times B$ satisfying the functional condition, i.e., such that, for every $x\in A$, there exists a unique $y\in B$ such that $(x,y)\in F$. A function $F$ from $A$ to $B$ is usually denoted by $(F_{x})_{x\in A}$. We will denote by $\mathrm{Hom}(A,B)$ (and, sometimes, also by $B^{A}$) the set of all mappings from $A$ to $B$. We recall that a mapping from $A$ to $B$ is an ordered triple $f = (A,F,B)$, denoted by $f\colon A\mor B$, in which $F$ is a function from $A$ to $B$. Therefore $\mathrm{Hom}(A,B) = \{A\}\times \mathrm{Fnc}(A,B)\times\{B\}$. We will denote by $\mathrm{Sub}(A)$ the set of all sets $X$ such that $X\subseteq A$ and if $X\in\mathrm{Sub}(A)$, then we will denote by $\complement_{A}X$ or $A-X$ the complement of $X$ in $A$. Moreover, if $f$ is a mapping from $A$ to $B$, then the mapping $f[\cdot]$ from $\mathrm{Sub}(A)$ to $\mathrm{Sub}(B)$, of $f$-\emph{direct image formation}, sends $X$ in $\mathrm{Sub}(A)$ to
$
f[X] = \{y\in B\mid \exists\,x\in X\,(y = f(x))\}
$
in $\mathrm{Sub}(B)$, and the mapping $f^{-1}[\cdot]$ from $\mathrm{Sub}(B)$ to $\mathrm{Sub}(A)$, of $f$-\emph{inverse image formation}, sends $Y$ in $\mathrm{Sub}(B)$ to
$
f^{-1}[Y] = \{x\in A\mid f(x)\in Y\}
$
in $\mathrm{Sub}(A)$. In the sequel, for a mapping $f$ from $A$ to $B$ and a subset $X$ of $A$, we will write $\mathrm{Ker}(f)$ for the kernel of $f$, $\mathrm{Im}(f)$ to mean $f[A]$, and the restriction of $f$ to $X$ will be denoted by $f\!\!\upharpoonright_{X}$.

%More specific assumptions, conditions, and conventions will be included in subsequent sections.

\section{Preliminaries}

In this section we collect the basic facts, mostly without proofs, about many-sorted sets, many-sorted algebras, and recognizability for arbitrary many-sorted algebras, that we will need.

\begin{assumption}
From now on $S$ stands for a set of sorts in $\boldsymbol{\mathcal{U}}$, fixed once and for all.
\end{assumption}

\begin{definition}
An $S$-\emph{sorted set} is a mapping $A = (A_{s})_{s\in S}$ from $S$ to $\boldsymbol{\mathcal{U}}$. If $A$ and $B$ are $S$-sorted sets, an $S$-\emph{sorted mapping from} $A$ \emph{to} $B$ is an $S$-indexed family $f = (f_{s})_{s\in S}$, where, for every $s$ in $S$, $f_{s}$ is a mapping from $A_{s}$ to  $B_{s}$. Thus, an $S$-sorted mapping from $A$ to $B$ is an element of $\prod_{s\in S}\mathrm{Hom}(A_{s}, B_{s})$. We will denote by $\mathrm{Hom}(A,B)$ the set of all $S$-sorted mappings from $A$ to $B$. $S$-sorted sets and $S$-sorted mappings form a category which we will denote henceforth by $\mathbf{Set}^{S}$. %From now on, $\mathbf{Set}^{S}$ stands for the category of $S$-sorted sets and $S$-sorted mappings.
\end{definition}

\begin{definition}
Let $I$ be a set in $\boldsymbol{\mathcal{U}}$ and $(A^{i})_{i\in I}$ an $I$-indexed family of $S$-sorted sets. Then the \emph{product} of $(A^{i})_{i\in I}$, denoted by $\prod_{i\in I}A^{i}$, is the $S$-sorted set defined, for every $s\in S$, as $\left(\prod\nolimits_{i\in I}A^{i}\right)_{s} = \prod\nolimits_{i\in I}A^{i}_{s}$, where
$$
\textstyle
\prod_{i\in I}A^{i}_{s} = \left\{(a_{i})_{i\in I}\in\mathrm{Fnc}(I,\bigcup_{i\in I}A^{i}_{s})\mid \forall\,i\in I\,(a_{i}\in A^{i}_{s})\right\}.
$$
For every $i\in I$, the \emph{i-th canonical projection}, $\mathrm{pr}^{i} = (\mathrm{pr}^{i}_{s})_{s\in S}$, is the $S$-sorted mapping from  $\prod_{i\in I}A^{i}$ to $A^{i}$ that, for every $s\in S$, sends $(a_{i})_{i\in I}$ in $\prod_{i\in I}A^{i}_{s}$ to $a_{i}$ in $A^{i}_{s}$. The ordered pair $(\prod_{i\in I}A^{i},(\mathrm{pr}^{i})_{i\in I})$ has the following universal property: For every $S$-sorted set $B$ and every $I$-indexed family of $S$-sorted mappings $(f^{i})_{i\in I}$, where, for every $i\in I$, $f^{i}$ is an $S$-sorted mapping from $B$ to $A^{i}$, there exists a unique $S$-sorted mapping $\left<f^{i}\right>_{i\in I}$ from $B$ to $\prod_{i\in I}A^{i}$ such that, for every $i\in I$, $\mathrm{pr}^{i}\circ \left<f^{i}\right>_{i\in I} = f^{i}$.

The \emph{coproduct} of $(A^{i})_{i\in I}$, denoted by $\coprod_{i\in I}A^{i}$, is the $S$-sorted set defined, for every $s\in S$, as $\left(\coprod\nolimits_{i\in I}A^{i}\right)_{s} = \coprod\nolimits_{i\in I}A^{i}_{s}$, where
$$
\textstyle
\coprod_{i\in I}A^{i}_{s} = \bigcup_{i\in I}(A^{i}_{s}\times\{i\}).
$$
For every $i\in I$, the \emph{i-th canonical injection}, $\mathrm{in}^{i}$, is the $S$-sorted mapping from $A^{i}$ to $\coprod_{i\in I}A^{i}$ that, for every $s\in S$, sends $a$ in $A^{i}_{s}$ to $(a,i)$ in $\coprod_{i\in I}A^{i}_{s}$. The ordered pair $(\coprod_{i\in I}A^{i},(\mathrm{in}^{i})_{i\in I})$ has the following universal property: For every $S$-sorted  set $B$ and every $I$-indexed family of $S$-sorted mappings $(f^{i})_{i\in I}$, where, for every $i\in I$, $f^{i}$ is an $S$-sorted mapping from $A^{i}$ to $B$, there exists a unique $S$-sorted mapping $[f^{i}]_{i\in I}$ from $\coprod_{i\in I}A^{i}$ to $B$ such that, for every $i\in I$, $ [f^{i}]_{i\in I}\circ\mathrm{in}^{i} = f^{i}$.

The remaining set-theoretic operations on $S$-sorted sets: $\times$ (binary product), $\amalg$ (binary coproduct), $\bigcup$ (union), $\cup$ (binary union), $\bigcap$ (intersection), $\cap$ (binary intersection), $-$ (difference), and $\complement_{A}$ (complement of an $S$-sorted set in a fixed $S$-sorted $A$), are defined in a similar way, i.e., componentwise.
\end{definition}

\begin{definition}
%An $S$-sorted set $A$ is \emph{subfinal} if $\mathrm{card}(A_{s})\leq 1$, for every $s\in S$.
We will denote by $1^{S}$ the (standard) final $S$-sorted set of $\mathbf{Set}^{S}$, which is $1^{S} = (1)_{s\in S}$, and by $\varnothing^{S}$ the initial $S$-sorted set, which is $\varnothing^{S} = (\varnothing)_{s\in S}$. We shall abbreviate $1^{S}$ to $1$ and $\varnothing^{S}$ to $\varnothing$ when this is unlikely to cause confusion.
\end{definition}

\begin{definition}
If $A$ and $X$ are $S$-sorted sets, then we will say that $X$ is a \emph{subset} of $A$, denoted by $X\subseteq A$, if, for every $s\in S$, $X_{s}\subseteq A_{s}$. We will denote by $\mathrm{Sub}(A)$ the set of all $S$-sorted sets $X$ such that $X\subseteq A$.
\end{definition}

\begin{remark}
For every $S$-sorted set $A$, the ordered pairs $(B,f)$, where $B$ is an $S$-sorted set and $f$ a monomorphisms from $B$ to $A$, are classified by letting $(B,f) \equiv (C,g)$ if and only if there exists an isomorphism $h$ from $B$ to $C$ such that $f = g\circ h$, and the corresponding equivalence classes are called \emph{subobjects} of $A$. Then $\mathrm{Sub}(A)$ is isomorphic to the set of all subobjects $A$.
\end{remark}

\begin{definition}
Let $\delta$ be the mapping from $S\times \boldsymbol{\mathcal{U}}$ to $\boldsymbol{\mathcal{U}}^{S}$ that sends $(t,X)$ in $S\times \boldsymbol{\mathcal{U}}$ to the $S$-sorted set $\delta^{t,X} = (\delta^{t,X}_{s})_{s\in S}$ defined, for every $s\in S$, as follows: $\delta^{t,X}_{s} = X$, if $s = t$; $\delta^{t,X}_{s} = \varnothing$, otherwise. We will call the value of $\delta$ at $(t,X)$ the \emph{delta of Kronecker associated to} $(t,X)$. If $X = \{x\}$, then, for simplicity of notation, we will write $\delta^{t,x}$ instead of $\delta^{t,\{x\}}$. Moreover, for a sort $t$ in $S$, $\delta^{t,1}$, the delta of Kronecker associated to $(t,1)$, will be denoted by $\delta^{t}$ and called \emph{delta of Kronecker}.

%Given a sort $t\in S$ we call \emph{delta of Kronecker in} $t$, the $S$-sorted set $\delta^{t} = (\delta^{t}_{s})_{s\in S}$ defined, for every $s\in S$, as follows: $\delta^{t}_{s} = 1$, if $s = t$; $\delta^{t}_{s} = \varnothing$, otherwise.
%%\begin{equation*}
%%\delta^{t}_{s} =
%%\begin{cases}
%%1, &\text{if $s = t$;}\\
%%\varnothing, & \text{otherwise.}
%%\end{cases}
%%\end{equation*}
%Let $t$ be a sort in $S$ and $X$ a set, then we will denote by $\delta^{t,X}$ the $S$-sorted set defined, for every $s\in S$, as follows: $\delta^{t,X}_{s} = X$, if $s = t$; $\delta^{t,X}_{s} = \varnothing$, otherwise.  If $X = \{x\}$, then, for simplicity of notation, we write $\delta^{t,x}$ instead of $\delta^{t,\{x\}}$.
%%\begin{equation*}
%%\delta^{t,X}_{s} =
%%\begin{cases}
%%X, &\text{if $s = t$;}\\
%%\varnothing, & \text{otherwise.}
%%\end{cases}
%%\end{equation*}
%
%Let us note that $\delta^{t}$ is $\delta^{t,1}$, i.e., the deltas of Kronecker, are particular cases of the $S$-sorted sets $\delta^{t,X}$ (however, see the remark immediately below). Therefore we will use indistinctly  $\delta^{t}$ or $\delta^{t,1}$.
\end{definition}

\begin{remark}
For a sort $t\in S$ and a set $X$, the $S$-sorted set $\delta^{t,X}$ is isomorphic to the $S$-sorted set $\coprod_{x\in X}\delta^{t}$, i.e., to the coproduct of the family $(\delta^{t})_{x\in X}$.

For every sort $t\in S$ we have a functor $\delta^{t,\cdot}$ from $\mathbf{Set}$ to $\mathbf{Set}^{S}$. In fact, for every set $X$, $\delta^{t,\cdot}(X) = \delta^{t,X}$, and, for every mapping $f\colon X\mor Y$, $\delta^{t,\cdot}(f) = \delta^{t,f}$, where, for $s\in S$, $\delta^{t,f}_{s} = \mathrm{id}_{\varnothing}$, if $s\neq t$, and $\delta^{t,f}_{t} = f$. Moreover, for every $t\in S$, the object mapping of the functor $\delta^{t,\cdot}$ is injective and $\delta^{t,\cdot}$ is full and faithful. Hence, for every $t\in S$, $\delta^{t,\cdot}$ is a full embedding from $\mathbf{Set}$ to $\mathbf{Set}^{S}$.
%Note that the diagonal functor $\Delta_{S}\colon \mathbf{Set}\mor \mathbf{Set}^{S}$ is $\coprod_{t_{S}}\delta^{t,\cdot}$.

The final object $1^{S}$ does not generate ($\equiv$ separate) the category $\mathbf{Set}^{S}$. However, the set $\{\,\delta^{s}\mid s\in S\,\}$, of the deltas of Kronecker, is a generating ($\equiv$ separating) set for the category $\mathbf{Set}^{S}$. Therefore, every $S$-sorted set $A$ can be represented as a coproduct of copowers of deltas of Kronecker, i.e., $A$ is naturally isomorphic to $\coprod_{s\in S}\mathrm{card}(A_{s})\boldsymbol{\cdot}\delta^{s}$, where, for every $s\in S$, $\mathrm{card}(A_{s})\boldsymbol{\cdot}\delta^{s}$ is the copower of the family $(\delta^{s})_{\alpha\in \mathrm{card}(A_{s})}$, i.e., the coproduct of $(\delta^{s})_{\alpha\in \mathrm{card}(A_{s})}$.
To this we add the following facts: (1) $\{\,\delta^{s}\mid s\in S\,\}$ is the set of atoms of the Boolean algebra $\mathbf{Sub}(1^{S})$, of subobjects of $1^{S}$; (2) the Boolean algebras $\mathbf{Sub}(1^{S})$ and $\mathbf{Sub}(S)$ are isomorphic; (3) for every $s\in S$, $\delta^{s}$ is a projective object; and (4) for every $s\in S$, every $S$-sorted mapping from $\delta^{s}$ to another $S$-sorted set is a monomorphism.

In view of the above, it must  be concluded that the deltas of Kronecker are of crucial importance for many-sorted sets and associated fields.
\end{remark}

Before proceeding any further, let us point out that it is no longer unusual to find in the works devoted to investigate both many-sorted algebras and many-sorted algebraic systems the following. (1) That an $S$-sorted set $A$ is defined in such a way that $\mathrm{Hom}(1^{S},A)\neq \varnothing$, or, what is equivalent, requiring that, for every $s\in S$, $A_{s}\neq \varnothing$. This has as an immediate consequence that the corresponding category is not even finite cocomplete. Since cocompleteness (and completeness) are desirable properties for a category, we exclude such a convention in our work (the admission of $\varnothing^{S}$ is crucial in many applications). And (2) that an $S$-sorted set $A$ must be such that, for every $s$, $t\in S$, if $s\neq t$, then $A_{s}\cap A_{t} = \varnothing$. We also exclude such a requirement (the possibility of a common underlying set for the different sorts is very important in many applications). The above conventions are possibly based on the untrue widespread belief that many-sorted equational logic and many-sorted first-order logic are \emph{inessential} variations of equational logic and first-order logic, respectively. One can find a definitive refutation to the just mentioned belief in~\cite{gm85} and \cite{m76}, regarding many-sorted equational logic, and in~\cite{Hook85}, with respect to many-sorted first-order logic.

%On this matter it seems that still is in force \emph{The Law of Conservation of Ignorance}, attributed by Kline, in~\cite{kl80}, on p.~88, to Cantor: ``A false conclusion once arrived at and widely accepted is not easily dislodged and the less it is understood the more tenaciously it is held.''

We next define for an $S$-sorted mapping the associated mappings of direct and inverse image formation, its kernel, its image, as well as its restriction to a subset of its domain.

\begin{definition}
Let $f\colon A\mor B$ be an $S$-sorted mapping. Then the mapping $f[\cdot]\colon\mathrm{Sub}(A)\mor\mathrm{Sub}(B)$, of $f$-\emph{direct image formation}, sends $X \in \mathrm{Sub}(A)$ to $f[X] = (f_{s}[X_{s}])_{s\in S} \in \mathrm{Sub}(B)$, and the mapping $f^{-1}[\cdot]\colon\mathrm{Sub}(B)\mor\mathrm{Sub}(A)$, of $f$-\emph{inverse image formation}, sends $Y \in \mathrm{Sub}(B)$ to $f^{-1}[Y] = (f_{s}^{-1}[Y_{s}] )_{s\in S} \in \mathrm{Sub}(A)$. The \emph{kernel} of $f$, denoted by $\mathrm{Ker}(f)$, is $(\mathrm{Ker}(f_{s}))_{s\in S}$ and the \emph{image} of $f$, denoted by $\mathrm{Im}(f)$, is $f[A]$. Moreover, if $X\subseteq A$, then the \emph{restriction of} $f$ \emph{to} $X$, denoted by $f\!\!\upharpoonright_{X}$, is $f\circ \mathrm{in}_{X,A}$, where $\mathrm{in}_{X,A} = (\mathrm{in}_{X_{s},A_{s}})_{s\in S}$ is the canonical embedding of $X$ into $A$.
%\begin{enumerate}
%\item The mapping $f[\cdot]$ of $f$-\emph{direct image formation} is the mapping defined as follows:
%$$
%f[\cdot] \nfunction
%{\mathrm{Sub}(A)}
%{\mathrm{Sub}(B)}
%{X}
%{ f[X] = (f_{s}[X_{s}])_{s\in S}  }
%$$

%\item The mapping $f^{-1}[\cdot]$ of $f$-\emph{inverse image formation} is the mapping defined as follows:
%$$
%f^{-1}[\cdot] \nfunction
%{\mathrm{Sub}(B)}
%{\mathrm{Sub}(A)}
%{Y}
%{ f^{-1}[Y] = (f_{s}^{-1}[Y_{s}] )_{s\in S} }
%$$
%\end{enumerate}
\end{definition}

Before stating the following proposition, we recall that, for an $S$-sorted set $A$, the $S$-sorted set $(\mathrm{Sub}(A_{s}))_{s\in S}$, usually denoted by $A^{\wp}$, is the power object of $A$ in the topos $\mathbf{Set}^{S}$. Therefore, one should take great care not to confuse $\mathrm{Sub}(A)$, which is a \emph{set}, i.e., an object of $\mathbf{Set}$---naturally isomorphic to the set $\prod_{s\in S}2^{A_{s}}$---, and $A^{\wp}$, which is an $S$-\emph{sorted set}, i.e., an object of $\mathbf{Set}^{S}$---naturally isomorphic to the $S$-sorted set $(2^{A_{s}})_{s\in S}$. On the other hand, it is clear that there exists a natural isomorphism between $\mathrm{Sub}(A)$ and $\prod A^{\wp}$.

%We next prove that every $S$-sorted mapping $f$ from an $S$-sorted set $X$ to the $S$-sorted set $(\mathrm{Sub}(Y_{s}))_{s\in S}$ canonically associated to an $S$-sorted set $Y$ has a distinguished extension up to an $S$-sorted  mapping $f^{\mathfrak{p}}$ from $(\mathrm{Sub}(X_{s}))_{s\in S}$ to $(\mathrm{Sub}(Y_{s}))_{s\in S}$ and that the canonical embedding of the category $\mathbf{Set}^{S}_{\wp,\mathrm{ca}}$, with objects the $S$-sorted sets $(\mathrm{Sub}(X_{s}))_{s\in S}$, for $X\in\boldsymbol{\mathcal{U}}^{S}$, and morphisms from $(\mathrm{Sub}(X_{s}))_{s\in S}$ to $(\mathrm{Sub}(Y_{s}))_{s\in S}$ the completely additive mappings (see below), into the category $\mathbf{Set}^{S}$ has a left adjoint.

\begin{proposition}
Let $X$ and $Y$ be $S$-sorted sets and $f$ an $S$-sorted mapping from $X$ to $Y^{\wp}$. Then there exists a unique $S$-sorted mapping $f^{\mathfrak{p}}$ from $(\mathrm{Sub}(X_{s}))_{s\in S}$ to $Y^{\wp}$ such that $f^{\mathfrak{p}}$ is completely additive, i.e., for every $s\in S$ and every $\mathcal{L}\subseteq \mathrm{Sub}(X_{s})$, $f^{\mathfrak{p}}_{s}(\bigcup\mathcal{L}) = \bigcup_{L\in \mathcal{L}}f^{\mathfrak{p}}_{s}(L)$, and $f^{\mathfrak{p}}\circ \{\cdot\}_{X} = f$, where $\{\cdot\}_{X}$ is the $S$-sorted mapping from $X$ to $X^{\wp}$ that, for every $s\in S$, sends $x\in X_{s}$ to $\{x\}\in\mathrm{Sub}(X_{s})$.
\end{proposition}

\begin{proof}
The $S$-sorted mapping $f^{\mathfrak{p}}$ from $X^{\wp}$ to $Y^{\wp}$ that, for every $s\in S$, assigns to $L\subseteq X_{s}$ the set $\bigcup_{x\in L}f_{s}(x)\subseteq Y_{s}$ is completely additive and $f^{\mathfrak{p}}\circ \{\cdot\}_{X} = f$. Moreover, $f^{\mathfrak{p}}$ is clearly the unique $S$-sorted mapping from $X^{\wp}$ to $Y^{\wp}$ satisfying the aforementioned  conditions.
\end{proof}

\begin{corollary}
Let $\mathbf{Set}^{S}_{\wp,\mathrm{ca}}$ be the category whose objects are the $S$-sorted sets $X^{\wp}$, where $X\in\boldsymbol{\mathcal{U}}^{S}$, and in which the set of morphisms from $X^{\wp}$ to $Y^{\wp}$ is the set of the completely additive $S$-sorted mappings from $X^{\wp}$ to $Y^{\wp}$. Then from $\mathbf{Set}^{S}_{\wp,\mathrm{ca}}$ to $\mathbf{Set}^{S}$ we have a canonical inclusion, denoted by $\mathrm{In}_{\mathbf{Set}^{S}_{\wp,\mathrm{ca}},\mathbf{Set}^{S}}$, and, for every $S$-sorted set $X$, the ordered pair $(X^{\wp},\{\cdot\}_{X})$ is a universal morphism from $X$ to $\mathrm{In}_{\mathbf{Set}^{S}_{\wp,\mathrm{ca}},\mathbf{Set}^{S}}$.
\end{corollary}

\begin{definition}
We will denote by $(\cdot)^{\wp}$ the functor from $\mathbf{Set}^{S}$ to $\mathbf{Set}^{S}_{\wp,\mathrm{ca}}$ that sends an $S$-sorted set $X$ to $X^{\wp}$ and an $S$-sorted  mapping $f$ from $X$ to $Y$ to the $S$-sorted mapping $f^{\wp} = (\{\cdot\}_{Y}\circ f)^{\mathfrak{p}}$ from $X^{\wp}$ to $Y^{\wp}$ (the $S$-sorted mappings $\{\cdot\}_{X}$, for $X$ in $\boldsymbol{\mathcal{U}}^{S}$, are the components of the unit of the adjunction: $\mathrm{Hom}_{\mathbf{Set}^{S}_{\wp,\mathrm{ca}}}(X^{\wp},Y^{\wp})\cong
\mathrm{Hom}_{\mathbf{Set}^{S}}(X,Y^{\wp})$). Thus, $(\cdot)^{\wp}$ is a left adjoint of $\mathrm{In}_{\mathbf{Set}^{S}_{\wp,\mathrm{ca}},\mathbf{Set}^{S}}$.
\end{definition}

\begin{remark}
For an $S$-sorted mapping $f$ from $X$ to $Y$ the $S$-sorted mapping $f^{\wp}$ from $X^{\wp}$ to $Y^{\wp}$ sends, for every $s\in S$, $L\subseteq X_{s}$ to $f_{s}[L]\subseteq Y_{s}$, i.e., $f^{\wp} = (f_{s}[\cdot])_{s\in S}$. One should be careful not to confuse $f[\cdot]$, which is a \emph{mapping} from the \emph{set} $\mathrm{Sub}(X)$ to the \emph{set} $\mathrm{Sub}(Y)$, and $f^{\wp}$, which is an $S$-\emph{sorted mapping} from the $S$-\emph{sorted set} $X^{\wp} = (\mathrm{Sub}(X_{s}))_{s\in S}$ to the $S$-\emph{sorted set} $Y^{\wp} = (\mathrm{Sub}(Y_{s}))_{s\in S}$. On the other hand, it is evident that $f[\cdot]$ and $\prod f^{\wp}$ are \emph{essentially} the same mapping.
\end{remark}

\begin{definition}
Let $A$ be an $S$-sorted set. Then the \emph{cardinal} of $A$, denoted by $\mathrm{card}(A)$, is $\mathrm{card}(\coprod A)$, i.e., the cardinal of the set $\coprod A = \bigcup_{s\in S}(A_{s}\times \{s\})$. An $S$-sorted set $A$ is \emph{finite} if $\mathrm{card}(A)<\aleph_{0}$. We will say that an $S$-sorted set $X$ is a \emph{finite} subset of $A$ if $X$ is finite and $X\subseteq A$. We will denote by $\mathrm{Sub}_{\mathrm{f}}(A)$ the set of all $S$-sorted sets $X$ in $\mathrm{Sub}(A)$ which are finite.
\end{definition}

\begin{remark}
For an object $A$ in the topos $\mathbf{Set}^{S}$ the following assertions are equivalent: (1) $A$ is finite, (2) $A$ is a finitary object of $\mathbf{Set}^{S}$, and (3) $A$ is a strongly finitary object of $\mathbf{Set}^{S}$ (for the notions of finitary and strongly finitary object of a category see~\cite{hs73}, Exercise 22E, on p.~155).
%An $S$-sorted set $A$ is finite if and only if the covariant hom-functor $\mathrm{H}(A,\cdot)\colon \mathbf{Set}^{S}\mor \mathbf{Set}$ is finitary, i.e., if and only if for every $\boldsymbol{\mathcal{U}}$-small upward-directed preordered set $\mathbf{I}$ and every functor $D$ from $\mathbf{C}(\mathbf{\mathbf{I}})$, the category canonically associated to $\mathbf{I}$, to $\mathbf{Set}^{S}$, if $((f^{i})_{i\in \mathrm{Ob}(\mathbf{C}(\mathbf{\mathbf{I}}))},L)$ is an inductive limit of $D$, then $((\mathrm{H}(A,f^{i}))_{i\in \mathrm{Ob}(\mathbf{C}(\mathbf{\mathbf{I}}))},\mathrm{H}(A,L))$ is an epi-sink. Moreover, the $S$-sorted set $A$ is finite if and only if the functor $\mathrm{H}(A,\cdot)$ is strongly finitary, i.e., if and only if for every $\boldsymbol{\mathcal{U}}$-small upward-directed preordered set $\mathbf{I}$ and every functor $D$ from $\mathbf{C}(\mathbf{\mathbf{I}})$ to $\mathbf{Set}^{S}$, if $((f^{i})_{i\in \mathrm{Ob}(\mathbf{C}(\mathbf{\mathbf{I}}))},L)$ is an inductive limit of $D$, then $((\mathrm{H}(A,f^{i}))_{i\in \mathrm{Ob}(\mathbf{C}(\mathbf{\mathbf{I}}))},\mathrm{H}(A,L))$ is an inductive limit.

In $\mathbf{Set}^{S}$ there is another notion of finiteness: An $S$-sorted set $A$ is called $S$-\emph{finite} or \emph{locally finite}, abbreviated as $\mathrm{lf}$, if and only if, for every $s\in S$, $A_{s}$ is finite.  We will denote by $\mathrm{Sub}_{\mathrm{lf}}(A)$ the set of all $S$-sorted sets $X$ in $\mathrm{Sub}(A)$ which are locally finite. Although, unless $S$ is finite, this notion of finiteness is not categorial in nature, however, it plays a relevant role in the field of many-sorted algebra and in computer science (see below, after the definition of many-sorted algebra and when we deal with the congruences of locally finite index, respectively).
\end{remark}

\begin{definition}
Let $A$ be an $S$-sorted set. Then the \emph{support of} $A$, denoted by $\mathrm{supp}_{S}(A)$, is the set $\{\,s\in S\mid A_{s}\neq \varnothing\,\}$.
%(i.e., the subset of $S$ on which $A$ does something nontrivial).
\end{definition}

\begin{remark}
An $S$-sorted set $A$ is finite if and only if $\mathrm{supp}_{S}(A)$ is finite and, for every $s\in \mathrm{supp}_{S}(A)$, $A_{s}$ is finite.
\end{remark}

We next recall, after fixing some notation with regard to an equivalence relation $\Phi$ on an $S$-sorted set $A$, the universal property of $(A/\Phi,\mathrm{pr}^{\Phi})$, where $A/\Phi$ is the quotient $S$-sorted set of $A$ by $\Phi$ and $\mathrm{pr}^{\Phi}$ the canonical projection from $A$ to $A/\Phi$; the notion of transversal of $A/\Phi$ in $A$; the notion of $\Phi$-saturation of a subset $X$ of $A$; and those properties of this last notion which will be used afterwards.

\begin{definition}
An $S$-\emph{sorted equivalence relation on} (or, to abbreviate, an $S$-\emph{sorted equivalence on}) an $S$-sorted set $A$ is an $S$-sorted relation $\Phi$ on $A$, i.e., a subset $\Phi = (\Phi_{s})_{s\in S}$ of the cartesian product $A\times A = (A_{s}\times A_{s})_{s\in S}$ such that, for every $s\in S$, $\Phi_{s}$ is an equivalence relation on $A_{s}$. We will denote by $\mathrm{Eqv}(A)$ the set of all $S$-sorted equivalences on $A$ (which is an algebraic closure system on $A\times A$), by $\mathbf{Eqv}(A)$ the algebraic lattice  $(\mathrm{Eqv}(A),\subseteq)$, by $\nabla_{A}$ the greatest element of $\mathbf{Eqv}(A)$, and by $\Delta_{A}$ the least element of $\mathbf{Eqv}(A)$.

For an $S$-sorted equivalence relation $\Phi$ on $A$, the $S$-\emph{sorted quotient set of} $A$ \emph{by} $\Phi$, denoted by  $A/\Phi$, is $(A_{s}/\Phi_{s})_{s\in S} = (\{[x]_{\Phi_{s}}\mid x\in A_{s}\})_{s\in S} (\subseteq A^{\wp})$, where, for every $s\in S$ and every $x\in A_{s}$, $[x]_{\Phi_{s}}$, the \emph{equivalence class of} $x$ \emph{with respect to} $\Phi_{s}$ (or, the $\Phi$-\emph{equivalence class of} $x$) is $\{y\in A_{s}\mid (x,y)\in \Phi_{s}\}$, and $\mathrm{pr}_{\Phi}\colon A\mor A/\Phi$, the \emph{canonical projection from} $A$ \emph{to} $A/\Phi$, is the $S$-sorted mapping $(\mathrm{pr}_{\Phi_{s}})_{s\in S}$, where, for every $s\in S$, $\mathrm{pr}_{\Phi_{s}}$ is the canonical projection from $A_{s}$ to $A_{s}/\Phi_{s}$ (which sends $x$ in $A_{s}$ to $\mathrm{pr}_{\Phi_{s}}(x) = [x]_{\Phi_{s}}$, the $\Phi_{s}$-equivalence class of $x$, in $A_{s}/\Phi_{s}$).

The ordered pair $(A/\Phi,\mathrm{pr}_{\Phi})$ has the following universal property: $\mathrm{Ker}(\mathrm{pr}_{\Phi})$ is $\Phi$ and, for every $S$-sorted set $B$ and every $S$-sorted mapping $f$ from $A$ to $B$, if $\mathrm{Ker}(f)\supseteq \Phi$, then there exists a unique $S$-sorted mapping $h$ from $A/\Phi$ to $B$ such that $h\circ\mathrm{pr}_{\Phi} = f$. In particular, if $\Psi$ is an  $S$-sorted equivalence relation on $A$ such that $\Phi\subseteq \Psi$, then we will denote by $\mathrm{p}_{\Phi,\Psi}$ the unique $S$-sorted mapping from $A/\Phi$ to $A/\Psi$ such that $\mathrm{p}_{\Phi,\Psi}\circ \mathrm{pr}_{\Phi} = \mathrm{pr}_\Psi$.

%Let $\Phi$ and $\Psi$ be $S$-sorted equivalence on $A$ such that $\Phi\subseteq\Psi$. Then the \emph{quotient of} $\Psi$ \emph{by} $\Phi$, denoted by $\Psi/\Phi$, is the $S$-sorted equivalence $(\Psi_{s}/\Phi_{s})_{s\in S}$ on $A/\Phi$ defined, for every $s\in S$, as follows:
%$$
%\Psi_{s}/\Phi_{s}=\{([a]_{\Phi_{s}},
%[b]_{\Phi_{s}})\in (A_{s}/\Phi_{s})\mid (a,b)\in\Psi_{s}\}.
%$$
\end{definition}

%\begin{remark}
%Let $A$ be an $S$-sorted set and $\Phi\in\mathrm{Eqv}(A)$. Then, by Proposition~\ref{propssupport}, $\mathrm{supp}_{S}(A) = \mathrm{supp}_{S}(A/\Phi)$.
%\end{remark}

\begin{remark}
Let $\mathbf{ClfdSet}^{S}$ be the category whose objects are the \emph{classified $S$-sorted sets}, i.e, the ordered pairs $(A,\Phi)$ where $A$ is an $S$-sorted set and $\Phi$ an $S$-sorted equivalence relation on $A$, and in which the set of morphisms from $(A,\Phi)$ to $(B,\Psi)$ is the set of all $S$-sorted mappings $f$ from $A$ to $B$ such that, for every $s\in S$ and every $(x,y)\in A^{2}_{s}$, if $(x,y)\in \Phi_{s}$, then $(f_{s}(x),f_{s}(y))\in \Psi_{s}$. Let $G$ be the functor from $\mathbf{Set}^{S}$ to $\mathbf{ClfdSet}^{S}$ whose object mapping sends $A$ to $(A,\Delta_{A})$ and whose morphism mapping sends $f\colon A\mor B$ to $f\colon (A,\Delta_{A})\mor (B,\Delta_{B})$. Then, for every classified $S$-sorted set $(A,\Phi)$, there exists a universal mapping from $(A,\Phi)$ to $G$, which is precisely the ordered pair $(A/\Phi,\mathrm{pr}_{\Phi})$ with $\mathrm{pr}_{\Phi}\colon (A,\Phi)\mor (A/\Phi,\Delta_{A/\Phi})$.
\end{remark}

\begin{definition}
Let $A$ be an $S$-sorted set and $\Phi\in \mathrm{Eqv}(A)$. Then a \emph{transversal of} $A/\Phi$ \emph{in} $A$ is a subset $X$ of $A$ such that, for every $s\in S$ and every $a\in A_{s}$, $\mathrm{card}(X_{s}\cap [a]_{\Phi_{s}}) = 1$.
\end{definition}

\begin{remark}
For an $S$-sorted equivalence relation $\Phi$ on $A$, the set of all transversals of $A/\Phi$ in $A$ is isomorphic to the set of all cross-sections of $\mathrm{pr}_{\Phi}$, where an $S$-sorted mappings $f$ from $A/\Phi$ to $A$ is a cross-section of $\mathrm{pr}_{\Phi}$ if $\mathrm{pr}_{\Phi}\circ f = \mathrm{id}_{A/\Phi}$. Moreover, if $\Psi$ is another equivalence relation on $A$, $\Psi$ is a refinement of $\Phi$, i.e., $\Psi\subseteq \Phi$, and $X^{\Phi}$ is a transversal of $A/\Phi$ in $A$, then, for every $s\in S$ and every $a\in A_{s}$, there exists a unique $x\in X^{\Phi}_{s}$ such that $[a]_{\Psi_{s}}\subseteq [x]_{\Phi_{s}}$.
\end{remark}

\begin{definition}
Let $A$ be an $S$-sorted set, $X$ a subset of $A$, and $\Phi\in\mathrm{Eqv}(A)$. Then the $\Phi$-\emph{saturation of} $X$ (or, the \emph{saturation of} $X$ \emph{with respect to} $\Phi$), denoted by $[X]^{\Phi}$, is the $S$-sorted set defined, for every $s\in S$, as follows:
$$
\textstyle[X]^{\Phi}_{s} = \{a\in A_{s}\mid
X_{s}\cap [a]_{\Phi_{s}}\neq\varnothing\}= \bigcup_{x\in X_{s}}[x]_{\Phi_{s}} = [X_{s}]^{\Phi_{s}}.
$$
Let $X$ be a subset of $A$ and $\Phi\in\mathrm{Eqv}(A)$. Then we will say that $X$ is $\Phi$-\emph{saturated} if and only if $X = [X]^{\Phi}$. We will denote by $\Phi\text{-}\mathrm{Sat}(A)$ the subset of $\mathrm{Sub}(A)$ defined as $\Phi\text{-}\mathrm{Sat}(A) = \{X\in \mathrm{Sub}(A)\mid X = [X]^{\Phi}\}$.
\end{definition}

\begin{remark}
Let $A$ be an $S$-sorted set and $\Phi\in\mathrm{Eqv}(A)$. Then, for a subset $X$ of
$A$, we have that the $\Phi$-saturation of $X$ is $(\mathrm{pr}_{\Phi})^{-1}[\mathrm{pr}_{\Phi}[X]]$. Therefore, $X$ is $\Phi$-saturated if and only if $X \supseteq [X]^{\Phi}$. Besides, $X$ is $\Phi$-saturated if and only if there exists a $\mathcal{Y}\subseteq A/\Phi$ such that $X = (\mathrm{pr}^{\Phi})^{-1}[\mathcal{Y}]$.
\end{remark}

\begin{proposition}\label{PropIncSat and PropsIncSat}
Let $A$ be an $S$-sorted set and $\Phi$, $\Psi\in\mathrm{Eqv}(A)$. Then
$$
\Phi\subseteq \Psi\, \text{if and only if }\, \forall X\subseteq A\;([[X]^{\Psi}]^{\Phi}=[X]^{\Psi}).
$$
Moreover, for a sort $s\in S$, we have that
$$
\Phi_{s}\subseteq \Psi_{s}\, \text{if and only if }\, \forall X\subseteq A_{s}\;([[X]^{\Psi_{s}}]^{\Phi_{s}}=[X]^{\Psi_{s}}).
$$
\end{proposition}

\begin{proof}
We restrict ourselves to proving the first assertion. Let us assume that $\Phi\subseteq \Psi$ and let $X$ be a subset of $A$. In order to prove that $[[X]^{\Psi}]^{\Phi}=[X]^{\Psi}$ it suffices to verify that $[[X]^{\Psi}]^{\Phi}\subseteq [X]^{\Psi}$. Let $s$ be an element of $S$. Then, by definition, $a\in [[X]^{\Psi}]^{\Phi}_{s}$ if and only if there exists some $b\in [X]^{\Psi}_{s}$ such that $a\in [b]_{\Phi_{s}}$. Since $\Phi\subseteq \Psi$, we have that $a\in [b]_{\Psi_{s}}$, therefore $a\in [X]^{\Psi}_{s}$.

To prove the converse, let us assume that $\Phi\not\subseteq \Psi$. Then there exists some sort $s\in S$ and elements $a$, $b$ in $A_{s}$ such that $(a,b)\in\Phi_{s}$ and $(a,b)\not\in \Psi_{s}$. Hence $b$ does not belong to $[\delta^{s,[a]_{\Psi_{s}}}]^{\Psi}_s$, whereas it does belong to $[[\delta^{s,[a]_{\Psi_{s}}}]^{\Psi}]^{\Phi}_{s}$. It follows that $[\delta^{s,[a]_{\Psi_{s}}}]^{\Psi}\neq [[\delta^{s,[a]_{\Psi_{s}}}]^{\Psi}]^{\Phi}$.
\end{proof}

\begin{corollary}\label{IncSat and sIncSat}
Let $A$ be an $S$-sorted set and $\Phi$, $\Psi\in\mathrm{Eqv}(A)$. If $\Phi\subseteq \Psi$, then  $\Psi\text{-}\mathrm{Sat}(A)\subseteq\Phi\text{-}\mathrm{Sat}(A)$. Moreover, for $s\in S$ and $L\subseteq A_{s}$, if $\Phi_{s}\subseteq \Psi_{s}$ and $L = [L]^{\Psi_{s}}$, then  $L = [L]^{\Phi_{s}}$.
\end{corollary}

\begin{remark}
If, for an $S$-sorted set $A$, we denote by $(\cdot)\text{-}\mathrm{Sat}(A)$ the mapping from $\mathrm{Eqv}(A)$ to $\mathrm{Sub}(\mathrm{Sub}(A))$ which sends $\Phi$ to $\Phi\text{-}\mathrm{Sat}(A)$, then the above corollary means that $(\cdot)\text{-}\mathrm{Sat}(A)$ is an antitone ($\equiv$ order-reversing) mapping from the ordered set $(\mathrm{Eqv}(A),\subseteq)$ to the ordered set $(\mathrm{Sub}(\mathrm{Sub}(A)),\subseteq)$.
\end{remark}

\begin{proposition}\label{NablaSat}
Let $A$ be an $S$-sorted set and $X\subseteq A$. Then $X\in \nabla_{A}\text{-}\mathrm{Sat}(A)$ if and only if, for every $s\in S$, if $s\in \mathrm{supp}_{S}(X)$, then $X_{s} = A_{s}$.
\end{proposition}

\begin{proof}
Let us suppose that there exists a $t\in S$ such that $X_{t}\neq \varnothing$ and $X_{t}\neq A_{t}$. Then, since $[X]_{t}^{\nabla_{A}} = \bigcup_{x\in X_{t}}[x]_{\nabla_{A_{t}}}$ and $X_{t}\neq \varnothing$, we have that, for some $y\in X_{t}$, $[y]_{\nabla_{A_{t}}} = A_{t}$. But $X_{t}\subset A_{t}$. Hence $[X]_{t}^{\nabla_{A}}\neq X_{t}$. Therefore $X\not\in \nabla_{A}\text{-}\mathrm{Sat}(A)$.

The converse implication is straightforward.
\end{proof}

\begin{remark}
Let $A$ be an $S$-sorted set. Then, from the above proposition, it follows that $\varnothing^{S}$, $A\in \nabla_{A}\text{-}\mathrm{Sat}(A)$. Moreover, for every subset $T$ of $S$, we have that $\bigcup_{t\in T}\delta^{t,A_{t}}\in \nabla_{A}\text{-}\mathrm{Sat}(A)$.
\end{remark}

\begin{proposition}
Let $A$ be an $S$-sorted set, $X\subseteq A$, and $\Phi$, $\Psi\in\mathrm{Eqv}(A)$. Then
$[X]^{\Phi\cap\Psi}\subseteq[X]^{\Phi}\cap[X]^{\Psi}$.
\end{proposition}

\begin{proof}
Let $s$ be a sort in $S$ and $b\in[X]^{\Phi\cap \Psi}_{s}$. Then, by definition, there exists an $a\in X_{s}$ such that $(a,b)\in (\Phi\cap\Psi)_{s} = \Phi_{s}\cap\Psi_{s}$. Hence, $(a,b)\in \Phi_{s}$ and $(a,b)\in \Psi_{s}$. Therefore $b\in [X]^{\Phi}_{s}$ and $b\in[X]^{\Psi}_{s}$. Consequently, $b\in ([X]^{\Phi}\cap [X]^{\Psi})_{s}$. Thus $[X]^{\Phi\cap\Psi}\subseteq[X]^{\Phi}\cap[X]^{\Psi}$.
\end{proof}

\begin{corollary}
Let $A$ be an $S$-sorted set and $\Phi$, $\Psi\in\mathrm{Eqv}(A)$. Then we have that  $\Phi\text{-}\mathrm{Sat}(A)\cap\Psi\text{-}\mathrm{Sat}(A)\subseteq(\Phi\cap\Psi)\text{-}\mathrm{Sat}(A)$.
\end{corollary}

We next state that the set $\Phi\text{-}\mathrm{Sat}(A)$ is the set of all fixed points of a suitable operator on $A$, i.e., of an endomapping of $\mathrm{Sub}(A)$.

\begin{proposition}\label{SatOperator}
Let $A$ be an $S$-sorted set and $\Phi\in\mathrm{Eqv}(A)$. Then the mapping $[\cdot]^{\Phi}$ from $\mathrm{Sub}(A)$ to $\mathrm{Sub}(A)$ that sends $X$ in $\mathrm{Sub}(A)$ to $[\cdot]^{\Phi}(X) = [X]^{\Phi}$ in $\mathrm{Sub}(A)$
%defined as follows:
%$$
%[\cdot]^{\Phi} \nfunction
%{\mathrm{Sub}(A)}
%{\mathrm{Sub}(A)}
%{X}
%{ [\cdot]^{\Phi}(X) = [X]^{\Phi}  }
%$$
is a completely additive closure operator on $A$. Moreover, for every nonempty set $I$ in $\boldsymbol{\mathcal{U}}$ and every $I$-indexed family $(X^{i})_{i\in I}$ in $\mathrm{Sub}(A)$, $[\bigcap_{i\in I}X^{i}]^{\Phi} \subseteq \bigcap_{i\in I}[X^{i}]^{\Phi}$ (and, obviously, $[A]^{\Phi} = A$), and, for every $X\subseteq A$, if $X = [X]^{\Phi}$, then  $\complement_{A}X = [\complement_{A}X]^{\Phi}$. Besides, $[\cdot]^{\Phi}$ is uniform, i.e., is such that, for every $X$, $Y\subseteq A$, if  $\mathrm{supp}_{S}(X) = \mathrm{supp}_{S}(Y)$, then $\mathrm{supp}_{S}([X]^{\Phi}) = \mathrm{supp}_{S}([Y]^{\Phi})$---hence, in particular, $[\cdot]^{\Phi}$ is a uniform algebraic closure operator on $A$. And $\Phi\text{-}\mathrm{Sat}(A) = \mathrm{Fix}([\cdot]^{\Phi})$, where $\mathrm{Fix}([\cdot]^{\Phi})$ is the set of all fixed point of the operator $[\cdot]^{\Phi}$.
\end{proposition}

%\begin{remark}
%Given a subset $X$ of $A$, if $X = [X]^{\Phi} = (\mathrm{pr}^{\Phi})^{-1}[\mathcal{Y}]$, for some $\mathcal{Y}\subseteq A/\Phi$, then $\complement_{A}X = (\mathrm{pr}^{\Phi})^{-1}[\complement_{A/\Phi}\mathcal{Y}]$. Hence $\complement_{A}X = [\complement_{A}X]^{\Phi}$.
%\end{remark}

\begin{proposition}\label{CABA Saturades} Let $A$ be an $S$-sorted set and $\Phi\in\mathrm{Eqv}(A)$. Then the ordered pair
$\Phi\text{-}\mathbf{Sat}(A) = (\Phi\text{-}\mathrm{Sat}(A),\subseteq)$ is a complete atomic Boolean algebra.
%(for brevity a CABA).
\end{proposition}

\begin{proof}
The proof is straightforward and we leave it to the reader. We only point out that the atoms of $\Phi\text{-}\mathbf{Sat}(A)$ are precisely the deltas of Kronecker $\delta^{t, [x]_{\Phi_t}}$, for some $t\in S$ and some $x\in A_{t}$, and that, obviously, every $\Phi$-saturated subset $X$ of $A$ is the join ($\equiv$ union) of all atoms smaller than $X$.
\end{proof}

%\subsection{On words}

We next define the concept of free monoid on a set and several notions associated to it that will be used afterwards to construct the free algebra on an $S$-sorted set and to define, in Section 3, diverse substitution operators.

\begin{definition}
Let $A$ be a set. The \emph{free monoid on} $A$, denoted by $\mathbf{A}^{\star}$, is $(A^{\star},\curlywedge,\lambda)$, where $A^{\star}$, the set of all \emph{words on} $A$, is $\bigcup_{n\in\mathbb{N}}\mathrm{Hom}(n,A)$, the set of all mappings $w\colon n\mor A$ from some $n\in \mathbb{N}$ to $A$, $\curlywedge$, the \emph{concatenation} of words on $A$, is the binary operation on $A^{\star}$ which sends a pair of words $(w,v)$ on $A$ to the mapping $w\curlywedge v$ from $\bb{w}+\bb{v}$ to $A$, where $\bb{w}$ and $\bb{v}$ are the lengths ($\equiv$ domains) of the mappings $w$ and $v$, respectively, defined as follows:
$$
w\bconcat v
\nfunction
{\bb{w}+\bb{v}}{S}
{i}{
\begin{cases}
  w_{i}, & \text{if $0\leq i < \bb{w}$;}\\
  v_{i-\bb{w}}, & \text{if $\bb{w}\leq i < \bb{w}+\bb{v}$.}
\end{cases}
   }
$$
and $\lambda$, the \emph{empty word on} $A$, is the unique mapping from $\varnothing$ to $A$. A word $w\in A^{\star}$ is usually denoted as a sequence $(a_{i})_{i\in\bb{w}}$, where, for $i\in\bb{w}$, $a_{i}$ is the letter in $A$ satisfying $w(i)=a_{i}$. We will denote by $\eta_{A}$ the mapping from $A$ to $A^{\star}$ that sends $a\in A$ to $(a)\in A^{\star}$, i.e., to the mapping $(a)\colon 1\mor A$ that sends $0$ to $a$. The ordered pair $(\mathbf{A}^{\star},\eta_{A})$ is a universal morphism from $A$ to the forgetful functor from the category $\mathbf{Mon}$, of monoids, to $\mathbf{Set}$.
\end{definition}

\begin{remark}
For a word $w\in A^{\star}$, $\bb{w}$ is the value at $w$ of the unique homomorphism $\bb{\cdot}$ from $\mathbf{A}^{\star}$ to $(\mathbb{N},+,0)$, the additive monoid of the natural numbers, such that $\bb{\cdot}\circ \eta_{A} = \kappa_{1}$, where $\kappa_{1}$ is the mapping from $A$ to $\mathbb{N}$ constantly $1$. Note that, for every $n\in \mathbb{N}$, $\bb{\cdot}$ sends $w\in \mathrm{Hom}(n,A)$ to $n$. Thus, for the family of mappings $(\kappa_{n})_{n\in \mathbb{N}}$, where, for every $n\in \mathbb{N}$, $\kappa_{n}$ is the mapping from $\mathrm{Hom}(n,A)$ to $\mathbb{N}$ constantly $n$, and by applying the universal property of the coproduct, we have that $\bb{\cdot} = [\kappa_{n}]_{n\in \mathbb{N}}$.
\end{remark}

\begin{definition}
Let $w$ and $w'$ be words in $A^{\star}$. We will say that $w'$ is a \emph{subword} of $w$ if there are words $u$ and $v$ such that $w = u\bconcat w'\bconcat v$. A word $w$ may have several subwords equal to $w'$. In that case, the equation $w = u\bconcat w'\bconcat v$ has several solutions $(u,v)$.
%The first subword of $w$ equal to $w'$ is called \emph{the first occurrence of }$w'$ \emph{in} $w$, the second subword of $w$ equal to $w'$ is \emph{the second occurrence of} $w'$ \emph{in} $w$ and so on.
If the pairs $(u_{i},v_{i})$ $(i\in n)$ are all solutions of $w = u\bconcat w'\bconcat v$ and if $\bb{u_{0}}< \bb{u_{1}}<\cdots<\bb{u_{n-1}}$, then $(u_{i},v_{i})$ determines the $i$-th occurrence of $w'$ in $w$. The solution $(u,v)$ in which either $u$ or $v$ is $\lambda$ is not excluded.
Let $w$ be a word in $A^{\star}$ and $a\in A$. We will say that \emph{$a$ occurs} in $w$ if there are words $u$, $v$ in $A^{\star}$ such that $w = u\bconcat(a)\bconcat v$. Note that $a$ occurs in $w$ if and only if there exists an $i\in \bb{w}$ such that $w(i) = a$. We will denote by $\bb{w}_{a}$ the natural number $\mathrm{card}(\{i\in \bb{w}\mid w(i) = a\}) = \mathrm{card}(w^{-1}[\{a\}])$, i.e., the number of occurrences of $a$ in $w$. Moreover, we let $(i_{\alpha})_{\alpha\in\bb{w}_{a}}$ stand for the  enumeration in ascending order of the occurrences of $a$ in $w$. Thus $(i_{\alpha})_{\alpha\in\bb{w}_{a}}$ is the order embedding of $(\bb{w}_{a},<)$ into $(\bb{w},<)$ defined  recursively as follows:
\begin{gather*}
i_{0} = \mathrm{min}\{i\in \bb{w}\mid w(i) = a\} \text{ and}, \\
\text{for } \alpha\in \bb{w}_{a}-\{\bb{w}_{a}-1\},\, i_{\alpha+1} = \mathrm{min}\{i\in \bb{w}-\{i_{0},\ldots,i_{\alpha}\}\mid w(i) = a\}.
\end{gather*}
If the pairs $(u_{i_{\alpha}},v_{i_{\alpha}})$ $(\alpha\in \bb{w}_{a})$ are all solutions of $w = u\bconcat (a)\bconcat v$ and if $\bb{u_{i_{0}}}< \bb{u_{i_{1}}}<\cdots<\bb{u_{i_{\bb{w}_{a}-1}}}$, then $(u_{i_{\alpha}},v_{i_{\alpha}})$ determines the $i_{\alpha}$-th occurrence of $(a)$ in $w$ and we will say that \emph{$a$ occurs at the $i_{\alpha}$-th place of $w$}.
\end{definition}

\begin{remark}
Let $a$ be an element of $A$. Then, for $w\in A^{\star}$, $\bb{w}_{a}$, the number of occurrences of $a$ in $w$, is the value at $w$ of the unique homomorphism $\bb{\cdot}_{a}$ from $\mathbf{A}^{\star}$ to $(\mathbb{N},+,0)$ such that $\bb{\cdot}_{a}\circ\eta_{A} = \delta_{a}$, where $\delta_{a}$ is the mapping from $A$ to $\mathbb{N}$ that sends $a\in A$ to $1\in\mathbb{N}$ and $b\in A-\{a\}$ to $0\in \mathbb{N}$. Therefore, since, for every $w\in A^{\star}$, the $A$-indexed family $(\bb{w}_{a})_{a\in A}$ in $\mathbb{N}$ is such that $\bb{w}_{a} = 0$ for all but a finite number of elements $a$ in $A$, i.e., is such that $\mathrm{card}(\{a\in A\mid \bb{w}_{a}\neq 0\})<\aleph_{0}$, we have that $\sum_{a\in A}\bb{w}_{a}\in \mathbb{N}$ and, obviously, for every $w\in A^{\star}$, $\bb{w} = \sum_{a\in A}\bb{w}_{a}$, i.e., $\bb{\cdot} = \sum_{a\in A}\bb{\cdot}_{a}$.
\end{remark}

\begin{definition} Let $w$ be a word in $A^{\star}$. Then we will denote by $\left(\!\begin{smallmatrix}a\\ \cdot\end{smallmatrix}\!\right)_{a\in A}(w)$ the mapping from $\prod_{a\in A}(A^{\star})^{\bb{w}_{a}}$ to $A^{\star}$ that assigns to $((q^{a}_{\alpha})_{\alpha\in \bb{w}_{a}})_{a\in A}$ in $\prod_{a\in A}(A^{\star})^{\bb{w}_{a}}$ the word $\left(\!\begin{smallmatrix}a\\(q^{a}_{\alpha})_{\alpha\in \bb{w}_{a}}\end{smallmatrix}\!\right)_{a\in A}(w)$ in $A^{\star}$ obtained by substituting in $w$, for every $a\in A$ and every $\alpha\in \bb{w}_a$, $q^{a}_{\alpha}$ for the $i_{\alpha}$-th occurrence of $a$ in $w$. We call  $\left(\!\begin{smallmatrix}a\\(q^{a}_{\alpha})_{\alpha\in \bb{w}_{a}}\end{smallmatrix}\!\right)_{a\in A}(w)$ \emph{the substitution of $(q_{\alpha})_{\alpha\in\bb{w}_{a}}$ for $a$ in $w$ for every $a$ in $A$}, and $\left(\!\begin{smallmatrix}a\\ \cdot\end{smallmatrix}\!\right)_{a\in A}(w)$ the \emph{substitution operator for $w$}. Let $c$ be an element of $A$. Then we will denote by $\left(\!\begin{smallmatrix}c\\ \cdot\end{smallmatrix}\!\right)(w)$ the mapping from $(A^{\star})^{\bb{w}_{c}}$ to $A^{\star}$ that assigns to $(q^{c}_{\alpha})_{\alpha\in \bb{w}_{c}}\in (A^{\star})^{\bb{w}_{c}}$ the term $\left(\!\begin{smallmatrix}c\\(q^{c}_{\alpha})_{\alpha\in \bb{w}_{c}}\end{smallmatrix}\!\right)(w)$ in $A^{\star}$ obtained by substituting in $w$, for every $\alpha\in \bb{w}_{c}$, $q^{c}_{\alpha}$ for the $i_{\alpha}$-th occurrence of $c$ in $w$. We call $\left(\!\begin{smallmatrix}c\\(q^{c}_{\alpha})_{\alpha\in \bb{w}_{c}}\end{smallmatrix}\!\right)(w)$ \emph{the substitution of $(q^{c}_{\alpha})_{\alpha\in\bb{w}_{c}}$ for $c$ in $w$}.
\end{definition}

\begin{remark}
Therefore, for every set $A$, we have obtained the family of substitution operators:
$$
\textstyle
(\left(\!\begin{smallmatrix}a\\ \cdot\end{smallmatrix}\!\right)_{a\in A}(w))_{w\in A^{\star}}\in \prod_{w\in A^{\star}}\mathrm{Hom}(\prod_{a\in A}(A^{\star})^{\bb{w}_{a}},A^{\star}).
$$
\end{remark}

Our next aim is to provide those notions from the field of many-sorted universal algebra that will be used afterwards.

\begin{definition}\label{$S$-sorted signature}
An $S$-\emph{sorted signature} is a function $\Sigma$ from $S^{\star}\times S$ to $\boldsymbol{\mathcal{U}}$ which sends a pair $(w,s)\in S^{\star}\times S$ to the set $\Sigma_{w,s}$ of the \emph{formal operations} of \emph{arity} $w$, \emph{sort} (or \emph{coarity}) $s$, and \emph{rank} (or \emph{biarity}) $(w,s)$.
%Sometimes we will write $\sigma\colon w\mor s$ to indicate that the formal operation $\sigma$ belongs to
%$\Sigma_{w,s}$.
If $\Sigma$ and $\Lambda$ are $S$-sorted signatures, then a \emph{morphism from} $\Sigma$ \emph{to} $\Lambda$ is an $S^{\star}\times S$-indexed family $d = (d_{w,s})_{(w,s)\in S^{\star}\times S}$, where, for every $(w,s)\in S^{\star}\times S$, $d_{w,s}$ is a mapping from $\Sigma_{w,s}$ to $\Lambda_{w,s}$. $S$-sorted signatures and morphisms between $S$-sorted signatures form a category which we will denote henceforth by $\mathbf{Sig}(S)$.

%From now on, $\mathbf{Sig}(S)$ stands for the category of $S$-sorted signatures and morphisms between $S$-sorted signatures.
\end{definition}

\begin{remark}
For every set of sorts $S$, the category $\mathbf{Sig}(S)$ is $\mathbf{Set}^{S^{\star}\times S}$.
\end{remark}

\begin{assumption}
From now on $\Sigma$ stands for an $S$-sorted signature, fixed once and for all.
\end{assumption}

%\begin{remark}
%To give an $S$-sorted signature as above is equivalent to give an ordered triple $(\Sigma,\mathrm{ar},\mathrm{car})$ where $\Sigma$ is a set (of operation symbols), $\mathrm{ar}$ a mapping from $\Sigma$ to $S^{\star}$ (the arity mapping), and $\mathrm{car}$ a mapping from $\Sigma$ to $S$ (the the coarity mapping). More accurately, the category $\mathbf{Set}^{S^{\star}\times S}$ is equivalent to the category $\mathbf{Sig}(S)$ whose objects are the just considered ordered triples $(\Sigma,\mathrm{ar},\mathrm{car})$ and in which the set of morphisms from $(\Sigma,\mathrm{ar},\mathrm{car})$ to $(\Sigma',\mathrm{ar}',\mathrm{car}')$ is the set of all mappings $d$ from $\Sigma$ to $\Sigma'$ such that $\mathrm{ar}'\circ d = \mathrm{ar}$ and $\mathrm{car}'\circ d = \mathrm{car}$. In all work with many-sorted signatures, we shall feel free to deal directly either with an $S$-sorted signature or with the correlated ordered triple, whichever is most convenient for the work at hand.
%
%%Therefore, although both notions, literally, are different, we will use them indistinctly.
%\end{remark}

We shall now give precise definitions of the concepts of many-sorted algebra and of homomorphism between many-sorted algebras.

\begin{definition}
The $S^{\star}\times S$-sorted set of the \emph{finitary operations on} an $S$-sorted set $A$ is $(\mathrm{Hom}(A_{w},A_{s}))_{(w,s)\in S^{\star}\times S}$, where, for every $w\in S^{\star}$, $A_{w} = \prod_{i\in \lvert w\rvert}A_{w_{i}}$, with $\lvert w\rvert$ denoting the length of the word $w$ (if $w = \lambda$, then $A_{\lambda}$ is a final set). A \emph{structure of} $\Sigma$-\emph{algebra on} an $S$-sorted set  $A$ is a family $(F_{w,s})_{(w,s)\in S^{\star}\times S}$, denoted by $F$, where, for $(w,s)\in S^{\star}\times S$, $F_{w,s}$ is a mapping from $\Sigma_{w,s}$ to $\mathrm{Hom}(A_{w},A_{s})$ (if $(w,s) = (\lambda,s)$ and $\sigma\in \Sigma_{\lambda,s}$, then $F_{w,s}(\sigma)$ picks out an element of $A_{s}$). For a pair $(w,s)\in S^{\star}\times S$ and a formal operation $\sigma\in \Sigma_{w,s}$, in order to simplify the notation, the operation $F_{w,s}(\sigma)$ from $A_{w}$ to $A_{s}$ will be written as $F_{\sigma}$. A $\Sigma$-\emph{algebra} is a pair $(A,F)$, abbreviated to $\mathbf{A}$, where $A$ is an $S$-sorted set and $F$ a structure of $\Sigma$-algebra on $A$. A $\Sigma$-\emph{homomorphism} from $\mathbf{A}$ to $\mathbf{B}$, where $\mathbf{B} = (B,G)$, is a triple $(\mathbf{A},f,\mathbf{B})$, abbreviated to $f\colon \mathbf{A}\mor \mathbf{B}$, where $f$ is an $S$-sorted mapping from $A$ to $B$ such that, for every $(w,s)\in S^{\star}\times S$, every  $\sigma\in \Sigma_{w,s}$, and every $(a_{i})_{i\in \lvert w\rvert}\in A_{w}$, we have that
$
f_{s}(F_{\sigma}((a_{i})_{i\in \lvert w\rvert})) = G_{\sigma}(f_{w}((a_{i})_{i\in \lvert w\rvert})),
$
where $f_{w}$ is the mapping $\prod_{i\in \lvert w\rvert}f_{w_{i}}$ from $A_{w}$ to $B_{w}$ that sends $(a_{i})_{i\in \lvert w\rvert}$ in $A_{w}$ to $(f_{w_{i}}(a_{i}))_{i\in \lvert w\rvert}$ in $B_{w}$. We will denote by $\mathbf{Alg}(\Sigma)$ the category of $\Sigma$-algebras and $\Sigma$-homomorphisms (or, to abbreviate, homomorphisms) and by $\mathrm{Alg}(\Sigma)$ the set of objects  of $\mathbf{Alg}(\Sigma)$.

In some cases, to avoid mistakes, we will denote by $F^{\mathbf{A}}$ the structure of $\Sigma$-algebra on $A$, and, for $(w,s)\in S^{\star}\times S$ and $\sigma\in \Sigma_{w,s}$, by $F^{\mathbf{A}}_{\sigma}$ the corresponding operation. Moreover, for $s\in S$ and $\sigma\in\Sigma_{\lambda,s}$, we will, usually, denote by $\sigma^{\mathbf{A}}$ the value of the mapping $F^{\mathbf{A}}_{\sigma}\colon 1\mor A_{s}$ at the unique element in $1$.

We will denote by $\mathbf{1}^{S}$ or, to abbreviate, by $\mathbf{1}$, the (standard) final $\Sigma$-algebra.
\end{definition}

\begin{definition}
Let $\mathbf{A}$ be a $\Sigma$-algebra. Then the \emph{support of} $\mathbf{A}$, denoted by $\mathrm{supp}_{S}(\mathbf{A})$, is $\mathrm{supp}_{S}(A)$, i.e., the support of the underlying $S$-sorted set $A$ of $\mathbf{A}$.
\end{definition}

\begin{remark}
The set $\{\mathrm{supp}_{S}(\mathbf{A})\mid \mathbf{A}\in \mathrm{Alg}(\Sigma)\}$ is a closure system on $S$.
\end{remark}

\begin{definition}
Let $\mathbf{A}$ be a $\Sigma$-algebra. We will say that $\mathbf{A}$ is \emph{finite} if $A$, the underlying $S$-sorted set of $\mathbf{A}$, is finite.
\end{definition}

\begin{remark}
In $\mathbf{Alg}(\Sigma)$, as was the case with $\mathbf{Set}^{S}$, there is another notion of finiteness: A $\Sigma$-algebra $\mathbf{A}$ is called $S$-\emph{finite} or \emph{locally finite}, abbreviated as $\mathrm{lf}$, if and only if the underlying $S$-sorted set of $\mathbf{A}$ is $S$-finite. As was noted above this notion of finiteness plays a relevant role in the field of many-sorted algebra, e.g., to define $S$-finite, also called locally finite, terms and to distinguish, on the one hand, between many-sorted varieties and finitary many-sorted  varieties, and, on the other hand, between many-sorted quasivarieties and finitary many-sorted quasivarieties (see, e.g., \cite{gm85} and \cite{{m76}}).
\end{remark}

We next define the subset many-sorted algebra associated to a many-sorted algebra and prove, in particular, that this is the object mapping of an endofunctor of $\mathbf{Alg}(\Sigma)$. We note that, in Section~3, the subset many-sorted algebra associated to a free many-sorted algebra will appear as the codomain of a substitution operator which will be used to prove recognizability results.

\begin{proposition}\label{subsetalg}
Let $(\cdot)^{\wp}$ be the mapping that sends (1) a $\Sigma$-algebra $\mathbf{A}$ to the $\Sigma$-algebra $\mathbf{A}^{\wp} = (A^{\wp},F^{\wp})$ where $A^{\wp}$ is $(\mathrm{Sub}(A_{s}))_{s\in S}$ and where, for every $(w,s)\in S^{\star}\times S$ and every $\sigma\in \Sigma_{w,s}$, $F_{\sigma}^{\wp}$ is the mapping from $A^{\wp}_{w} = \prod_{i\in \bb{w}}\mathrm{Sub}(A_{w_{i}})$ to $\mathrm{Sub}(A_{s})$ that sends $(L_{i})_{i\in \bb{w}}$ in $A^{\wp}_{w}$ to  $\{F_{\sigma}((x_{i})_{i\in \bb{w}})\mid (x_{i})_{i\in \bb{w}}\in \prod_{i\in \bb{w}}L_{i}\}$ in $\mathrm{Sub}(A_{s})$, and (2) a homomorphism $f$ from $\mathbf{A}$ to $\mathbf{B}$ to the $S$-sorted mapping $f^{\wp} = (f_{s}[\cdot])_{s\in S}$ from $A^{\wp}$ to $B^{\wp}$. Then $(\cdot)^{\wp}$ is an endofunctor of $\mathbf{Alg}(\Sigma)$. The $\Sigma$-algebra $\mathbf{A}^{\wp}$ is called the \emph{subset algebra} associated to $\mathbf{A}$ (this notion, but for single-sorted algebras, is due to Mezei and Wright, cf.~\cite{MW67}, Definition 2.2). Moreover, there exists a pointwise monomorphic natural transformation $\{\cdot\}^{\Sigma}$ from $\mathrm{Id}_{\mathbf{Alg}(\Sigma)}$ to $(\cdot)^{\wp}$, displayed as:
$$
\xymatrix@C=20ex@R=8ex{
\mathbf{Alg}(\Sigma)
\ar@/^20pt/[r]^{\mathrm{Id}_{\mathbf{Alg}(\Sigma)}}="f"
  \ar@/^-20pt/[r]_{(\cdot)^{\wp}}="g" &
\mathbf{Alg}(\Sigma)
\ar @{} "f";"g" |{\dir{=>}}^{\,\{\cdot\}^{\Sigma}}
}
$$
\end{proposition}

\begin{proof}
To show that $(\cdot)^{\wp}$ is an endofunctor of $\mathbf{Alg}(\Sigma)$ it suffices to verify that, for every homomorphism $f\colon \mathbf{A}\mor \mathbf{B}$, $f^{\wp}$ is, actually, a homomorphism from $\mathbf{A}^{\wp}$ to $\mathbf{B}^{\wp}$. Let $f$ be a homomorphism from $\mathbf{A}$ to $\mathbf{B}$, $(w,s)\in S^{\star}\times S$, and $\sigma\in \Sigma_{w,s}$. Then the following diagram commutes
$$
\xymatrix@C=60pt{
A^{\wp}_{w} \ar[d]_-{F_{\sigma}^{\wp}}\ar[r]^-{f_{{w}}[\cdot]}&
B^{\wp}_{w} \ar[d]^-{G_{\sigma}^{\wp}}\\
\mathrm{Sub}(A_{s})\ar[r]_-{f_{s}[\cdot]}&\mathrm{Sub}(B_{s})
}
$$
that is, for every sequence $(L_{i})_{i\in \bb{w}}$ in $A^{\wp}_{w}$, we have that
$$
f_{s}[F_{\sigma}^{\wp}((L_{i})_{i\in \bb{w}})] = G_{\sigma}^{\wp}((f_{w_{i}}[L_{i}])_{i\in \bb{w}}).
$$

Now let $\{\cdot\}^{\Sigma}$ be the mapping from $\mathrm{Alg}(\Sigma)$ to $\mathrm{Mor}(\mathbf{Alg}(\Sigma))$ that assigns to a $\Sigma$-algebra $\mathbf{A}$ the $S$-sorted mapping $\{\cdot\}^{\Sigma}_{\mathbf{A}}$ from $A$ to $A^{\wp}$ that, for every $s\in S$, sends $a$ in $A_{s}$ to $\{a\}$ in $A^{\wp}_{s}$. It is easily seen that, for every $\Sigma$-algebra $\mathbf{A}$, $\{\cdot\}^{\Sigma}_{\mathbf{A}}$ is an injective homomorphism from $\mathbf{A}$ to $\mathbf{A}^{\wp}$ and that $\{\cdot\}^{\Sigma} = (\{\cdot\}^{\Sigma}_{\mathbf{A}})_{\mathbf{A}\in \mathrm{Alg}(\Sigma)}$ is, in fact, a natural transformation from $\mathrm{Id}_{\mathbf{Alg}(\Sigma)}$ to $(\cdot)^{\wp}$
\end{proof}

We shall now go on to define the notion of subalgebra of a $\Sigma$-algebra $\mathbf{A}$ and the subalgebra generating operator for $\mathbf{A}$.

\begin{definition}\label{Subalg}
Let $\mathbf{A}$ be a $\Sigma$-algebra and $X\subseteq A$. Given $(w,s)\in S^{\star}\times S$ and $\sigma\in\Sigma_{w,s}$, we will say that $X$ is \emph{closed under the operation} $F_{\sigma}\colon A_{w}\mor A_{s}$ if, for every $(a_{i})_{i\in\bb{w}}\in X_{w}$, $F_{\sigma}((a_{i})_{i\in\bb{w}})\in X_{s}$. We will say that $X$ is a \emph{subalgebra} of $\mathbf{A}$ if $X$ is closed under the operations of $\mathbf{A}$. We will denote by $\mathrm{Sub}(\mathbf{A})$ the set of all subalgebras of $\mathbf{A}$ (which is an algebraic closure system on $A$) and by  $\mathbf{Sub}(\mathbf{A})$ the algebraic lattice $(\mathrm{Sub}(\mathbf{A}),\subseteq)$. We also say, equivalently, that a $\Sigma$-algebra $\mathbf{B}$ is a \emph{subalgebra} of $\mathbf{A}$ if $B\subseteq A$ and the canonical embedding of $B$ into $A$ determines an embedding of $\mathbf{B}$ into $\mathbf{A}$.
\end{definition}

\begin{definition}
Let $\mathbf{A}$ be a $\Sigma$-algebra. Then we will denote by $\mathrm{Sg}_{\mathbf{A}}$ the algebraic closure operator canonically associated to the algebraic closure system $\mathrm{Sub}(\mathbf{A})$ on $A$ and we call it the \emph{subalgebra  generating operator for} $\mathbf{A}$. Moreover, if $X\subseteq A$, then we call $\mathrm{Sg}_{\mathbf{A}}(X)$ the \emph{subalgebra of} $\mathbf{A}$ \emph{generated by} $X$, and if $X$ is such that $\mathrm{Sg}_{\mathbf{A}}(X) = A$, then we will say that $X$ is a \emph{generating subset of} $\mathbf{A}$. Besides, $\mathbf{Sg}_{\mathbf{A}}(X)$ denotes the algebra determined by $\mathrm{Sg}_{\mathbf{A}}(X)$.
\end{definition}

\begin{remark}
Let $\mathbf{A}$ be a $\Sigma$-algebra. Then the algebraic closure operator $\mathrm{Sg}_{\mathbf{A}}$ is uniform, i.e., for every  $X$, $Y\subseteq A$, if $\mathrm{supp}_{S}(X) = \mathrm{supp}_{S}(Y)$, then we have that $\mathrm{supp}_{S}(\mathrm{Sg}_{\mathbf{A}}(X)) = \mathrm{supp}_{S}(\mathrm{Sg}_{\mathbf{A}}(Y))$.
\end{remark}

We next recall the Principle of Proof by Algebraic Induction. This principle will be used in Section~3.

\begin{proposition}[Principle of Proof by Algebraic Induction]\label{PPAI}
Let $\mathbf{A}$ be a $\Sigma$-algebra generated by $X$. Then to prove that a subset $Y$ of $A$ is equal to $A$ it suffices to show: (1) $X\subseteq Y$ (algebraic induction basis) and (2) $Y$ is a subalgebra of $\mathbf{A}$ (algebraic induction step).
\end{proposition}

We next state that the forgetful functor $\mathrm{G}_{\Sigma}$ from $\mathbf{Alg}(\Sigma)$ to
$\mathbf{Set}^{S}$ has a left adjoint $\mathbf{T}_{\Sigma}$ which assigns to an $S$-sorted set $X$ the free $\Sigma$-algebra $\mathbf{T}_{\Sigma}(X)$ on $X$. Let us note that in what follows, to construct the algebra of $\Sigma$-rows in $X$, and  the free $\Sigma$-algebra on $X$, since neither the $S$-sorted signature $\Sigma$ nor the $S$-sorted set $X$ are subject to any constraint, coproducts must necessarily be used.

\begin{definition}
Let $X$ be an $S$-sorted set. The \emph{algebra of} $\Sigma$-\emph{rows in} $X$, denoted by $\mathbf{W}_{\Sigma}(X)$, is defined as follows:
\begin{enumerate}
\item The underlying $S$-sorted set of $\mathbf{W}_{\Sigma}(X)$, written as $\mathrm{W}_{\Sigma}(X)$, is precisely the $S$-sorted set $((\coprod\Sigma \amalg \coprod X)^{\star})_{s\in S}$, i.e., the mapping from $S$ to $\boldsymbol{\mathcal{U}}$ constantly $(\coprod\Sigma \amalg \coprod X)^{\star}$, where $(\coprod\Sigma \amalg \coprod X)^{\star}$ is the set of all words on the set $\coprod\Sigma \amalg \coprod X$, i.e., on the set
      $$
      \textstyle
      [(\bigcup_{(w,s)\in S^{\star}\times S}(\Sigma_{w,s}\times\{(w,s)\}))\times\{0\}]\cup
      [(\bigcup_{s\in S}(X_{s}\times\{s\}))\times\{1\}].
      $$

\item For every $(w,s)\in S^{\star}\times S$, and every $\sigma\in\Sigma_{w,s}$, the structural operation $F_{\sigma}$ associated to $\sigma$ is the mapping from $\mathrm{W}_{\Sigma}(X)_{w}$ to ${\mathrm{W}_{\Sigma}(X)}_{s}$ which sends $(P_{i})_{i\in\lvert w \rvert} \in \mathrm{W}_{\Sigma}(X)_{w}$ to $(\sigma)\curlywedge\concat_{i\in\lvert w \rvert}P_{i} \in {\mathrm{W}_{\Sigma}(X)}_{s}$, where, for every $(w,s)\in S^{\star}\times S$, and every $\sigma\in\Sigma_{w,s}$,
    %defined as follows:
%      $$
%      F_{\sigma}
%      \nfunction
%      {\mathrm{W}_{\Sigma}(X)_{w}}
%      {{\mathrm{W}_{\Sigma}(X)}_{s}}
%      {(P_{i})_{i\in\lvert w \rvert}}
%      {(\sigma)\curlywedge\concat_{i\in\lvert w \rvert}P_{i}}
%      $$
      $(\sigma)$ stands for $(((\sigma,(w,s)),0))$, which is the value at $\sigma$ of the canonical mapping from $\Sigma_{w,s}$ to $(\coprod\Sigma \amalg \coprod X)^{\star}$.
%      mappings indicated in the following figure
%      $$\xymatrix@C=2.50pc@R=1pc{
%      \Sigma_{w,s} \ar[r]^-{\inc_{\Sigma_{w,s}}} & \coprod \Sigma
%      \ar[r]^-{\inc_{\coprod\Sigma}} &
%      \coprod\Sigma\amalg\coprod X
%      \ar[r]^-{\eta_{\coprod\Sigma\amalg\coprod X}} &
%      (\coprod\Sigma \amalg \coprod X)^{\star} %\\
%      \sigma \ar@{|->}[r]  & (\sigma,(w,s)) \ar@{|->}[r] &
%      ((\sigma,(w,s)),0) \ar@{|->}[r]  &
%      (((\sigma,(w,s)),0))\equiv(\sigma)
%      }
%      $$
\end{enumerate}
\end{definition}

\begin{definition}
The \emph{free} $\Sigma$-\emph{algebra on} an $S$-sorted set $X$, denoted by $\mathbf{T}_{\Sigma}(X)$, is the $\Sigma$-algebra determined by $\mathrm{Sg}_{\mathbf{W}_{\Sigma}(X)}((\{(x)\mid x\in X_{s}\})_{s\in S})$, the subalgebra of $\mathbf{W}_{\Sigma}(X)$ generated by $(\{(x)\mid x\in X_{s}\})_{s\in S}$, where, for every $s\in S$ and every $x\in X_{s}$, $(x)$ stands for $(((x,s),1))$, which is the value at $x$ of the canonical mapping from $X_{s}$ to $(\coprod\Sigma \amalg \coprod X)^{\star}$.
%mappings indicated in the following figure
%$$
%\xymatrix@C=3pc@R=1pc{
%X_{s} %\ar `u[r] `[rrr]^{\eta_{X}} [rrr]
%      %\ar@(u,u)[rrr]^{\eta_{X}}
%      \ar[r]^-{\inc_{X_{s}}} &
%\coprod X \ar[r]^-{\inc_{\coprod X}} &
%\coprod\Sigma\amalg\coprod X
%\ar[r]^-{\eta_{\coprod\Sigma\amalg\coprod X}} &
%(\coprod\Sigma \amalg \coprod X)^{\star} %\\
%%x \ar@{|->}[r]  & (x,s) \ar@{|->}[r]  & ((x,s),1) \ar@{|->}[r]  &
%%(((x,s),1))\equiv(x)
%}
%$$
We will denote by $\mathrm{T}_{\Sigma}(X)$ the underlying $S$\nobreakdash-sorted of $\mathbf{T}_{\Sigma}(X)$ and, for $s\in S$, we will call the elements of $\mathrm{T}_{\Sigma}(X)_{s}$ \emph{terms of type} $s$ \emph{with variables in} $X$ or  $(X,s)$-\emph{terms}.
\end{definition}

\begin{remark}
Since $(\{(x)\mid x\in X_{s}\})_{s\in S}$ is a generating subset of $\mathbf{T}_{\Sigma}(X)$, to prove that a subset $\mathcal{T}$ of $\mathrm{T}_{\Sigma}(X)$ is equal to $\mathrm{T}_{\Sigma}(X)$ it suffices, by Proposition~\ref{PPAI}, to show: (1) $(\{(x)\mid x\in X_{s}\})_{s\in S}\subseteq \mathcal{T}$ (algebraic induction basis) and (2) $\mathcal{T}$ is a subalgebra of $\mathbf{T}_{\Sigma}(X)$ (algebraic induction step).
\end{remark}

In the many-sorted case we have, as in the single-sorted case, the following characterization of the elements of $\mathrm{T}_{\Sigma}(X)_{s}$, for $s\in S$.

\begin{proposition}\label{rut}
Let $X$ be an $S$-sorted set. Then, for every $s\in S$ and every $P\in
\mathrm{W}_{\Sigma}(X)_{s}$, we have that $P$ is a term of type $s$ with variables in $X$ if and only if $P = (x)$, for a unique $x\in X_{s}$, or $P = (\sigma)$, for a unique $\sigma\in\Sigma_{\lambda,s}$, or $P = (\sigma)\curlywedge\concat(P_{i})_{i\in\lvert w \rvert}$, for a unique $w\in S^{\star}-\{\lambda\}$, a unique $\sigma\in\Sigma_{w,s}$, and a unique family $(P_{i})_{i\in\lvert w \rvert}\in\mathrm{T}_{\Sigma}(X)_{w}$.
%\begin{enumerate}
%\item $P = (x)$, for a unique $x\in
%      X_{s}$, or
%
%\item $P = (\sigma)$, for a unique $\sigma\in\Sigma_{\lambda,s}$,
%      or
%
%\item $P = (\sigma)\curlywedge\concat(P_{i})_{i\in\lvert w \rvert}$, for
%      a unique $w\in S^{\star}-\{\lambda\}$, a unique
%      $\sigma\in\Sigma_{w,s}$, and a unique family
%      $(P_{i})_{i\in\lvert w \rvert}\in\mathrm{T}_{\Sigma}(X)_{w}$.
%\end{enumerate}
Moreover, the three possibilities are mutually exclusive.
\end{proposition}

From now on, for simplicity of notation, we will write $x$, $\sigma$, and $\sigma(P_{0},\ldots,P_{\lvert w \rvert-1})$ or $\sigma((P_{i})_{i\in \bb{w}})$ instead of $(x)$, $(\sigma)$, and $(\sigma)\curlywedge\concat(P_{i})_{i\in\lvert w \rvert}$, respectively.

From the above proposition it follows, immediately, the universal property of the free $\Sigma$-algebra on an $S$-sorted set $X$, as stated in the subsequent proposition.

\begin{proposition}
For every $S$-sorted set $X$, the pair $(\eta_{X},\mathbf{T}_{\Sigma}(X))$, where $\eta_{X}$, the
\emph{insertion of (the $S$-sorted set of generators)} $X$ \emph{into} $\mathrm{T}_{\Sigma}(X)$, is the co-restric\-tion to
$\mathrm{T}_{\Sigma}(X)$ of the canonical embedding of $X$ into $\mathrm{W}_{\Sigma}(X)$, has the following universal property: for every $\Sigma$-algebra $\mathbf{A}$ and every $S$-sorted mapping $f\colon X\mor A$, there exists a unique homomorphism $f^{\sharp}\colon\mathbf{T}_{\Sigma}(X)\mor\mathbf{A}$ such that $f^{\sharp}\circ \eta_{X} = f$.
\end{proposition}

\begin{proof}
For every $s\in S$ and every $(X,s)$-term $P$, the $s$-th coordinate $f^{\sharp}_{s}$ of $f^{\sharp}$ is defined recursively as follows: $f^{\sharp}_{s}(x) = f_{s}(x)$, if $P = x$; $f^{\sharp}_{s}(\sigma) = \sigma$, if $P = \sigma$; and, finally,  $f^{\sharp}_{s}(\sigma(P_{0},\ldots,P_{\lvert w \rvert-1})) = F_{\sigma}(f^{\sharp}_{w_{0}}(P_{0}),\ldots,f^{\sharp}_{w_{\lvert w \rvert-1}}(P_{\lvert w \rvert-1}))$, if $P = \sigma(P_{0},\ldots,P_{\lvert w \rvert-1})$.
%$$
%P\,\lmapsto\,f^{\sharp}_{s}(P) =
%\begin{cases}
%  f_{s}(x),&\text{if }P = x; \\
%  \sigma^{\mathbf{A}},&\text{if }P = \sigma;\\
%  F_{\sigma}^{\mathbf{A}}(f^{\sharp}_{w(0)}(P_{0}),\ldots,
%  f^{\sharp}_{w(\lvert w \rvert-1)}(P_{\lvert w \rvert-1})),
%  &\text{if }P = \sigma(P_{0},\ldots,P_{\lvert w \rvert-1}).
%\end{cases}
%$$
\end{proof}

The just stated proposition allows us to carry out definitions by algebraic recursion on a free many-sorted algebra as indeed we will be doing throughout  this paper.

%The just stated proposition provides a sound basis for the principle of definition by algebraic recursion.

\begin{corollary} \label{FladjG}
The functor $\mathbf{T}_{\Sigma}$, which sends an $S$-sorted set $X$ to $\mathbf{T}_{\Sigma}(X)$ and an $S$-sorted mapping $f$ from $X$ to $Y$ to $f^{@} (= (\eta_{Y}\circ f)^{\sharp})$, the unique homomorphism from $\mathbf{T}_{\Sigma}(X)$ to $\mathbf{T}_{\Sigma}(Y)$ such that $f^{@}\circ\eta_{X} = \eta_{Y}\circ f$, is left adjoint for the forgetful functor $\mathrm{G}_{\Sigma}$ from $\mathbf{Alg}(\Sigma)$ to $\mathbf{Set}^{S}$.
%$$
%\xymatrix@=5pc{ \mathbf{Alg}(\Sigma)
%  \ar@<1.5ex>[r]^{\mathrm{G}_{\Sigma}}
%  \ar@{}[r]|{\uadj}
% &
%\mathbf{Set}^{S}
%  \ar@<1.5ex>[l]^{\mathbf{T}_{\Sigma}}
%}
%$$
\end{corollary}

For every $\Sigma$-algebra it is possible to define a preorder on the coproduct of its underlying $S$-sorted set. Moreover, for the case of free algebras, such a preorder is, in fact, an order and this allows us to define the notion of subterm of a given term.

\begin{definition}
Let $\mathbf{A}=(A,F)$ be a $\Sigma$-algebra. Then $<_{\mathbf{A}}$ denotes the binary relation on $\coprod A$ consisting of the ordered pairs $((a,s),(b,t))\in(\coprod A)^{2}$ for which there exists a $w\in S^{\star}-\{\lambda\}$, a
$\sigma\in\Sigma_{w,t}$, and an $x\in A_{w}$ such that $F_{\sigma}(x) = b$ and, for some $i\in \bb{w}$, $w_{i}=s$ and  $x_{i}=a$. We will denote by $\leq_{\mathbf{A}}$ the reflexive and transitive closure of $<_{\mathbf{A}}$, i.e., the preorder on $\coprod A$  generated by $<_{\mathbf{A}}$.
\end{definition}

\begin{remark}
The preorder $\leq_{\mathbf{A}}$ on $\coprod A$ is defined by letting $((a,s),(b,t))\in \leq_{\mathbf{A}}$ mean that $s = t$ and $a = b$ or there exists an $n\in \mathbb{N}-\{0\}$, a $u\in S^{\star}$, and a family $(c_{i})_{i\in \bb{u}}\in A_{u}$ such that $\bb{u} = n+1$, $u_{0} = s$, $u_{n} = t$, $c_{0} = a$, $c_{n} = b$, and, for every $i\in n$, $((c_{i},u_{i}),(c_{i+1},u_{i+1}))\in <_{\mathbf{A}}$.
\end{remark}

\begin{proposition}
Let $X$ be an $S$-sorted set.  Then $\leq_{\mathbf{T}_{\Sigma}(X)}$ is antisymmetric and does not have strictly descending $\omega_{0}$-chains, i.e., is an Artinian order.
\end{proposition}

\begin{definition}
Let $X$ be an $S$-sorted set, $t\in S$, and $P\in\mathrm{T}_{\Sigma}(X)_{t}$. Then the $S$-sorted set of all \emph{subterms} of $P$, denoted by $\mathrm{Subt}(P)$, is defined as follows:
$$
\mathrm{Subt}(P) = (\{Q\in \mathrm{T}_{\Sigma}(X)_{s} \mid (Q,s)
\leq_{\mathbf{T}_{\Sigma}(X)} (P,t)\})_{s \in S }.
$$
\end{definition}

\begin{remark}
$\mathrm{Subt}(P)\in \mathrm{Sub}_{\mathrm{f}}(\mathrm{T}_{\Sigma}(X))$. Moreover, $\mathrm{Subt}(P)$ can also be characterized as the smallest subset $\mathcal{L}$ of $\mathrm{T}_{\Sigma}(X)$ which satisfies the following conditions: (1) $P\in \mathcal{L}_{t}$ and (2) for every $(w,s)\in S^{\star}\times S$, every $\xi\in \Sigma_{w,s}$, and every $(Q_{i})_{i\in\bb{w}}\in \mathrm{T}_{\Sigma}(X)_{w}$, if $\xi((Q_{i})_{i\in\bb{w}})\in \mathcal{L}_{s}$, then, for every $i\in\bb{w}$, $Q_{i}\in \mathcal{L}_{w_{i}}$. Note that the second condition is exactly the converse of the defining condition of the concept of subalgebra of a $\Sigma$-algebra.
\end{remark}

Following this we associate to every term for $\Sigma$ of type $(X,s)$ its $S$-sorted set of variables. This will be used, in Section 3, in the proof of diverse propositions.

\begin{definition}
Let $X$ be an $S$-sorted set.  Then $\mathbf{Fin}(X)$ is the $\Sigma$-algebra which has as underlying
$S$-sorted set $(\mathrm{Sub}_{\mathrm{f}}(X))_{s\in S}$, i.e., the $S$-sorted set constantly $\mathrm{Sub}_{\mathrm{f}}(X)$, and, for every $(w,s)\in S^{\star}\times S$ and every $\sigma\in \Sigma_{w,s}$, as operation $F_{\sigma}\colon \mathrm{Sub}_{\mathrm{f}}(X)^{\bb{w}}\mor
\mathrm{Sub}_{\mathrm{f}}(X)$ that one defined as $F_{\sigma}((K^{i})_{i\in\bb{w}})= \bigcup_{i\in\bb{w}}K^{i}$. Let
$\delta^{X}=(\delta^{X}_{s})_{s\in S}$ be the $S$-sorted mapping defined, for every $s\in S$, as
$$
\delta^{X}_{s}
  \nfunction
  {X_{s}} {\mathrm{Sub}_{\mathrm{f}}(X)}
  {x} {\delta^{s,x}}
$$
Then we will denote by $\mathrm{Var}^{X}$ the unique homomorphism $\ext{(\delta^{X})}$ from $\mathbf{T}_{\Sigma}(X)$ to
$\mathbf{Fin}(X)$ such that the following diagram commutes
$$
\xymatrix{ X \ar[r]^-{\eta_{X}}
  \ar[rd]_-{\delta^{X}}
 &
\mathrm{T}_{\Sigma}(X)
  \ar[d]^{\ext{(\delta^{X})} = \mathrm{Var}^{X} = (\mathrm{Var}^{X}_{s})_{s\in S}} \\
 &
{\mathrm{Fin}}(X) = (\mathrm{Sub}_{\mathrm{f}}(X))_{s\in S}}
$$
For a sort $s\in S$ and a term $P\in \mathrm{T}_{\Sigma}(X)_{s}$ we will call $\mathrm{Var}^{X}_{s}(P) (\in \mathrm{Sub}_{\mathrm{f}}(X))$ the $S$-sorted set of \emph{variables} of $P$. Moreover, when this is unlikely to cause confusion, we will write $\mathrm{Var}^{X}(P)$ or, simply, $\mathrm{Var}(P)$ for $\mathrm{Var}^{X}_{s}(P)$.
\end{definition}

Our next goal is to define the concepts of congruence on a $\Sigma$-algebra and of quotient of a $\Sigma$-algebra by a congruence on it. Moreover, we recall the notion of kernel of a homomorphism between $\Sigma$-algebras and the universal property of the quotient of a $\Sigma$-algebra by a congruence on it.

\begin{definition}
Let $\mathbf{A}$ be a $\Sigma$-algebra and $\Phi$ an $S$-sorted equivalence on $A$. We will say that $\Phi$ is an
$S$-\emph{sorted congruence on} (or, to abbreviate, a \emph{congruence on}) $\mathbf{A}$ if, for every $(w,s)\in (S^{\star}-\{\lambda\})\times S$, every $\sigma\in \Sigma_{w,s}$,
and every $(a_{i})_{i\in\bb{w}},(b_{i})_{i\in\bb{w}}\in A_{w}$, if, for every $i\in \lvert w\rvert$, $(a_{i}, b_{i})\in\Phi_{w_{i}}$, then $(F_{\sigma}((a_{i})_{i\in\bb{w}}), F_{\sigma}((b_{i})_{i\in\bb{w}}))\in \Phi_{s}$.
%we have that
%$$
%\frac
%{\forall i\in \lvert w\rvert\mathrm{, }\,\, (a_{i}, b_{i})\in\Phi_{w_{i}} }
%{(F_{\sigma}(a), F_{\sigma}(b))\in \Phi_{s}}\cdot
%$$
We will denote by $\mathrm{Cgr}(\mathbf{A})$ the set of all $S$-sorted congruences on $\mathbf{A}$ (which is an algebraic closure system on $A\times A$), by $\mathbf{Cgr}(\mathbf{A})$ the algebraic lattice  $(\mathrm{Cgr}(\mathbf{A}),\subseteq)$, by $\nabla_{\mathbf{A}}$ the greatest element of $\mathbf{Cgr}(\mathbf{A})$, and by $\Delta_{\mathbf{A}}$ the least element of $\mathbf{Cgr}(\mathbf{A})$.

For a congruence $\Phi$ on $\mathbf{A}$, the \emph{quotient $\Sigma$-algebra of} $\mathbf{A}$ \emph{by} $\Phi$, denoted by $\mathbf{A}/\Phi$, is the $\Sigma$-algebra $(A/\Phi, F^{\mathbf{A}/\Phi})$, where, for every  $(w,s)\in S^{\star}\times S$ and every $\sigma\in \Sigma_{w,s}$, the operation $F_{\sigma}^{\mathbf{A}/\Phi}$ from $(A/\Phi)_{w}$ to $A_{s}/\Phi_{s}$, also denoted, to simplify, by $F_{\sigma}$, sends $([a_{i}]_{\Phi_{w_{i}}})_{i\in\lvert w\rvert}$ in $(A/\Phi)_{w}$ to $[F_{\sigma}((a_{i})_{i\in \lvert w\rvert})]_{\Phi_{s}}$ in $A_{s}/\Phi_{s}$,  %
%$$
%F_{\sigma}
%\nfunction
%{(A/\Phi)_{w}} {A_{s}/\Phi_{s}}
%{([a_{i}]_{\Phi_{w_{i}}})_{i\in\lvert w\rvert}}
%{[F_{\sigma}((a_{i})_{i\in \lvert w\rvert})]_{\Phi_{s}}}
%$$
and the \emph{canonical projection} from $\mathbf{A}$ to $\mathbf{A}/\Phi$, denoted by $\mathrm{pr}_{\Phi}\colon \mathbf{A}\mor \mathbf{A}/\Phi$, is the homomorphism determined by the projection from $A$ to $A/\Phi$.
The ordered pair $(\mathbf{A}/\Phi,\mathrm{pr}_{\Phi})$ has the following universal property: $\mathrm{Ker}(\mathrm{pr}_{\Phi})$ is $\Phi$ and, for every $\Sigma$-algebra $\mathbf{B}$ and every homomorphism $f$ from $\mathbf{A}$ to $\mathbf{B}$, if $\mathrm{Ker}(f)\supseteq \Phi$, then there exists a unique homomorphism $h$ from $\mathbf{A}/\Phi$ to $\mathbf{B}$ such that $h\circ\mathrm{pr}_{\Phi} = f$. In particular, if $\Psi$ is a congruence on $A$ such that $\Phi\subseteq \Psi$, then we will denote by $\mathrm{p}_{\Phi,\Psi}$ the unique homomorphism from $\mathbf{A}/\Phi$ to $\mathbf{A}/\Psi$ such that $\mathrm{p}_{\Phi,\Psi}\circ \mathrm{pr}_{\Phi} = \mathrm{pr}_\Psi$.
\end{definition}

%\begin{remark}
%Let $\mathbf{A}$ be a $\Sigma$-algebra and $\Phi$ an $S$-sorted equivalence on $A$. Then $\Phi$ is a congruence on $\mathbf{A}$ if and only if $\Phi$ is a subalgebra of $\mathbf{A}\times \mathbf{A}$.
%\end{remark}

\begin{remark}
Let $\mathbf{ClfdAlg}(\Sigma)$ be the category whose objects are the \emph{classified $\Sigma$-algebras}, i.e, the ordered pairs $(\mathbf{A},\Phi)$ where $\mathbf{A}$ is a $\Sigma$-algebra and $\Phi$ a congruence on $\mathbf{A}$, and in which the set of morphisms from $(\mathbf{A},\Phi)$ to $(\mathbf{B},\Psi)$ is the set of all homomorphisms $f$ from $\mathbf{A}$ to $\mathbf{B}$ such that, for every $s\in S$ and every $(x,y)\in A^{2}_{s}$, if $(x,y)\in \Phi_{s}$, then $(f_{s}(x),f_{s}(y))\in \Psi_{s}$. Let $G$ be the functor from $\mathbf{Alg}(\Sigma)$ to $\mathbf{ClfdAlg}(\Sigma)$ whose object mapping sends $\mathbf{A}$ to $(\mathbf{A},\Delta_{\mathbf{A}})$ and whose morphism mapping sends $f\colon \mathbf{A}\mor \mathbf{B}$ to $f\colon (\mathbf{A},\Delta_{\mathbf{A}})\mor (\mathbf{B},\Delta_{\mathbf{B}})$. Then, for every classified $\Sigma$-algebra $(\mathbf{A},\Phi)$, there exists a universal mapping from $(\mathbf{A},\Phi)$ to $G$, which is precisely the ordered pair $(\mathbf{A}/\Phi,\mathrm{pr}_{\Phi})$ with $\mathrm{pr}_{\Phi}\colon (\mathbf{A},\Phi)\mor (\mathbf{A}/\Phi,\Delta_{\mathbf{A}/\Phi})$.
\end{remark}

We next define for a $\Sigma$-algebra the concepts of elementary translation and of translation with respect to it, and provide, by using the just mentioned concepts, two characterizations of the congruences on a $\Sigma$-algebra. To this we add that  the concept of translation will allow us to define the concept of congruence cogenerated by an $S$-sorted subset of the underlying $S$-sorted set of a $\Sigma$-algebra, which will be all-important in our congruence based proofs of the recognizability theorems for free many-sorted algebras, in Section 3.

\begin{definition}
Let $\mathbf{A}$ be a $\Sigma$-algebra and $t\in S$. Then we will denote by $\mathrm{Etl}_{t}(\mathbf{A})$ the subset $(\mathrm{Etl}_{t}(\mathbf{A})_{s})_{s\in S}$ of $(\mathrm{Hom}(A_{t},A_{s}))_{s\in S}$ defined, for every $s\in S$, as follows: For every mapping $T\in \mathrm{Hom}(A_{t},A_{s})$, $T\in \mathrm{Etl}_{t}(\mathbf{A})_{s}$ if and only if there is a word $w\in  S^{\star}-\{\lambda\}$, an $i\in \lvert w \rvert$, a $\sigma\in \Sigma_{w,s}$, a family $(a_{j})_{j\in i}\in\prod_{j\in i}A_{w_{j}}$, and a family $(a_{k})_{k\in \lvert w \rvert-(i+1)} \in\prod_{k\in \lvert w \rvert-(i+1)}A_{w_{k}}$ (recall that $i+1 = \{0, 1,\ldots,i\}$ and that $\lvert w \rvert-(i+1) = \{i+1,\ldots,\lvert w \rvert-1\}$) such that $w_{i} = t$ and, for every $x\in A_{t}$, $T(x) =
F_{\sigma}(a_{0},\ldots,a_{i-1},x,a_{i+1},\ldots,a_{\lvert w \rvert-1})$. We call the elements of $\mathrm{Etl}_{t}(\mathbf{A})_{s}$ the $t$-\emph{elementary translations of sort} $s$ \emph{for} $\mathbf{A}$.
\end{definition}

\begin{definition}
Let $\mathbf{A}$ be a $\Sigma$-algebra and $t\in S$. Then we will denote by $\mathrm{Tl}_{t}(\mathbf{A})$ the subset
$(\mathrm{Tl}_{t}(\mathbf{A})_{s})_{s\in S}$ of $(\mathrm{Hom}(A_{t},A_{s}))_{s\in S}$ defined, for every $s\in S$, as follows: For every mapping $T\in \mathrm{Hom}(A_{t},A_{s})$, $T\in \mathrm{Tl}_{t}(\mathbf{A})_{s}$ if and only if there is an $n\in \mathbb{N}-1$, a word $(s_{j})_{j\in n+1}\in S^{n+1}$, and a family $(T_{j})_{j\in n}$ such that $s_{0} = t$, $s_{n} = s$, $T_{0}\in \mathrm{Etl}_{t}(\mathbf{A})_{s_{1}}$, $T_{1}\in \mathrm{Etl}_{s_{1}}(\mathbf{A})_{s_{2}}$, \ldots, $T_{n-1}\in \mathrm{Etl}_{s_{n-1}}(\mathbf{A})_{s}$ and $T = T_{n-1}\circ\cdots\circ T_{0}$. We call the elements of $\mathrm{Tl}_{t}(\mathbf{A})_{s}$ the $t$-\emph{translations of sort} $s$ \emph{for} $\mathbf{A}$. Besides, for every $t\in S$, the mapping  $\mathrm{id}_{A_{t}}$ will be viewed as an element of $\mathrm{Tl}_{t}(\mathbf{A})_{t}$.
\end{definition}

\begin{remark}
The $S\times S$-sorted set $(\mathrm{Tl}_{t}(\mathbf{A})_{s})_{(t,s)\in S\times S}$ determines a category $\mathbf{Tl}(\mathbf{A})$ whose object set is $S$ and in which, for every $(t,s)\in S\times S$, $\mathrm{Hom}_{\mathbf{Tl}(\mathbf{A})}(t,s)$, the hom-set from $t$ to $s$, is $\mathrm{Tl}_{t}(\mathbf{A})_{s}$. Therefore, for every $t\in S$, $\mathrm{End}_{\mathbf{Tl}(\mathbf{A})}(t)$ is equipped with a structure of monoid.
\end{remark}

Given a $\Sigma$-algebra $\mathbf{A}$ and a translation $T\in \mathrm{Tl}_{t}(\mathbf{A})_{s}$ we next define, as\-so\-ci\-at\-ed to the mappings $T[\cdot]\colon \mathrm{Sub}(A_{t})\mor \mathrm{Sub}(A_{s})$ and $T^{-1}[\cdot]\colon \mathrm{Sub}(A_{s})\mor \mathrm{Sub}(A_{t})$, operators: $T[\cdot]$ and $T^{-1}[\cdot]$ from $\mathrm{Sub}(A)$ to $\mathrm{Sub}(A)$, $T[\cdot]$ from $\mathrm{Sub}(A_{t})$ to $\mathrm{Sub}(A)$, and $T^{-1}[\cdot]$ from $\mathrm{Sub}(A_{s})$ to $\mathrm{Sub}(A)$. This will be used below in the characterization of the recognizable languages.

\begin{definition}\label{DAntiTrans}
Let $\mathbf{A}$ be a $\Sigma$-algebra, $L\subseteq A$, $s$, $t\in S$, $M\subseteq A_{t}$, $N\subseteq A_{s}$, and $T\in \mathrm{Tl}_{t}(\mathbf{A})_{s}$. Then
\begin{enumerate}
\item $T[L]$ is the subset of $A$ defined as follows: $T[L]_{s} = T[L_{t}]$ and $T[L]_{u} = \varnothing$, if $u\neq s$. Therefore, $T[L] = \delta^{s,T[L_{t}]}$;
\item $T^{-1}[L]$ is the subset of $A$ defined as follows: $T^{-1}[L]_{t} = T^{-1}[L_{s}]$ and $T^{-1}[L]_{u} = \varnothing$, if $u\neq t$. Therefore, $T^{-1}[L] = \delta^{t,T^{-1}[L_{s}]}$;
\item $T[M]$ is $\delta^{s,T[M]} (= T[\delta^{t,M}])$; and
\item $T^{-1}[N]$ is $\delta^{t,T^{-1}[N]} (= T^{-1}[\delta^{s,N}])$.
\end{enumerate}
\end{definition}

\begin{remark}
Let $\mathbf{A}$ be a $\Sigma$-algebra, $K$, $L\subseteq A$, $s$, $t\in S$, $M\subseteq A_{t}$, $N\subseteq A_{s}$, and $T\in \mathrm{Tl}_{t}(\mathbf{A})_{s}$. Then $T[M]\subseteq \delta^{s,N}$ if and only if $\delta^{t,M}\subseteq T^{-1}[N]$. Moreover, $T[K]\subseteq \delta^{s,L_{s}}$ if and only if $\delta^{t,K_{t}}\subseteq T^{-1}[L]$. Had the operators $T[\cdot]$ and $T^{-1}[\cdot]$ been defined (for $u\in S-\{s,t\}$) differently, then the suitably modified counterparts of the just stated connections would not be met.
\end{remark}

As announced above, we next provide, by using the notions of elementary trans\-la\-tion and of translation, two characterizations of the congruences on a $\Sigma$-algebra. This shows, in particular, the significance of the notions of elementary translation and of translation. We note that in~\cite{m76}, on p.~199, it was announced without proof a proposition similar to that set out below.

\begin{proposition}\label{CharacCong}
Let $\mathbf{A}$ be a $\Sigma$-algebra and $\Phi$ an $S$-sorted equivalence on $A$. Then the following conditions are equivalent:
\begin{enumerate}
\item $\Phi$ is a congruence on $\mathbf{A}$.
\item $\Phi$ is closed under the elementary translations on
$\mathbf{A}$, i.e., for every every $t$, $s\in S$, every $x$, $y\in A_{t}$, and every $T\in \mathrm{Etl}_{t}(\mathbf{A})_{s}$, if $(x,y)\in \Phi_{t}$, then $(T(x),T(y))\in \Phi_{s}$.
\item $\Phi$ is closed under the translations on $\mathbf{A}$, i.e., for every every $t$, $s\in S$, every $x$, $y\in A_{t}$, and every $T\in \mathrm{Tl}_{t}(\mathbf{A})_{s}$, if $(x,y)\in \Phi_{t}$, then $(T(x),T(y))\in \Phi_{s}$.
\end{enumerate}
\end{proposition}

\begin{proof}
Let us first prove that (1) and (2) are equivalent.

Let us suppose that $\Phi$ is a congruence on $\mathbf{A}$. We want to show that $\Phi$ is closed under the elementary translations on $\mathbf{A}$. Let $t$ and $s$ be elements of $S$ and $T$ a $t$-elementary translation of sort $s$ for $\mathbf{A}$. Then $T\colon A_{t}\mor A_{s}$ and there is a word $w\in S^{\star}-\{\lambda\}$, an $i\in \lvert w \rvert$, a $\sigma\in
\Sigma_{w,s}$, a family $(a_{j})_{j\in i}\in\prod_{j\in i}A_{w_{j}}$, and a family $(a_{k})_{k\in \lvert w \rvert-(i+1)}
\in\prod_{k\in \lvert w \rvert-(i+1)}A_{w_{k}}$ such that $w_{i} = t$ and, for every $z\in A_{t}$, $T(z) =
F_{\sigma}(a_{0},\ldots,a_{i-1},z,a_{i+1},\ldots,a_{\lvert w \rvert-1})$. Let $x$ and $y$ be elements of $A_{t}$ such that $(x,y)\in \Phi_{t}$. Since, for every $j\in i$, $(a_{j},a_{j})\in \Phi_{w_{j}}$, for every $k\in \lvert w \rvert-(i+1)$, $(a_{k},a_{k})\in \Phi_{w_{k}}$, and, in addition, $(x,y)\in \Phi_{t} = \Phi_{w_{i}}$, then $(T(x),T(y))\in \Phi_{s}$.

Reciprocally, let us suppose that, for every $t$, $s\in S$, every $x$, $y\in A_{t}$, and every $T\in \mathrm{Etl}_{t}(\mathbf{A})_{s}$, if $(x,y)\in \Phi_{t}$, then $(T(x),T(y))\in \Phi_{s}$. We want to show that  $\Phi$ is a congruence on $\mathbf{A}$. Let $(w,u)\in (S^{\star}-\{\lambda\})\times S$, $\sigma\colon w\mor u$,
and $a = (a_{i})_{i\in\lvert w \rvert}$, $b = (b_{i})_{i\in\lvert w \rvert}\in A_{w}$ such that, for every $i\in \lvert w \rvert$ we have that $(a_{i}, b_{i})\in \Phi_{w_{i}}$. We now define, for every $i\in\lvert w \rvert$, $T_{i}$, the $w_{i}$-elementary translation of sort $u$ for $\mathbf{A}$, as the mapping from $A_{w_{i}}$ to $A_{u}$ which sends $x\in A_{w_{i}}$ to $F_{\sigma}(b_{0},\ldots,b_{i-1},x,a_{i+1},\ldots,a_{\lvert w \rvert-1}) \in A_{u}$. Then $F_{\sigma}(a_{0},\ldots,a_{\lvert w \rvert-1}) = T_{0}(a_{0})$ and $(T_{0}(a_{0}),T_{0}(b_{0}))\in \Phi_{w_{0}}$. But $T_{0}(b_{0}) = T_{1}(a_{1})$ and $(T_{1}(a_{1}),T_{1}(b_{1}))\in \Phi_{w_{1}}$. By proceeding in the same way we, finally, come to $T_{\lvert w \rvert-2}(b_{\lvert w \rvert-2}) = T_{\lvert w \rvert-1}(a_{\lvert w \rvert-1})$, $(T_{\lvert w \rvert-1}(a_{\lvert w \rvert-1}), T_{\lvert w \rvert-1}(b_{\lvert w \rvert-1}))\in \Phi_{w_{\lvert w \rvert-1}}$, and $T_{\lvert w \rvert-1}(b_{\lvert w \rvert-1})\! =\! F_{\sigma}(b_{0},\ldots,b_{\lvert w \rvert-1})$. Therefore $(F_{\sigma}(a), F_{\sigma}(b))\in \Phi_{u}$.

We shall now proceed to verify that (2) and (3) are equivalent.

Since every elementary translations on $\mathbf{A}$ is a translation on $\mathbf{A}$, it is obvious that if $\Phi$ is closed under the translations on $\mathbf{A}$, then $\Phi$ is closed under the elementary translations on $\mathbf{A}$.

Reciprocally, let us suppose that $\Phi$ is closed under the elementary translations on $\mathbf{A}$. We want to show that $\Phi$ is closed under the translations on $\mathbf{A}$. Let $t$ and $s$ be elements of $S$, $x$, $y$ elements of $A_{t}$, $T\in \mathrm{Tl}_{t}(\mathbf{A})_{s}$, and let us suppose that $(x,y)\in \Phi_{t}$. Then there is an $n\in \mathbb{N}-1$, a word $(s_{j})_{j\in n+1}\in S^{n+1}$, and a family $(T_{j})_{j\in n}$ such that $s_{0} = t$, $s_{n} = s$, $T_{0}\in \mathrm{Etl}_{t}(\mathbf{A})_{s_{1}}$, $T_{1}\in \mathrm{Etl}_{s_{1}}(\mathbf{A})_{s_{2}}$, \ldots, $T_{n-1}\in \mathrm{Etl}_{s_{n-1}}(\mathbf{A})_{s}$ and $T = T_{n-1}\circ\cdots\circ T_{0}$. Then, from $(x,y)\in \Phi_{t} = \Phi_{s_{0}}$, we infer that $(T_{0}(x),T_{0}(y))\in \Phi_{s_{1}}$. By proceeding in the same way we, finally, come to $(T_{n-1}(\ldots(T_{0}(x))\ldots),T_{n-1}(\ldots(T_{0}(y))\ldots))\in \Phi_{s} = \Phi_{s_{n}}$, i.e., to $(T(x),T(y))\in \Phi_{s}$.
\end{proof}

Our next aim is to assign, in a functional way, to every subset $L$ of the underlying $S$-sorted set $A$ of a $\Sigma$-algebra $\mathbf{A}$ the so-called congruence on $\mathbf{A}$ cogenerated by $L$, and to investigate its properties. But before doing it we need to recall the following result. For every $S$-sorted set $A$, the mapping from $\mathrm{Sub}(A)$ to $\mathrm{Hom}(A,(2)_{s\in S})$ that assigns to $L\in \mathrm{Sub}(A)$ precisely $\mathrm{ch}_{L}$, the character of $L$, i.e., the $S$-sorted mapping from $A$ to $(2)_{s\in S}$ whose $s$-th coordinate, for $s\in S$, is $\mathrm{ch}_{L_{s}}$, the characteristic mapping of $L_{s}$, is a natural isomorphism, in other words, $(2)_{s\in S}$ together with $\top^{S} = (\top)_{s\in S}\colon 1\mor (2)_{s\in S}$, where, for every $s\in S$, $\top$ is the mapping from $1$ to $2$ that sends $0$ to $1$, is a subobject classifier for $\mathbf{Set}^{S}$. Then, given a $\Sigma$-algebra  $\mathbf{A}$ and a subset $L$ of $A$, we associate to $L$ the $S$-sorted equivalence $\mathrm{Ker}(\mathrm{ch}_{L})$ on $A$ determined by $\mathrm{ch}_{L}$. So, for every $s\in S$, we have that:
$$
  \mathrm{Ker}(\mathrm{ch}_{L})_{s} = \mathrm{Ker}(\mathrm{ch}_{L_{s}}) = \{\,(x,y)\in A^{2}_{s}\mid x\in
  L_{s}\leftrightarrow y\in L_{s}\,\}.
$$

We next prove that there exists a congruence $\Omega^{\mathbf{A}}(L)$ on $\mathbf{A}$
%, the congruence on $\mathbf{A}$ cogenerated by the $S$-sorted equivalence $\mathrm{Ker}(\mathrm{ch}_{L})$ on $A$,
that saturates $L$, i.e., which is such that  $\Omega^{\mathbf{A}}(L)\subseteq \mathrm{Ker}(\mathrm{ch}_{L})$, and is the largest congruence on $\mathbf{A}$ having such a property.

\begin{definition}
Let $\mathbf{A}$ be a $\Sigma$-algebra and $L\subseteq A$. Then $\Omega^{\mathbf{A}}(L)$ is the binary relation on $A$ defined, for every $t\in S$, as follows:
$$
\Omega^{\mathbf{A}}(L)_{t} = \biggl\{ (x,y)\in A^{2}_{t}\biggm|
\begin{gathered}
\forall\, s\in S\,\,\forall\, T\in \mathrm{Tl}_{t}(\mathbf{A})_{s}\,
\\[-3pt]
(T(x)\in L_{s}\leftrightarrow T(y)\in L_{s})
\end{gathered}
\biggr\}.
$$
\end{definition}

\begin{proposition}\label{CharacCogenCong}
Let $\mathbf{A}$ be a $\Sigma$-algebra and $L\subseteq A$. Then
\begin{enumerate}
\item $\Omega^{\mathbf{A}}(L)$ is a congruence on $\mathbf{A}$.

\item $\Omega^{\mathbf{A}}(L)\subseteq \mathrm{Ker}(\mathrm{ch}_{L})$.

\item For every congruence $\Phi$ on $\mathbf{A}$, if $\Phi\subseteq \mathrm{Ker}(\mathrm{ch}_{L})$, then $\Phi\subseteq \Omega^{\mathbf{A}}(L)$.
\end{enumerate}
In other words, $\Omega^{\mathbf{A}}(L)$ is the greatest congruence on $\mathbf{A}$ which saturates $L$.
\end{proposition}

\begin{proof}
To prove (1) it suffices to take into account Proposition~\ref{CharacCong}. To prove (2), given $t\in S$ and $(x,y)\in \Omega^{\mathbf{A}}(L)_{t}$, it suffices to consider $\mathrm{id}_{A_{t}}\in \mathrm{Tl}_{t}(\mathbf{A})_{t}$, to conclude that $x\in L_{t}$ if and only if $y\in L_{t}$, i.e., that $(x,y)\in \mathrm{Ker}(\mathrm{ch}_{L})_{t}$. We now proceed to prove (3).
Let $\Phi$ be a congruence on $\mathbf{A}$ such that $\Phi\subseteq \mathrm{Ker}(\mathrm{ch}_{L})$, i.e., such that, for every $s\in S$ and every $x$, $y\in A_{s}$, if $(x,y)\in \Phi_{s}$, then $x\in L_{s}$ if and only if $y\in L_{s}$. We want to show that, for every $t\in S$, $\Phi_{t}\subseteq \Omega^{\mathbf{A}}(L)_{t}$. Let $t$ be an element of $S$ and $(x,y)\in \Phi_{t}$. Then, since $\Phi$ is a congruence on $\mathbf{A}$, for every $s\in S$ and every $T\in \mathrm{Tl}_{t}(\mathbf{A})_{s}$, we have that $(T(x),T(y))\in \Phi_{s}$. Hence, by the hypothesis on $\Phi$, $T(x)\in L_{s}$ if and only if $T(y)\in L_{s}$. Therefore $\Phi\subseteq \Omega^{\mathbf{A}}(L)$.
\end{proof}

\begin{definition}
Let $\mathbf{A}$ be a $\Sigma$-algebra and $L\subseteq A$. Then we call $\Omega^{\mathbf{A}}(L)$ the congruence on $\mathbf{A}$ \emph{cogenerated} by $L$, or the \emph{syntactic} congruence on $\mathbf{A}$ determined by $L$
(which is shorthand for ``the congruence on $\mathbf{A}$ cogenerated by the equivalence $\mathrm{Ker}(\mathrm{ch}_{L})$ on $\mathbf{A}$ canonically associated to $L$'').
\end{definition}

\begin{remark}
For every subset $L$ of the underlying $S$-sorted set of a $\Sigma$-algebra $\mathbf{A}$, $\Omega^{\mathbf{A}}(L)$ can be viewed as the value at $L$ of a mapping $\Omega^{\mathbf{A}}$ from $\mathrm{Sub}(A)$ to $\mathrm{Cgr}(\mathbf{A})$. We will call $\Omega^{\mathbf{A}}$ the \emph{congruence cogenerating operator for} $\mathbf{A}$ (with regard to subsets of $A$).
\end{remark}

%\begin{remark}
%Let $L$ be a subset of a semigroup (or monoid). Then the syntactic (or principal) congruence of $L$ (also called the two-sided principal congruence of $L$) falls under the notion of congruence cogenerated by $L$.
%\end{remark}

With regard to the congruence cogenerated by an equivalence it is worthwhile to quote what B\"{u}chi, in~\cite{{rb89}}, on p.~113, wrote: ``The notion of induced [$\equiv$ cogenerated, \emph{we add}] congruence therefore is clearly relevant to automata theory. For some reason it seems to have escaped the attention of algebraists. In contrast, its mate, the generated congruence, is widely used in algebra.''

It may be worth reminding the reader that in the theory of formal languages, a congruence of the type $\Omega^{\mathbf{A}^{\star}}(L)$, where $L$ is a subset of the underlying set of a free monoid $\mathbf{A}^{\star}$ on an alphabet $A$, is called the syntactic congruence determined by $L$ (or the two-sided principal congruence of $L$). According to Lallement (in~\cite{Lall79}, on p. 175): ``the fact that the principal equivalence of $L$ is the largest congruence for which $L$ is a union of equivalence classes is due to Teissier in~\cite{Tei51} [for demi-groups, i.e., sets with an associative binary operation, \emph{we add}].'' To this we append that these congruences were defined by Sch\"{u}tzenberger (in~\cite{Sch55}, on p.~10) for monoids (he speaks of: ``demi-groupes contenant un \'{e}l\'{e}ment neutre''). It should also be pointed out that, for a \emph{single-sorted} algebra $\mathbf{A}$ and for an equivalence relation $\Phi$ on $A$, in~\cite{gree52}, on p.~65, Green defined the congruence on $\mathbf{A}$ \emph{analysed} ($\equiv$ cogenerated, \emph{we add}) by $\Phi$ as the join ($\equiv$ supremum) of the set of all congruences $\Psi$ on $\mathbf{A}$ contained in $\Phi$ and asserted that it is the greatest congruence on $\mathbf{A}$ contained in $\Phi$; and  in~\cite{sl74}, on pp.~32--33, S{\l}omi\'{n}ski proved that there exists the greatest congruence on $\mathbf{A}$ contained in $\Phi$. In this respect it should be noted that in~\cite{ja90}, and for \emph{single-sorted} algebras, Almeida addressed, in particular, the issue of the definition and basic properties of the syntactic congruence.

The theorem of Green and S{\l}omi\'{n}ski is also valid for the many-sorted case, as shown below.

\begin{proposition}\label{CogenCong}
Let $\mathbf{A}$ be a $\Sigma$-algebra and $\Phi$ an $S$-sorted equivalence on $A$. Then there exists a congruence $\widetilde{\Omega}^{\mathbf{A}}(\Phi)$ on $\mathbf{A}$ such that $\widetilde{\Omega}^{\mathbf{A}}(\Phi)\subseteq \Phi$ and, for every congruence $\Psi$ on $\mathbf{A}$, if $\Psi\subseteq \Phi$, then $\Psi\subseteq \widetilde{\Omega}^{\mathbf{A}}(\Phi)$. We will call $\widetilde{\Omega}^{\mathbf{A}}(\Phi)$ \emph{the congruence on} $\mathbf{A}$ \emph{cogenerated by} $\Phi$.
\end{proposition}

\begin{proof}
Since: (1) the join in $\mathbf{Eqv}(A)$ of a family $(\Psi^{i})_{i\in I}$ of congruences on $\mathbf{A}$ is the $S$-sorted equivalence $\Psi$ on $A$ whose $s$-th coordinate, $\Psi_{s}$, for $s\in S$, is:
$$
\Psi_{s} =
\biggl\{ (a,b)\in {A_{s}^{2}}\biggm|
\begin{gathered}
\exists n\geq 1\, \exists x\in {A_{s}^{n+1}}( x_{0}=a \And x_{n}=b
\And
\\[-3pt]
\forall p\in n ( (x_{p},x_{p+1})\in\textstyle\bigcup_{i\in
I}\Psi^{i}_{s}))
\end{gathered}
\biggr\};
$$
(2) for the join $\Psi$ in $\mathbf{Eqv}(A)$ of a family $(\Psi^{i})_{i\in I}$ of congruences on $\mathbf{A}$, we have that, for every $(w,s)\in (S^{\star}-\{\lambda\})\times S$, every $\sigma\in\Sigma_{w,s}$, every $j\in\bb{w}$, every $a$, $b\in A_{w_{j}}$, every $(c_{i})_{i\in j}\in\prod_{i\in j}A_{w_{i}}$, and every $(c_{k})_{k\in \lvert w \rvert-(j+1)}
\in\prod_{k\in \lvert w \rvert-(j+1)}A_{w_{k}}$, if $(a,b)\in\Psi_{w_{j}}$, then
$$
(F_{\sigma}(c_{0},\ldots,c_{j-1},a,c_{j+1},\ldots,c_{\bb{w}-1}),
F_{\sigma}(c_{0},\ldots,c_{j-1},b,c_{j+1},\ldots,c_{\bb{w}-1}))\in \Psi_{s};
$$
(3) the join in $\mathbf{Eqv}(A)$ of a family $(\Psi^{i})_{i\in I}$ of congruences on $\mathbf{A}$ is a congruence on $\mathbf{A}$ (this follows from (2)); and (4) $\Delta_{\mathbf{A}}$ is a congruence on $\mathbf{A}$ contained in $\Phi$, it suffices to take as $\widetilde{\Omega}^{\mathbf{A}}(\Phi)$ the join of $\{\Psi\in \mathrm{Cgr}(\mathbf{A})\mid \Psi\subseteq\Phi\}$ in $\mathbf{Eqv}(A)$.
%It suffices to take into account the following facts
\end{proof}

%The following picture illustrates the position of the congruence $\Omega^{\mathbf{A}}(\Phi)$ in the lattice $\mathbf{Cgr}(\mathbf{A})$.
%$$
%\xymatrix@R=2.5pc@C=2.5pc{
%{} &
%*[o]{\circ} \save[]+<0pt,10pt>*{\nabla_{\mathbf{A}}}\restore
%\ar@{-}@/^1.0pc/[rd]
%\ar@{-}@/_2pc/[dd]
%&
%{}  \\
%{} & {} &
%*[o]{\circ} \save[]+<20pt,0pt>*{\Omega^{\mathbf{A}}(\Phi)}\restore
%\ar@{-}@/_1pc/[ld]
%\ar@{-}@/^1pc/[ld]
%\\
%%{}  & {} & {} \\
%{} &
%*[o]{\circ} \save[]+<0pt,-10pt>*{\Delta_{\mathbf{A}}}\restore
%& {}
%}
%$$

\begin{remark}
For every $S$-sorted equivalence $\Phi$ on the underlying $S$-sorted set of a $\Sigma$-algebra $\mathbf{A}$, $\widetilde{\Omega}^{\mathbf{A}}(\Phi)$ can be viewed as the value at $\Phi$ of a mapping $\widetilde{\Omega}^{\mathbf{A}}$ from $\mathrm{Eqv}(A)$ to $\mathrm{Cgr}(\mathbf{A})$. We will call $\widetilde{\Omega}^{\mathbf{A}}$ the \emph{congruence cogenerating operator} for $\mathbf{A}$ (with regard to $S$-sorted equivalences on $A$).
\end{remark}

\begin{remark}
For every subset $L$ of the underlying $S$-sorted set of a $\Sigma$-algebra $\mathbf{A}$, $\Omega^{\mathbf{A}}(L) = \widetilde{\Omega}^{\mathbf{A}}(\mathrm{Ker}(\mathrm{ch}_{L}))$. In other words, the mapping $\Omega^{\mathbf{A}}$ from $\mathrm{Sub}(A)$ to $\mathrm{Cgr}(\mathbf{A})$ is the composition of the mapping from $\mathrm{Sub}(A)$ to $\mathrm{Eqv}(A)$ that sends $L$ in $\mathrm{Sub}(A)$ to $\mathrm{Ker}(\mathrm{ch}_{L})$ in $\mathrm{Eqv}(A)$ (which in general is neither injecive, nor surjective, nor isotone) and the mapping $\widetilde{\Omega}^{\mathbf{A}}$ from $\mathrm{Eqv}(A)$ to $\mathrm{Cgr}(\mathbf{A})$.
\end{remark}

\begin{remark}
It is possible to provide a proof of Proposition~\ref{CogenCong} along the lines of Proposition~\ref{CharacCogenCong}. Let $\Phi$ be an $S$-sorted equivalence on the underlying $S$-sorted set of a $\Sigma$-algebra $\mathbf{A}$. Then the $S$-sorted binary relation $\widetilde{\Omega}^{\mathbf{A}}(\Phi)$ on $A$ defined, for every $t\in S$, as follows:
$$
\widetilde{\Omega}^{\mathbf{A}}(\Phi)_{t} = \biggl\{ (x,y)\in A^{2}_{t}\biggm|
\begin{gathered}
\forall\, s\in S\,\,\forall\, T\in \mathrm{Tl}_{t}(\mathbf{A})_{s}\,
\\[-3pt]
(T(x),T(y))\in \Phi_{s})
\end{gathered}
\biggr\},
$$
is the greatest congruence on $\mathbf{A}$ contained in $\Phi$.
\end{remark}

%\begin{remark}
%Proposition~\ref{CharacCogenCong} is, obviously, a corollary of Proposition~\ref{CogenCong}. In spite of this, we have decided to explicitly state and prove it because of its constructive character and since it is the particular case we will make use of in this work.
%\end{remark}

We next provide, for a $\Sigma$-algebra $\mathbf{A}$, some basic properties of the congruence cogenerating operator $\widetilde{\Omega}^{\mathbf{A}}$.

\begin{proposition}
Let $\mathbf{A}$ be a $\Sigma$-algebra. Then $\widetilde{\Omega}^{\mathbf{A}}$, considered as an endomapping of $\mathrm{Eqv}(A)$, is a kernel ($\equiv$ interior) operator, i.e., it is contractive ($\equiv$ deflationary), isotone, and idempotent. Moreover, $\widetilde{\Omega}^{\mathbf{A}}(\Delta_{\mathbf{A}}) = \Delta_{\mathbf{A}}$, $\widetilde{\Omega}^{\mathbf{A}}(\nabla_{\mathbf{A}}) = \nabla_{\mathbf{A}}$ and, for every nonempty set $I$ in $\boldsymbol{\mathcal{U}}$ and every $(\Phi^{i})_{i\in I}\in \mathrm{Eqv}(A)^{I}$, $\widetilde{\Omega}^{\mathbf{A}}(\bigcap_{i\in I}\Phi^{i}) = \bigcap_{i\in I}\widetilde{\Omega}^{\mathbf{A}}(\Phi^{i})$.
\end{proposition}

\begin{proposition}
Let $\mathbf{A}$ be a $\Sigma$-algebra, $\Phi$ an $S$-sorted equivalence on $A$, $t$, $s\in S$, and $T\in\mathrm{Tl}_{t}(\mathbf{A})_{s}$. Then $\Omega^{\mathbf{A}}(\Phi)\subseteq \Omega^{\mathbf{A}}((T\times T)^{-1}[\Phi])$, where $(T\times T)^{-1}[\Phi]$ stands for the $S$-sorted equivalence on $A$ defined, for every $u\in S$, as follows:
$$
\Omega^{\mathbf{A}}((T\times T)^{-1}[\Phi])_{u} =
\begin{cases}
(T\times T)^{-1}[\Phi_{s}] \text{, if }u=t; \\
\nabla_{A_{u}} \text{, otherwise.} \\
\end{cases}
$$
%$\delta^{(t|(T\times T)^{-1}[\Phi_{s}])}_{t}$ is $(T\times T)^{-1}[\Phi_{s}]$ and $\delta^{(t|(T\times T)^{-1}[\Phi_{s}])}_{u}$ is $\Phi_{u}$, if $u\neq t$.
\end{proposition}

%\begin{proposition}
%Let $\mathbf{A}$ be a $\Sigma$-algebra, $\Phi$ an $S$-sorted equivalence on $A$, $t$, $s\in S$, and $T\in\mathrm{Tl}_{t}(\mathbf{A})_{s}$. Then $\Omega^{\mathbf{A}}(\Phi)\subseteq \Omega^{\mathbf{A}}(\Phi^{(t|(T\times T)^{-1}[\Phi_{s}])})$, where $\Phi^{(t|(T\times T)^{-1}[\Phi_{s}])}$ stands for the $S$-sorted equivalence on $A$ defined, for every $u\in S$, as follows:
%$$
%\Phi^{(t|(T\times T)^{-1}[\Phi_{s}])}_{u} =
%\begin{cases}
%(T\times T)^{-1}[\Phi_{s}] \text{, if }u=t; \\
%\Phi_{u} \text{, otherwise.} \\
%\end{cases}
%$$
%%$\delta^{(t|(T\times T)^{-1}[\Phi_{s}])}_{t}$ is $(T\times T)^{-1}[\Phi_{s}]$ and $\delta^{(t|(T\times T)^{-1}[\Phi_{s}])}_{u}$ is $\Phi_{u}$, if $u\neq t$.
%\end{proposition}

\begin{proposition}\label{TAntiHom}
Let $f$ be a homomorphism from $\mathbf{A}$ to $\mathbf{B}$ and $\Upsilon$ an $S$-sorted equivalence on $B$. Then
$(f\times f)^{-1}[\widetilde{\Omega}^{\mathbf{B}}(\Upsilon)]\subseteq \widetilde{\Omega}^{\mathbf{A}}((f\times f)^{-1}[\Upsilon])$. Moreover, if $f$ is an epimorphism, then $(f\times f)^{-1}[\widetilde{\Omega}^{\mathbf{B}}(M)] = \widetilde{\Omega}^{\mathbf{A}}((f\times f)^{-1}[\Upsilon])$.
\end{proposition}

\begin{remark}
For every $\Sigma$-algebra $\mathbf{A}$, $\widetilde{\Omega}^{\mathbf{A}}$ can be regarded as the component at $\mathbf{A}$ of a natural transformation $\widetilde{\Omega}$ between two contravariant functors from a suitable category of $\Sigma$-algebras to the category $\mathbf{Set}$. In fact, let $\mathbf{Alg}(\Sigma)_{\mathrm{epi}}$ be the category whose objects are the $\Sigma$-algebras and whose morphisms are the epimorphisms between $\Sigma$-algebras. Then we have, on the one hand, the functor $\mathrm{Eqv}$ from $\mathbf{Alg}(\Sigma)^{\mathrm{op}}_{\mathrm{epi}}$, the dual of $\mathbf{Alg}(\Sigma)_{\mathrm{epi}}$, to $\mathbf{Set}$ which assigns to a $\Sigma$-algebra $\mathbf{A}$ the set $\mathrm{Eqv}(A)$, and to an epimorphism $f\colon \mathbf{A}\mor \mathbf{B}$ the mapping $(f\times f)^{-1}[\cdot]$ from $\mathrm{Eqv}(B)$ to  $\mathrm{Eqv}(A)$, and, on the other hand, the functor $\mathrm{Cgr}$ from $\mathbf{Alg}(\Sigma)^{\mathrm{op}}_{\mathrm{epi}}$ to $\mathbf{Set}$ which assigns to a $\Sigma$-algebra $\mathbf{A}$ the set $\mathrm{Cgr}(\mathbf{A})$, and to an epimorphism $f\colon \mathbf{A}\mor \mathbf{B}$ the mapping $(f\times f)^{-1}[\cdot]$ from $\mathrm{Cgr}(\mathbf{B})$ to $\mathrm{Cgr}(\mathbf{A})$. Then the mapping $\widetilde{\Omega}$ from $\mathrm{Alg}(\Sigma)$, the set of objects of $\mathbf{Alg}(\Sigma)$, to $\mathrm{Mor}(\mathbf{Set})$, the set of morphisms of $\mathbf{Set}$, which assigns to a $\Sigma$-algebra $\mathbf{A}$ the mapping $\widetilde{\Omega}^{\mathbf{A}}$ from $\mathrm{Eqv}(A)$ to $\mathrm{Cgr}(\mathbf{A})$ is a natural transformation from $\mathrm{Eqv}$ to $\mathrm{Cgr}$, because, for every epimorphism $f\colon \mathbf{A}\mor \mathbf{B}$, we have that $(f\times f)^{-1}[\cdot]\circ \widetilde{\Omega}^{\mathbf{B}} = \widetilde{\Omega}^{\mathbf{A}}\circ (f\times f)^{-1}[\cdot]$,
%the following diagram
%$$\xymatrix{
%\mathrm{Sub}(A) \ar[r]^-{\Omega^{\mathbf{A}}} & \mathrm{Cgr}(\mathbf{A})\\
%\mathrm{Sub}(B)\ar[u]^-{f^{-1}[\cdot]}\ar[r]_-{\Omega^{\mathbf{B}}} & \mathrm{Cgr}(\mathbf{B})\ar[u]_-{(f\times f)^{-1}[\cdot]}
%   }
%$$
%commutes,
i.e., for every $S$-sorted equivalence $\Upsilon$ on $B$, $(f\times f)^{-1}[\widetilde{\Omega}^{\mathbf{B}}(\Upsilon)] = \widetilde{\Omega}^{\mathbf{A}}((f\times f)^{-1}[\Upsilon])$.
%(the inclusion $(f\times f)^{-1}[\Omega^{\mathbf{B}}(M)] \subseteq \Omega^{\mathbf{A}}(f^{-1}[M])$ is evident and the converse inclusion follows from Proposition~\ref{TlandHom}).
%We summarize the above by means of the following diagram:
%$$
%\xymatrix@C=20ex@R=8ex{
%\mathbf{Alg}(\Sigma)^{\mathrm{op}}_{\mathrm{epi}}
%\ar@/^20pt/[r]^{\mathrm{P}^{-}}="f"
%\ar@/^-20pt/[r]_{\mathrm{Cgr}}="g" &
%\mathbf{Set}
%\ar @{} "f";"g" |{\dir{=>}}^{\,\Omega}
%}
%$$
\end{remark}

We next gather together, for a $\Sigma$-algebra $\mathbf{A}$, some basic properties of the con\-gru\-ence cogenerating operator $\Omega^{\mathbf{A}}$.

\begin{proposition}\label{CharacSatCCog and CharacsSatCCog}\label{PSat and sPSat}
Let $\mathbf{A}$ be a $\Sigma$-algebra, $L$ a subset of $A$, and $\Phi\in\mathrm{Cgr}(\mathbf{A})$. Then $L\in \Phi\text{-}\mathrm{Sat}(A)$, i.e., $L = [L]^{\Phi}$, if and only if $\Phi\subseteq\Omega^{\mathbf{A}}(L)$. Moreover, for $s\in S$ and $L\subseteq A_{s}$, we have that $L = [L]^{\Phi_{s}}$ if and only if $\Phi_{s}\subseteq\Omega^{\mathbf{A}}(\delta^{s,L})_{s}$.
\end{proposition}

\begin{proof}
Let us suppose that $L = [L]^{\Phi}$. Then, since $\Omega^{\mathbf{A}}(L)$ is the greatest congruence on $\mathbf{A}$ such that $L = [L]^{\Omega^{\mathbf{A}}(L)}$, we have that $\Phi\subseteq\Omega^{\mathbf{A}}(L)$.

Reciprocally, let us suppose that $\Phi\subseteq\Omega^{\mathbf{A}}(L)$ then, by Corollary~\ref{IncSat and sIncSat}, we have that  $\Omega^{\mathbf{A}}(L)\text{-}\mathrm{Sat}(A)\subseteq\Phi\text{-}\mathrm{Sat}(A)$. Since $L$ belongs to $\Omega^{\mathbf{A}}(L)\text{-}\mathrm{Sat}(A)$, we conclude that $L\in \Phi\text{-}\mathrm{Sat}(A)$.
%Let us suppose that $L = [L]^{\Phi}$. Then, since $\Omega^{\mathbf{A}}(L)$ is the greatest congruence on $\mathbf{A}$ such that $L = [L]^{\Omega^{\mathbf{A}}(L)}$, we have that $\Phi\subseteq\Omega^{\mathbf{A}}(L)$.
%
%Reciprocally, let us suppose that $\Phi\subseteq\Omega^{\mathbf{A}}(L)$. Let $t$ be a sort in $S$ and $y\in [L]^{\Phi}_{t} = \bigcup_{x\in L_{t}}[x]_{\Phi_{t}}$. Then there exists an $x\in L_{t}$ such that $y\in [x]_{\Phi_{t}}$. Hence $(x,y)\in \Phi_{t}$. Therefore $(x,y)\in \Omega^{\mathbf{A}}(L)_{t}$. Hence, for every $s\in S$ and every $T\in \mathrm{Tl}_{t}(\mathbf{A})_{s}$,
%$T(x)\in L_{s}$ if and only if $T(y)\in L_{s}$. Thus, for $s = t$ and $T = \mathrm{id}_{A_{t}}$, $x\in L_{t}$ if and only if $y\in L_{t}$. Consequently $y\in L_{t}$. This proves that $[L]^{\Phi}_{t}\subseteq L_{t}$. So $L\in \Phi\text{-}\mathrm{Sat}(A)$.
\end{proof}

%\begin{proposition}\label{RepCongInterCCogKroneckerDelta}
%Let $\mathbf{A}$ be a $\Sigma$-algebra and $\Phi\in\mathrm{Cgr}(\mathbf{A})$. Then we have that
%$$
%\textstyle
%\Phi = \bigcap\{\Omega^{\mathbf{A}}(\delta^{s,[a]_{\Phi_{s}}})\mid s\in S \And a\in A_{s}\}.
%$$
%\end{proposition}

%\begin{proof}
%It is straightforward to verify that, for every $s\in S$ and every $a\in A_{s}$, $\delta^{s,[a]_{\Phi_{s}}}$ is $\Phi$-saturated. Hence $\Phi \subseteq \bigcap\{\Omega^{\mathbf{A}}(\delta^{s,[a]_{\Phi_{s}}})\mid s\in S \And a\in A_{s}\}$.
%
%Reciprocally, let $s$ be an element of $S$ and $a$, $b\in A_{s}$. If $(a,b)\not\in \Phi_{s}$, then $(a,b)\not\in \mathrm{Ker}(\mathrm{ch}^{\delta^{s,[a]_{\Phi_{s}}}})_{s}$. Hence $(a,b)\not\in \Omega^{\mathbf{A}}(\delta^{s,[a]_{\Phi_{s}}})_{s}$. Therefore we have that $\bigcap\{\Omega^{\mathbf{A}}(\delta^{s,[a]_{\Phi_{s}}})\mid s\in S \! \And\! a\in A_{s}\}\subseteq \Phi$.
%\end{proof}

%\begin{remark}
%Let $\mathbf{A}$ be a $\Sigma$-algebra. Then we have that
%$$
%\textstyle
%\Delta_{\mathbf{A}} = \bigcap\{\Omega^{\mathbf{A}}(\delta^{s,a})\mid s\in S\! \And\! a\in A_{s}\}.
%$$
%\end{remark}

\begin{proposition}\label{Compl}
Let $\mathbf{A}$ be a $\Sigma$-algebra and $L$ a subset of $A$. Then it happens that $\Omega^{\mathbf{A}}(L) = \Omega^{\mathbf{A}}(\complement_{A}L)$.
\end{proposition}

\begin{proposition}\label{InterCCog and CCogInter}
Let $\mathbf{A}$ be a $\Sigma$-algebra, $I$ a nonempty set in $\boldsymbol{\mathcal{U}}$, and $(L^{i})_{i\in I}$ an $I$-indexed family of subsets of $A$. Then $\bigcap_{i\in I}\Omega^{\mathbf{A}}(L^{i})\subseteq \Omega^{\mathbf{A}}(\bigcap_{i\in I}L^{i})$, i.e., $\bigcap_{i\in I}\widetilde{\Omega}^{\mathbf{A}}(\mathrm{Ker}(\mathrm{ch}_{L^{i}}))\subseteq \widetilde{\Omega}^{\mathbf{A}}(\mathrm{Ker}(\mathrm{ch}_{\bigcap_{i\in I}L^{i}}))$.
\end{proposition}

\begin{proof}
To prove that $\bigcap_{i\in I}\Omega^{\mathbf{A}}(L^{i}) \subseteq \Omega^{\mathbf{A}}(\bigcap_{i\in I}L^{i})$ it suffices, by part (3) of Propo\-si\-tion~\ref{CharacCogenCong}, to verify that $\bigcap_{i\in I}\Omega^{\mathbf{A}}(L^{i}) \subseteq \mathrm{Ker}(\mathrm{ch}_{\bigcap_{i\in I}L^{i}})$. But  $\bigcap_{i\in I}\Omega^{\mathbf{A}}(L^{i}) \subseteq \bigcap_{i\in I}\mathrm{Ker}(\mathrm{ch}_{L^{i}})$ and, in addition, we have that $\bigcap_{i\in I}\mathrm{Ker}(\mathrm{ch}_{L^{i}}) \subseteq \mathrm{Ker}(\mathrm{ch}_{\bigcap_{i\in I}L^{i}})$. Therefore $\bigcap_{i\in I}\Omega^{\mathbf{A}}(L^{i}) \subseteq \mathrm{Ker}(\mathrm{ch}_{\bigcap_{i\in I}L^{i}})$.
\end{proof}

%\begin{remark}
%Since, in general, $\bigcap_{i\in I}\mathrm{Ker}(\mathrm{ch}_{L^{i}}) \neq \mathrm{Ker}(\mathrm{ch}_{\bigcap_{i\in I}L^{i}})$, we cannot claim that $\bigcap_{i\in I}\Omega^{\mathbf{A}}(L^{i}) = \Omega^{\mathbf{A}}(\bigcap_{i\in I}L^{i})$.
%\end{remark}

%For the following proposition one should bear in mind that,  It must be borne in mind that

With regard to the following proposition, it must be borne in mind that, for every set $C$, $\mathrm{Ker}(\mathrm{ch}_{\varnothing}) = \mathrm{Ker}(\mathrm{ch}_{C}) = \nabla_{C}$.

\begin{proposition}\label{TAntiTrans}
Let $\mathbf{A}$ be a $\Sigma$-algebra, $L$ a subset of $A$, $t$, $s\in S$, and $T\in\mathrm{Tl}_{t}(\mathbf{A})_{s}$. Then $\Omega^{\mathbf{A}}(L)\subseteq \Omega^{\mathbf{A}}(T^{-1}[L])$, i.e., $\widetilde{\Omega}^{\mathbf{A}}(\mathrm{Ker}(\mathrm{ch}_{L}))\subseteq \widetilde{\Omega}^{\mathbf{A}}(\mathrm{Ker}(\mathrm{ch}_{T^{-1}[L]}))$, where, for every $u\in S-\{t\}$,  $\mathrm{Ker}(\mathrm{ch}_{T^{-1}[L]})_{u} = \mathrm{Ker}(\mathrm{ch}_{\varnothing}) = \nabla_{A_{u}}$ and
$$
\mathrm{Ker}(\mathrm{ch}_{T^{-1}[L]})_{t} = \mathrm{Ker}(\mathrm{ch}_{T^{-1}[L_{s}]}) = \{(x,y)\in A_{t}^{2}\mid T(x)\in L_{s}\leftrightarrow T(y)\in L_{s}\}.
$$
\end{proposition}

\begin{proposition}\label{TAntiHom}
Let $f$ be a homomorphism from $\mathbf{A}$ to $\mathbf{B}$ and $M$ a subset of $B$. Then
$(f\times f)^{-1}[\Omega^{\mathbf{B}}(M)]\subseteq \Omega^{\mathbf{A}}(f^{-1}[M])$. Moreover, if $f$ is an epimorphism, then
$(f\times f)^{-1}[\Omega^{\mathbf{B}}(M)] = \Omega^{\mathbf{A}}(f^{-1}[M])$.
\end{proposition}

\begin{remark}
Let $\mathrm{P}^{-}$ be the functor from $\mathbf{Alg}(\Sigma)^{\mathrm{op}}_{\mathrm{epi}}$ to $\mathbf{Set}$ which assigns to a $\Sigma$-algebra $\mathbf{A}$ the set $\mathrm{Sub}(A)$, and to an epimorphism $f\colon \mathbf{A}\mor \mathbf{B}$ the mapping $f^{-1}[\cdot]$ from $\mathrm{Sub}(B)$ to  $\mathrm{Sub}(A)$.  Then the mapping $\Omega$ from $\mathrm{Alg}(\Sigma)$  to $\mathrm{Mor}(\mathbf{Set})$ which assigns to a $\Sigma$-algebra $\mathbf{A}$ the mapping $\Omega^{\mathbf{A}}$ from $\mathrm{Sub}(A)$ to $\mathrm{Cgr}(\mathbf{A})$ is a natural transformation from $\mathrm{P}^{-}$ to $\mathrm{Cgr}$, since, for every epimorphism $f\colon \mathbf{A}\mor \mathbf{B}$, we have that $(f\times f)^{-1}[\cdot]\circ \Omega^{\mathbf{B}} = \Omega^{\mathbf{A}}\circ f^{-1}[\cdot]$,
%the following diagram
%$$\xymatrix{
%\mathrm{Sub}(A) \ar[r]^-{\Omega^{\mathbf{A}}} & \mathrm{Cgr}(\mathbf{A})\\
%\mathrm{Sub}(B)\ar[u]^-{f^{-1}[\cdot]}\ar[r]_-{\Omega^{\mathbf{B}}} & \mathrm{Cgr}(\mathbf{B})\ar[u]_-{(f\times f)^{-1}[\cdot]}
%   }
%$$
%commutes,
i.e., for every $M\subseteq B$, $(f\times f)^{-1}[\Omega^{\mathbf{B}}(M)] = \Omega^{\mathbf{A}}(f^{-1}[M])$.
%(the inclusion $(f\times f)^{-1}[\Omega^{\mathbf{B}}(M)] \subseteq \Omega^{\mathbf{A}}(f^{-1}[M])$ is evident and the converse inclusion follows from Proposition~\ref{TlandHom}).
\end{remark}

We finish this section by reviewing a few aspects of recognizability for subsets of the underlying many-sorted set of an arbitrary many-sorted algebra. But before going any further it is worth noting what G\'{e}cseg and Steinby, in~\cite{GS84}, at the beginning of Chapter 2, wrote: ``\ldots one should note that there are often many ways to generalize from languages [sets of words of the underlying set of a free monoid on a set, \emph{we add}] to forest [sets of terms of the underlying set of a free algebra on a set, \emph{we add}], and a right choice among the alternatives is essential if one wants to gen\-er\-al\-ize the corresponding results, too.'' In this regard, concerning the con\-gru\-ences on a many-sorted algebra---on which, ultimately, the notion of recognizability will be founded---we have two nonequivalent ways of defining the concept of congruence of ``finite'' index on it, depending on the notion of finiteness we choose: The categorial or the non-categorial notion of finiteness.

%If we pick the non-categorial notion of finiteness, i.e., the notion of $S$-finiteness, then we will say that a congruence $\Phi$ on $\mathbf{A}$ is of $S$\nobreakdash-finite index if $A/\Phi$ is $S$-finite, i.e., if, for every $s\in S$, $\mathrm{card}(A_{s}/\Phi_{s})<\aleph_{0}$. If we choose the categorial notion of finiteness, then we will say that $\Phi$ is of finite index if the $S$-sorted set $A/\Phi$ is finite, i.e., if $\mathrm{card}(\mathrm{supp}_{S}(A/\Phi))<\aleph_{0}$ and, for every $s\in \mathrm{supp}_{S}(A/\Phi)$, $\mathrm{card}(\mathrm{supp}_{S}(A_{s}/\Phi_{s}))<\aleph_{0}$.

%We shall now go on to define the notion of congruence of finite index which will be crucial to provide congruence based proofs of the recognizability theorems.

We shall now go on to define both notions of congruence of finite index on a many-sorted algebra.
%However, we do point out that in Section~3 we will only use the categorial one, and this will be crucial to provide congruence based proofs of the recognizability theorems.

\begin{definition}
Let $\mathbf{A}$ be a $\Sigma$-algebra and $\Phi\in\Cgr(\mathbf{A})$. We will say that $\Phi$ is of \emph{finite index}, abbreviated as $\mathrm{fi}$, if  $A/{\Phi}\in \mathrm{Sub}_{\mathrm{f}}(A^{\wp})$, i.e., if $\mathrm{card}(\mathrm{supp}_{S}(A/\Phi))$ is finite and, for every $s\in \mathrm{supp}_{S}(A/\Phi)$, $A_{s}/\Phi_{s}$ is finite. We will denote by $\mathrm{Cgr}_{\mathrm{fi}}(\mathbf{A})$ the set of all congruences on $\mathbf{A}$ of finite index. Moreover, we will say that $\Phi$ is of $S$-\emph{finite index} or of \emph{locally finite index}, abbreviated as $\mathrm{lfi}$, if $A/{\Phi}\in \mathrm{Sub}_{\mathrm{lf}}(A^{\wp})$ i.e., if, for every $s\in S$, $\mathrm{card}(A_{s}/\Phi_{s})<\aleph_{0}$. We will denote by $\mathrm{Cgr}_{\mathrm{lfi}}(\mathbf{A})$ the set of all congruences on $\mathbf{A}$ of locally finite index.
\end{definition}

Let us note that in~\cite{Cou89}, on p.~99, and in~\cite{Cou96}, on p.~30, and for a set of sorts $S$, eventually infinite, and a many-sorted algebra $\mathbf{A}$ such that, for every $s\in S$, $A_{s}\neq\varnothing$, Courcelle says that a  congruence on $\mathbf{A}$ is \emph{locally finite} if it has finitely many classes of each sort. To this we add that he uses such a type of congruence to investigate, for a many-sorted algebra $\mathbf{A}$, subject to satisfy the above condition, and a sort $s\in S$, the recognizable subsets of $A_{s}$. So, if we disregard the condition imposed by Courcelle on the many-sorted algebras, our notion of congruence of locally finite index coincides with Courcelle's notion of locally finite congruence.

\begin{proposition}
Let $\mathbf{A}$ be a $\Sigma$-algebra. Then $\mathrm{Cgr}_{\mathrm{fi}}(\mathbf{A})\neq\varnothing$ if and only if  $\mathrm{supp}_{S}(\mathbf{A})$ is finite. Moreover, $\mathrm{card}(S)<\aleph_{0}$ if and only if, for every $\Sigma$-algebra $\mathbf{A}$, $\mathrm{Cgr}_{\mathrm{fi}}(\mathbf{A})\neq\varnothing$. Therefore, the following conditions are equivalent: (1) for every $\Sigma$-algebra $\mathbf{A}$, $\mathrm{supp}_{S}(\mathbf{A})$ is finite, (2) for every $\Sigma$-algebra $\mathbf{A}$, $\mathrm{Cgr}_{\mathrm{fi}}(\mathbf{A})\neq\varnothing$, and (3) $S$ is  finite.
\end{proposition}

\begin{proof}
Since the first assertion is straightforward, we restrict ourselves to verify the second one. If the set of sorts $S$ is finite, then, for every $\Sigma$-algebra $\mathbf{A}$, $\nabla_{\mathbf{A}}\in \mathrm{Cgr}_{\mathrm{fi}}(\mathbf{A})$. If $S$ is infinite, then the final $\Sigma$-algebra $\mathbf{1}$ is such that $\mathrm{Cgr}_{\mathrm{fi}}(\mathbf{1}) = \varnothing$.
\end{proof}

\begin{remark}
If $S$ is finite, then, obviously, for every $X\in \boldsymbol{\mathcal{U}}^{S}$, $\mathrm{supp}_{S}(\mathbf{T}_{\Sigma}(X))$ is finite. If $S$ is infinite, then there exists an $X\in \boldsymbol{\mathcal{U}}^{S}$ such that $\mathrm{supp}_{S}(\mathbf{T}_{\Sigma}(X))$ is infinite, e.g., for $X = 1 = (1)_{s\in S}$, we have that $\mathrm{supp}_{S}(\mathbf{T}_{\Sigma}(1)) = S$, thus $\mathrm{supp}_{S}(\mathbf{T}_{\Sigma}(1))$ is infinite. Hence, if, for every $X\in \boldsymbol{\mathcal{U}}^{S}$, $\mathrm{supp}_{S}(\mathbf{T}_{\Sigma}(X))$ is finite, then $S$ is finite. Therefore, the following conditions are equivalent: (1) for every $X\in \boldsymbol{\mathcal{U}}^{S}$, $\mathrm{supp}_{S}(\mathbf{T}_{\Sigma}(X))$ is finite, (2) for every $X\in \boldsymbol{\mathcal{U}}^{S}$,  $\mathrm{Cgr}_{\mathrm{fi}}(\mathbf{T}_{\Sigma}(X))\neq\varnothing$, and (3) $S$ is  finite.
\end{remark}

%\begin{example} Let $\mathbf{A}$ be a $\Sigma$-algebra, then $\nabla^{\mathbf{A}}\in\mathrm{Cgr}_{\mathrm{fi}}(\mathbf{A})$ as $\mathbf{A}/{\nabla^{\mathbf{A}}}$ is isomorphic to $\mathbf{1}$.
%\end{example}

\begin{proposition}\label{Filter}
Let $\mathbf{A}$ be a $\Sigma$-algebra. Then
\begin{enumerate}
\item for every $n\in \mathbb{N}-\{0\}$ and every $(\Phi^{i})_{i\in n}\in \mathrm{Cgr}_{\mathrm{fi}}(\mathbf{A})^{n}$,  $\bigcap_{i\in n}\Phi^{i}\in \mathrm{Cgr}_{\mathrm{fi}}(\mathbf{A})$;
\item $\mathrm{Cgr}_{\mathrm{fi}}(\mathbf{A})$ is an upward closed set of the lattice $\mathbf{Cgr}\mathbf{(A)}$, i.e., for every $\Phi,\Psi\in \mathrm{Cgr}\mathbf{(A)}$, if $\Phi\subseteq \Psi$ and $\Phi \in \mathrm{Cgr}_{\mathrm{fi}}(\mathbf{A})$, then $\Psi \in \mathrm{Cgr}_{\mathrm{fi}}(\mathbf{A})$.
\end{enumerate}
Therefore, if $\mathrm{supp}_{S}(\mathbf{A})$ is finite, then $\mathrm{Cgr}_{\mathrm{fi}}(\mathbf{A})$ is a filter of the lattice $\mathbf{Cgr}\mathbf{(A)}$.
\end{proposition}

\begin{proof}
Let $n$ be a non-zero natural number and $(\Phi^{i})_{i\in n}\in \mathrm{Cgr}_{\mathrm{fi}}(\mathbf{A})^{n}$. Then the congruence $\bigcap_{i\in n}\Phi^{i}$ on $\mathbf{A}$ is of finite index because $\prod_{i\in n}\mathbf{A}/\Phi^{i}$ is finite and the homomorphism $\mathrm{p}^{(\Phi^{i})_{i\in n}}$ from $\mathbf{A}/\bigcap_{i\in n}\Phi^{i}$ to $\prod_{i\in n}\mathbf{A}/\Phi^{i}$, defined, for every $s\in S$, as follows:
$$
\textstyle
\mathrm{p}^{(\Phi^{i})_{i\in n}}_{s}
\nfunction
{A_{s}/\bigcap_{i\in n}\Phi^{i}_{s}}
{\prod_{i\in n}A_{s}/\Phi^{i}_{s}}
{\left[a\right]_{\bigcap_{i\in n}\Phi^{i}_{s}}}
{\left(\left[a\right]_{\Phi^{i}_{s}}\right)_{i\in n}}
$$
is a subdirect embedding of $\mathbf{A}/\bigcap_{i\in n}\Phi^{i}$ into $\prod_{i\in n}\mathbf{A}/\Phi^{i}$. Let us note that $\mathrm{p}^{(\Phi^{i})_{i\in n}}$ is $\left<\mathrm{p}_{\bigcap_{i\in n}\Phi^{i},\Phi^{i}}\right>_{i\in n}$ (the unique homomorphism from $\mathbf{A}/\bigcap_{i\in n}\Phi^{i}$ to $\prod_{i\in n}\mathbf{A}/\Phi^{i}$ such that, for every $i\in n$, $\mathrm{pr}_{\Phi^{i}}\circ \left<\mathrm{p}_{\bigcap_{i\in n}\Phi^{i},\Phi^{i}}\right>_{i\in n} = \mathrm{p}_{\bigcap_{i\in n}\Phi^{i},\Phi^{i}}$, where $\mathrm{p}_{\bigcap_{i\in n}\Phi^{i},\Phi^{i}}$ and  $\mathrm{pr}_{\Phi^{i}}$ are the canonical homomorphisms from $\mathbf{A}/\bigcap_{i\in n}\Phi^{i}$ and $\prod_{i\in n}\mathbf{A}/\Phi^{i}$, respectively, to $\mathbf{A}/\Phi^{i}$).

Now let $\Phi$ be a congruence on $\mathbf{A}$ of finite index and $\Psi$ a congruence on $\mathbf{A}$ such that $\Phi\subseteq \Psi$. Then, since $\mathbf{A}/\Phi$ is finite and $\mathrm{pr}_{\Phi,\Psi}$ (the unique homomorphism from $\mathbf{A}/\Phi$ to $\mathbf{A}/\Psi$ such that $\mathrm{pr}_{\Phi,\Psi}\circ \mathrm{pr}_{\Phi} = \mathrm{pr}_{\Psi}$) is surjective, $\mathbf{A}/\Psi$ is finite. Hence $\Psi\in\mathrm{Cgr}_{\mathrm{fi}}(\mathbf{A})$.
\end{proof}

%For congruences of locally finite index there are results similar to those we have just stated for congruences of finite index. We leave the details to the reader.

\begin{proposition}
Let $\mathbf{A}$ be a $\Sigma$-algebra. Then $\mathrm{Cgr}_{\mathrm{lfi}}(\mathbf{A})$ is a filter of the lattice $\mathbf{Cgr}\mathbf{(A)}$.
\end{proposition}

%As in the case of the congruences of finite index on a many-sorted algebra, two different notions of recognizability arise for a many-sorted language: Those that are recognizable and those that are locally finite, abbreviated as $\mathrm{lf}$,  recognizable. However, as for congruences, we do point out that in Section~3 we will only use the recognizable languages.

As for congruences of finite index on a many-sorted algebra, for many-sorted languages we have, on the one hand, those whose definition is based on the categorial notion of finiteness and which will be called recognizable, and, on the other, those whose definition is founded on the notion of local finiteness and which will be called locally finitely recognizable, abbreviated to $\mathrm{lf}$ recognizable.
%However, as for congruences, we do point out that in Section~3 we will only use the recognizable languages.

\begin{definition}
Let $\mathbf{A}$ be a $\Sigma$-algebra, $T\subseteq S$, and $L\subseteq A\!\!\upharpoonright_{T} = (A_{t})_{t\in T}$.
\begin{enumerate}
\item We will say that $L$ is $T\text{-}$\emph{recognizable} if there exists a finite $\Sigma$-algebra $\mathbf{B}$, a homomorphism $f\colon\mathbf{A}\mor\mathbf{B}$, and a subset $M$ of $B\!\!\upharpoonright_{T}$ such that
    $(f\!\!\upharpoonright_{T})^{-1}[M] = L$, where $(f\!\!\upharpoonright_{T})^{-1}[M] = (f^{-1}_{t}[M_{t}])_{t\in T}$. We will denote by $\mathrm{Rec}_{T}(\mathbf{A})$ the set of all subsets of $A\!\!\upharpoonright_{T}$ which are $T\text{-}$recognizable.
\item We will say that $L$ is $(\mathrm{lf},T)\text{-}$\emph{recognizable} if there exists a locally finite ($\equiv$ $S$-finite) $\Sigma$-algebra $\mathbf{B}$, a homomorphism $f\colon\mathbf{A}\mor\mathbf{B}$, and a subset $M$ of $B\!\!\upharpoonright_{T}$ such that $L = ((f\!\!\upharpoonright_{T})^{-1}[M]$. We will denote by $\mathrm{Rec}_{\mathrm{lf},T}(\mathbf{A})$ the set of all subsets of $A\!\!\upharpoonright_{T}$ which are $(\mathrm{lf},T)$-recognizable.
\end{enumerate}

For $T = S$ we will denote by $\mathrm{Rec}(\mathbf{A})$ and $\mathrm{Rec}_{\mathrm{lf}}(\mathbf{A})$ the sets $\mathrm{Rec}_{S}(\mathbf{A})$ and $\mathrm{Rec}_{\mathrm{lf},S}(\mathbf{A})$, respectively. We will call the elements of $\mathrm{Rec}(\mathbf{A})$ and $\mathrm{Rec}_{\mathrm{lf}}(\mathbf{A})$ \emph{rec\-og\-nizable} and $\mathrm{lf}$-\emph{recognizable}, respectively.

For $s\in S$ we will denote by $\mathrm{Rec}_{s}(\mathbf{A})$ and $\mathrm{Rec}_{\mathrm{lf},s}(\mathbf{A})$ the sets $\mathrm{Rec}_{\{s\}}(\mathbf{A})$ and $\mathrm{Rec}_{\mathrm{lf},\{s\}}(\mathbf{A})$, respectively.  We will call the elements of $\mathrm{Rec}_{s}(\mathbf{A})$ and $\mathrm{Rec}_{\mathrm{lf},s}(\mathbf{A})$ $s$-\emph{recognizable} and $(\mathrm{lf},s)$-\emph{recognizable}, respectively.
\end{definition}

\begin{remark}
Let $\mathbf{A}$ be a $\Sigma$-algebra, $T\subseteq S$, and $s\in S$. Then $\mathrm{Rec}_{T}(\mathbf{A})\subseteq \mathrm{Rec}_{\mathrm{lf},T}(\mathbf{A})$, $\mathrm{Rec}(\mathbf{A})\subseteq \mathrm{Rec}_{\mathrm{lf}}(\mathbf{A})$, and $\mathrm{Rec}_{s}(\mathbf{A})\subseteq \mathrm{Rec}_{\mathrm{lf},s}(\mathbf{A})$. If $S$ is infinite, then the converse inclusions are not valid. Let $T$ be an infinite subset of $S$ and $\mathbf{A}$ a $\Sigma$-algebra such that $T\subseteq\mathrm{supp}_{S}(\mathbf{A})$. Then $A\!\!\upharpoonright_{T}$ is a language in $\mathrm{Rec}_{l.\mathrm{f},T}(\mathbf{A})$, but not in $\mathrm{Rec}_{T}(\mathbf{A})$.
\end{remark}

\begin{remark}
If $S$ is finite, $T\subseteq S$, $s\in S$, and $\mathbf{A}$ is a $\Sigma$-algebra, then $\mathrm{Rec}_{T}(\mathbf{A}) = \mathrm{Rec}_{\mathrm{lf},T}(\mathbf{A})$, $\mathrm{Rec}(\mathbf{A}) = \mathrm{Rec}_{\mathrm{lf}}(\mathbf{A})$, and $\mathrm{Rec}_{s}(\mathbf{A}) = \mathrm{Rec}_{\mathrm{lf},s}(\mathbf{A})$.
\end{remark}

In what follows, among other things, we investigate, for a many-sorted algebra and a language of it, the relationships between the different notions of recognizability for the language and the two notions of congruence of finite index on the many-sorted algebra saturating the language. In particular, we investigate such a relationship for the case of the congruence cogenerated by a many-sorted language.

%succinctly summarize the different propositions that express the connection of the different notions of recognizability of a many-sorted language and the existence of a finite index congruence on the many-sorted algebra saturating the given language. In particular, we express these properties for the special case of the congruence cogenerated by the language.

%\begin{proposition}
%Let $\mathbf{A}$ be a $\Sigma$-algebra, $T\subseteq S$, and $L\subseteq A\!\!\upharpoonright_{T}$. Then
%\begin{enumerate}
%\item $L$ is $T$-recognizable if and only if there exists a congruence $\Phi$ on $\mathbf{A}$ of finite index such that $L = [L]^{\Phi\!\upharpoonright_{T}}$, where $\Phi\!\!\upharpoonright_{T} = (\Phi_{t})_{t\in T}$ and, for every $t\in T$, $[L]^{\Phi\!\upharpoonright_{T}}_{t} = [L_{t}]^{\Phi_{t}}$;
%\item $L$ is $(\mathrm{lf},T)$-recognizable if and only if there exists a congruence $\Phi$ on $\mathbf{A}$ of $S\text{-}$finite index such that $L = [L]^{\Phi\!\upharpoonright_{T}}$.
%\end{enumerate}
%\end{proposition}

\begin{proposition}
Let $\mathbf{A}$ be a $\Sigma$-algebra, $T\subseteq S$, and $L\subseteq A\!\!\upharpoonright_{T}$. Then the following assertions are equivalent:
\begin{enumerate}
\item $L$ is $T$-recognizable.
\item There exists a congruence $\Phi$ on $\mathbf{A}$ of finite index such that $L = [L]^{\Phi\!\upharpoonright_{T}}$, where $\Phi\!\!\upharpoonright_{T} = (\Phi_{t})_{t\in T}$ and, for every $t\in T$, $[L]^{\Phi\!\upharpoonright_{T}}_{t} = [L_{t}]^{\Phi_{t}}$.
\item $\Omega^{\mathbf{A}}([L,\varnothing^{S-T}])$ is of finite index, where $\varnothing^{S-T}$ is the mapping from $S-T$ to $\boldsymbol{\mathcal{U}}$ constantly $\varnothing$ and $[L,\varnothing^{S-T}]$ the unique mapping from $S$ to $\boldsymbol{\mathcal{U}}$ such that $[L,\varnothing^{S-T}]\!\!\upharpoonright_{T} = L$ and $[L,\varnothing^{S-T}]\!\!\upharpoonright_{S-T} = \varnothing^{S-T}$.
\end{enumerate}
\end{proposition}

\begin{proposition}
Let $\mathbf{A}$ be a $\Sigma$-algebra and $L\subseteq A$. Then the following assertions are equivalent:
\begin{enumerate}
\item $L$ is recognizable.
\item There exists a congruence $\Phi$ on $\mathbf{A}$ of finite index such that $L = [L]^{\Phi}$.
\item $\Omega^{\mathbf{A}}(L)$ is of finite index.
\end{enumerate}
\end{proposition}

\begin{proof}
We first prove that $(1)\Rightarrow (2)$.

Let us suppose that $L$ is recognizable, i.e., that there exists a finite $\Sigma$-algebra $\mathbf{B}$, a homomorphism $f\colon \mathbf{A}\mor \mathbf{B}$, and a subset $M$ of $B$ such that $L = f^{-1}[M]$. Then, since $\mathbf{Im}(f)$ is a subalgebra of a finite algebra, $\mathbf{Im}(f)$ is finite. Therefore $\mathbf{A}/\mathrm{Ker}(f)$ is finite, because it is canonically isomorphic to $\mathbf{Im}(f)$. Hence $\mathrm{Ker}(f)$ is a congruence on $\mathbf{A}$ of finite index. Let us check that $L = [L]^{\mathrm{Ker}(f)}$. The inclusion $L\subseteq [L]^{\mathrm{Ker}(f)}$ is always true. To show that $[L]^{\mathrm{Ker}(f)}\subseteq L$, let $s$ be a sort in $S$, $a\in A_{s}$, $x\in L_{s}$, and let us suppose that $(a,x)\in \mathrm{Ker}(f_{s})$, i.e., that $f_{s}(a) = f_{s}(x)$. Then, since $L = f^{-1}[M]$, we have that $L_{s} = f^{-1}_{s}[M_{s}]$, thus $f_{s}(x)\in M_{s}$, but $f_{s}(a) = f_{s}(x)$, therefore $f_{s}(a)\in M_{s}$, consequently $a\in L_{s}$. From this it follows that $[L]^{\mathrm{Ker}(f)}\subseteq L$.

We next prove that $(2)\Rightarrow (3)$.

Let us suppose that there exists a congruence $\Phi$ on $\mathbf{A}$ such that $\Phi$ is of finite index and $L = [L]^{\Phi}$. Then, by Proposition~\ref{PSat and sPSat}, we have that $\Phi\subseteq\Omega^{\mathbf{A}}(L)$. Thus, by Proposition~\ref{Filter}, since, by hypothesis, $\Phi\in\mathrm{Cgr}_{\mathrm{fi}}(\mathbf{A})$, $\Omega^{\mathbf{A}}(L)\in\mathrm{Cgr}_{\mathrm{fi}}(\mathbf{A})$.

Finally, we prove that $(3)\Rightarrow (1)$.

Let us suppose that $\Omega^{\mathbf{A}}(L)\in\mathrm{Cgr}_{\mathrm{fi}}(\mathbf{A})$. Then $\mathbf{A}/{\Omega^{\mathbf{A}}(L)}$ is finite. Let $\mathcal{M}$ be the subset $\mathrm{pr}_{\Omega^{\mathbf{A}}(L)}[L]$ of $A/{\Omega^{\mathbf{A}}(L)}$, where $\mathrm{pr}_{\Omega^{\mathbf{A}}(L)}$ is the canonical projection from $\mathbf{A}$ to $\mathbf{A}/{\Omega^{\mathbf{A}}(L)}$. To verify that $L = (\pr_{\Omega^{\mathbf{A}}(L)})^{-1}[\mathcal{M}]$ it suffices to check the inclusion from right to left. Let $s$ be an element of $S$ and $b\in \mathrm{pr}_{\Omega^{\mathbf{A}}(L)_{s}}^{-1}[\mathcal{M}_{s}]$ (recall that $(\mathrm{pr}_{\Omega^{\mathbf{A}}(L)})^{-1}[\mathcal{M}] = (\mathrm{pr}_{\Omega^{\mathbf{A}}(L)_{s}}^{-1}[\mathcal{M}_{s})_{s\in S}$). Then $\mathrm{pr}_{\Omega^{\mathbf{A}}(L)_{s}}[b]\in \mathcal{M}_{s}$. Hence, there exists an $a\in L_{s}$ such that $[a]_{\Omega^{\mathbf{A}}(L)_{s}} = [b]_{\Omega^{\mathbf{A}}(L)_{s}}$. But, by Proposition~\ref{PSat and sPSat}, $L$ is $\Omega^{\mathbf{A}}(L)$-saturated and, since $a\in L_{s}$, we have that $b\in L_{s}$. Therefore   $(\pr_{\Omega^{\mathbf{A}}(L)})^{-1}[\mathcal{M}]\subseteq L$.
\end{proof}

%\begin{proposition}
%Let $\mathbf{A}$ be a $\Sigma$-algebra, $s\in S$, and $L\subseteq A_{s}$. Then
%\begin{enumerate}
%\item $L$ is $s$-recognizable if and only if there exists a congruence $\Phi$ on $\mathbf{A}$ of finite index such that $L = [L]^{\Phi_{s}}$;
%\item $L$ is $(\mathrm{lf},s)$-recognizable if and only if there exists a congruence $\Phi$ on $\mathbf{A}$ of $S$-finite index such that $L = [L]^{\Phi_{s}}$.
%\end{enumerate}
%\end{proposition}

\begin{proposition}
Let $\mathbf{A}$ be a $\Sigma$-algebra, $s\in S$, and $L\subseteq A_{s}$. Then the following assertions are equivalent:
\begin{enumerate}
\item $L$ is $s$-recognizable.
\item There exists a congruence $\Phi$ on $\mathbf{A}$ of finite index such that $L = [L]^{\Phi_{s}}$.
\item $\Omega^{\mathbf{A}}(\delta^{s,L})$ is of finite index.
\end{enumerate}
\end{proposition}

\begin{proposition}\label{Rec is Bool}
Let $\mathbf{A}$ and $\mathbf{B}$ be $\Sigma$-algebras. Then
\begin{enumerate}
\item $\varnothing^{S}, A\in \mathrm{Rec}(\mathbf{A})$.
\item If $K,L\in \mathrm{Rec}(\mathbf{A})$, then $K\cup L, K\cap L, K-L\in \mathrm{Rec}(\mathbf{A})$.
\item If $T\in \mathrm{Tl}_{t}(\mathbf{A})_{s}$, and $L\in \mathrm{Rec}(\mathbf{A})$, then $T^{-1}[L] = \delta^{t,T^{-1}[L_{s}]}\in \mathrm{Rec}(\mathbf{A})$.
\item If $f\colon \mathbf{A}\mor \mathbf{B}$ and $M\in \mathrm{Rec}(\mathbf{B})$, then $f^{-1}[M]\in \mathrm{Rec}(\mathbf{A})$.
\end{enumerate}
\end{proposition}

%\begin{proposition}\label{s-Rec iff Rec}
%Let $\mathbf{A}$ be a $\Sigma$-algebra, $s\in S$, and $L\subseteq A_{s}$. Then
%\begin{enumerate}
%\item $L$ is $s$-recognizable if and only if  $\delta^{s,L}$ is recognizable, i.e., $L\in \mathrm{Rec}_{s}(\mathbf{A})$ if and only if $\delta^{s,L}\in \mathrm{Rec}(\mathbf{A})$. Thus $\mathrm{Rec}_{s}(\mathbf{A})$ is isomorphic to a subset of $\mathrm{Rec}(\mathbf{A})$;
%\item $L$ is $(\mathrm{lf},s)$-recognizable if and only if  $\delta^{s,L}$ is $\mathrm{lf}$-recognizable, i.e., $L\in \mathrm{Rec}_{\mathrm{lf},s}(\mathbf{A})$ if and only if $\delta^{s,L}\in S\text{-}\mathrm{Rec}_{\mathrm{lf}}(\mathbf{A})$. Thus $\mathrm{Rec}_{\mathrm{lf},s}(\mathbf{A})$ is isomorphic to a subset of $\mathrm{Rec}_{\mathrm{lf}}(\mathbf{A})$.
%\end{enumerate}
%\end{proposition}

%\begin{proof}
%This follows immediately from the definitions.
%\end{proof}

\begin{proposition}\label{s-Rec iff Rec}
Let $\mathbf{A}$ be a $\Sigma$-algebra, $s\in S$, and $L\subseteq A_{s}$. Then $L\in \mathrm{Rec}_{s}(\mathbf{A})$ if and only if $\delta^{s,L}\in \mathrm{Rec}(\mathbf{A})$. Thus $\mathrm{Rec}_{s}(\mathbf{A})$ is isomorphic to a subset of $\mathrm{Rec}(\mathbf{A})$.
\end{proposition}

\begin{assumption}
To prove the following proposition we will assume that the set of sorts $S$ is finite.
\end{assumption}

\begin{proposition}\label{Rec SubdirectProd sRec}
Let $\mathbf{A}$ be a $\Sigma$-algebra and $L\subseteq A$. Then $L\in \mathrm{Rec}(\mathbf{A})$ if and only if, for every $s\in S$, $L_{s}\in \mathrm{Rec}_{s}(\mathbf{A})$. Thus there exists an embedding from $\mathrm{Rec}(\mathbf{A})$ into $\prod_{s\in S}\mathrm{Rec}_{s}(\mathbf{A})$ and $\mathrm{Rec}(\mathbf{A})$ is a subdirect product of $(\mathrm{Rec}_{s}(\mathbf{A}))_{s\in S}$.
\end{proposition}

\begin{proof}
If $L\in \mathrm{Rec}(\mathbf{A})$, then, from the definitions, it follows that, for every $s\in S$, $L_{s}\in \mathrm{Rec}_{s}(\mathbf{A})$.

Conversely, let us suppose that, for every $s\in S$, $L_{s}\in \mathrm{Rec}_{s}(\mathbf{A})$. Then, for every $s\in S$, there exists a finite $\Sigma$-algebra $\mathbf{B}^{s}$, a homomorphism $f^{s}\colon\mathbf{A}\mor \mathbf{B}^{s}$ and $M_{s}\subseteq B^{s}_{s}$ such that $(f^{s}_{s})^{-1}[M_{s}] = L_{s}$. Let $\prod_{s\in S}\mathbf{B}^{s}$ be the product of $(\mathbf{B^{s}})_{s\in S}$, which is finite---since, for every $s\in S$, $\mathbf{B}^{s}$ is finite and, by hypothesis, $S$ is finite---, and $\langle f^{s}\rangle_{s\in S}$ the canonical homomorphism from $\mathbf{A}$ to $\prod_{s\in S}\mathbf{B}^{s}$. Let $N = (N_{t})_{t\in S}$ be the subset of $\prod_{s\in S}B^{s} = (\prod_{s\in S}B^{s}_{t})_{t\in S}$ defined, for every $t\in S$, up to isomorphism, as: $N_{t} = M_{t}\times \prod_{s\in S-\{t\}}B^{s}_{t}$. Then $(\langle f^{s}\rangle_{s\in S})^{-1}[N] = L$.
%, for every $i\in R$,
%\begin{align*}
%(\langle f^{r}\rangle_{r\in R})_{i}^{-1}[N_{i}]&=
%\{a\in A_{i}\mid (\langle f^{r}\rangle_{r\in R})_{i}(a)\in N_{i}\}\\
%&=
%\{a\in A_{i}\mid \forall r\in R\, ((\mathrm{pr}^{r}\circ\langle f^{r}\rangle_{r\in R})_{i}(a)\in N_{i,r})\}\\
%&=
%\{a\in A_{i}\mid \forall r\in R\, ( f^{r}_{i}(a)\in N_{i,r})\}\\
%&=
%\{a\in A_{i}\mid  (f^{i}_{i}(a)\in M^{i})\}\\
%&=f^{i,-1}_{i}[M^{i}]\\
%&=L_{i}
%\end{align*}
Consequently, $L$ is recognizable.

Before proceeding any further, we provide another proof of the just proved result. If, for every $s\in S$, $L_{s}\in \mathrm{Rec}_{s}(\mathbf{A})$, then, for every $s\in S$, by Proposition~\ref{s-Rec iff Rec}, $\delta^{s,L_{s}}\in \mathrm{Rec}(\mathbf{A})$. Therefore, by part (2) of Proposition~\ref{Rec is Bool} and since, by hypothesis, $S$ is finite, $\bigcup_{s\in S}\delta^{s,L_{s}} = L\in \mathrm{Rec}(\mathbf{A})$.

The mapping from $\mathrm{Rec}(\mathbf{A})$ to $\prod_{s\in S}\mathrm{Rec}_{s}(\mathbf{A})$ that sends $L$ in $\mathrm{Rec}(\mathbf{A})$ to $L$ in $\prod_{s\in S}\mathrm{Rec}_{s}(\mathbf{A})$ is well-defined and injective. Moreover, for every $s\in S$ and every $K\in \mathrm{Rec}_{s}(\mathbf{A})$, the $S$-sorted set $\delta^{s,K}$ is in $\mathrm{Rec}(\mathbf{A})$ and its $s$-th projection is $K$.
\end{proof}

From the just stated proposition (which, we recall, has been obtained under the assumption that $S$ is finite) and since, for a $\Sigma$-algebra $\mathbf{A}$ and an $L\subseteq A$, $L = \bigcup_{s\in S}\delta^{s,L_{s}}$, it follows that in order to investigate the languages in $\mathrm{Rec}(\mathbf{A})$ it suffices to investigate, for every $s\in S$, the languages in $\mathrm{Rec}_{s}(\mathbf{A})$.

\begin{proposition}\label{Rec s is Bool}
Let $\mathbf{A}$ and $\mathbf{B}$ be $\Sigma$-algebras and $s, t\in S$. Then
\begin{enumerate}
\item $\varnothing, A_{s}\in \mathrm{Rec}_{s}(\mathbf{A})$.
\item If $K,L\in \mathrm{Rec}_{s}(\mathbf{A})$, then $K\cup L, K\cap L, K-L\in \mathrm{Rec}_{s}(\mathbf{A})$.
\item If $T\in \mathrm{Tl}_{t}(\mathbf{A})_{s}$, and $L\in \mathrm{Rec}_{s}(\mathbf{A})$, then $T^{-1}[L]\in \mathrm{Rec}_{t}(\mathbf{A})$.
\item If $f\colon \mathbf{A}\mor \mathbf{B}$ and $M\in \mathrm{Rec}_{s}(\mathbf{B})$, then $f^{-1}_{s}[M]\in \mathrm{Rec}_{s}(\mathbf{A})$.
\end{enumerate}
\end{proposition}

For the notions of $\mathrm{lf}$-recognizability and congruence of locally finite index there are results similar to those we have just stated for recognizability and congruence of finite index. By way of illustration, here are some examples of it.

%To illustrate, let us state the following propositions.

%The details are left to the reader.

\begin{proposition}
Let $\mathbf{A}$ be a $\Sigma$-algebra, $T\subseteq S$, and $L\subseteq A\!\!\upharpoonright_{T}$. Then the following assertions are equivalent:
\begin{enumerate}
\item $L$ is $\mathrm{lf}$-recognizable.
\item There exists a congruence $\Phi$ on $\mathbf{A}$ of $S$-finite index such that $L = [L]^{\Phi\!\upharpoonright_{T}}$, where $\Phi\!\!\upharpoonright_{T} = (\Phi_{t})_{t\in T}$ and, for every $t\in T$, $[L]^{\Phi\!\upharpoonright_{T}}_{t} = [L_{t}]^{\Phi_{t}}$.
\item $\Omega^{\mathbf{A}}([L,\varnothing^{S-T}])$ is of $S$-finite index.
\end{enumerate}
\end{proposition}

\begin{proposition}\label{lf,s-Rec iff lf-Rec}
Let $\mathbf{A}$ be a $\Sigma$-algebra, $s\in S$, and $L\subseteq A_{s}$. Then $L$ is $(\mathrm{lf},s)$-recognizable if and only if  $\delta^{s,L}$ is $\mathrm{lf}$-recognizable, i.e., $L\in \mathrm{Rec}_{\mathrm{lf},s}(\mathbf{A})$ if and only if $\delta^{s,L}\in \mathrm{Rec}_{\mathrm{lf}}(\mathbf{A})$. Thus $\mathrm{Rec}_{\mathrm{lf},s}(\mathbf{A})$ is isomorphic to a subset of $\mathrm{Rec}_{\mathrm{lf}}(\mathbf{A})$.
\end{proposition}

\begin{proposition}
Let $\mathbf{A}$ be a $\Sigma$-algebra and $L\subseteq A$. Then the following assertions are equivalent:
\begin{enumerate}
\item $L\in \mathrm{Rec}_{\mathrm{lf}}(\mathbf{A})$.
\item There exists a congruence $\Phi$ on $\mathbf{A}$ of $S$-finite index such that $L = [L]^{\Phi}$.
\item $\Omega^{\mathbf{A}}(L)\in \mathrm{Cgr}_{\mathrm{lfi}}(\mathbf{A})$.
\end{enumerate}
\end{proposition}

\section{Recognizable subsets of a free many-sorted algebra}

This section is devoted to provide congruence based proofs of the recognizability theorems for free many-sorted algebras.
Actually, in order to deal with the different cases of recognizability, classified according to the type of operator under consideration, we have divided  this section into several subsections. In Subsection 1, entitled \emph{Basic terms}, we prove that the final sets containing a variable, a constant or an operation symbol applied to a suitable family of variables are recognizable. These results will be used later on to prove that the set of recognizable languages of a many-sorted term algebra forms an algebra of the same signature. In Subsection 2, entitled \emph{Substitutions}, after defining several substitution operators associated to a free many-sorted algebra and investigating the relationships between them, we prove that, if all input languages for a given substitution are recognizable, then the output language is recognizable as well. In Subsection 3, entitled \emph{Iterations}, we introduce the notion of iteration of a language with respect to a variable and we prove that, if the input language is recognizable then its iteration with respect to a variable is also recognizable. In Subsection 4, entitled \emph{Quotients}, we introduce the notion of quotient of a language by another with respect to a variable and we prove that if the input language is recognizable then its quotient is also recognizable.  In Subsection 5, entitled \emph{Tree Homomorphisms}, we introduce the notion of hyperderivor as a way of transforming terms relative to a signature into terms relative to another signature. We show that tree homomorphisms are, really, homomorphisms from a free many-sorted algebra to another many-sorted algebra, itself derived from a many-sorted algebra relative to another signature. Then we prove that the inverse image of a recognizable language under a tree homomorphism is recognizable and that the direct image of a recognizable language by a linear tree homomorphism, which is a particular type of tree homomorphism, is also recognizable. Finally, Subsection~6 is devoted to the study of derivors. We first introduce the category of many-sorted signatures and derivors. Next, after defining a contravariant functor from this category to the category of $\boldsymbol{\mathcal{U}}$-locally small categories and functors, we obtain by means of the Grothendieck contruction, the category of ordered pairs consisting of a many-sorted signature and an algebra of the same signature and the pairs of derivors and derived homomorphisms of algebras as morphisms. As a consequence, after showing that every derivor, together with some additional data, gives rise to a hyperderivor, the counterparts of the two main results in Subsection~5 are immediately obtained.

%\begin{assumptions}
%In the remainder of this section we assume that every signature is finite and that the free generating many-sorted set of every free algebra is finite.
%%Moreover, for an signature $\Sigma$ and a free generating many-sorted set $X$, we assume that no sort is isolated, i.e, that, for every $s\in S$, $\mathrm{T}_{\Sigma}(\varnothing^{S})_{s}\neq \varnothing$ or that, for every $s\in S$, $X_{s}\neq\varnothing$.
%\end{assumptions}

From now on, following a strongly rooted tradition in the fields of formal languages and automata, we agree to call \emph{languages} the subsets of the underlying many-sorted set of the free many-sorted algebra on a many-sorted set.

\subsection{Basic terms}
In this subsection we prove some basic recognizability results relative to a free algebra on an $S$-sorted set. Concretely, we prove that the final sets consisting, respectively, of a variable, a constant, and the action of an operator symbol on a family of variables, are recognizable.

\begin{assumption}
In this subsection we will assume that $S$ is finite.
\end{assumption}

\begin{proposition}\label{PRecVar}
Let $\Sigma$ be an $S$-sorted signature, $X$ an $S$-sorted set, $s\in S$, and $x$ and element in $X_{s}$. Then the language $\{x\}\subseteq \mathrm{T}_{\Sigma}(X)_{s}$ is recognizable.
\end{proposition}

\begin{proof}
We will denote by $2^{S}$ or, to abbreviate, by $2$, the $S$-sorted set $(2)_{s\in S}$, where $2=\{0,1\}$. Let $\mathbf{2}=(2,G^{\mathbf{2}})$ be the finite $\Sigma$-algebra defined as follows:
For every $(w,t)\in S^{\star}\times S$ and every $\sigma\in\Sigma_{w,t}$, the structural operation associated to $\sigma$, $G^{\mathbf{2}}_{\sigma}\colon 2_{w}\mor 2$, is the constant mapping with value $0$. Let $f\colon X\mor 2$ be the $S$-sorted mapping such that $f_{s}\colon X\mor 2$ is $\mathrm{ch}_{\{x\}}$, the characteristic mapping of $\{x\}$, and, for $t\in S-\{s\}$, $f_{t}$ is the mapping constantly $0$. Then, for $f^{\sharp}$, the unique homomorphism from $\mathbf{T}_{\Sigma}(X)$ to $\mathbf{2}$ such that $f^{\sharp}\circ \eta_{X} = f$,
we have that, for every $P\in \mathrm{T}_{\Sigma}(X)_{s}$, $f^{\sharp}_{s}(P)= \mathrm{ch}_{\{x\}}(P)  = 1$ if and only if $P = x$. Hence, since $\{x\}=(f^{\sharp}_{s})^{-1}[\{1\}]$,  the language $\{x\}$ is recognizable.
\end{proof}

\begin{proposition}\label{PRecConst}
Let $\Sigma$ be an $S$-sorted signature, $X$ an $S$-sorted set, and $\sigma\in\Sigma_{\lambda,s}$. Then the language $\{\sigma\}\subseteq \mathrm{T}_{\Sigma}(X)_{s}$ is recognizable.
\end{proposition}

\begin{proof}
Let $\mathbf{2}=(2,H^{\mathbf{2}})$ be the finite $\Sigma$-algebra defined as follows: $H^{\mathbf{2}}_{\sigma}$ is the mapping from $2_{\lambda}$ to $2$ that sends the unique member of $2_{\lambda}$ to $1$, and, for every $(w,t)\in S^{\star}\times S$ and every $\tau\in\Sigma_{w,t}$, $H^{\mathbf{2}}_{\tau}\colon 2_{w}\mor 2$ is the constant mapping to $0$. Let $\kappa\colon X\mor 2$ be the $S$-sorted mapping constantly zero at each sort. Then, for $\kappa^{\sharp}$, the unique homomorphism from $\mathbf{T}_{\Sigma}(X)$ to $\mathbf{2}$ such that
$\kappa^{\sharp}\circ \eta_{X} = \kappa$,
we have that, for every $P\in \mathrm{T}_{\Sigma}(X)_{s}$, $\kappa_{s }^{\sharp}(P) = 1$ if and only if $P = \sigma$. Hence, since $\{\sigma\}=(\kappa_{s}^{\sharp})^{-1}[\{1\}]$, the language $\{\sigma\}$ is recognizable.
\end{proof}

\begin{proposition}\label{PRecOp}
Let $\Sigma$ be an $S$-sorted signature, $X$ an $S$-sorted set, $(w,s)\in (S^{\star}-\{\lambda\})\times S$, $\sigma\in\Sigma_{w,s}$, and $(x_{i})_{i\in \bb{w}}$ a family of variables in $X_{w}$. Then the language $\{\sigma((x_{i})_{i\in \bb{w}})\}\subseteq\mathrm{T}_{\Sigma}(X)_{s}$ is recognizable.
\end{proposition}

\begin{proof}
For $w\in S^{\star}-\{\lambda\}$, let $\downarrow\!\!w$ be the $S$-sorted set defined, for every $t\in S$ as follows: $(\downarrow\!\!w)_{t}=\{i\in\bb{w}\mid w_{i}=t\}$, i.e., $(\downarrow\!\!w)_{t}=w^{-1}[\{t\}]$. Therefore, for every $t\in S-\mathrm{Im}(w)$, $(\downarrow\!\!w)_{t}=\varnothing$, while, for every $t\in\mathrm{Im}(w)$, $(\downarrow\!\!w)_{t}\neq\varnothing$. For every $S$-sorted set $X$, the sets $\mathrm{Hom}(\downarrow\!\!w,X)$ and $X_{w}$ are naturally isomorphic. If $x=(x_{i})_{i\in\bb{w}}\in X_{w}$, i.e., if $x$ is a mapping from $\bb{w}$ to $\bigcup_{i\in\bb{w}}X_{w_{i}}$ such that, for every $i\in\bb{w}$, $x_{i}\in X_{w_{i}}$, then we will denote by $\overline{x}$ the $S$-sorted mapping from $\downarrow\!\!w$ to $X$ which is associated, in virtue of the natural isomorphism, to $x$ and is defined as follows: For every $t\in S$ and every $i\in (\downarrow\!\!w)_{t}$, then $\overline{x}_{t}(i) = x_{i}$. Let us note that, for every $t\in S-\mathrm{Im}(w)$, $\overline{x}_{t}$ is the unique mapping from $\varnothing$ to $X_{t}$. Thus, given $w\in S^{\star}-\{\lambda\}$ and $x=(x_{i})_{i\in\bb{w}}\in X_{w}$, let $k = (k_{t})_{t\in S}$ be the $S$-sorted set defined, for every $t\in S$, as $k_{t} = \mathrm{card}((\mathrm{Im}(\overline{x}_{t}))) = \mathrm{card}(\{x_{i}\mid i\in w^{-1}[\{t\}]\})$. Then, for every $t\in S$, $k_{t}$ is the number of different variables of type $t$ in $x$.

 %Let us note that, for every $s\in S-\mathrm{Im}(w)$, $k_{s} = 0$. Let us note that $\{x_{i}\mid i\in w^{-1}[\{s\}]\} = \mathrm{Im}(x\!\!\upharpoonright_{w^{-1}[\{s\}]})$.
Let $\varphi\colon\mathrm{Im}(\overline{x})\mor k$ be a fixed $S$-sorted mapping such that, for every $t\in S$, $\varphi_{t}\colon \mathrm{Im}(\overline{x}_{t})\mor k_{t}$ is a bijection.  Let $K$ be the $S$-sorted set such that, for every $t\in S-\{s\}$, $K_{t}=k_{t}+1$ and $K_{s}=k_{s}+2$. Then let $\mathbf{K}=(K,I^{\mathbf{K}})$ be the $\Sigma$-algebra defined as follows: For $(u,t)\in S^{\star}\times S$ and $\tau\in\Sigma_{u,t}$, if $(u,t)\neq (w,s)$ or $\tau\neq \sigma$,  then $I^{\mathbf{K}}\colon K_{u}\mor K_{t}$ is the constant mapping with value $k_{t}$, and for $\tau=\sigma$,  $I^{\mathbf{K}}$ is the mapping:
$$
I^{\mathbf{K}}_{\sigma}
\nfunction
{K_{w}}
{K_{s}}
{a}
{
\begin{cases}
k_{s}+1, & \text{if $a=\varphi(x)$;}\\
k_{s}, & \text{otherwise.}
\end{cases}
}
$$
Let $f\colon X\mor K$ be the $S$-sorted mapping defined, for every $s\in S$, as follows:
$$
f_{s}
\nfunction
{X_{s}}
{K_{s}}
{z}
{
\begin{cases}
\varphi_{s}(z), & \text{if $z$ is a variable in $\mathrm{Im}(\overline{x})_{s}$;}\\
k_{s}, & \text{otherwise.}
\end{cases}
}
$$
Then, for the unique homomorphism $f^{\sharp}$ from $\mathbf{T}_{\Sigma}(X)$ to $\mathbf{K}$ such that
$f^{\sharp}\circ \eta_{X} = f$,
we have that, for every $P\in \mathrm{T}_{\Sigma}(X)_{s}$, $f^{\sharp}(P) = k_{s}+1$ if and only if $P = \sigma((x_{i})_{i\in \bb{w}})$. Hence, since $\{\sigma((x_{i})_{i\in \bb{w}})\}=(f^{\sharp})^{-1}_{s}[\{k_{s}+1\}]$, the language $\{\sigma((x_{i})_{i\in n})\}$ is recognizable.
\end{proof}

\subsection{Substitutions}
In this subsection we introduce several substitution operators associated to a free algebra with the aim of proving that, if the languages under consideration are recognizable, then the language that results from the substitution operator applied to these languages is also recognizable.

\begin{definition}\label{GlobSubstOp}
Let $X$ be an $S$-sorted set, $s\in S$, and $P\in \mathrm{T}_{\Sigma}(X)_{s}$. Then we will denote by $\left(\!\begin{smallmatrix}(x,t)\\ \cdot\end{smallmatrix}\!\right)_{(x,t)\in \coprod X}(P)$ the mapping from $\prod_{(x,t)\in \coprod X}\mathrm{T}_{\Sigma}(X)_{t}^{\bb{P}_{x}}$ to $\mathrm{T}_{\Sigma}(X)_{s}$ that assigns to $\left((Q^{x,t}_{\alpha})_{\alpha\in \bb{P}_{x}}\right)_{(x,t)\in \coprod X}$ in $\prod_{(x,t)\in \coprod X}\mathrm{T}_{\Sigma}(X)_{t}^{\bb{P}_{x}}$ the term $\left(\!\begin{smallmatrix}(x,t)\\(Q^{x,t}_{\alpha})_{\alpha\in \bb{P}_{x}}\end{smallmatrix}\!\right)_{(x,t)\in \coprod X}(P)$ in $\mathrm{T}_{\Sigma}(X)_{s}$ obtained by substituting in $P$, for every $t\in S$, every $x\in X_{t}$, and every $\alpha\in \bb{P}_x$, $Q^{x,t}_{\alpha}$ for the $i_{\alpha}$-th occurrence of $x$ in $P$. We will call $\left(\!\begin{smallmatrix}(x,t)\\(Q^{x,t}_{\alpha})_{\alpha\in \bb{P}_{x}}\end{smallmatrix}\!\right)_{(x,t)\in \coprod X}(P)$ \emph{the substitution of $(Q^{x,t}_{\alpha})_{\alpha\in\bb{P}_{x}}$ for $x$ in $P$ for every $t\in S$ and every $x$ in $X_{t}$}, and $\left(\!\begin{smallmatrix}(x,t)\\ \cdot\end{smallmatrix}\!\right)_{(x,t)\in \coprod X}(P)$ the \emph{global substitution operator for $P$}.

Let $u$ be a sort in $S$ and $z\in X_{u}$. Then we will denote by $\left(\!\begin{smallmatrix}z\\ \cdot\end{smallmatrix}\!\right)(P)$ the mapping from $\mathrm{T}_{\Sigma}(X)_{u}^{\bb{P}_{z}}$ to $\mathrm{T}_{\Sigma}(X)_{s}$ that assigns to $(Q^{z}_{\alpha})_{\alpha\in \bb{P}_{z}}$ in $\mathrm{T}_{\Sigma}(X)_{u}^{\bb{P}_{z}}$ the term $\left(\!\begin{smallmatrix}z\\(Q^{z}_{\alpha})_{\alpha\in \bb{P}_{z}}\end{smallmatrix}\!\right)(P)$ in $\mathrm{T}_{\Sigma}(X)_{s}$ obtained by substituting in $P$, for every $\alpha\in \bb{P}_{z}$, $Q^{z}_{\alpha}$ for the $i_{\alpha}$-th occurrence of $z$ in $P$. We will call $\left(\!\begin{smallmatrix}z\\(Q^{z}_{\alpha})_{\alpha\in \bb{P}_{z}}\end{smallmatrix}\!\right)(P)$ \emph{the substitution of $(Q^{z}_{\alpha})_{\alpha\in\bb{P}_{z}}$ for $z$ in $P$}.
\end{definition}

\begin{remark}
For an $S$-sorted set $X$, a sort $s\in S$, and a term $P$ in $\mathrm{T}_{\Sigma}(X)_{s}$, the domain of $\left(\!\begin{smallmatrix}(x,t)\\ \cdot\end{smallmatrix}\!\right)_{(x,t)\in \coprod X}(P)$ is the set of all choice functions for $(\mathrm{T}_{\Sigma}(X)_{t}^{\bb{P}_{x}})_{(x,t)\in \coprod X}$. Regarding the index set $\coprod X$, since $\mathrm{Var}^{X}(P)$, the $S$-sorted of the variables of $P$, is finite, we can, without loss of generality, replace it by $\bigcup_{t\in \mathrm{sort}(\mathrm{Var}^{X}(P))}(X_{t}\times\{t\})$, where $\mathrm{sort}(\mathrm{Var}^{X}(P))$ is the finite set of the sorts of the variables which appear in $P$. However, for uniformity, we will continue using $\coprod X$.
\end{remark}

\begin{remark}
Let $X$ be an $S$-sorted set, $s$ a sort in $S$, and $P$ a term in $\mathrm{T}_{\Sigma}(X)_{s}$. Then we will denote by $\left(\!\begin{smallmatrix}x\\ \cdot\end{smallmatrix}\!\right)_{x\in X_{s}}(P)$ the mapping from $\prod_{x\in X_{s}}\mathrm{T}_{\Sigma}(X)_{s}^{\bb{P}_{x}}$ to $\mathrm{T}_{\Sigma}(X)_{s}$ that assigns to $\left((Q^{x}_{\alpha})_{\alpha\in \bb{P}_{x}}\right)_{x\in X_{s}}$ in $\prod_{x\in X_{s}}\mathrm{T}_{\Sigma}(X)_{s}^{\bb{P}_{x}}$ the term $\left(\!\begin{smallmatrix}x\\(Q^{x}_{\alpha})_{\alpha\in \bb{P}_{x}}\end{smallmatrix}\!\right)_{x\in X_{s}}(P)$ in $\mathrm{T}_{\Sigma}(X)_{s}$ obtained by substituting in $P$, for every $x\in X_{s}$ and every $\alpha\in \bb{P}_x$, $Q^{x}_{\alpha}$ for the $i_{\alpha}$-th occurrence of $x$ in $P$. We will call $\left(\!\begin{smallmatrix}x\\(Q^{x}_{\alpha})_{\alpha\in \bb{P}_{x}}\end{smallmatrix}\!\right)_{x\in X_{s}}(P)$ the substitution of $(Q^{x}_{\alpha})_{\alpha\in\bb{P}_{x}}$ for $x$ in $P$ for every $x$ in $X_{s}$, and $\left(\!\begin{smallmatrix}x\\ \cdot\end{smallmatrix}\!\right)_{x\in X_{s}}(P)$ the substitution operator for $P$.

Since (1) for every $s\in S$, $X_{s}\cong X_{s}\times \{s\}$, (2) by the associativity of the product, $\prod_{(x,t)\in \coprod X}\mathrm{T}_{\Sigma}(X)_{t}^{\bb{P}_{x}} \cong \prod_{t\in S}(\prod_{x\in X_{t}}\mathrm{T}_{\Sigma}(X)_{t}^{\bb{P}_{x}})$, and (3) $\prod_{x\in X_{s}}\mathrm{T}_{\Sigma}(X)_{s}^{\bb{P}_{x}}$ is naturally isomorphic to a subset of $\prod_{t\in S}(\prod_{x\in X_{t}}\mathrm{T}_{\Sigma}(X)_{t}^{\bb{P}_{x}})$, we have that $\left(\!\begin{smallmatrix}x\\ \cdot\end{smallmatrix}\!\right)_{x\in X_{s}}(P)$ is, essentially, the restriction of $\left(\!\begin{smallmatrix}(x,t)\\ \cdot\end{smallmatrix}\!\right)_{(x,t)\in \coprod X}(P)$ to $\prod_{x\in X_{s}}\mathrm{T}_{\Sigma}(X)_{s}^{\bb{P}_{x}}$. Similarly $\left(\!\begin{smallmatrix}z\\ \cdot\end{smallmatrix}\!\right)(P)$ is, essentially, the restriction of $\left(\!\begin{smallmatrix}(x,t)\\ \cdot\end{smallmatrix}\!\right)_{(x,t)\in \coprod X}(P)$ to $\mathrm{T}_{\Sigma}(X)_{u}^{\bb{P}_{z}}$.
\end{remark}

We next prove that, for every $s\in S$ and every $P\in \mathrm{T}_{\Sigma}(X)$, $\left(\!\begin{smallmatrix}(x,t)\\ \cdot\end{smallmatrix}\!\right)_{(x,t)\in \coprod X}(P)$ is, actually, a mapping from $\prod_{(x,t)\in \coprod X}\mathrm{T}_{\Sigma}(X)_{t}^{\bb{P}_{x}}$ to $\mathrm{T}_{\Sigma}(X)_{s}$.

\begin{proposition}
Let $X$ be an $S$-sorted set, $s\in S$, and $P$ a term in $\mathrm{T}_{\Sigma}(X)_{s}$. Then, for every $\left((Q^{x,t}_{\alpha})_{\alpha\in \bb{P}_{x}}\right)_{(x,t)\in \coprod X}$ in $\prod_{(x,t)\in \coprod X}\mathrm{T}_{\Sigma}(X)_{t}^{\bb{P}_{x}}$, we have that $\left(\!\begin{smallmatrix}(x,t)\\(Q^{x,t}_{\alpha})_{\alpha\in \bb{P}_{x}}\end{smallmatrix}\!\right)_{(x,t)\in \coprod X}(P)$ is a term in $\mathrm{T}_{\Sigma}(X)_{s}$. Therefore $\left(\!\begin{smallmatrix}(x,t)\\ \cdot\end{smallmatrix}\!\right)_{(x,t)\in \coprod X}(P)$ is a mapping from $\prod_{(x,t)\in \coprod X}\mathrm{T}_{\Sigma}(X)_{t}^{\bb{P}_{x}}$ to $\mathrm{T}_{\Sigma}(X)_{s}$.
\end{proposition}

\begin{proof}
Let $\mathcal{T}$ be the subset of $\mathrm{T}_{\Sigma}(X)$ defined, for every $s\in S$, as follows:
$$
\mathcal{T}_{s} = \Big\lbrace P\in \mathrm{T}_{\Sigma}(X)_{s}\,\,\Big|\,\, \mathrm{Im}(\left(\!\begin{smallmatrix}(x,t)\\ \cdot\end{smallmatrix}\!\right)_{(x,t)\in \coprod X}(P))\subseteq \mathrm{T}_{\Sigma}(X)_{s}\Big\rbrace.
$$
To prove that $\mathcal{T} = \mathrm{T}_{\Sigma}(X)$ it suffices to show, by Proposition~\ref{PPAI}, that $X\subseteq \mathcal{T}$ and that $\mathcal{T}$ is a subalgebra of $\mathbf{T}_{\Sigma}(X)$.

We first prove that $X\subseteq \mathcal{T}$.

Let $s$ be a sort in $S$ and $z\in X_{s}$. Then we have that $\bb{z}_{z}=1$, whilst, for every $x\in X_{s}-\{z\}$, $\bb{z}_{x}=0$, and, for every $t\in S-\{s\}$ and every $x\in X_{t}$ $\bb{z}_{x}=0$. Therefore, for a term $Q^{z,s}$ in $\mathrm{T}_{\Sigma}(X)_{s}$ associated to $z$, the global substitution of $Q^{z,s}$ for $z$ in $z$ is equal to $Q^{z,s}$. Hence is a term in $\mathrm{T}_{\Sigma}(X)_{s}$. Consequently $X\subseteq \mathcal{T}$.

We next prove that $\mathcal{T}$ is a subalgebra of $\mathbf{T}_{\Sigma}(X)$.

Let $s$ be a sort in $S$ and $\sigma\in\Sigma_{\lambda,s}$. Then, for every $t\in S$ and every $x\in X_{t}$, we have that $\bb{\sigma}_{x}=0$. Hence the global substitution operator leaves $\sigma$ invariant.

Let $(w,s)\in (S^{\star}-\{\lambda\})\times S$, $\sigma\in\Sigma_{w,s}$, and let $(P_{i})_{i\in\bb{w}}\in \mathrm{T}_{\Sigma}(X)_{w}$ be such that, for every $i\in \bb{w}$, $P_{i}$ satisfies the given requirement. Then, for every $t\in S$ and every $x\in X_{t}$, we have that
$$
\textstyle
\bb{\sigma((P_{i})_{i\in\bb{w}})}_{x}=\sum_{i\in\bb{w}}\bb{P_{i}}_{x}.
$$
But we have that
\begin{multline}\label{Eq1}
\left(\!\begin{smallmatrix}(x,t)\\(Q^{x,t}_{\alpha})_{\alpha\in \bb{\sigma((P_{i})_{i\in\bb{w}})}_{x}}\end{smallmatrix}\!\right)_{(x,t)\in \coprod X}(\sigma((P_{i})_{i\in\bb{w}})) = \\
\sigma\left(\left(\left(\!\begin{smallmatrix}(x,t)\\(Q^{x,t}_{\alpha+\sum_{k\in i}\bb{P_{k}}_{x}})_{\alpha\in \bb{P_{i}}_{x}}\end{smallmatrix}\!\right)_{(x,t)\in \coprod X}(P_{i})\right)_{i\in\bb{w}}\right)
\end{multline}
and, by induction hypothesis, for every $i\in\bb{w}$, $\left(\!\begin{smallmatrix}(x,t)\\(Q^{x,t}_{\alpha+\sum_{k\in i}\bb{P_{k}}_{x}})_{\alpha\in \bb{P_{i}}_{x}}\end{smallmatrix}\!\right)_{(x,t)\in \coprod X}(P_{i})$ is a term in $\mathrm{T}_{\Sigma}(X)_{w_{i}}$. Therefore the left side of the above equation is a term in $\mathrm{T}_{\Sigma}(X)_{s}$. Consequently $\mathcal{T}$ is a subalgebra of $\mathbf{T}_{\Sigma}(X)$. Thus $\mathcal{T} = \mathrm{T}_{\Sigma}(X)$. Hereby completing our proof.
\end{proof}

\begin{corollary}
Let $X$ be an $S$-sorted set, $u\in S$, $z\in X_{u}$, and $P$ a term in $\mathrm{T}_{\Sigma}(X)_{s}$. Then, for every
$(Q^{z}_{\alpha})_{\alpha\in \bb{P}_{z}}$ in $\mathrm{T}_{\Sigma}(X)_{u}^{\bb{P}_{z}}$, we have that  $\left(\!\begin{smallmatrix}z\\(Q^{z}_{\alpha})_{\alpha\in \bb{P}_{z}}\end{smallmatrix}\!\right)(P)$ is a term in $\mathrm{T}_{\Sigma}(X)_{s}$. Therefore $\left(\!\begin{smallmatrix}z\\ \cdot\end{smallmatrix}\!\right)(P)$ is a mapping from $\mathrm{T}_{\Sigma}(X)^{\bb{P}_{z}}_{u}$ to $\mathrm{T}_{\Sigma}(X)_{s}$.
\end{corollary}

\begin{remark}
For every $S$-sorted set $X$ the family $\left(\left(\!\begin{smallmatrix}(x,t)\\ \cdot\end{smallmatrix}\!\right)_{(x,t)\in \coprod X}(P)\right)_{s\in S}$ of global substitution operators is in
$\prod_{(P,s)\in  \coprod\mathrm{T}_{\Sigma}(X)}\mathrm{Hom}(\prod_{(x,t)\in \coprod X}\mathrm{T}_{\Sigma}(X)_{t}^{\bb{P}_{x}},\mathrm{T}_{\Sigma}(X)_{s})$.
%\textcolor{red}{Los operadores de substituciÛn est·n definidos para ·lgebras libres. Por lo tanto, para los cocientes de las mismas, no disponemos, en principio, de ellos. Adem·s, para que tuviera sentido considerar operadores de substituciÛn relativos a cocientes, se tendrÌa que demostrar que tales operadores son independientes de los representates de clase elegidos, lo cual es absurdo, ya que las substituciones tienen un caracter puramente \emph{sint·ctico}.}
\end{remark}

We next prove that every homomorphism from $\mathbf{T}_{\Sigma}(X)$ is compatible with the substitution operator.

\begin{lemma}\label{LIGlobSubstOp}
Let $X$ be an $S$-sorted set, $s\in S$, $W\in\mathrm{T}_{\Sigma}(X)_{s}$, $\mathbf{A}$ a $\Sigma$-algebra, and $g$ a homomorphism from $\mathbf{T}_{\Sigma}(X)$ to $\mathbf{A}$. Then, for every $\left((P^{x,t}_{\alpha})_{\alpha\in \bb{W}_{x}}\right)_{(x,t)\in \coprod X}$, $\left((Q^{x,t}_{\alpha})_{\alpha\in \bb{W}_{x}}\right)_{(x,t)\in \coprod X}\in\prod_{(x,t)\in \coprod X}\mathrm{T}_{\Sigma}(X)_{t}^{\bb{W}_{x}}$, if, for every $(x,t)\in\coprod X$ and every  $\alpha\in\bb{W}_{x}$, we have that $g_{t}(P^{x,t}_{\alpha})=g_{t}(Q^{x,t}_{\alpha})$, then
$$
g_{s}\left(\left(\!\begin{smallmatrix}(x,t)\\(P^{x,t}_{\alpha})_{\alpha\in \bb{W}_{x}}\end{smallmatrix}\!\right)_{(x,t)\in \coprod X}(W)\right)
=
g_{s}\left(\left(\!\begin{smallmatrix}(x,t)\\(Q^{x,t}_{\alpha})_{\alpha\in \bb{W}_{x}}\end{smallmatrix}\!\right)_{(x,t)\in \coprod X}(W)\right).
$$
\end{lemma}

\begin{proof}
We prove it by algebraic induction on $W$. But for $W$, either (1) $W=z$, for a unique $z\in X_{s}$ , or (2) $W=\sigma$, for a unique $\sigma\in\Sigma_{\lambda, s}$, or (3) $W=\sigma((W_{i})_{i\in\bb{w}})$, for a unique $w\in S^{\star}-\{\lambda\}$, a unique $\sigma\in\Sigma_{w,s}$, and a unique family $(W_{i})_{i\in\bb{w}}\in\mathrm{T}_{\Sigma}(X)_{w}$.

In case (1), $\bb{W}_{z}=1$, whilst for every $x\in X_{s}-\{z\}$, $\bb{W}_{x}=0$, and for every $t\in S-\{s\}$ and every $x\in X_{t}$, $\bb{W}_{x}=0$. Therefore, for two terms $P^{z,s}$ and $Q^{z,s}$ in $\mathrm{T}_{\Sigma}(X)_{s}$ with $g_{s}(P^{z,s})=g(Q^{z,s})$ then it holds that
$$
g_{s}(P^{z,s})=g_{s}\left(\left(\!\begin{smallmatrix}z\\P^{z,s}_{\alpha}\end{smallmatrix}\!\right)(z)\right)=
g_{s}\left(\left(\!\begin{smallmatrix}z\\Q^{z,s}_{\alpha}\end{smallmatrix}\!\right)(z)\right)=g_{s}(Q^{z,s}).
$$

In case (2), we note that no constant symbol $\sigma\in\Sigma_{\lambda,s}$ has variables  in $X$ and, consequently no proper substitution is made. To prove the statement, we need to show that $g_{s}(\sigma)=g_{s}(\sigma)$, which trivially holds.

Finally, in case (3), we have the following equations

$g_{s}\left(\left(\!\begin{smallmatrix}(x,t)\\(P^{x,t}_{\alpha})_{\alpha\in \bb{W}_{x}}\end{smallmatrix}\!\right)_{(x,t)\in \coprod X}(W)\right)$
\begin{alignat}{2}
&=g_{s}\left(\left(\!\begin{smallmatrix}(x,t)\\(P^{x,t}_{\alpha})_{\alpha\in \bb{\sigma((W_{i})_{i\in\bb{w}})}_{x}}\end{smallmatrix}\!\right)_{(x,t)\in \coprod X}(\sigma((W_{i})_{i\in\bb{w}}))\right)
\notag\\
&=g_{s}\left(
\sigma\left(\left(\left(\!\begin{smallmatrix}(x,t)\\(P^{x,t}_{\alpha+\sum_{k\in i}\bb{W_{k}}_{x}})_{\alpha\in \bb{W_{i}}_{x}}\end{smallmatrix}\!\right)_{(x,t)\in \coprod X}(W_{i})\right)_{i\in\bb{w}}\right)
\right)&&(\dagger_{1})
\notag\\
&=\sigma^{\mathbf{A}}\left(\left(g_{w_{i}}
\left(\left(\!\begin{smallmatrix}(x,t)\\(P^{x,t}_{\alpha+\sum_{k\in i}\bb{W_{k}}_{x}})_{\alpha\in \bb{W_{i}}_{x}}\end{smallmatrix}\!\right)_{(x,t)\in \coprod X}(W_{i})
\right)
\right)_{i\in\bb{w}}
\right)&&(\dagger_{2})
\notag\\
&=\sigma^{\mathbf{A}}\left(\left(g_{w_{i}}
\left(\left(\!\begin{smallmatrix}(x,t)\\(Q^{x,t}_{\alpha+\sum_{k\in i}\bb{W_{k}}_{x}})_{\alpha\in \bb{W_{i}}_{x}}\end{smallmatrix}\!\right)_{(x,t)\in \coprod X}(W_{i})
\right)
\right)_{i\in\bb{w}}
\right)&&(\dagger_{3})
\notag\\
&=g_{s}\left(
\sigma\left(\left(\left(\!\begin{smallmatrix}(x,t)\\(Q^{x,t}_{\alpha+\sum_{k\in i}\bb{W_{k}}_{x}})_{\alpha\in \bb{W_{i}}_{x}}\end{smallmatrix}\!\right)_{(x,t)\in \coprod X}(W_{i})\right)_{i\in\bb{w}}\right)
\right)&&(\dagger_{2})
\notag\\
&=g_{s}\left(\left(\!\begin{smallmatrix}(x,t)\\(Q^{x,t}_{\alpha})_{\alpha\in \bb{\sigma((W_{i})_{i\in\bb{w}})}_{x}}\end{smallmatrix}\!\right)_{(x,t)\in \coprod X}(\sigma((W_{i})_{i\in\bb{w}}))\right)&&(\dagger_{1})
\notag\\
&=g_{s}\left(\left(\!\begin{smallmatrix}(x,t)\\(Q^{x,t}_{\alpha})_{\alpha\in \bb{W}_{x}}\end{smallmatrix}\!\right)_{(x,t)\in \coprod X}(W)\right),
\notag
\end{alignat}
where, to shorten notation, we let $(\dagger_{1})$, $(\dagger_{2})$, and $(\dagger_{3})$ stand for (by equation~\ref{Eq1}), (by definition of homomorphism), and (by inductive hypothesis), respectively.
\end{proof}

As a consequence of the last lemma we have the following corollary, that states that every congruence on $\mathbf{T}_{\Sigma}(X)$ is compatible with the substitution operator.

\begin{corollary} Let $X$ be an $S$-sorted set, $s\in S$, $W\in\mathrm{T}_{\Sigma}(X)_{s}$, and $\Phi$ be a congruence on $\mathbf{T}_{\Sigma}(X)$. Then, for every $\left((P^{x,t}_{\alpha})_{\alpha\in \bb{W}_{x}}\right)_{(x,t)\in \coprod X}$, $\left((Q^{x,t}_{\alpha})_{\alpha\in \bb{W}_{x}}\right)_{(x,t)\in \coprod X}\in\prod_{(x,t)\in \coprod X}\mathrm{T}_{\Sigma}(X)_{t}^{\bb{W}_{x}}$, if, for every $(x,t)\in\coprod X$ and every  $\alpha\in\bb{W}_{x}$, we have that $[P^{x,t}_{\alpha}]_{\Phi_{t}}=[Q^{x,t}_{\alpha}]_{\Phi_{t}}$, then
$$
\left[\left(\!\begin{smallmatrix}(x,t)\\(P^{x,t}_{\alpha})_{\alpha\in \bb{W}_{x}}\end{smallmatrix}\!\right)_{(x,t)\in \coprod X}(W)\right]_{\Phi_{s}}
=
\left[\left(\!\begin{smallmatrix}(x,t)\\(Q^{x,t}_{\alpha})_{\alpha\in \bb{W}_{x}}\end{smallmatrix}\!\right)_{(x,t)\in \coprod X}(W)\right]_{\Phi_{s}}.
$$
\end{corollary}

Since it will be used in the following definition we recall that, as a particular case of Proposition~\ref{subsetalg}, for an $S$-sorted set $X$, we have the power $\Sigma$-algebra $\mathbf{T}_{\Sigma}(X)^{\wp}$ associated to   $\mathbf{T}_{\Sigma}(X)$, which has as underlying $S$-sorted set $\mathrm{T}_{\Sigma}(X)^{\wp} = (\mathrm{Sub}(\mathrm{T}_{\Sigma}(X)_{s}))_{s\in S}$ and, for every $(w,s)\in S^{\star}\times S$ and every $\sigma\in \Sigma_{w,s}$, as structural operation associated to $\sigma$, the mapping $\sigma^{\wp}$ from $\mathrm{T}_{\Sigma}(X)^{\wp}_{w} = \prod_{i\in \bb{w}}\mathrm{Sub}(\mathrm{T}_{\Sigma}(X)_{w_{i}})$ to $\mathrm{T}_{\Sigma}(X)^{\wp}_{s} = \mathrm{Sub}(\mathrm{T}_{\Sigma}(X)_{s})$ that sends $(L_{i})_{i\in \bb{w}}$ in $\mathrm{T}_{\Sigma}(X)^{\wp}_{w}$ to
$$
\textstyle
\sigma^{\wp}((L_{i})_{i\in \bb{w}}) = \{\sigma((P_{i})_{i\in \bb{w}})\mid (P_{i})_{i\in \bb{w}}\in \prod_{i\in \bb{w}}L_{i}\}
$$
in $\mathrm{T}_{\Sigma}(X)^{\wp}_{s}$. We remind the reader that in accordance with what is established in  Proposition~\ref{rut}, we have let, for abbreviation, for every $(w,s)\in S^{\star}\times S$ and every $\sigma\in\Sigma_{w,s}$, $\sigma$ stand for $F_{\sigma}^{\mathbf{T}_{\Sigma}(X)}$, the structural operation of $\mathbf{T}_{\Sigma}(X)$ associated to $\sigma$.

\begin{definition}
For every $t\in S$, let $(L_{x})_{x\in X_{t}}$ be a mapping from $X_{t}$ to $\mathrm{T}_{\Sigma}(X)^{\wp}_{t} = \mathrm{Sub}(\mathrm{T}_{\Sigma}(X)_{t})$, also written, in this context, as $\left(\!\begin{smallmatrix}x\\L_{x}\end{smallmatrix}\!\right)_{x\in X_{t}}$. Then, for the $S$-sorted mapping $\left(\left(\!\begin{smallmatrix}x\\L_{x}\end{smallmatrix}\!\right)_{x\in X_{t}}\right)_{t\in S}$ from $X$ to $\mathrm{T}_{\Sigma}(X)^{\wp}$, we will denote by  $\left(\left(\left(\!\begin{smallmatrix}x\\L_{x}\end{smallmatrix}\!\right)_{x\in X_{t}}\right)_{t\in S}\right)^{\sharp}$ the unique homomorphism from $\mathbf{T}_{\Sigma}(X)$ to $\mathbf{T}_{\Sigma}(X)^{\wp}$ such that
$$
\left(\left(\left(\!\begin{smallmatrix}x\\L_{x}\end{smallmatrix}\!\right)_{x\in X_{t}}\right)_{t\in S}\right)^{\sharp}\circ \eta_ {X} = \left(\left(\!\begin{smallmatrix}x\\L_{x}\end{smallmatrix}\!\right)_{x\in X_{t}}\right)_{t\in S}.
$$
Let $u$ be a sort in $S$, $z$ an element of $X_{u}$, and $L\in \mathrm{T}_{\Sigma}(X)^{\wp}_{u}$. Then, for the $S$-sorted mapping $\left(\!\begin{smallmatrix}z\\L\end{smallmatrix}\!\right)$ from $X$ to $\mathrm{T}_{\Sigma}(X)^{\wp}$ defined as: $\left(\!\begin{smallmatrix}z\\L\end{smallmatrix}\!\right)_{u}$ is the mapping from $X_{u}$ to $\mathrm{Sub}(\mathrm{T}_{\Sigma}(X)_{u})$ that sends $z$ to $L$ and $y\in X_{u}-\{z\}$ to $\{y\}$, while, for $t\in S-\{u\}$, $\left(\!\begin{smallmatrix}z\\L\end{smallmatrix}\!\right)_{t}$ is the mapping from $X_{t}$ to $\mathrm{Sub}(\mathrm{T}_{\Sigma}(X)_{t})$ that sends $x\in X_{t}$ to $\{x\}$, we will denote by $\left(\!\begin{smallmatrix}z\\L\end{smallmatrix}\!\right)^{\sharp}$ the unique homomorphism from $\mathbf{T}_{\Sigma}(X)$ to $\mathbf{T}_{\Sigma}(X)^{\wp}$ such that $\left(\!\begin{smallmatrix}z\\L\end{smallmatrix}\!\right)^{\sharp}\circ \eta_ {X} = \left(\!\begin{smallmatrix}z\\L\end{smallmatrix}\!\right)$.
\end{definition}

We recall that the homomorphism $\left(\left(\left(\!\begin{smallmatrix}x\\L_{x}\end{smallmatrix}\!\right)_{x\in X_{t}}\right)_{t\in S}\right)^{\sharp}$ from  $\mathbf{T}_{\Sigma}(X)$ to $\mathbf{T}_{\Sigma}(X)^{\wp}$ is defined, by algebraic recursion, as follows. Let $s$ be a sort in $S$ and $P\in\mathrm{T}_{\Sigma}(X)_{s}$. Then we know that $P$ either has the form (1) $z$, for a unique $z\in X_{s}$, or (2), $\sigma$, for a unique $\sigma\in\Sigma_{\lambda, s}$, or (3) $\sigma((P_{i})_{i\in \bb{w}})$, for a unique $w\in S^{\star}-\{\lambda\}$, a unique $\sigma\in \Sigma_{w,s}$, and a unique family $(P_{i})_{i\in\bb{w}}\in\mathrm{T}_{\Sigma}(X)_{w}$.

In case (1), we have that $\left(\left(\left(\!\begin{smallmatrix}x\\L_{x}\end{smallmatrix}\!\right)_{x\in X_{t}}\right)_{t\in S}\right)^{\sharp}_{s}(z) = L_{z}$.

In case (2), we have that $\left(\left(\left(\!\begin{smallmatrix}x\\L_{x}\end{smallmatrix}\!\right)_{x\in X_{t}}\right)_{t\in S}\right)^{\sharp}_{s}(\sigma) = \{\sigma\}$.

Finally, in case (3), and under the hypothesis that, for every $i\in \bb{w}$, the subset $\left(\left(\left(\!\begin{smallmatrix}x\\L_{x}\end{smallmatrix}\!\right)_{x\in X_{t}}\right)_{t\in S}\right)^{\sharp}_{w_{i}}(P_{i})$ of $\mathrm{T}_{\Sigma}(X)_{w_{i}}$ has been defined, we have that
\begin{multline*}
 \left(\left(\left(\!\begin{smallmatrix}x\\L_{x}\end{smallmatrix}\!\right)_{x\in X_{t}}\right)_{t\in S}\right)^{\sharp}_{s}(\sigma((P_{i})_{i\in \bb{w}})) = \\
\sigma^{\wp}\left(\left(\left(\left(\left(\!\begin{smallmatrix}x\\L_{x}\end{smallmatrix}\!\right)_{x\in X_{t}}\right)_{t\in S}\right)^{\sharp}_{w_{i}}(P_{i})\right)_{i\in \bb{w}}\right) = \\
\textstyle\Big\lbrace\sigma((Q_{i})_{i\in \bb{w}})\,\,\Big|\,\, (Q_{i})_{i\in \bb{w}}\in \prod_{i\in \bb{w}}\left(\left(\left(\!\begin{smallmatrix}x\\L_{x}\end{smallmatrix}\!\right)_{x\in X_{t}}\right)_{t\in S}\right)^{\sharp}_{w_{i}}(P_{i})\Big\rbrace.
\end{multline*}

%\begin{align}
%\left(\left(\left(\!\begin{smallmatrix}x\\L_{x}\end{smallmatrix}\!\right)_{x\in X_{t}}\right)_{t\in S}\right)^{\sharp}(\sigma((P_{i})_{i\in \bb{w}})) &= \sigma^{\wp}\left(\left(\left(\left(\left(\!\begin{smallmatrix}x\\L_{x}\end{smallmatrix}\!\right)_{x\in X_{t}}\right)_{t\in S}\right)^{\sharp}(P_{i})\right)_{i\in \bb{w}}\right)\notag\\
%{}                               &= \textstyle\Big\lbrace\sigma((Q_{i})_{i\in \bb{w}})\mid (Q_{i})_{i\in \bb{w}}\in \prod_{i\in \bb{w}}\left(\left(\left(\!\begin{smallmatrix}x\\L_{x}\end{smallmatrix}\!\right)_{x\in X_{t}}\right)_{t\in S}\right)^{\sharp}(P_{i})\Big\rbrace. \notag
%\end{align}

We next prove that, for every $S$-sorted mapping $\left(\left(\!\begin{smallmatrix}x\\L_{x}\end{smallmatrix}\!\right)_{x\in X_{t}}\right)_{t\in S}$ from $X$ to $\mathrm{T}_{\Sigma}(X)^{\wp}$, every $s\in S$, and every $P\in \mathrm{T}_{\Sigma}(X)_{s}$, $\left(\left(\left(\!\begin{smallmatrix}x\\L_{x}\end{smallmatrix}\!\right)_{x\in X_{t}}\right)_{t\in S}\right)^{\sharp}_{s}(P)$ is given by collecting together the terms $\left(\!\begin{smallmatrix}(x,t)\\(Q^{x,t}_{\alpha})_{\alpha\in \bb{P}_{x}}\end{smallmatrix}\!\right)_{(x,t)\in \coprod X}(P)$ with $\left((Q^{x,t}_{\alpha})_{\alpha\in \bb{P}_{x}}\right)_{(x,t)\in \coprod X}$ varying in $\prod_{(x,t)\in \coprod X}L_{x}^{\bb{P}_{x}}$.
Let us note that, since, for every $x\in X_{t}$, $L_{x}^{\bb{P}_{x}}$ is embedded into $\mathrm{T}_{\Sigma}(X)_{t}^{\bb{P}_{x}}$, $\prod_{x\in X_{t}}L_{x}^{\bb{P}_{x}}$ is also embedded into $\prod_{x\in X_{t}}\mathrm{T}_{\Sigma}(X)_{t}^{\bb{P}_{x}}$. Therefore $\prod_{(x,t)\in \coprod X}L_{x}^{\bb{P}_{x}}$ is embedded into $\prod_{(x,t)\in \coprod X}\mathrm{T}_{\Sigma}(X)_{t}^{\bb{P}_{x}}$ and, consequently, the restriction of $\left(\!\begin{smallmatrix}(x,t)\\ \cdot\end{smallmatrix}\!\right)_{(x,t)\in \coprod X}(P)$ to $\prod_{(x,t)\in \coprod X}L_{x}^{\bb{P}_{x}}$ is available.

\begin{proposition}
Let $X$ be an $S$-sorted set, $s\in S$, and $P$ a term in $\mathrm{T}_{\Sigma}(X)_{s}$. Then, for every $S$-sorted mapping $\left(\left(\!\begin{smallmatrix}x\\L_{x}\end{smallmatrix}\!\right)_{x\in X_{t}}\right)_{t\in S}$ from $X$ to $\mathrm{T}_{\Sigma}(X)^{\wp}$, we have that
$$
\left(\left(\left(\!\begin{smallmatrix}x\\L_{x}\end{smallmatrix}\!\right)_{x\in X_{t}}\right)_{t\in S}\right)^{\sharp}_{s}(P) = \mathrm{Im}\left(\left(\!\begin{smallmatrix}(x,t)\\ \cdot\end{smallmatrix}\!\right)_{(x,t)\in \coprod X}(P){\textstyle\upharpoonright}_{\prod_{(x,t)\in \coprod X}L_{x}^{\bb{P}_{x}}}\right).
$$
\end{proposition}

\begin{proof}
Let $\mathcal{T}$ be the subset of $\mathrm{T}_{\Sigma}(X)$ defined, for every $s\in S$, as follows: For every $P\in \mathrm{T}_{\Sigma}(X)_{s}$, $P\in \mathcal{T}_{s}$ if and only if
$$
\left(\left(\left(\!\begin{smallmatrix}x\\L_{x}\end{smallmatrix}\!\right)_{x\in X_{t}}\right)_{t\in S}\right)^{\sharp}_{s}(P) = \mathrm{Im}\left(\left(\!\begin{smallmatrix}(x,t)\\ \cdot\end{smallmatrix}\!\right)_{(x,t)\in \coprod X}(P){\textstyle\upharpoonright}_{\prod_{(x,t)\in \coprod X}L_{x}^{\bb{P}_{x}}}\right).
$$
%$$
%\mathcal{T}_{s} = \Big\lbrace P\in \mathrm{T}_{\Sigma}(X)_{s}\,\,\Big|\,\,\left(\left(\left(\!\begin{smallmatrix}x\\L_{x}\end{smallmatrix}\!\right)_{x\in X_{t}}\right)_{t\in S}\right)^{\sharp}_{s}(P) = \mathrm{Im}\left(\left(\!\begin{smallmatrix}(x,t)\\ \cdot\end{smallmatrix}\!\right)_{(x,t)\in \coprod X}(P){\textstyle\upharpoonright}_{\prod_{(x,t)\in \coprod X}L_{x}^{\bb{P}_{x}}}\right)\Big\rbrace.
%$$
To prove that $\mathcal{T} = \mathrm{T}_{\Sigma}(X)$ it suffices to show, by Proposition~\ref{PPAI}, that $X\subseteq \mathcal{T}$ and that $\mathcal{T}$ is a subalgebra of $\mathbf{T}_{\Sigma}(X)$.

We first prove that $X\subseteq \mathcal{T}$.

Let $s$ be a sort in $S$ and $z\in X_{s}$. Then $\left(\left(\left(\!\begin{smallmatrix}x\\L_{x}\end{smallmatrix}\!\right)_{x\in X_{t}}\right)_{t\in S}\right)^{\sharp}_{s}(z) = L_{z}$. On the other hand, we have that $\bb{z}_{z}=1$, whilst, for every $x\in X_{s}-\{z\}$, $\bb{z}_{x}=0$, and, for every $t\in S-\{s\}$ and every $x\in X_{t}$ $\bb{z}_{x}=0$. Therefore,  since, for every $Q\in L_{z}$, it happens that $\left(\!\begin{smallmatrix}z\\Q\end{smallmatrix}\!\right)(z) = Q$, we have that $\mathrm{Im}\left(\left(\!\begin{smallmatrix}(x,t)\\ \cdot\end{smallmatrix}\!\right)_{(x,t)\in \coprod X}(P)\right) = L_{z}$. Thus $\mathrm{Im}\left(\left(\!\begin{smallmatrix}(x,t)\\ \cdot\end{smallmatrix}\!\right)_{(x,t)\in \coprod X}(P){\textstyle\upharpoonright}_{\prod_{(x,t)\in \coprod X}L_{x}^{\bb{P}_{x}}}\right) = L_{z}$. Hence both sets coincide. Consequently $X\subseteq \mathcal{T}$.

We next prove that $\mathcal{T}$ is a subalgebra of $\mathbf{T}_{\Sigma}(X)$.

Let $s$ be a sort in $S$ and $\sigma\in\Sigma_{\lambda,s}$. Then $\left(\left(\left(\!\begin{smallmatrix}x\\L_{x}\end{smallmatrix}\!\right)_{x\in X_{t}}\right)_{t\in S}\right)^{\sharp}_{s}(\sigma) = \{\sigma\}$, which is the interpretation of the constant symbol $\sigma$ in the subset algebra. On the other hand, $\mathrm{Im}\left(\left(\!\begin{smallmatrix}(x,t)\\ \cdot\end{smallmatrix}\!\right)_{(x,t)\in \coprod X}(P){\textstyle\upharpoonright}_{\prod_{(x,t)\in \coprod X}L_{x}^{\bb{P}_{x}}}\right)$, which is the result of collecting together the terms $\left(\!\begin{smallmatrix}(x,t)\\(Q^{x,t}_{\alpha})_{\alpha\in \bb{P}_{x}}\end{smallmatrix}\!\right)_{(x,t)\in \coprod X}(\sigma)$ when $\left((Q^{x,t}_{\alpha})_{\alpha\in \bb{P}_{x}}\right)_{(x,t)}$ varies in $\prod_{(x,t)\in\coprod X}L_{x}^{\bb{P}_{x}}$, is, simply, $\{\sigma\}$. Consequently $$
\mathrm{Im}\left(\left(\!\begin{smallmatrix}(x,t)\\ \cdot\end{smallmatrix}\!\right)_{(x,t)\in \coprod X}(P){\textstyle\upharpoonright}_{\prod_{(x,t)\in \coprod X}L_{x}^{\bb{P}_{x}}}\right) = \{\sigma\}.
$$

Let $(w,s)\in (S^{\star}-\{\lambda\})\times S$, $\sigma\in\Sigma_{w,s}$, and let $(P_{i})_{i\in\bb{w}}\in \mathrm{T}_{\Sigma}(X)_{w}$ be such that, for every $i\in \bb{w}$, $P_{i}$ satisfies the given requirement. Then, since $\left(\left(\left(\!\begin{smallmatrix}x\\L_{x}\end{smallmatrix}\!\right)_{x\in X_{t}}\right)_{t\in S}\right)^{\sharp}$ is a homomorphism from $\mathbf{T}_{\Sigma}(X)$ to $\mathbf{T}_{\Sigma}(X)^{\wp}$, we have that
$$
\left(\left(\left(\!\begin{smallmatrix}x\\L_{x}\end{smallmatrix}\!\right)_{x\in X_{t}}\right)_{t\in S}\right)^{\sharp}_{s}(\sigma((P_{i})_{i\in\bb{w}})) = \sigma^{\wp}\left(\left(\left(\left(\left(\!\begin{smallmatrix}x\\L_{x}\end{smallmatrix}\!\right)_{x\in X_{t}}\right)_{t\in S}\right)^{\sharp}_{w_{i}}(P_{i})\right)_{i\in \bb{w}}\right).
$$
On the other hand, by the induction hypothesis, bearing in mind that, for every $t\in S$, every $x\in X_{t}$, $\bb{\sigma((P_{i})_{i\in n})}_{x} = \sum_{i\in n}\bb{P_{i}}_{x}$, and taking into account Equation~\ref{Eq1}, we have that
$$
\mathrm{Im}\left(\left(\!\begin{smallmatrix}(x,t)\\ \cdot\end{smallmatrix}\!\right)_{(x,t)\in \coprod X}(P){\textstyle\upharpoonright}_{\prod_{(x,t)\in \coprod X}L_{x}^{\bb{P}_{x}}}\right) = \sigma^{\wp}\left(\left(\left(\left(\left(\!\begin{smallmatrix}x\\L_{x}\end{smallmatrix}\!\right)_{x\in X_{t}}\right)_{t\in S}\right)^{\sharp}_{w_{i}}(P_{i})\right)_{i\in \bb{w}}\right).
$$
Consequently $\mathcal{T}$ is a subalgebra of $\mathbf{T}_{\Sigma}(X)$. Thus $\mathcal{T} = \mathrm{T}_{\Sigma}(X)$. Hereby completing our proof.
\end{proof}

\begin{remark}
Given an $S$-sorted mapping $\left(\left(\!\begin{smallmatrix}x\\L_{x}\end{smallmatrix}\!\right)_{x\in X_{t}}\right)_{t\in S}$ from $X$ to $\mathrm{T}_{\Sigma}(X)^{\wp}$, and $P$ a term in $\mathrm{T}_{\Sigma}(X)_{s}$, for some $s\in S$, we have that
$
\textstyle
\left(\left(\left(\!\begin{smallmatrix}x\\L_{x}\end{smallmatrix}\!\right)_{x\in X_{t}}\right)_{t\in S}\right)^{\sharp}_{s}(P)
$
is equal to
$$
\textstyle
\bigcup_{\left((Q^{x,t}_{\alpha})_{\alpha\in \bb{P}_{x}}\right)_{(x,t)\in \coprod X}\in \prod_{(x,t)\in \coprod X}L_{x}^{\bb{P}_{x}}}\left\lbrace\left(\!\begin{smallmatrix}(x,t)\\(Q^{x,t}_{\alpha})_{\alpha\in \bb{P}_{x}}\end{smallmatrix}\!\right)_{(x,t)\in \coprod X}(P)\right\rbrace.
$$
In other words, the value of $\left(\left(\left(\!\begin{smallmatrix}x\\L_{x}\end{smallmatrix}\!\right)_{x\in X_{t}}\right)_{t\in S}\right)^{\sharp}_{s}$ at $P$ is the union of the image of the mapping $\{\cdot\}_{\mathrm{T}_{\Sigma}(X)_{s}}\circ \left(\left(\!\begin{smallmatrix}(x,t)\\ \cdot\end{smallmatrix}\!\right)_{(x,t)\in \coprod X}(P)\!\upharpoonright_{\prod_{(x,t)\in \coprod X}L_{x}^{\bb{P}_{x}}}\right)$ from $\prod_{(x,t)\in \coprod X}L_{x}^{\bb{P}_{x}}$ to $\mathrm{Sub}(\mathrm{T}_{\Sigma}(X))$, where $\left(\!\begin{smallmatrix}(x,t)\\ \cdot\end{smallmatrix}\!\right)_{(x,t)\in \coprod X}(P)\!\upharpoonright_{\prod_{(x,t)\in \coprod X}L_{x}^{\bb{P}_{x}}}$ is the restriction of the  substitution mapping $\left(\!\begin{smallmatrix}(x,t)\\ \cdot\end{smallmatrix}\!\right)_{(x,t)\in \coprod X}(P)$ to $\prod_{(x,t)\in \coprod X}L_{x}^{\bb{P}_{x}}$, and $\{\cdot\}_{\mathrm{T}_{\Sigma}(X)_{s}}$ the canonical embedding of $\mathrm{T}_{\Sigma}(X)_{s}$ into $\mathrm{T}_{\Sigma}(X)_{s}^{\wp}$.
\end{remark}

\begin{corollary}
Let $s$ and $u$  be sorts in $S$, $z\in X_{u}$, $L\in \mathrm{T}_{\Sigma}(X)^{\wp}_{u}$, and $P$ a term in $\mathrm{T}_{\Sigma}(X)_{s}$. Then
$$
\textstyle
\left(\!\begin{smallmatrix}z\\L\end{smallmatrix}\!\right)^{\sharp}_{s}(P) = \mathrm{Im}\left(\left(\!\begin{smallmatrix}z\\ \cdot\end{smallmatrix}\!\right)(P){\textstyle\upharpoonright}_{L^{\bb{P}_{z}}}\right).
$$
\end{corollary}

\begin{definition}
Let $\left(\left(\!\begin{smallmatrix}x\\L_{x}\end{smallmatrix}\!\right)_{x\in X_{t}}\right)_{t\in S}$ be an $S$-sorted mapping from $X$ to $\mathrm{T}_{\Sigma}(X)^{\wp}$. Then we will denote by  $\left(\left(\left(\!\begin{smallmatrix}x\\L_{x}\end{smallmatrix}\!\right)_{x\in X_{t}}\right)_{t\in S}\right)^{\sharp\mathfrak{p}}$ the $S$-sorted mapping from $\mathrm{T}_{\Sigma}(X)^{\wp}$ to $\mathrm{T}_{\Sigma}(X)^{\wp}$ that, for every $s\in S$, sends $K$ in $\mathrm{T}_{\Sigma}(X)^{\wp}_{s}$ to
$$
\textstyle
\left(\left(\left(\!\begin{smallmatrix}x\\L_{x}\end{smallmatrix}\!\right)_{x\in X_{t}}\right)_{t\in S}\right)^{\sharp\mathfrak{p}}_{s}(K) = \bigcup_{P\in K}\left(\left(\left(\!\begin{smallmatrix}x\\L_{x}\end{smallmatrix}\!\right)_{x\in X_{t}}\right)_{t\in S}\right)^{\sharp}_{s}(P)
$$
in $\mathrm{T}_{\Sigma}(X)^{\wp}_{s}$. Therefore $\left(\left(\left(\!\begin{smallmatrix}x\\L_{x}\end{smallmatrix}\!\right)_{x\in X_{t}}\right)_{t\in S}\right)^{\sharp\mathfrak{p}}$ is the canonical extension of the underlying $S$-sorted mapping of the homomorphism $\left(\left(\left(\!\begin{smallmatrix}x\\L_{x}\end{smallmatrix}\!\right)_{x\in X_{t}}\right)_{t\in S}\right)^{\sharp}$ from $\mathbf{T}_{\Sigma}(X)$ to $\mathbf{T}_{\Sigma}(X)^{\wp}$.

Let $u$ be a sort in $S$, $z\in X_{u}$, and $L\in \mathrm{T}_{\Sigma}(X)^{\wp}_{u}$. Then we will denote by $\left(\!\begin{smallmatrix}z\\L\end{smallmatrix}\!\right)^{\sharp\mathfrak{p}}$ the canonical extension of the underlying mapping of the homomorphism $\left(\!\begin{smallmatrix}z\\L\end{smallmatrix}\!\right)^{\sharp}$ from $\mathbf{T}_{\Sigma}(X)$ to $\mathbf{T}_{\Sigma}(X)^{\wp}$.
\end{definition}

\begin{remark}
Let $\left(\left(\!\begin{smallmatrix}x\\L_{x}\end{smallmatrix}\!\right)_{x\in X_{t}}\right)_{t\in S}$ be an $S$-sorted mapping from $X$ to $\mathrm{T}_{\Sigma}(X)^{\wp}$, $s\in S$, $K\subseteq \mathrm{T}_{\Sigma}(X)_{s}$, and $W\in\mathrm{T}_{\Sigma}(X)_{s}$. Then $W\in \left(\left(\left(\!\begin{smallmatrix}x\\L_{x}\end{smallmatrix}\!\right)_{x\in X_{t}}\right)_{t\in S}\right)^{\sharp\mathfrak{p}}_{s}(K)$ if and only if there are $P\in K$ and $\left((Q^{x,t}_{\alpha})_{\alpha\in \bb{P}_{x}}\right)_{(x,t)\in \coprod X}\in \prod_{(x,t)\in \coprod X}L_{x}^{\bb{P}_{x}}$ such that
$$
W = \left(\!\begin{smallmatrix}(x,t)\\(Q^{x,t}_{\alpha})_{\alpha\in \bb{P}_{x}}\end{smallmatrix}\!\right)_{(x,t)\in \coprod X}(P).
$$
In particular, for $z\in X_{s}$ and $L\in \mathrm{Sub}(\mathrm{T}_{\Sigma}(X)_{s})$, we have that $W\in \left(\!\begin{smallmatrix}z\\L\end{smallmatrix}\!\right)^{\sharp\mathfrak{p}}_{s}(K)$ if and only if there are $P\in K$ and $(Q^{z}_{\alpha})_{\alpha\in \bb{P}_{z}}\in L^{\bb{P}_{z}}$ such that $W = \left(\!\begin{smallmatrix}z\\(Q^{z}_{\alpha})_{\alpha\in \bb{P}_{z}}\end{smallmatrix}\!\right)(P)$.
\end{remark}

%\textcolor{red}{Para un conjunto de sorts $S$ finito, con $\mathrm{card}(S) = n$, un $S$-conjunto $X$ finito, y una congruencia $\Phi$ sobre $\mathbf{T}_{\Sigma}(X)$ de Ìndice finito, convenimos, si, para cada $i\in n$, $\mathrm{card}(\mathrm{T}_{\Sigma}(X)_{s_{i}}/\Phi_{s_{i}}) = k_{i}$, que un transversal $\mathcal{W}_{\Phi} = (\mathcal{W}_{\Phi_{s_{i}}})_{s_{i}\in S}$ es tal que, para cada $i\in n$, $\mathcal{W}_{\Phi_{s_{i}}} = \{W_{i,j}\mid j\in k_{i}\}$.}

We next prove the many-sorted version of Theorem~4.6, on p.~74, in \cite{GS84}. The proposition states that, for every sort $s\in S$, every $s$-recognizable language $K$, and every operator  $\left(\left(\!\begin{smallmatrix}x\\L_{x}\end{smallmatrix}\!\right)_{x\in X_{t}}\right)_{t\in S}$ such that, for every $t\in S$ and every $x\in X_{t}$, $L_{x}\in \mathrm{Rec}_{t}(\mathbf{T}_{\Sigma}(X))$,
$\left(\left(\left(\!\begin{smallmatrix}x\\L_{x}\end{smallmatrix}\!\right)_{x\in X_{t}}\right)_{t\in S}\right)^{\sharp\mathfrak{p}}_{s}(K)$ is $s$-recognizable.

\begin{assumption}
To prove the following proposition we will assume that $S$ and $X$ are finite.
\end{assumption}

%pp.~66--67
\begin{proposition}\label{PRecSubs}
Let $s$ be a sort in $S$, $K\in\mathrm{Rec}_{s}(\mathbf{T}_{\Sigma}(X))$, and $\left(\left(\!\begin{smallmatrix}x\\L_{x}\end{smallmatrix}\!\right)_{x\in X_{t}}\right)_{t\in S}$ an $S$-sorted mapping from $X$ to $(\mathrm{Rec}_{t}(\mathbf{T}_{\Sigma}(X)))_{t\in S}$. Then
$$
\left(\left(\left(\!\begin{smallmatrix}x\\L_{x}\end{smallmatrix}\!\right)_{x\in X_{t}}\right)_{t\in S}\right)^{\sharp\mathfrak{p}}_{s}(K)\in\mathrm{Rec}_{s}(\mathbf{T}_{\Sigma}(X)).
$$
\end{proposition}

\begin{proof}
Let  $\Phi$ be the congruence on $\mathbf{T}_{\Sigma}(X)$ defined as follows:
$$
\textstyle
\Phi = \bigcap((\bigcup_{t\in S}\{\Omega^{\mathbf{T}_{\Sigma}(X)}(\delta^{t,L_{x}})\mid x\in X_{t}\})\cup\{\Omega^{\mathbf{T}_{\Sigma}(X)}(\delta^{s,K})\}).
$$
By the assumption and Proposition~\ref{Filter}, $\Phi$ is of finite index. Moreover, for every $t\in S-\{s\}$ and every $x\in X_{t}$, $\Phi_{t}$ saturates $L_{x}$, and, for $t = s$, $\Phi_{s}$ saturates $K$.

From now on, for every $r\in S$, $k_{r}$ and $\mathcal{W}_{\Phi_{r}}=\{W_{r,l}\mid l\in k_{r}\}$ stand for the index of $\Phi_{r}$ and a fixed transversal of $\mathrm{T}_{\Sigma}(X)_{r}/\Phi_{r}$ in $\mathrm{T}_{\Sigma}(X)_{r}$, respectively. Moreover, $k$ and $\mathcal{W}_{\Phi}$ denote the $S$-sorted sets $(k_{r})_{r\in S}$ and $(\mathcal{W}_{\Phi_{r}})_{r\in S}$, respectively.

Let $\Psi=(\Psi_{r})_{r\in S}$ be the binary relation on $\mathrm{T}_{\Sigma}(X)$ defined as follows: For every $r\in S$, $\Psi_{r}$ is the subset of $\mathrm{T}_{\Sigma}(X)_{r}^{2}$ consisting of all ordered pairs $(P,Q)$ in $\mathrm{T}_{\Sigma}(X)_{r}^{2}$ such that the following two conditions are satisfied:
\begin{enumerate}
\item $(P,Q)\in\Phi_{r}$.
\item For every $l\in k_{r}$
\[\left(
P\in
\left(\left(\left(\!\begin{smallmatrix}x\\L_{x}\end{smallmatrix}\!\right)_{x\in X_{t}}\right)_{t\in S}\right)^{\sharp\mathfrak{p}}_{r}([W_{r,l}]_{\Phi_{r}})
\leftrightarrow
Q\in
\left(\left(\left(\!\begin{smallmatrix}x\\L_{x}\end{smallmatrix}\!\right)_{x\in X_{t}}\right)_{t\in S}\right)^{\sharp\mathfrak{p}}_{r}([W_{r,l}]_{\Phi_{r}})
\right).\]
\end{enumerate}

By definition, for every $r\in S$, $\Psi_{r}$ is a refinement of $\Phi_{r}$ and an equivalence relation on $\mathrm{T}_{\Sigma}(X)_{r}$. Moreover, for every $r\in S$, 	the index of $\Psi_{r}$ on $\mathrm{T}_{\Sigma}(X)_{r}$ is bounded by $k_{r}2^{k_{r}}$. Consequently, the $S$-sorted set $\mathrm{T}_{\Sigma}(X)/{\Psi}$ is finite.

Let us check that $\Psi$ is a congruence on $\mathbf{T}_{\Sigma}(X)$. Let $(w,u)\in(S^{\star}-\{\lambda\})\times S$, $\sigma\in \Sigma_{w,u}$, and let $(P_{i})_{i\in\bb{w}}$ and $(Q_{i})_{in\bb{w}}$ be sequences of terms in $\mathrm{T}_{\Sigma}(X)_{w}$ such that, for every $i\in\bb{w}$, $(P_{i}, Q_{i})\in \Psi_{w_{i}}$. We want to show that
$$
\left(\sigma((P_{i})_{i\in \bb{w}}), \sigma((Q_{i})_{i\in\bb{w}})\right)\in\Psi_{u}.
$$

Let us note that, by definition of $\Psi$, for every $i\in \bb{w}$, we have that $(P_{i}, Q_{i})\in\Phi_{w_{i}}$. Since $\Phi$ is a congruence on $\mathrm{T}_{\Sigma}(X)$, we conclude that $\left(\sigma((P_{i})_{i\in \bb{w}}), \sigma((Q_{i})_{i\in\bb{w}})\right)$ is a pair in $\Phi_{u}$, so $\sigma((P_{i})_{i\in \bb{w}})$ and $\sigma((Q_{i})_{i\in \bb{w}})$ satisfy the first condition for being related under $\Psi$.

Regarding the second condition, let $l$ be an element of $k_{u}$. Assume that
$$
\sigma((P_{i})_{i\in\bb{w}})\in\left(\left(\left(\!\begin{smallmatrix}x\\L_{x}\end{smallmatrix}\!\right)_{x\in X_{t}}\right)_{t\in S}\right)^{\sharp\mathfrak{p}}_{u}([W_{u,l}]_{\Phi_{u}}).
$$
Then there are $W^{\dagger}\in[W_{u,l}]_{\Phi_{u}}$ and
$\left((U^{x,t}_{\alpha})_{\alpha\in \bb{W^{\dagger}}_{x}}\right)_{(x,t)\in\coprod X}$ in $\prod_{(x,t)\in\coprod X}L_{x}^{\bb{W^{\dagger}}_{x}}$ such that
\begin{equation}\label{EqSubsHtg}
\sigma((P_{i})_{i\in\bb{w}})=\left(\!\begin{smallmatrix}(x,t)\\(U^{x,t}_{\alpha})_{\alpha\in \bb{W^{\dagger}}_{x}}\end{smallmatrix}\!\right)_{(x,t)\in \coprod X}(W^{\dagger}).
\end{equation}

But for $W^{\dagger}$, either (a) $\mathrm{Var}(W^{\dagger})\cap X=\varnothing^{S}$ or (b) $\mathrm{Var}(W^{\dagger})\cap X\neq\varnothing^{S}$.

In case (a), since for every $t\in S$ and $x\in X_{t}$, $\bb{W^{\dagger}}_{x}=0$, we derive from Equation~\ref{EqSubsHtg} that $W^{\dagger}=\sigma((P_{i})_{i\in\bb{w}})$. It follows that, for every $i\in\bb{w}$, for every $t\in S$ and $x\in X_{t}$, $\bb{P_{i}}_{x}=0$. Therefore, for every $i\in\bb{w}$, the term $P_{i}$ is a term in $\left(\left(\left(\!\begin{smallmatrix}x\\L_{x}\end{smallmatrix}\!\right)_{x\in X_{t}}\right)_{t\in S}\right)^{\sharp\mathfrak{p}}_{w_{i}}([P_{i}]_{\Phi_{w_{i}}})$. Thus, by assumption, for every $i\in\bb{w}$, the term $Q_{i}$ is in $\left(\left(\left(\!\begin{smallmatrix}x\\L_{x}\end{smallmatrix}\!\right)_{x\in X_{t}}\right)_{t\in S}\right)^{\sharp\mathfrak{p}}_{w_{i}}([P_{i}]_{\Phi_{w_{i}}})$. Hence, for every $i\in\bb{w}$, there are $W^{\ddagger_{i}}\in [P_{i}]_{\Phi_{w_{i}}}$ and $\left((V^{x,t,i}_{\beta})_{\beta\in \bb{W^{\ddagger_{i}}}_{x}}\right)_{(x,t)\in\coprod X}$ in $\prod_{(x,t)\in\coprod X}L_{x}^{\bb{W^{\ddagger_{i}}}_{x}}$ such that
$$
Q_{i}=\left(\!\begin{smallmatrix}(x,t)\\(V^{x,t,i}_{\alpha})_{\alpha\in \bb{W^{\ddagger_{i}}}_{x}}\end{smallmatrix}\!\right)_{(x,t)\in \coprod X}(W^{\ddagger_{i}}).
$$
Let us denote by $W^{\ddagger}$ the term $\sigma((W^{\ddagger_{i}})_{i\in \bb{w}})$. Since, for every $i\in \bb{w}$, $W^{\ddagger_{i}}$ is a term in $[P_{i}]_{\Phi_{w_{i}}}$ and $\Phi$ is a congruence on $\mathrm{T}_{\Sigma}(X)$, we have that $(W^{\dagger}, W^{\ddagger})\in\Phi_{w_{u}}$. Thus $W^{\ddagger}\in [W_{u,l}]_{\Phi_{u}}$. Moreover, let us note that, for every $t\in S$ and $x\in X_{t}$, $\bb{W^{\ddagger}}_{x}=\sum_{i\in \bb{w}}\bb{W^{\ddagger_{i}}}_{x}$. Let $\left((V^{x,t}_{\beta})_{\beta\in \bb{W^{\ddagger}}_{x}}\right)_{(x,t)\in\coprod X}$ be the element in $\prod_{(x,t)\in\coprod X}L_{x}^{\bb{W^{\ddagger}}_{x}}$ obtained by joining, in order, the family $\left(\left((V^{x,t,i}_{\beta})_{\beta\in \bb{W^{\ddagger_{i}}}_{x}}\right)_{(x,t)\in\coprod X}\right)_{i\in \bb{w}}$. Then, from Equation~\ref{Eq1}, it follows that
\begin{align*}
\sigma((Q_{i})_{i\in\bb{w}})&=
\sigma\left(\left(
\left(\!\begin{smallmatrix}(x,t)\\(V^{x,t,i}_{\alpha})_{\alpha\in \bb{W^{\ddagger_{i}}}_{x}}\end{smallmatrix}\!\right)_{(x,t)\in \coprod X}(W^{\ddagger_{i}})
\right)_{i\in \bb{w}}\right)\\
&=
\left(\!\begin{smallmatrix}(x,t)\\(V^{x,t}_{\alpha})_{\alpha\in \bb{W^{\ddagger}}_{x}}\end{smallmatrix}\!\right)_{(x,t)\in \coprod X}(W^{\ddagger}).
\end{align*}
Therefore, $\sigma((Q_{i})_{i\in\bb{w}})\in \left(\left(\left(\!\begin{smallmatrix}x\\L_{x}\end{smallmatrix}\!\right)_{x\in X_{t}}\right)_{t\in S}\right)^{\sharp\mathfrak{p}}_{u}([W_{u,l}]_{\Phi_{u}})$. The other implication follows by a similar reasoning, thus proving case (a).

In case (b), either (b.1) $W^{\dagger}$ is a variable or (b.2) $W^{\dagger}$ has the form of an operation symbol applied to a suitable family of terms. Let us note that the case in which $W^{\dagger}$ is a constant is excluded because a term of such type does not contain any variable.

In case (b.1), we have that $W^{\dagger}=z$, for some $z\in X_{u}$. Thus, $\bb{W^{\dagger}}_{z}=1$, whilst, for every $x\in X_{u}-\{z\}$, $\bb{W^{\dagger}}_{x}=0$, and, for every $r\in S-\{u\}$ and every $x\in X_{r}$, $\bb{W^{\dagger}}_{x}=0$. Therefore, from Equation~\ref{EqSubsHtg} we obtain the following equation
$$
\left(\!\begin{smallmatrix}(x,t)\\(U^{x,t}_{\alpha})_{\alpha\in \bb{W^{\dagger}}_{x}}\end{smallmatrix}\!\right)_{(x,t)\in \coprod X}(W^{\dagger})=
\left(\!\begin{smallmatrix}(x,t)\\(U^{x,t}_{\alpha})_{\alpha\in \bb{W^{\dagger}}_{x}}\end{smallmatrix}\!\right)_{(x,t)\in \coprod X}(z)=
U^{z,u}_{0}=
\sigma((P_{i})_{i\in\bb{w}}).
$$
Hence $\sigma((P_{i})_{i\in\bb{w}})$ is a term in $L_{z}$. Let us consider the term $V^{z,u}_{0}=\sigma((Q_{i})_{i\in\bb{w}})$ in $L_{z}$. Since, for every $i\in\bb{w}$, $(P_{i}, Q_{i})\in\Phi_{w_{i}}$, we have that $(U^{z,u}_{0}, V^{z,u}_{0})\in\Phi_{u}$, and, since $L_{z}$ is $\Phi_{u}$-saturated, we have that $V^{z,u}_{0}$ is a term in $L_{x}$. Moreover, $\sigma((Q_{i})_{i\in\bb{w}})=\left(\!\begin{smallmatrix}z\\V^{z,u}_{0}\end{smallmatrix}\!\right)(z)$. Therefore $\sigma((Q_{i})_{i\in\bb{w}})\in\left(\left(\left(\!\begin{smallmatrix}x\\L_{x}\end{smallmatrix}\!\right)_{x\in X_{t}}\right)_{t\in S}\right)^{\sharp\mathfrak{p}}_{u}([W_{u,l}]_{\Phi_{u}})$. The other implication follows from a similar reasoning, thus proving case (b.1).

In case (b.2), we conclude, from Equation~\ref{EqSubsHtg}, that the term $W^{\dagger}$ has the form $\sigma((W^{\dagger_{i}})_{i\in\bb{w}})$ for some sequence of terms $(W^{\dagger_{i}})_{i\in\bb{w}}\in \mathrm{T}_{\Sigma}(X)_{w}$, since no sub\-sti\-tu\-tion replaces the operation symbol. From Equation~\ref{Eq1}, for every $i\in\bb{w}$,
$$
P_{i}=\left(\left(\!\begin{smallmatrix}(x,t)\\(U^{x,t}_{\alpha+\sum_{j\in i}\bb{W^{\dagger_{j}}}_{x}})_{\alpha\in \bb{W^{\dagger_{i}}}_{x}}\end{smallmatrix}\!\right)_{(x,t)\in \coprod X}(W^{\dagger_{i}})\right).
$$
Hence, for every $i\in \bb{w}$,
$P_{i}\in\left(\left(\left(\!\begin{smallmatrix}x\\L_{x}\end{smallmatrix}\!\right)_{x\in X_{t}}\right)_{t\in S}\right)^{\sharp\mathfrak{p}}_{w_{i}}([W^{\dagger_{i}}]_{\Phi_{w_{i}}})$.

By definition of $\Psi$, for every $i\in\bb{w}$,  $Q_{i}\in\left(\left(\left(\!\begin{smallmatrix}x\\L_{x}\end{smallmatrix}\!\right)_{x\in X_{t}}\right)_{t\in S}\right)^{\sharp\mathfrak{p}}_{w_{i}}([W^{\dagger_{i}}]_{\Phi_{w_{i}}})$. There\-fore, for every $i\in\bb{w}$, there are $W^{\ddagger_{i}}\in[W^{\dagger_{i}}]_{\Phi_{w_{i}}}$ and $\left((V^{x,t,i}_{\beta})_{\beta\in \bb{W^{\ddagger_{i}}}_{x}}\right)_{(x,t)\in\coprod X}$ in $\prod_{(x,t)\in\coprod X}L_{x}^{\bb{W^{\ddagger_{i}}}_{x}}$ such that
$$
Q_{i}=\left(\!\begin{smallmatrix}(x,t)\\(V^{x,t,i}_{\alpha})_{\alpha\in \bb{W^{\ddagger_{i}}}_{x}}\end{smallmatrix}\!\right)_{(x,t)\in \coprod X}(W^{\ddagger_{i}}).
$$
For $W^{\ddagger}=\sigma((W^{\ddagger_{i}})_{i\in\bb{w}})$ and  $\left((V^{x,t}_{\beta})_{\beta\in \bb{W^{\ddagger}}_{x}}\right)_{(x,t)\in\coprod X}$ of $\prod_{(x,t)\in\coprod X}L_{x}^{\bb{W^{\ddagger}}_{x}}$ ob\-tained by joining, in order, the family  $\left(\left((V^{x,t,i}_{\beta})_{\beta\in \bb{W^{\ddagger_{i}}}_{x}}\right)_{(x,t)\in\coprod X}\right)_{i\in \bb{w}}$, we have that $\sigma((Q_{i})_{i\in\bb{w}})$ is in $\left(\left(\left(\!\begin{smallmatrix}x\\L_{x}\end{smallmatrix}\!\right)_{x\in X_{t}}\right)_{t\in S}\right)^{\sharp\mathfrak{p}}_{u}([W_{u,l}]_{\Phi_{u}})$. The other implication follows by a similar reasoning, thus proving case (b.2).

Therefore $\Psi$ is a congruence of finite index in $\mathrm{Cgr}_{\mathrm{fi}}(\mathbf{T}_{\Sigma}(X))$.

Finally, we prove that $\left(\left(\left(\!\begin{smallmatrix}x\\L_{x}\end{smallmatrix}\!\right)_{x\in X_{t}}\right)_{t\in S}\right)^{\sharp\mathfrak{p}}_{s}(K)$ is $\Psi_{s}$-saturated. Note that, by definition of $\Psi$, for every $r\in S$ and every $l\in k_{r}$,   $\left(\left(\left(\!\begin{smallmatrix}x\\L_{x}\end{smallmatrix}\!\right)_{x\in X_{t}}\right)_{t\in S}\right)^{\sharp\mathfrak{p}}_{r}([W_{r,l}]_{\Phi_{r}})$ is $\Psi_{r}$-saturated. Moreover,
$$
\textstyle
\left(\left(\left(\!\begin{smallmatrix}x\\L_{x}\end{smallmatrix}\!\right)_{x\in X_{t}}\right)_{t\in S}\right)^{\sharp\mathfrak{p}}_{s}(K)=\bigcup_{W_{s,l}\in\mathcal{W}_{\Phi_{s}}\And [W_{s,l}]_{\Phi_{s}}\subseteq K}\left(\left(\left(\!\begin{smallmatrix}x\\L_{x}\end{smallmatrix}\!\right)_{x\in X_{t}}\right)_{t\in S}\right)^{\sharp\mathfrak{p}}_{s}([W_{s,l}]_{\Phi_{s}}).
$$
Hence, $\left(\left(\left(\!\begin{smallmatrix}x\\L_{x}\end{smallmatrix}\!\right)_{x\in X_{t}}\right)_{t\in S}\right)^{\sharp\mathfrak{p}}_{s}(K)$ is $\Psi_{s}$-saturated because it is a finite union of $\Psi_{s}$-saturated languages. The statement follows from Proposition~\ref{Rec is Bool}.
\end{proof}

\begin{corollary}\label{CRecSubs}
Let $u$ and $s$ be sorts in $S$, $z\in X_{u}$, $L\in \mathrm{Rec}_{u}(\mathbf{T}_{\Sigma}(X))$, and $K\in\mathrm{Rec}_{s}(\mathbf{T}_{\Sigma}(X))$. Then $\left(\!\begin{smallmatrix}z\\L\end{smallmatrix}\!\right)^{\sharp\mathfrak{p}}_{s}(K)\in \mathrm{Rec}_{s}(\mathbf{T}_{\Sigma}(X))$.
\end{corollary}
%\begin{proof}This is a particular case of Proposition~\ref{PRecSubs} in which all variables $z\in X\setminus\{x\}$ are assigned to the languages $\{z\}$. Proposition~\ref{PRecIt}\end{proof}

We next state the many-sorted version of Corollary 4.12, on p.~78, in \cite{GS84}.

\begin{corollary}\label{CRecOp}
Let $(w,s)$ be an element of $S^{\star}\times S$, $\sigma\in\Sigma_{w,s}$, and $(L_{i})_{i\in \bb{w}}\in \prod_{i\in \bb{w}}\mathrm{Rec}_{w_{i}}(\mathbf{T}_{\Sigma}(X))$. Then the language $\sigma^{\wp}((L_{i})_{i\in \bb{w}})\in \mathrm{Rec}_{s}(\mathbf{T}_{\Sigma}(X))$.
\end{corollary}

\begin{proof}
It follows from Proposition~\ref{PRecOp} and Proposition~\ref{PRecSubs}.
\end{proof}

\begin{corollary}
The $S$-sorted set $(\mathrm{Rec}_{s}(\mathbf{T}_{\Sigma}(X)))_{s\in S}$ is a subalgebra of $\mathbf{T}_{\Sigma}(X)^{\wp}$. We will denote by $\mathbf{Rec}_{\boldsymbol{\cdot}}(\mathbf{T}_{\Sigma}(X))$ the $\Sigma$-algebra canonically associated to the subalgebra $(\mathrm{Rec}_{s}(\mathbf{T}_{\Sigma}(X)))_{s\in S}$ of $\mathbf{T}_{\Sigma}(X)^{\wp}$.
\end{corollary}

%\begin{definition}\label{Rec Alg}
%We will denote by $\mathbf{Rec}_{\boldsymbol{\cdot}}(\mathbf{T}_{\Sigma}(X))$ the $\Sigma$-algebra canonically associated to the subalgebra $(\mathrm{Rec}_{s}(\mathbf{T}_{\Sigma}(X)))_{s\in S}$ of $\mathbf{T}_{\Sigma}(X)^{\wp}$.
%\end{definition}

\subsection{Iterations}
In this subsection we introduce the notion of iteration of a language with respect to a variable with the aim of proving that, when the considered language is recognizable, then its iteration with respect to a variable is also a recognizable language.

%, p.\,76
We begin by stating the many-sorted counterpart of the single-sorted notion of $z$-iteration as defined by G\'{e}cseg and Steinby in \cite{GS84}, Definition~4.7, on p.~76.

\begin{definition}
Let $s$ be a sort in $S$, $z\in X_{s}$, and $L\in \mathrm{T}_{\Sigma}(X)^{\wp}_{s}$. Then the \emph{$z$-iteration of $L$} is the language
$$
\textstyle
L^{\star\, z} = \bigcup_{j\in\mathbb{N}}L^{j,z},
$$
where $(L^{j,z})_{j\in \mathbb{N}}$ is the family of subsets of $\mathrm{T}_{\Sigma}(X)_{s}$ defined recursively as follows:
$$
L^{0,z}=\{z\},\, \text{and, for } j\in \mathbb{N},\, L^{j+1,z} = L^{j,z}\cup \left(\!\begin{smallmatrix}z\\L^{j,z}\end{smallmatrix}\!\right)^{\sharp\mathfrak{p}}_{s}(L).
$$
\end{definition}

\begin{remark}
The language $L^{\star\, z}$ is obtained as follows. First include $z$. New members of $L^{\star\, z}$ are obtained by substituting in some term of $L$, for every occurrence of $z$, some term already known to be in $L^{\star\, z}$. Let us note that $L^{1,z} = L\cup \{z\}$ and that $(L^{j,z})_{j\in \mathbb{N}}$ is an ascending chain of subsets of  $\mathrm{T}_{\Sigma}(X)_{s}$, i.e., that, for every $j\in \mathbb{N}$, $L^{j,z}\subseteq L^{j+1,z}$.
\end{remark}

We next prove the many-sorted version of Theorem~4.8., on p.~76, in \cite{GS84}. The proposition states that, for every sort $s\in S$ and $z\in X_{s}$, if the input language $L\subseteq \mathrm{T}_{\Sigma}(X)_{s}$ is recognizable, then its $z$-iteration is also recognizable.

\begin{assumption}
To prove the following proposition we will assume that $S$ is finite.
\end{assumption}
%p.~68
\begin{proposition}
Let $s$ be a sort in $S$ and $z\in X_{s}$. If $L\in\mathrm{Rec}_{s}(\mathbf{T}_{\Sigma}(X))$, then $L^{\star\,z}\in\mathrm{Rec}_{s}(\mathbf{T}_{\Sigma}(X))$.
\end{proposition}

\begin{proof}
Let $\Phi$ be the congruence on $\mathbf{T}_{\Sigma}(X)$ defined as follows:
$$
\Phi = \Omega^{\mathbf{T}_{\Sigma}(X)}(\delta^{s,L})\cap\Omega^{\mathbf{T}_{\Sigma}(X)}(\delta^{s,z}).
$$
By Proposition~\ref{Filter}, $\Phi$ is of finite index (recall that, by Proposition~\ref{PRecVar}, $\{z\}\subseteq\mathrm{T}_{\Sigma}(X)_{s}$ is recognizable). Moreover, for the sort $s\in S$, $\Phi_{s}$ saturates $L$ and $\{z\}$.

From now on, for every $r\in S$, $k_{r}$ and $\mathcal{W}_{\Phi_{r}}=\{W_{r,l}\mid l\in k_{r}\}$ stand for the index of $\Phi_{r}$ and a fixed transversal of $\mathrm{T}_{\Sigma}(X)_{r}/\Phi_{r}$ in $\mathrm{T}_{\Sigma}(X)_{r}$, respectively. Moreover, $k$ and $\mathcal{W}_{\Phi}$ denote the $S$-sorted sets $(k_{r})_{r\in S}$ and $(\mathcal{W}_{\Phi_{r}})_{r\in S}$, respectively.

Let $\Psi=(\Psi_{r})_{r\in S}$ be the binary relation on $\mathrm{T}_{\Sigma}(X)$ that is defined as follows: For every $r\in S$, $\Psi_{r}$ is the binary relation on $\mathrm{T}_{\Sigma}(X)_{r}$ consisting of all ordered pairs $(P,Q)\in\mathrm{T}_{\Sigma}(X)^{2}_{r}$ such that the following two conditions are satisfied:
\begin{enumerate}
\item $(P,Q)\in\Phi_{r}$.
\item For every $l\in k_{r}$,
$\left(P\in \left(\!\begin{smallmatrix}z\\L^{\star\, z}\end{smallmatrix}\!\right)^{\sharp\mathfrak{p}}_{r}([W_{r,l}]_{\Phi_{r}})\leftrightarrow Q\in \left(\!\begin{smallmatrix}z\\L^{\star\, z}\end{smallmatrix}\!\right)^{\sharp\mathfrak{p}}_{r}([W_{r,l}]_{\Phi_{r}})\right)$.
\end{enumerate}

By definition, for every $r\in S$, $\Psi_{r}$ is a refinement of $\Phi_{r}$ and an equivalence relation on $\mathrm{T}_{\Sigma}(X)_{r}$. Moreover, for every $r\in S$, the index of $\Psi_{r}$ on $\mathrm{T}_{\Sigma}(X)_{r}$ is bounded by $k_{r}2^{k_{r}}$. Consequently, the $S$-sorted $\mathrm{T}_{\Sigma}(X)/{\Psi}$ is finite.

Let us check that $\Psi$ is a congruence on $\mathbf{T}_{\Sigma}(X)$.  Let $(w,u)\in (S^{\star}-\{\lambda\}\times S)$, $\sigma\in\Sigma_{w,u}$ and let $(P_{i})_{i\in\bb{w}}$ and $(Q_{i})_{i\in \bb{w}}$ be sequences of terms in $\mathrm{T}_{\Sigma}(X)_{w}$ such that for every $i\in\bb{w}$, $(P_{i},Q_{i})\in\Psi_{w_{i}}$. We want to show that $(\sigma((P_{i})_{i\in\bb{w}}),\sigma((Q_{i})_{i\in\bb{w}}))\in\Psi_{u}$.

Let us note that, by definition of $\Psi$, for every $i\in\bb{w}$, we have that $(P_{i},Q_{i})\in\Phi_{w_{i}}$. Since $\Phi$ is a congruence on $\mathrm{T}_{\Sigma}(X)$, we conclude that $(\sigma((P_{i})_{i\in\bb{w}}), \sigma((Q_{i})_{i\in\bb{w}}))$ is a pair in $\Phi_{u}$, so $\sigma((P_{i})_{i\in \bb{w}})$ and $\sigma((Q_{i})_{i\in \bb{w}})$ satisfy the first condition for being related under $\Psi_{u}$.

Regarding the second condition, let $l$ be an element of $k_{u}$. Assume that
\[
\sigma((P_{i})_{i\in\bb{w}}))\in \left(\!\begin{smallmatrix}z\\L^{\star\, z}\end{smallmatrix}\!\right)^{\sharp\mathfrak{p}}_{u}([W_{u,l}]_{\Phi_{u}}).
\]
Then there are $W^{\dagger}\in [W_{u,l}]_{\Phi_{s}}$ and $(U^{z}_{\alpha})_{\alpha\in\bb{W^{\dagger}}_{z}}$ in $(L^{\star z})^{\bb{W^{\dagger}}_{z}}$ such that
\begin{equation}\label{EqItHtg}
\sigma((P_{i})_{i\in \bb{w}}) = \left(\!\begin{smallmatrix}z\\(U^{z}_{\alpha})_{\alpha\in\bb{W^{\dagger}}_{z}}\end{smallmatrix}\!\right)\left(W^{\dagger}\right).
\end{equation}

Note that, for $W^{\dagger}$, either (a) $z\not\in \mathrm{Var}(W^{\dagger})_{s}$ or (b) $z\in\mathrm{Var}(W^{\dagger})_{s}$.

Case (a) follows by a similar argument to that presented in case (a) of Proposition~\ref{PRecSubs}.

In case (b), either (b.1) $W^{\dagger}$ is the variable $z$ (and $u=s$) or (b.2) $W^{\dagger}$ has the form of an operation symbol applied to a suitable family of terms.

In case (b.1), from Equation~\ref{EqItHtg}, we obtain the following equation:
$$
\left(\!\begin{smallmatrix}z\\(U^{z}_{\alpha})_{\alpha\in\bb{W^{\dagger}}_{z}}\end{smallmatrix}\!\right)\left(W^{\dagger}\right) = \left(\!\begin{smallmatrix}z\\(U^{z}_{\alpha})_{\alpha\in\bb{W^{\dagger}}_{z}}\end{smallmatrix}\!\right)(z)=U^{z}_{0} = \sigma((P_{i})_{i\in \bb{w}}).
$$
It follows that $\sigma((P_{i})_{i\in \bb{w}})$ is a term in $L^{\star\, z}$. Thus, there exists a $j\in\mathbb{N}$ such that  $\sigma((P_{i})_{i\in \bb{w}})\in L^{j,z}$. Let $j\in\mathbb{N}$ be the smallest natural number satisfying the just mentioned  property. Since $L^{0,z}=\{z\}$ it follows that $j\neq 0$. Therefore, we have that
$$
\sigma((P_{i})_{i\in \bb{w}})\in L^{j,z}=L^{j-1,z}\cup \left(\!\begin{smallmatrix}z\\L^{j-1,z}\end{smallmatrix}\!\right)^{\sharp\mathfrak{p}}_{s}(L).
$$
By the minimality of $j$, we conclude that $\sigma((P_{i})_{i\in \bb{w}})\in \left(\!\begin{smallmatrix}z\\L^{j-1,z}\end{smallmatrix}\!\right)^{\sharp\mathfrak{p}}_{s}(L)$. Then there are $\overline{W}^{\dagger}\in L$ and $(\overline{U}^{z}_{\alpha})_{\alpha\in\bb{\overline{W}^{\dagger}}_{z}}\in (L^{j-1,z})^{\bb{\overline{W}^{\dagger}}_{z}}$ such that
\begin{equation}\label{EqIt2Htg}
\sigma((P_{i})_{i\in \bb{w}}) = \left(\!\begin{smallmatrix}z\\ \left(\overline{U}^{z}_{\alpha}\right)_{\alpha\in\bb{\overline{W}^{\dagger}}_{z}}\end{smallmatrix}\!\right)\left(\overline{W}^{\dagger}\right).
\end{equation}
For $\overline{W}^{\dagger}$, either (b.1.i) $z\notin \mathrm{Var}(\overline{W}^{\dagger})_{s}$  or (b.1.ii) $z\in \mathrm{Var}(\overline{W}^{\dagger})_{s}$.

In case (b.1.i). since we are assuming that $\bb{\overline{W}^{\dagger}}_{z}=0$, the substitution leaves $\overline{W}^{\dagger}$ invariant. Hence $\overline{W}^{\dagger}=\sigma((P_{i})_{i\in \bb{w}})$. It follows that $\sigma((P_{i})_{i\in \bb{w}})\in L$. Let $V^{z}_{0}=\sigma((Q_{i})_{i\in \bb{w}})$. Let us note that, for every $i\in \bb{w}$, $(P_{i}, Q_{i})\in\Phi_{w_{i}}$. Therefore, $(U^{z}_{0}, V^{z}_{0})\in \Phi_{s}$ and, since $L$ is $\Phi_{s}$-saturated, we conclude that $V^{z}_{0}\in L\subseteq L^{\star\, z}$. It follows that $\sigma((Q_{i})_{i\in \bb{w}})$ is a term in $\left(\!\begin{smallmatrix}z\\L^{\star\, z}\end{smallmatrix}\!\right)^{\sharp\mathfrak{p}}_{s}([W_{s,l}]_{\Phi_{s}})$, as desired.

In case (b.1.ii), where we are assuming that $\bb{\overline{W}^{\dagger}}_{z}\neq 0$, we claim that $\overline{W}^{\dagger}$ cannot be the term $z$. Otherwise, from Equation~\ref{EqIt2Htg}, we can conclude that $\sigma((P_{i})_{i\in \bb{w}})$ is a term in $L^{j-1,z}$, contradicting the minimality of $j$. Thus, from Equation~\ref{EqIt2Htg}  we conclude that $\overline{W}^{\dagger} = \sigma((\overline{W}^{\dagger_{i}})_{i\in \bb{w}})$, for a unique $(\overline{W}^{\dagger_{i}})_{i\in \bb{w}}\in \mathrm{T}_{\Sigma}(X)_{w}$. Hence, from Equation~\ref{Eq1}, for every $i\in \bb{w}$, we have that
$$
P_{i} = \left(\!\begin{smallmatrix}z\\ \left(\overline{U}^{z}_{\alpha + \sum_{k\in i}\bb{\overline{W}^{\dagger_{k}}}_{z}}\right)_{\alpha\in\bb{\overline{W}^{\dagger_{i}}}_{z}}\end{smallmatrix}\!\right)\left(\overline{W}^{\dagger_{i}}\right)
\in \left(\!\begin{smallmatrix}z\\L^{\star\, z}\end{smallmatrix}\!\right)^{\sharp\mathfrak{p}}_{w_{i}}\left([\overline{W}^{\dagger_{i}}]_{\Phi_{w_{i}}}\right).
$$

By construction, for every $i\in \bb{w}$, $Q_{i}\in  \left(\!\begin{smallmatrix}z\\L^{\star\, z}\end{smallmatrix}\!\right)^{\sharp\mathfrak{p}}_{w_{i}}([\overline{W}^{\dagger_{i}}]_{\Phi_{w_{i}}})$. Therefore, for every $i\in \bb{w}$, there are $\overline{W}^{\ddagger_{i}}\in [\overline{W}^{\dagger_{i}}]_{\Phi_{w_{i}}}$ and $(\overline{V}^{z,i}_{\beta})_{\beta\in\bb{\overline{W}^{\ddagger_{i}}}_{z}}\in (L^{\star\, z})^{\bb{\overline{W}^{\ddagger_{i}}}_{z}}$ such that
$$
Q_{i} = \left(\!\begin{smallmatrix}z\\ \left(\overline{V}^{z,i}_{\beta}\right)_{\beta\in\bb{\overline{W}^{\ddagger_{i}}}_{z}}\end{smallmatrix}\!\right)\left(\overline{W}^{\ddagger_{i}}\right).
$$
Let us note that the term $\overline{W}^{\ddagger}=\sigma((\overline{W}^{\ddagger_{i}})_{i\in n})$ is in $L$ because, for every $i\in \bb{w}$, the pair $(\overline{W}^{\dagger_{i}}, \overline{W}^{\ddagger_{i}})$ is in $\Phi_{s}$, $\overline{W}^{\dagger}\in L$, and $L$ is $\Phi_{s}$-saturated.

Let $(\overline{V}^{z}_{\beta})_{\beta\in \bb{\overline{W}^{\ddagger}}_{z}}$ be the element of $(L^{\star\, z})^{\bb{\overline{W}^{\ddagger}}_{z}}$ obtained by joining, in order, the $\bb{w}$-indexed family $\left((\overline{V}^{z,i}_{\beta})_{\beta\in \bb{\overline{W}^{\dagger_{i}}}_{z}}\right)_{i\in \bb{w}}$. Since $(\overline{V}^{z}_{\beta})_{\beta\in \bb{\overline{W}^{\dagger}}_{z}}$ is finite, there exists a  $t\in\mathbb{N}$ such that, for every  $\beta\in\bb{\overline{W}^{\ddagger}}_{x}$, the term $\overline{V}^{z}_{\beta}$ is in $L^{t, z}$.
On the whole, we conclude that
$$
\sigma((Q_{i})_{i\in \bb{w}})\in \left(\!\begin{smallmatrix}z\\L^{t, z}\end{smallmatrix}\!\right)^{\sharp\mathfrak{p}}_{s}(L)\subseteq L^{\star\, z}.
$$
Therefore $\sigma((Q_{i})_{i\in \bb{w}})$ is a term in $\left(\!\begin{smallmatrix}z\\L^{\star\, z}\end{smallmatrix}\!\right)^{\sharp\mathfrak{p}}_{s}([W_{s,l}]_{\Phi_{s}})$, as desired. The other implication follows by a similar reasoning, thus proving case (1.b.ii).

Case (b.2) follows by an argument similar to that used in case (b.2) of Proposition~\ref{PRecSubs}.

Therefore $\Psi$ is a congruence of finite index in $\mathrm{Cgr}_{\mathrm{fi}}(\mathbf{T}_{\Sigma}(X))$.

Finally, we prove that $L^{\star\, z}$ is $\Psi_{s}$-saturated. Note that, by definition of $\Psi$, for every $r\in S$ and $l\in k_{r}$, the language $\left(\!\begin{smallmatrix}z\\ L^{\star\, z}\end{smallmatrix}\!\right)^{\sharp\mathfrak{p}}_{r}([W_{r,l}]_{\Phi_{r}})$ is $\Psi_{r}$-saturated. In particular, since $\Phi_{s}$ recognizes $\{z\}$, we conclude that $L^{\star\, z}=\left(\!\begin{smallmatrix}z\\ L^{\star\, z}\end{smallmatrix}\!\right)^{\sharp\mathfrak{p}}_{s}(\{z\})$ is $\Psi_{s}$-saturated, thus proving the stamement.
\end{proof}

\subsection{Quotients}
We next define the notion of quotient of a language by another with respect to a variable of a specified sort with the aim of proving that, when one of the languages is recognizable, then the resulting quotient is also recognizable.

We begin by stating the many-sorted counterpart of the single-sorted notion of $z$-quotient as defined by G\'{e}cseg and Steinby in \cite{GS84}, Definition~4.9., on p.~77.

\begin{definition}
Let $s$ be a sort in $S$ and $L\in\mathrm{T}_{\Sigma}(X)^{\wp}_{s}$.
Let $t$ be a sort in $S$, $z$ an element of $X_{t}$, and $K\in \mathrm{T}_{\Sigma}(X)^{\wp}_{t}$.  Then the $z$-quotient of $L$ by $K$ is the language
$$
K^{-z}L = \left\lbrace U \in \mathrm{T}_{\Sigma}(X)_{s}\mid \left(\!\begin{smallmatrix}z\\ K\end{smallmatrix}\!\right)^{\sharp}_{s}(U)\cap L\neq \varnothing\right\rbrace.
$$
This operation may be seen as a converse of the $z$-substitution. If $K = \{x\}$ is a final set, then we will write $x^{-z}L$ instead of $K^{-z}L$.
\end{definition}

We prove now the many-sorted version of Theorem~4.10., on p.~77, in \cite{GS84}. The proposition states that, for a sort $s$ in $S$, if  $L\in \mathrm{T}_{\Sigma}(X)^{\wp}_{s}$ is recognizable, then, for every $t\in S$, $z\in X_{t}$ and every $K\in \mathrm{T}_{\Sigma}(X)^{\wp}_{t}$, the $z$-quotient of $L$ by $K$ is also recognizable.

\begin{assumption}
To prove the following proposition we will assume that $S$ is finite.
\end{assumption}

\begin{proposition}\label{PRecQ}
Let $s$ be a sort in $S$. If $L\in\mathrm{Rec}_{s}(\mathbf{T}_{\Sigma}(X))$, then, for every $t\in S$, every $z\in X_{t}$, and every $K\in \mathrm{T}_{\Sigma}(X)^{\wp}_{t}$, $K^{-z}L\in\mathrm{Rec}_{s}(\mathbf{T}_{\Sigma}(X))$. Moreover, the number of different $z$-quotients $K^{-z}L$ for any fixed $L\in\mathrm{Rec}_{s}(\mathbf{T}_{\Sigma}(X))$ is finite.
\end{proposition}

\begin{proof}
Let $\Phi$ be the congruence on $\mathbf{T}_{\Sigma}(X)$ defined as follows:
$$
\Phi=\Omega^{\mathbf{T}_{\Sigma}(X)}(\delta^{s,L}).
$$
Then $\Phi$ is of finite index because it is the syntactic congruence of a recognizable language. Moreover, for the sort $s\in S$, $\Phi_{s}$ saturates $L$.

From now on, for every $r\in S$, let $k_{r}$  and $\mathcal{W}_{\Phi_{r}}=\{W_{r,l}\mid l\in k_{r}\}$ stand for the index of $\Phi_{r}$ and a fixed transversal of $\mathrm{T}_{\Sigma}(X)_{r}/{\Phi_{r}}$ in $\mathrm{T}_{\Sigma}(X)_{r}$, respectively. Moreover, $k$ and $\mathcal{W}_{\Phi}$ denote the $S$-sorted sets $(k_{r})_{r\in S}$ and $(\mathcal{W}_{\Phi_{r}})_{r\in S}$, respectively.

Let $\Psi=(\Psi_{r})_{r\in S}$ be the binary relation on $\mathrm{T}_{\Sigma}(X)$ that is defined as follows: For every $r\in S$, $\Psi_{r}$ is the binary relation on $\mathrm{T}_{\Sigma}(X)_{r}$ consisting of all ordered pairs $(P,Q)\in\mathrm{T}_{\Sigma}(X)_{r}^{2}$ such that the following two conditions are satisfied:
\begin{enumerate}
\item $(P,Q)\in\Phi_{r}$.
\item for every $l\in k_{r}\, ( P\in K^{-z}[W_{r,l}]_{\Phi_{r}}\leftrightarrow Q\in K^{-z}[W_{r,l}]_{\Phi_{r}})$.
\end{enumerate}
By definition, for every $r\in S$, $\Psi_{r}$ is a refinement of $\Phi_{r}$ and an equivalence relation on $\mathrm{T}_{\Sigma}(X)_{r}$. Moreover, for every $r\in S$, the index of $\Psi_{r}$ on $\mathrm{T}_{\Sigma}(X)_{r}$ is bounded by $k_{r}2^{k_{r}}$. Consequently, the $S$-sorted set $\mathrm{T}_{\Sigma}(X)/{\Psi}$ is finite.

Analysis similar to that in the proof of Proposition~\ref{PRecSubs} shows that $\Psi$ is a congruence on $\mathbf{T}_{\Sigma}(X)$.

Finally, we prove that $K^{-z}L$ is $\Psi_{s}$-saturated. Note that, by definition of $\Psi$, for every $r\in S$ and every $l\in k_{r}$, the language $K^{-z}[W_{r,l}]_{\Phi_{r}}$ is $\Psi_{r}$-saturated. Moreover
$$
\textstyle
K^{-z}L=\bigcup_{W_{s,l}\in \mathcal{W}_{\Phi_{s}} \And [W_{s,l}]_{\Phi_{s}}\subseteq L}K^{-z}[W_{s,l}]_{\Phi_{s}}.
$$
Hence, $K^{-z}L$ is $\Psi_{s}$-saturated because it is a finite union of $\Psi_{s}$-saturated languages. The statement follows from Proposition~\ref{Rec is Bool}.
\end{proof}

\subsection{Tree Homomorphisms}

Tree automata and tree homomorphisms were defined for the first time by Thatcher in~\cite{Tha69}. In the just cited paper  Thatcher proved, among other things, that linear tree homomorphisms preserve recognizability. We shall now consider a class of many-sorted \emph{homomorphisms}, the tree homomorphisms---which are the generalization to the many-sorted field of the  tree homomorphism defined by G\'{e}cseg and Steinby in~\cite{GS84}, Definition~4.13., on p.~78--- which go from a \emph{free} many-sorted algebra (of a certain many-sorted signature $(S,\Sigma)$) to another many-sorted algebra (of the same many-sorted signature), itself derived from a \emph{free} many-sorted algebras (of another many-sorted signature $(T,\Xi)$). These tree homomorphisms, as we will prove, have the property of reflecting suitable recognizable languages. Moreover, we will define a proper subset of the set of the tree homomorphisms, the linear tree homomorphisms, and we will prove that they have the remarkable property of preserving recognizable languages.

\begin{definition}
Let $\varphi\colon S\mor T$ be a mapping. Then we will denote by $\Delta_{\varphi}$ the functor from $\mathbf{Set}^{T}$ to $\mathbf{Set}^{S}$ that sends a $T$-sorted set $A$ the $S$-sorted set $A_{\varphi} = (A_{\varphi(s)})_{s\in S}$, i.e., $A\circ \varphi$, and a $T$-sorted mapping $f\colon A\mor B$ to the $S$-sorted mapping $f_{\varphi} = (f_{\varphi(s)})_{s\in  S}\colon A_{\varphi}\mor B_{\varphi}$.
\end{definition}

\begin{remark}
The functor $\Delta_{\varphi}$ has a left and a right adjoint.
\end{remark}

Let $\varphi\colon S\mor T$ be a mapping, $\varphi^{\star}$ the canonical homomorphism from $\mathbf{S}^{\star}$, the free monoid on $S$, to $\mathbf{T}^{\star}$, the free monoid on $T$, and $w\in S^{\star}$. Then, for a standard $T$-infinite countable $T$-sorted set of variables $V^{T} = (\{v^{t}_{n}\mid n\in\mathbb{N}\})_{t\in T}$, which is assumed to be disjoint from all other alphabets, we will denote by $V^{T}_{\downarrow\varphi^{\star}(w)} = (V^{T}_{(\downarrow\varphi^{\star}(w))_{t}})_{t\in T}$ the $T$-sorted set, where, for every $t\in T$, $(\downarrow\varphi^{\star}(w))_{t} = (\varphi^{\star}(w))^{-1}[\{t\}]$ and
$V^{T}_{(\downarrow\varphi^{\star}(w))_{t}}$ is the subset of $V^{T}_{t}$ defined as follows:
$$
V^{T}_{(\downarrow\varphi^{\star}(w))_{t}} = \{v^{t}_{i}\mid i\in (\downarrow\!\varphi^{\star}(w))_{t}\} = \{v^{t}_{i}\mid i\in \bb{\varphi^{\star}(w)} \!\And\! \varphi(w_{i}) = t\}.
$$
Since $V^{T}_{\downarrow\varphi^{\star}(w)}$ is isomorphic to $\downarrow\!\varphi^{\star}(w)$, we abbreviate $V^{T}_{\downarrow\varphi^{\star}(w)}$ to $\downarrow\!\varphi^{\star}(w)$. Let us note that $\mathrm{card}(V^{T}_{\downarrow\varphi^{\star}(w)})$, the total number of variables, is $\bb{w} = \bb{\varphi^{\star}(w)}$. Moreover, for every $i\in \bb{w}$, the number of variables of type $\varphi(w_{i})$ is $\bb{\varphi^{\star}(w)}_{\varphi(w_{i})}$ (while, for $t\in T-\mathrm{Im}(\varphi^{\star}(w))$, $V^{T}_{(\downarrow\varphi^{\star}(w))_{t}} = \varnothing$).

\begin{remark}
For $V^{T}_{\downarrow\varphi^{\star}(w)}$, if we disregard the classification into types, then we have the following variables:
$$
v^{\varphi(w_{0})}_{0},\ldots,v^{\varphi(w_{i})}_{i},\ldots,v^{\varphi(w_{\bb{w}-1})}_{\bb{w}-1}.
$$

On the other hand, instead of $V^{T}_{\downarrow\varphi^{\star}(w)}$ we can, equivalently, use the $T$-sorted set $(\downarrow\! v^{t}_{\bb{\varphi^{\star}(w)}_{t}})_{t\in T}$, where, for every $t\in T$, $\downarrow\! v^{t}_{\bb{\varphi^{\star}(w)}_{t}}$ is the set of the first $\bb{\varphi^{\star}(w)}_{t}$ variables in $V^{T}_{t}$. Therefore, $\downarrow\! v^{t}_{\bb{\varphi^{\star}(w)}_{t}} = \varnothing$, if $t\notin \mathrm{Im}(\varphi^{\star}(w))$; while $\downarrow\! v^{t}_{\bb{\varphi^{\star}(w)}_{t}} = \{v^{t}_{j}\mid j\in \bb{\varphi^{\star}(w)}_{t}\}$, if  $t\in \mathrm{Im}(\varphi^{\star}(w))$.

Thus, preserving the classification into types of the variables, we have:
$$
%\left(
\begin{matrix}
v^{\varphi(w_{0})}_{0} &\dots & v^{\varphi(w_{0})}_{\bb{\varphi^{\star}(w)}_{\varphi(w_{0})}-1}\\
\vdots &\ddots & \vdots \\
v^{\varphi(w_{\bb{w}-1})}_{0} &\dots & v^{\varphi(w_{\bb{w}-1})}_{\bb{\varphi^{\star}(w)}_{\varphi(w_{\bb{w}-1})}-1}
\end{matrix}
%\right)
$$
\end{remark}

%we will denote by $(\downarrow\! v^{t}_{\bb{\varphi^{\star}(w)}_{t}})_{t\in T}$ the $T$-sorted set, where, for every $t\in T$, $\downarrow\! v^{t}_{\bb{\varphi^{\star}(w)}_{t}}$ is the set of the first $\bb{\varphi^{\star}(w)}_{t}$ variables in $V^{T}_{t}$ (the infinite countable set of the variables of type $t$). Therefore, $\downarrow\! v^{t}_{\bb{\varphi^{\star}(w)}_{t}} = \varnothing$, if $t\notin \mathrm{Im}(\varphi^{\star}(w))$; while $\downarrow\! v^{t}_{\bb{\varphi^{\star}(w)}_{t}} = \{v^{t}_{j}\mid j\in \bb{\varphi^{\star}(w)}_{t}\}$, if  $t\in \mathrm{Im}(\varphi^{\star}(w))$. Since $(\downarrow\! v^{t}_{\bb{\varphi^{\star}(w)}_{t}})_{t\in T}$ is isomorphic to $\downarrow\!\varphi^{\star}(w)$, we abbreviate $(\downarrow\! v^{t}_{\bb{\varphi^{\star}(w)}_{t}})_{t\in T}$ to $\downarrow\!\varphi^{\star}(w)$. Let us note that $\mathrm{card}((\downarrow\! v^{t}_{\bb{\varphi^{\star}(w)}_{t}})_{t\in T})$, the total number of variables, is $\bb{w} = \bb{\varphi^{\star}(w)}$. Moreover, for every $i\in \bb{w}$, the number of variables of type $\varphi(w_{i})$ is $\bb{\varphi^{\star}(w)}_{\varphi(w_{i})}$.

\begin{definition}
A \emph{many-sorted signature} is an ordered pair $\mathbf{\Sigma} = (S,\Sigma)$ where $S$ is a set (of sorts) and $\Sigma$ an $S$-sorted signature.
\end{definition}

We next define the notion of hyperderivor from a pair $(\mathbf{\Sigma},X)$, where $\mathbf{\Sigma}$ is a many-sorted signature and $X$ an $S$-sorted set, to another pair $(\mathbf{\Xi},Y)$, where $\mathbf{\Xi} = (T,\Xi)$ is a many-sorted signature and $Y$ a $T$-sorted set.

\begin{definition}%[\protect{\cite[Definition~2.4.13.]{GS84}}]
Let $\mathbf{\Sigma}$ and $\mathbf{\Xi}$ be many-sorted signatures, $X$ an $S$-sorted set, and $Y$ a $T$-sorted set. A \emph{hyperderivor from} $(\mathbf{\Sigma},X)$ \emph{to} $(\mathbf{\Xi},Y)$ is an ordered pair $((\varphi,c),f)$, denoted by $(\mathbf{c},f)$, where $\varphi$ is a mapping from $S$ to $T$,
$c = (c_{w,s})_{(w,s)\in S^{\star}\times S}$ an $S^{\star}\times S$-sorted mapping from $\Sigma$ to $(\mathrm{T}_{\Xi}(Y\cup \downarrow\!\varphi^{\star}(w))_{\varphi(s)})_{(w,s)\in S^{\star}\times S}$, and $f$ an $S$-sorted mapping from $X$ to $\mathrm{T}_{\Xi}(Y)_{\varphi}$. We will say that a hyperderivor $(\mathbf{c},f)$ from $(\mathbf{\Sigma},X)$ \emph{to} $(\mathbf{\Xi},Y)$ is \emph{linear} if, for every $(w,s)\in S^{\star}\times S$, every $\sigma\in\Sigma_{w,s}$, and every $i\in \bb{w}$, no variable $v^{\varphi(w_{i})}_{i}$ appears more than once in $c_{w,s}(\sigma)$, i.e., $\bb{c_{w,s}(\sigma)}_{v^{\varphi(w_{i})}_{i}}\leq 1$.
%
%%for every $(w,s)\in S^{\star}\times S$, every $\sigma\in \Sigma_{w,s}$, every $t\in T$, and every $i\in \varphi^{\star}(w)^{-1}[\{t\}]$, we have that $\mathrm{card}(\bb{d_{w,s}(\sigma)}_{v_{i}^{t}})\leq 1$.
\end{definition}

\begin{remark}
The reason for the terminology \emph{hyperderivor} lies in the analogy with the notion of derivor defined by Goguen, Thatcher, and Wagner in~\cite{gtw85}, on p.~137.
\end{remark}

We show now that, for every hyperderivor $(\mathbf{c},f)$ from $(\mathbf{\Sigma},X)$ to $(\mathbf{\Xi},Y)$, the $S$-sorted set $\mathrm{T}_{\Xi}(Y)_{\varphi}$ is equipped, in a natural way, with a structure of $\Sigma$-algebra. But for this we need to show that, given a mapping $\varphi$ from a set of sorts $S$ to another $T$ and a $T$-sorted set $Y$, it happens that, for every $w\in S^{\star}$ and every $(P_{i})_{i\in\bb{w}}\in \mathrm{T}_{\Xi}(Y)_{\varphi^{\star}(w)}$, there exists a canonical homomorphism $\mathrm{S}^{w}_{(P_{i})_{i\in\bb{w}}}$ from $\mathbf{T}_{\Xi}(Y\cup \downarrow\!\varphi^{\star}(w))$ to $\mathbf{T}_{\Xi}(Y)$. In what follows we will assume that, for every $w\in S^{\star}$, $Y\cap\downarrow\!\varphi^{\star}(w) = \varnothing^{T}$.

\begin{remark}
Let us note that the just stated assumption is not, in anyway, a loss in generality. Actually, given a $T$-sorted set $A$ and an $I$-indexed family $(B^{i})_{i\in I}$ of $T$-sorted sets there exists a $T$-sorted set $C$ such that $A\cong C$ and, for every $i\in I$, $C\cap B^{i} = \varnothing^{T}$. In fact, it suffices, by the Axiom of Regularity, to take as $C$ the $T$-sorted set defined, for every $t\in T$, as $C_{t} = A_{t}\times \{\{(i,B^{i}_{t})\mid i\in I\}\}$.
\end{remark}

Let $w$ be an element of $S^{\star}$. Then the sets $\mathrm{T}_{\Xi}(Y)_{\varphi^{\star}(w)}$ and $\mathrm{Hom}(\downarrow\!\varphi^{\star}(w),\mathrm{T}_{\Xi}(Y))$ are naturally isomorphic. On the basis of this isomorphism, we let $\left(\!\begin{smallmatrix}v^{\varphi(w_{i})}_{i}\\ P_{i}\end{smallmatrix}\!\right)_{i\in \bb{w}}$ stand for the $T$-sorted mapping from $\downarrow\!\varphi^{\star}(w)$ to $\mathrm{T}_{\Xi}(Y)$ canonically associated to $(P_{i})_{i\in\bb{w}}\in \mathrm{T}_{\Xi}(Y)_{\varphi^{\star}(w)}$. And then we let $\left(\left(\!\begin{smallmatrix}v^{\varphi(w_{i})}_{i}\\ P_{i}\end{smallmatrix}\!\right)_{i\in \bb{w}}\right)^{\sharp}$ stand for the unique homomorphism from $\mathbf{T}_{\Xi}(\downarrow\!\varphi^{\star}(w))$ to $\mathbf{T}_{\Xi}(Y)$ such that
$$
\left(\left(\!\begin{smallmatrix}v^{\varphi(w_{i})}_{i}\\ P_{i}\end{smallmatrix}\!\right)_{i\in \bb{w}}\right)^{\sharp}\circ \eta_{\downarrow\varphi^{\star}(w)} = \left(\!\begin{smallmatrix}v^{\varphi(w_{i})}_{i}\\ P_{i}\end{smallmatrix}\!\right)_{i\in \bb{w}}.
$$
On the other hand, since $\mathbf{T}_{\Xi}$ has a right adjoint, the $\Xi$-algebras $\mathbf{T}_{\Xi}(Y\cup\downarrow\!\varphi^{\star}(w))$ and $\mathbf{T}_{\Xi}(Y)\coprod\mathbf{T}_{\Xi}(\downarrow\!\varphi^{\star}(w))$ are naturally isomorphic. Then, finally, we let $\mathrm{S}^{w}_{(P_{i})_{i\in\bb{w}}}$ stand for $\left[\mathrm{id}_{\mathbf{T}_{\Xi}(Y)},\left(\left(\!\begin{smallmatrix}v^{\varphi(w_{i})}_{i}\\ P_{i}\end{smallmatrix}\!\right)_{i\in \bb{w}}\right)^{\sharp}\right]$, the unique homomorphism from $\mathbf{T}_{\Xi}(Y\cup
\downarrow\!\varphi^{\star}(w))$ to $\mathbf{T}_{\Xi}(Y)$ obtained, from $\mathrm{id}_{\mathbf{T}_{\Xi}(Y)}$ and $\left(\left(\!\begin{smallmatrix}v^{\varphi(w_{i})}_{i}\\ P_{i}\end{smallmatrix}\!\right)_{i\in \bb{w}}\right)^{\sharp}$, by the universal property of the coproduct.

\begin{proposition}
Let $(\mathbf{c},f)$ be a hyperderivor from $(\mathbf{\Sigma},X)$ to $(\mathbf{\Xi},Y)$. Then the $S$-sorted set $\mathrm{T}_{\Xi}(Y)_{\varphi}$ is equipped, in a natural way, with a structure of $\Sigma$-algebra.
\end{proposition}

\begin{proof}
Let $\mathbf{c}(\mathbf{T}_{\Xi}(Y))$ be the $\Sigma$-algebra defined as follows: The underlying $S$-sorted set of $\mathbf{c}(\mathbf{T}_{\Xi}(Y))$ is $\mathrm{T}_{\Xi}(Y)_{\varphi}$ while, for $(w,s)\in S^{\star}\times S$ and $\sigma\in \Sigma_{w,s}$, the operation $\sigma^{\mathbf{c}(\mathbf{T}_{\Xi}(Y))}$ from $\mathbf{T}_{\Xi}(Y)_{\varphi^{\star}(w)}$ to $\mathbf{T}_{\Xi}(Y)_{\varphi(s)}$ associated to $\sigma$ is defined as:
$$
\sigma^{\mathbf{c}(\mathbf{T}_{\Xi}(Y))}
\nfunction
{\mathbf{T}_{\Xi}(Y)_{\varphi^{\star}(w)}}
{\mathbf{T}_{\Xi}(Y)_{\varphi(s)}}
{(P_{i})_{i\in\bb{w}}}
{\mathrm{S}^{w}_{(P_{i})_{i\in\bb{w}}}(c_{w,s}(\sigma))}
$$
Thus $\sigma^{\mathbf{c}(\mathbf{T}_{\Xi}(Y))}((P_{i})_{i\in\bb{w}})$ is the term obtained by substituting in $c_{w,s}(\sigma)$, for every $i\in \bb{w}$, $P_{i}$ for $v^{\varphi(w_{i})}_{i}$.
Let us note that since $(\bigcup_{i\in\bb{w}}\mathrm{Var}(P_{i}))\cap V^{T} = \varnothing^{T}$,  $\mathrm{Var}(\sigma^{\mathbf{c}(\mathbf{T}_{\Xi}(Y))}((P_{i})_{i\in\bb{w}}))\cap V^{T} = \varnothing^{T}$. Therefore $\sigma^{\mathbf{c}(\mathbf{T}_{\Xi}(Y))}((P_{i})_{i\in\bb{w}})\in \mathrm{T}_{\Xi}(Y)_{\varphi(s)}$. Consequently, the operation $\sigma^{\mathbf{c}(\mathbf{T}_{\Xi}(Y))}$ is well-defined.

%We next prove that $tm^{(\mathbf{d},f)}$ is a homomorphism from $\mathbf{T}_{\Sigma}(X)$ to $\mathbf{d}(\mathbf{T}_{\Lambda}(Y))$. Let $(w,s)\in S^{\star}\times S$, $\sigma\in \Sigma_{w,s}$, and $(P_{i})_{i\in\bb{w}}\in \mathrm{T}_{\Sigma}(X)_{w}$, then we have that:
%\begin{align*}
%tm^{(\mathbf{d},f)}_{s}(\sigma^{\mathbf{T}_{\Sigma}(X)}((P_{i})_{i\in\bb{w}})) &= tm^{(\mathbf{d},f)}_{s}(\sigma((P_{i})_{i\in\bb{w}}))\\
% &= \left(\!\begin{smallmatrix}v^{\varphi(w_{i})}_{i}\\tm^{(\mathbf{d},f)}_{w_{i}}(P_{i})\end{smallmatrix}\!\right)_{i\in \bb{w}}(d_{w,s}(\sigma))\\
%  &= \sigma^{\mathbf{d}(\mathbf{T}_{\Lambda}(Y))}((tm_{w_{i}}^{(\mathbf{d},f)}(P_{i}))_{i\in\bb{w}}).
%\end{align*}
%Moreover, for every $s\in S$ and every $x\in X_{s}$, we have that:
%$$
%tm^{(\mathbf{d},f)}_{s}(\eta_{X_{s}}(x)) = tm^{(\mathbf{d},f)}_{s}(x) = f_{s}(x).
%$$
%Hence $tm^{(\mathbf{d},f)}\circ \eta_{X} = f$. Therefore, by the uniqueness of the canonical extensiÛn, $tm^{(\mathbf{d},f)} = f^{\sharp}$.
\end{proof}

\begin{definition}
Let $(\mathbf{c},f)$ be a hyperderivor from $(\mathbf{\Sigma},X)$ to $(\mathbf{\Xi},Y)$. Then the unique homomorphism $f^{\sharp}$ from $\mathbf{T}_{\Sigma}(X)$ to $\mathbf{c}(\mathbf{T}_{\Xi}(Y))$ such that $f^{\sharp}\circ\eta_{X} = f$ will be called the \emph{tree homomorphism} determined by the hyperderivor $(\mathbf{c},f)$. Moreover, $f^{\sharp}$ will be called \emph{linear} if $(\mathbf{c},f)$ is linear.
\end{definition}

%We will adopt the following notational convention: We let $f^{\sharp}\colon \mathbf{T}_{\Sigma}(X)\umor \mathbf{T}_{\Xi}(Y)$, single-headed arrow, stand for the tree homomorphism $f^{\sharp}\colon\mathbf{T}_{\Sigma}(X)\mor\mathbf{c}(\mathbf{T}_{\Xi}(Y))$ determined by $(\mathbf{c},f)$.

The just defined tree homomorphisms are a generalization of those proposed by G\'{e}cseg and Steinby, in~\cite{GS84}, Definition~4.13., on p.~78, for single-sorted algebras, which, in its turn, generalize those of Thatcher in~\cite{Tha69}.

G\'{e}cseg and Steinby in~\cite{GS84}, on p.~70, wrote: ``Tree homomorphisms are not really homomorphisms in the sense of algebra.'' Literally speaking the above assertion is true due, simply, to the fact that the signatures of the domain and codomain of a tree homomorphism are, in general, not identical. However, as we have seen $\mathrm{T}_{\Xi}(Y)_{\varphi}$ is equipped with a structure of $\Sigma$-algebra and $f^{\sharp}$ is a homomorphism of $\Sigma$-algebras from $\mathbf{T}_{\Sigma}(X)$ to $\mathbf{c}(\mathbf{T}_{\Xi}(Y))$.

In~\cite{GS84}, Theorem~4.18., on p.~82, GÈcseg and Steinby proved, in the single-sorted case, that the inverse image of a recognizable language under a tree homomorphism is a recognizable language. We next extend GÈcseg and Steinby's result to the many-sorted case.

\begin{proposition}\label{IndAlgStrucImHom}
Let $(\mathbf{c},f)$ be a hyperderivor from $(\mathbf{\Sigma},X)$ to $(\mathbf{\Xi},Y)$, $\mathbf{A}$ a $\Xi$-algebra, and $g$ a surjective homomorphism from $\mathbf{T}_{\Xi}(Y)$ to $\mathbf{A}$. Then there exists a structure of $\Sigma$-algebra  $F^{\mathbf{c}(\mathbf{A})}$ on $A_{\varphi}$ such that $g_{\varphi}$ is a homomorphism from $\mathbf{c}(\mathbf{T}_{\Lambda}(Y))$ to $\mathbf{c}(\mathbf{A}) = (A_{\varphi},F^{\mathbf{c}(\mathbf{A})})$.
\end{proposition}

\begin{proof}
Let $(w,s)$ be an element of $S^{\star}\times S$ and $\sigma\in\Sigma_{w,s}$. Then we denote by $F_{\sigma}^{\mathbf{c}(\mathbf{A})}$ the mapping from $A_{\varphi^{\star}(w)}$ to $A_{\varphi(s)}$ defined as follows:
$$
F_{\sigma}^{\mathbf{c}(\mathbf{A})}
\nfunction
{A_{\varphi^{\star}(w)}}
{A_{\varphi(s)}}
{(a_{i})_{i\in\bb{w}}}
{g_{\varphi(s)}(\sigma^{\mathbf{c}(\mathbf{T}_{\Xi}(Y))}((P_{i})_{i\in\bb{w}}))}
$$
where, for every $i\in\bb{w}$, $P_{i}\in \mathrm{T}_{\Lambda}(Y)_{\varphi(w_{i})}$ such that  $g_{\varphi(w_{i})}(P_{i})=a_{i}$.

The operation $F_{\sigma}^{\mathbf{c}(\mathbf{A})}$ is well-defined. In fact, let $(a_{i})_{i\in\bb{w}}$ be an element of $A_{\varphi^{\star}(w)}$ and $(P_{i})_{i\in\bb{w}}$, $(Q_{i})_{i\in\bb{w}}$ elements of $\mathrm{T}_{\Lambda}(Y)_{\varphi^{\star}(w)}$ such that, for every $i\in\bb{w}$,
$$
g_{\varphi(w_{i})}(P_{i}) = a_{i} = g_{\varphi(w_{i})}(Q_{i}).
$$
But we have that:
\begin{align*}
g_{\varphi(s)}(\sigma^{\mathbf{T}_{\Lambda}(Y)^{\Sigma}_{\varphi}}((P_{i}^{\varphi(w_{i})})_{i\in\bb{w}}))
  &=
g_{\varphi(s)}\left(
\left(\!\begin{smallmatrix}v^{\varphi(w_{i})}_{i}\\P^{\varphi(w_{i})}_{i}\end{smallmatrix}\!\right)_{i\in \bb{w}}(c_{w,s}(\sigma))\right) \quad(\dagger_{1}) \\
  &=
g_{\varphi(s)}\left(
\left(\!\begin{smallmatrix}v^{\varphi(w_{i})}_{i}\\Q^{\varphi(w_{i})}_{i}\end{smallmatrix}\!\right)_{i\in \bb{w}}(c_{w,s}(\sigma))\right) \quad(\dagger_{2}) \\
  &=
g_{\varphi(s)}(\sigma^{\mathbf{T}_{\Lambda}(Y)^{\Sigma}_{\varphi}}((Q_{i}^{\varphi(w_{i})})_{i\in\bb{w}}))  \,\,\quad\quad(\dagger_{1})
\end{align*}
%$
%g^{\sharp}_{\varphi(s)}(\sigma^{\mathbf{T}_{\Lambda}(Y)^{\Sigma}_{\varphi}}((P_{i}^{\varphi(w_{i})})_{i\in\bb{w}}))=$
%\begin{alignat}{2}
%& g^{\sharp}_{\varphi(s)}\left(
%\left(\!\begin{smallmatrix}v^{\varphi(w_{i})}_{i}\\P^{\varphi(w_{i})}_{i}\end{smallmatrix}\!\right)_{i\in \bb{w}}(d_{w,s}(\sigma))
%\right) = &\qquad\qquad&(\clubsuit)
%\notag\\
%& g^{\sharp}_{\varphi(s)}\left(
%\left(\!\begin{smallmatrix}v^{\varphi(w_{i})}_{i}\\Q^{\varphi(w_{i})}_{i}\end{smallmatrix}\!\right)_{i\in \bb{w}}(d_{w,s}(\sigma))
%\right) = & &(\heartsuit)
%\notag\\
%& g^{\sharp}_{\varphi(s)}(\sigma^{\mathbf{T}_{\Lambda}(Y)^{\Sigma}_{\varphi}}((P_{i}^{\varphi(w_{i})})_{i\in\bb{w}})) & &(\clubsuit)
%\notag
%\end{alignat}
where, to shorten notation, $(\dagger_{1})$ and $(\dagger_{2})$ stand for (by definition of $\mathbf{c}(\mathbf{T}_{\Xi}(Y))$) and (by Lemma~\ref{LIGlobSubstOp}), respectively. Therefore
$$
g_{\varphi(s)}(\sigma^{\mathbf{c}(\mathbf{T}_{\Xi}(Y))}((P_{i})_{i\in\bb{w}})) = g_{\varphi(s)}(\sigma^{\mathbf{c}(\mathbf{T}_{\Xi}(Y))}((Q_{i})_{i\in\bb{w}})).
$$

It is obvious that $g_{\varphi}$ is a homomorphism of $\Sigma$-algebras.

By construction $g^{\sharp}_{\varphi}$ is a homomorphism from $\mathbf{T}_{\Lambda}(Y)^{\Sigma}_{\varphi}$ to $g^{\sharp}[\mathbf{T}_{\Lambda}(Y)]^{\Sigma}_{\varphi}$. Indeed, for $\sigma\in \Sigma_{w,s}$ and  a family $(P^{\varphi(w_{i})}_{i})_{i\in\bb{w}}\in\mathrm{T}_{\Lambda}(Y)_{\varphi^{\star}(w)}$ we have that
$$
g^{\sharp}_{\varphi(s)}
(
\sigma^{\mathbf{T}_{\Lambda}(Y)^{\Sigma}_{\varphi}}
(
(P^{\varphi(w_{i})}_{i})_{i\in\bb{w}}
)
)=
\sigma^{g^{\sharp}[\mathbf{T}_{\Lambda}(Y)]^{\Sigma}_{\varphi}}
(
(
g^{\sharp}
(
P^{\varphi(w_{i})}_{i}
)
)_{i\in\bb{w}}
).
$$
\end{proof}

\begin{proposition}\label{PRecH}
Let $(\mathbf{c},f)$ be a hyperderivor from $(\mathbf{\Sigma},X)$ to $(\mathbf{\Xi},Y)$, $f^{\sharp}$ the tree homomorphism from $\mathbf{T}_{\Sigma}(X)$ to $\mathbf{c}(\mathbf{T}_{\Xi}(Y))$ determined by $(\mathbf{c},f)$, $s\in S$, and $L\in\mathrm{Rec}_{\varphi(s)}(\mathbf{T}_{\Xi}(Y))$. Then $\left(f^{\sharp}_{\varphi(s)}\right)^{-1}[L]\in\mathrm{Rec}_{s}(\mathbf{T}_{\Sigma}(X))$.
\end{proposition}

\begin{proof}
It follows from the definition of recognizability relative to a sort and from Proposition~\ref{IndAlgStrucImHom}.
\end{proof}

Our immediate goal is to prove the many-sorted version of Theorem~4.16., on p.~80, in \cite{GS84}. The proposition states that, for a sort $s\in S$, if the language $L$ is $s$-recognizable, then its direct image by the $s$-th coordinate of a linear tree homomorphism $f^{\sharp}$ is $\varphi(s)$-recognizable.

\begin{assumption}
To prove the following proposition we will assume that $S$, $T$, $\Sigma$ and $X$ are finite.
\end{assumption}

\begin{proposition}\label{PRecLH}
Let $(\mathbf{c},f)$ be a linear hyperderivor from $(\mathbf{\Sigma},X)$ to $(\mathbf{\Xi},Y)$, $f^{\sharp}$ the linear tree homomorphism from $\mathbf{T}_{\Sigma}(X)$ to $\mathbf{c}(\mathbf{T}_{\Xi}(Y))$ determined by $(\mathbf{c},f)$, $s\in S$, and $L\in \mathrm{T}_{\Sigma}(X)^{\wp}_{s}$. If $L\in\mathrm{Rec}_{s}(\mathbf{T}_{\Sigma}(X))$, then $f^{\sharp}_{s}[L]\in\mathrm{Rec}_{\varphi(s)}(\mathbf{T}_{\Xi}(Y))$.
\end{proposition}

\begin{proof}
Let $\Phi$ be the congruence on $\mathbf{T}_{\Xi}(Y)$ defined as follows:
$$
\textstyle
\Phi = \bigcap_{(x,r)\in\coprod X}\Omega^{\mathbf{T}_{\Xi}(Y)}(\delta^{\varphi(r),\{f_{r}(x)\}}).
$$
Since $f\colon X\mor \mathrm{T}_{\Xi}(Y)_{\varphi}$ we have that, for every $r\in S$ and every $x\in X_{r}$, the language $\{f_{r}(x)\}$ is in $\mathrm{Rec}_{\varphi(r)}(\mathbf{T}_{\Xi}(Y))$. Hence, $\Omega^{\mathbf{T}_{\Xi}(Y)}(\delta^{\varphi(r),\{f_{r}(x)\}})$ is of finite index. Therefore, by the assumption and Proposition~\ref{Filter}, $\Phi$ is of finite index. Moreover, for every $r\in S$ and every $x\in X_{r}$, the language $\{f_{r}(x)\}$ is $\Phi_{\varphi(r)}$-saturated.

For abbreviation we let $\Theta$ stand for $\Omega^{\mathbf{T}_{\Sigma}(X)}(\delta^{s,L})$, the syntactic congruence on $\mathbf{T}_{\Sigma}(X)$ determined by $L$. Since $L$ is a language in $\mathrm{Rec}_{s}(\mathbf{T}_{\Sigma}(X))$, $\Theta$ is of finite index. Moreover, for $s\in S$, $L$ is $\Theta_{s}$-saturated.

From now on, for every $r\in S$, $k_{r}$ and $\mathcal{W}_{\Theta_{r}}=\{W_{r,l}\mid l\in k_{r}\}$ stand for the index of $\Theta_{r}$  and a fixed transversal of $\mathrm{T}_{\Sigma}(X)/{\Theta_{r}}$ in $\mathrm{T}(X)_{r}$, respectively.
Moreover, $k$ and $\mathcal{W}_{\Theta}$ denote the $S$-sorted sets $(k_{r})_{r\in S}$ and
$(\mathcal{W}_{\Theta_{r}})_{r\in S}$, respectively.

Let $\Psi = (\Psi_{t})_{t\in T}$ be the binary relation on $\mathrm{T}_{\Xi}(Y)$ defined as follows: For every $t\in T$, $\Psi_{t}$ is the binary relation on $\mathrm{T}_{\Xi}(Y)_{t}$ consisting of all ordered pairs $(M,N)\in\mathrm{T}_{\Xi}(Y)_{t}^{2}$ such that the following two conditions are satisfied:
\begin{enumerate}
\item $(M,N)\in\Phi_{t}$.
\item $\forall\, (w,r)\in S^{\star}\times S$,  $\forall\, \sigma\in \Sigma_{w,r}$, $\forall\, R\in\mathrm{Subt}(c_{w,r}(\sigma))_{t}$, $\forall\, i\in\bb{w}$, $\forall\, l_{i}\in k_{w_{i}}$
$$
\left(M\in \left(\left(\!\begin{smallmatrix}v^{\varphi(w_{i})}_{i}\\ f^{\sharp}_{w_{i}}[[W_{w_{i},l_{i}}]_{\Theta_{w_{i}}}]\end{smallmatrix}\!\right)_{i\in\bb{w}}\right)^{\sharp}_{t}(R) \leftrightarrow N\in \left(\left(\!\begin{smallmatrix}v^{\varphi(w_{i})}_{i}\\ f^{\sharp}_{w_{i}}[[W_{w_{i},l_{i}}]_{\Theta_{w_{i}}}]\end{smallmatrix}\!\right)_{i\in\bb{w}}
\right)^{\sharp}_{t}(R)\right).
$$
\end{enumerate}

By definition, for every $t\in T$, $\Psi_{t}$ is a refinement of $\Phi_{t}$ and an equivalence relation on $\mathrm{T}_{\Xi}(Y)_{t}$. Moreover, if $a=\card(\mathrm{T}_{\Xi}(Y)/{\Phi})$, $b=\card(\Sigma)$,
$$
d = \max\{\mathrm{card}(\mathrm{Subt}(c_{w,r}(\sigma))_{t})\mid (w,r)\in S^{\star}\times S,\, \sigma\in\Sigma_{w,r},\, t\in T\}
$$
and $e=\max\{\bb{w}\mid (w,r)\in S^{\star}\times S \!\And \!\Sigma_{w,r}\neq\varnothing\}$, then the index of $\Psi_{t}$ is bounded by $a2^{bd\card(k)^{e}}$. In addition, all different possibilities defining $\Psi$ are bounded.  Consequently, the $T$-sorted set $\mathrm{T}_{\Xi}(Y)/{\Psi}$ is finite.

Let us check that $\Psi$ is a congruence on $\mathbf{T}_{\Xi}(Y)$. Let $(u,t)$ be an element of  $(T^{\star}-\{\lambda\})\times T$, $\xi\in\Xi_{u,t}$, and let $(M_{j})_{j\in\bb{u}}$ and $(N_{j})_{j\in \bb{u}}$ be elements of $\mathrm{T}_{\Xi}(Y)_{u}$ such that, for every $j\in\bb{u}$, $(M_{j},N_{j})\in\Psi_{u_{j}}$. We want to show that $(\xi((M_{j})_{j\in\bb{u}}),\xi((N_{j})_{j\in\bb{u}}))$ is an element of $\Psi_{t}$. Let us note that, by definition of $\Psi$, for every $j\in\bb{u}$, we have that $(M_{j}, N_{j})\in\Phi_{u_{j}}$. Therefore,  since $\Phi$ is a congruence on $\mathrm{T}_{\Xi}(Y)$, $(\xi((M_{j})_{j\in\bb{u}}),\xi((N_{j})_{j\in\bb{u}}))$ belongs to  $\Phi_{t}$. So the pair $(\xi((M_{j})_{j\in\bb{u}}),\xi((N_{j})_{j\in\bb{u}}))$ satisfies the first condition for being related under $\Psi_{t}$.

Regarding the second condition, let us assume that, for $(w,r)\in S^{\star}\times S$, $\sigma\in \Sigma_{w,r}$, $R\in\mathrm{Subt}(c_{w,r}(\sigma))_{t}$, $i\in\bb{w}$, and $l_{i}\in k_{w_{i}}$, we have that
$$
\xi((M_{j})_{j\in \bb{u}})
\in \left(\left(\!\begin{smallmatrix}v^{\varphi(w_{i})}_{i}\\ f^{\sharp}_{w_{i}}[[W_{w_{i},l_{i}}]_{\Theta_{w_{i}}}]\end{smallmatrix}\!\right)_{i\in\bb{w}}\right)^{\sharp}_{t}(R).
$$

However, given that the tree homomorphism $f^{\sharp}$ is linear, we have that, for every $(w,r)\in S^{\star}\times S$, every $\sigma\in\Sigma_{w,r}$, and every $i\in \bb{w}$, no variable $v^{\varphi(w_{i})}_{i}$ appears more than once in $c_{w,r}(\sigma)$. Thus, for every $i\in \bb{w}$, $\bb{R}_{v^{\varphi(w_{i})}_{i}}\leq 1$. Hence there exists a family $(W_{i})_{i\in \bb{w}}$ in $\mathrm{T}_{\Sigma}(X)_{w}$ such that, for every $i\in\bb{w}$, $W_{i}\in[W_{w_{i},l_{i}}]_{\Theta_{w_{i}}}$ and
\begin{equation}\label{EqtmHtg}
\xi((M_{j})_{j\in \bb{u}})= \left(\!\begin{smallmatrix}v^{\varphi(w_{i})}_{i}\\ f^{\sharp}_{w_{i}}(W_{i})\end{smallmatrix}\!\right)_{i\in \bb{w}}(R).
\end{equation}

Now, for $R$, either (a) $\downarrow\! \varphi^{\star}(w)\cap\mathrm{Var}(R) = \varnothing^{T}$ or (b) $\downarrow\! \varphi^{\star}(w)\cap\mathrm{Var}(R)\neq \varnothing^{T}$.

In case (a), Equation~\ref{EqtmHtg} turns into $\xi((M_{j})_{j\in \bb{u}})=R$. It follows that, for every $i\in\bb{w}$ and every $j\in \bb{u}$, $\bb{M_{j}}_{v^{\varphi(w_{i})}}=0$. Moreover, since $R\in\mathrm{Subt}(c_{w,r}(\sigma))_{t}$, we have that, for every $j\in \bb{u}$, $M_{j}\in\mathrm{Subt}(c_{w,r}(\sigma))_{u_{j}}$. In addition, for every $j\in\bb{u}$,
$$
M_{j}\in \left(\left(\!\begin{smallmatrix}v^{\varphi(w_{i})}_{i}\\ f^{\sharp}_{w_{i}}[[W_{w_{i},l_{i}}]_{\Theta_{w_{i}}}]\end{smallmatrix}\!\right)_{i\in \bb{w}}\right)^{\sharp}_{u_{j}}(M_{j}) = \{M_{j}\}.
$$
Since, for every $j\in \bb{u}$, $(M_{j}, N_{j})\in \Psi_{u_{j}}$, it happens that $N_{j}=M_{j}$. Thus $\xi((N_{j})_{j\in \bb{u}})=\xi((M_{j})_{j\in \bb{u}})=R$. Therefore this case is settled.

In case (b), in its turn, either (b.1) there exists an $i\in \bb{w}$ such that $R = v^{\varphi(w_{i})}_{i}$ or (b.2) there exists a unique $(R_{j})_{j\in \bb{u}}\in \mathrm{T}_{\Xi}(Y)_{u}$ such that $R = \xi((R_{j})_{j\in \bb{u}})$.  Note that if $R$ starts with an operation symbol from $\Xi$, then, from Equation~\ref{EqtmHtg}, it follows that $\xi$ is the only possibility.

In case (b.1), Equation~\ref{EqtmHtg} turns into
\begin{equation}\label{Eqtm2Htg}
\xi((M_{j})_{j\in \bb{u}}) = f^{\sharp}_{w_{i}}(W_{i}).
\end{equation}
Note that in this case $t$ must be equal to $\varphi(w_{i})$.

Case (b.1) still needs to be further refined attending to the different possibilities for $W_{i}$ according to the definition of the tree homomorphism $f^{\sharp}$. Either (b.1.i) there exists a unique $x\in X_{w_{i}}$ such that $W_{i} = x$ or (b.1.ii) there exists a unique  $w'\in S^{\star}$ a unique $\nu\in\Sigma_{w',w_{i}}$, and a unique $((Q_{i'})_{i'\in \bb{w'}})\in \mathrm{T}_{\Sigma}(X)_{w'}$ such that $W_{i}=\nu((Q_{i'})_{i'\in\bb{w'}})$.

In case (b.1.i), the value of the tree homomorphism $f^{\sharp}_{w_{i}}$ at $W_{i}$ is $f_{w_{i}}(x)$. Let us recall that $\Phi$ was defined so that the set $\{f_{w_{i}}(x)\}$ is recognized by $\Phi_{\varphi(w_{i})}$. Since, for every $j\in \bb{u}$, $(M_{j},N_{j})\in \Phi_{u}$, it follows that $\xi((N_{j})_{j\in \bb{u}})\in\{f_{w_{i}}(x)\}$. Therefore this case is settled.

In case (b.1.ii), the value of the tree homomorphism $f^{\sharp}_{w_{i}}$ at $W_{i}$ is given by
$$
f^{\sharp}_{w_{i}}(W_{i})=\left(\!\begin{smallmatrix}v^{\varphi(w'_{i'})}_{i'}\\ f^{\sharp}_{w'_{i'}}(Q_{i'})\end{smallmatrix}\!\right)_{i'\in\bb{w'}}
(d_{w',w_{i}}(\nu)).
$$

In virtue of Equation~\ref{Eqtm2Htg}, this last term is equal to $\xi((M_{j})_{j\in \bb{u}})$. On the other hand, since no substitution changes the operation symbol, we have that $d_{w', w_{i}}(\nu)$ has the form $\xi((R_{j})_{j\in \bb{u}})$ for a unique  $(R_{j})_{j\in \bb{u}}\in \mathrm{T}_{\Xi}(Y\cup{\downarrow}\varphi^{\star}(w))_{u}$. Moreover, for every $j\in \bb{u}$, $R_{j}$ is a subterm of sort $u_{j}$ of $d_{w', w_{i}}(\nu)$. Hence, by the linearity of the tree homomorphism $f^{\sharp}$, we have that
$$
\left(\!\begin{smallmatrix}v^{\varphi(w'_{i'})}_{i'}\\ f^{\sharp}_{w'_{i'}}(Q_{i'})\end{smallmatrix}\!\right)_{i'\in\bb{w'}}(\xi((R_{j})_{j\in \bb{u}}))=\xi\left(\left(\left(\!\begin{smallmatrix}v^{\varphi(w'_{i'})}_{i'}\\ f^{\sharp}_{w'_{i'}}(Q_{i'})\end{smallmatrix}\!\right)_{i'\in\bb{w'}}(R_{j})\right)_{j\in \bb{u}}\right).
$$

Consequently, for every $j\in \bb{u}$, $M_{j} = \left(\!\begin{smallmatrix}v^{\varphi(w'_{i'})}_{i'}\\ f^{\sharp}_{w'_{i'}}(Q_{i'})\end{smallmatrix}\!\right)_{i'\in\bb{w'}}(R_{j})$. Thus, for every $j\in \bb{u}$,
$M_{j}\in \left(\left(\!\begin{smallmatrix}v^{\varphi(w'_{i'})}_{i'}\\ f^{\sharp}_{w'_{i'}}[[Q_{i'}]_{\Theta_{w'_{i'}}}]\end{smallmatrix}\!\right)_{i'\in\bb{w'}}\right)^{\sharp}_{u_{j}}(R_{j})$. However, since, for every $j\in \bb{u}$, $(M_{j},N_{j})\in \Psi_{u_{j}}$, we can assert that $N_{j}\in\left(\left(\!\begin{smallmatrix}v^{\varphi(w'_{i'})}_{i'}\\ f^{\sharp}_{w'_{i'}}[[Q_{i'}]_{\Theta_{w'_{i'}}}]\end{smallmatrix}\!\right)_{i'\in\bb{w'}}\right)^{\sharp}_{u_{j}}(R_{j})$. Therefore, for every $i'\in \bb{w'}$, there exists a $\overline{Q}_{i'}\in [Q_{i'}]_{\Theta_{w'_{i'}}}$ such that $N_{j}=\left(\!\begin{smallmatrix}v^{\varphi(w'_{i'})}_{i'}\\ f^{\sharp}_{w'_{i'}}(\overline{Q}_{i'})\end{smallmatrix}\!\right)_{i'\in \bb{w'}}(R_{j})$. Let $\overline{W}_{i}$ stand for $\nu((\overline{Q}_{i'})_{i'\in \bb{w'}})$.

Before proceeding any further, let us note that
$$
f^{\sharp}_{w_{i}}(\overline{W}_{i})=\left(\!\begin{smallmatrix}v^{\varphi(w'_{i'})}_{i'}\\ f^{\sharp}_{w'_{i'}}(\overline{Q}_{i'})\end{smallmatrix}\!\right)_{i'\in \bb{w'}}(d_{w',w_{i}}(\nu))=\xi((N_{j})_{j\in \bb{u}}).
$$

Now, since, for every $i'\in w'$, $\overline{Q}_{i'}\in [Q_{i'}]_{\Theta_{w'_{i'}}}$, it follows that $\overline{W}_{i}\in [W_{i}]_{\Theta_{w_{i}}}$ and $\xi((N_{j})_{j\in \bb{u}})\in \left(\left(\!\begin{smallmatrix}v^{\varphi(w_{i})}_{i}\\ f^{\sharp}_{w_{i}}[[W_{i}]_{\Theta_{w_{i}}}]\end{smallmatrix}\!\right)_{i\in\bb{w}}\right)^{\sharp}_{t}(v^{\varphi(w_{i})}_{i})$. Therefore this case is settled.

It only remains to consider the case (b.2), in which $R=\xi((R_{j})_{j\in \bb{u}})$. Again, by the linearity of the tree homomorphism $f^{\sharp}$, it follows, from Equation~\ref{EqtmHtg}, that
$$
\xi((M_{j})_{j\in \bb{u}})=\xi\left(\left(\left(\!\begin{smallmatrix}v^{\varphi(w_{i})}_{i}\\ f^{\sharp}_{w_{i}}(Q_{i})\end{smallmatrix}\!\right)_{i\in\bb{w}}(R_{j})\right)_{j\in \bb{u}}\right).
$$
Hence, for every $j\in \bb{u}$, $M_{j}=\left(\!\begin{smallmatrix}v^{\varphi(w_{i})}_{i}\\ f^{\sharp}_{w_{i}}(Q_{i})\end{smallmatrix}\!\right)_{i\in\bb{w}}(R_{j})$. Then the statement may be handled in much the same way as in case (b.1.ii).

Therefore $\Psi$ is a congruence on $\mathbf{T}_{\Xi}(Y)$ of $T$-finite index.

Finally, we prove that  $[f^{\sharp}_{s}[L]]^{\Psi}_{\varphi(s)}=f^{\sharp}_{s}[L]$, i.e., that $f^{\sharp}_{s}[L]$ is $\Psi_{\varphi(s)}$-saturated. Since, obviously, $f^{\sharp}_{s}[L]\subseteq [f^{\sharp}_{s}[L]]^{\Psi}_{\varphi(s)}$, we restrict ourselves to prove the other inclusion.

Let $M$ be a term in $[f^{\sharp}_{s}[L]]^{\Psi}_{\varphi(s)}$, then there exists a term $P\in L$ such that $M$ and $f^{\sharp}_{s}(P)$ are $\Psi_{\varphi(s)}$-related. Now, for $P$, either (a) $P=x$, for a unique variable $x\in X_{s}$, or (b) $P=\sigma((P_{i})_{i\in\bb{w}})$, for a unique $w\in S^{\star}$, a unique $\sigma\in\Sigma_{w,s}$, and a unique  $(P_{i})_{i\in \bb{w}}\in\mathrm{T}_{\Sigma}(X)_{w}$.
Before proceeding any further, let us explain the reason why,  in the second case, it is unnecessary to take into account sorts $r\in S- \{s\}$.  The reason is simple, terms of the form $\sigma((P_{i})_{i\in\bb{w}})$, for $\sigma\in\Sigma_{w,r}$ with $r\neq s$ and $(P_{i})_{i\in \bb{w}}\in\mathrm{T}_{\Sigma}(X)_{w}$ are terms of sort $r$, and, consequently,  terms that are not in $L$.

In case (a), $f^{\sharp}_{s}(x)=f_{s}(x)$. By construction of $\Phi$, the set $\{f_{s}(x)\}$ is $\Phi_{\varphi(s)}$-saturated. Moreover, since $\Psi_{\varphi(s)}$ is a refinement of $\Phi_{\varphi(s)}$, we conclude that $\{f_{s}(x)\}$ is $\Psi_{\varphi(s)}$-saturated. In this case, sinece $M$ is related to $f_{s}(x)$ for $\Psi_{\varphi(s)}$, we conclude that $M$ must be equal to $f_{s}(x)$ and, consequently, $M$ is an element in $f^{\sharp}_{s}[L]$.

In case (b), $f^{\sharp}_{s}(\sigma((P_{i})_{i\in\bb{w}}))  = \left(\!\begin{smallmatrix}v^{\varphi(w_{i})}_{i}\\f^{\sharp}_{w_{i}}(P_{i})\end{smallmatrix}\!\right)_{i\in \bb{w}}(c_{w,s}(\sigma))$. Note that
$$
 \left(\!\begin{smallmatrix}v^{\varphi(w_{i})}_{i}\\f^{\sharp}_{w_{i}}(P_{i})\end{smallmatrix}\!\right)_{i\in \bb{w}}(c_{w,s}(\sigma))\in \left(\left(\!\begin{smallmatrix}v^{\varphi(w_{i})}_{i}\\ f^{\sharp}_{w_{i}}[[P_{i}]_{\Theta_{w_{i}}}]\end{smallmatrix}\!\right)_{i\in\bb{w}}\right)^{\sharp}_{s}(c_{w,s}(\sigma)).
$$
By construction of $\Psi$, it follows that $M\in \left(\left(\!\begin{smallmatrix}v^{\varphi(w_{i})}_{i}\\ f^{\sharp}_{w_{i}}[[P_{i}]_{\Theta_{w_{i}}}]\end{smallmatrix}\!\right)_{i\in\bb{w}}\right)^{\sharp}_{s}(c_{w,s}(\sigma))$. Therefore, there exists, for every $i\in\bb{w}$, terms $Q_{i}$ in $[P_{i}]_{\Theta_{w_{i}}}$ such that
$$
M=\left(\!\begin{smallmatrix}v^{\varphi(w_{i})}_{i}\\f^{\sharp}_{w_{i}}(Q_{i})\end{smallmatrix}\!\right)_{i\in \bb{w}}(c_{w,s}(\sigma)).
$$

It follows that $M=f^{\sharp}_{s}(\sigma((Q_{i})_{i\in\bb{w}}))$. Since $Q_{i}\in [P_{i}]_{\Theta_{w_{i}}}$, and $\Theta$ is a congruence on $\mathrm{T}_{\Sigma}(X)$, we conclude that $\sigma((P_{i})_{i\in\bb{w}})$ and $\sigma((Q_{i})_{i\in\bb{w}})$ are $\Theta_{s}$-related. Moreover, since $L$ is $\Theta_{s}$-saturated and $\sigma((P_{i})_{i\in\bb{w}})$ is a term in $L$, we can assert that $\sigma((Q_{i})_{i\in\bb{w}})$ is also a term in $L$. Consequently, $M$ is a term in $f^{\sharp}_{s}[L]$.

Hence, $f^{\sharp}_{s}[L]$ is $\Psi_{\varphi(s)}$-saturated.
\end{proof}

In~\cite{GS97}, Proposition 7.7, on p.~18, GÈcseg and Steinby state that, for a homomorphism $f\colon \mathbf{T}_{\Sigma}(X)\mor \mathbf{T}_{\Sigma}(Y)$ between single-sorted algebras and a subset $L$ of $\mathrm{T}_{\Sigma}(X)$, if $L\in \mathrm{Rec}(\mathbf{T}_{\Sigma}(X))$, then $f[L]\in \mathrm{Rec}(\mathbf{T}_{\Sigma}(Y))$. Moreover, they, seemingly, provide a proof of it. One can obtain a plain proof of such a result, as an immediate corollary of Proposition~\ref{PRecLH} and particularizing it to the single-sorted case, taking into account the following fact: Every homomorphism $f\colon \mathbf{T}_{\Sigma}(X)\mor \mathbf{T}_{\Sigma}(Y)$ be\-tween many-sorted algebras  is an instance of a linear tree homomorphism. Indeed, for $((\mathrm{id}_{S},c),f\circ\eta_{X})$, where $c = (c_{w,s})_{(w,s)\in S^{\star}\times S}$ is the family of mappings defined, for every $(w,s)\in S^{\star}\times S$, as follows:
$$
c_{w,s}
\nfunction
{\Sigma_{w,s}}
{\mathrm{T}_{\Sigma}(Y\cup \downarrow\! w)_{s}}
{\sigma}
{\sigma\left(v^{w_{0}}_{0},\ldots,v^{w_{\bb{w}-1}}_{\bb{w}-1}\right)}
$$
we have that $(f\circ\eta_{X})^{\sharp}\colon \mathbf{T}_{\Sigma}(X)\mor \mathbf{T}_{\Sigma}(Y)$, the tree homomorphism determined by the pair $((\mathrm{id}_{S},c),f\circ\eta_{X})$, is $f$.

\begin{corollary}
Let $f$ be a homomorphism from $\mathbf{T}_{\Sigma}(X)$ to $\mathbf{T}_{\Sigma}(Y)$, $s\in S$ and $L\in \mathrm{T}_{\Sigma}(X)^{\wp}_{s}$. If $L\in\mathrm{Rec}_{s}(\mathbf{T}_{\Sigma}(X))$, then $f_{s}[L]\in\mathrm{Rec}_{s}(\mathbf{T}_{\Sigma}(Y))$.
\end{corollary}

We next provide, among other things, a categorial rendering of the just indicated result, but for a suitable class of homomorphisms.

\begin{proposition}
Let $\mathrm{Rec}_{\Sigma}$ be the mapping that sends an $S$-sorted set $X$ to the $\Sigma$-algebra $\mathbf{Rec}_{\boldsymbol{\cdot}}(\mathbf{T}_{\Sigma}(X))$ and an $S$-sorted mapping $f$ from $X$ to $Y$ to the $S$-sorted mapping $(f^{@})^{\wp} = (f^{@}_{s}[\cdot])_{s\in S}$ from $(\mathrm{Rec}_{s}(\mathbf{T}_{\Sigma}(X)))_{s\in S}$ to $(\mathrm{Rec}_{s}(\mathbf{T}_{\Sigma}(Y)))_{s\in S}$, where, we recall, $f^{@}$ is the unique homomorphism from $\mathbf{T}_{\Sigma}(X)$ to $\mathbf{T}_{\Sigma}(Y)$ such that $f^{@}\circ\eta_{X} = \eta_{Y}\circ f$---let us note that  $(f^{@})^{\wp}$ is an $S$-sorted mapping from $\mathrm{T}_{\Sigma}(X)^{\wp}$ to $\mathrm{T}_{\Sigma}(Y)^{\wp}$, and that to simplify notation we have used the same symbol for its restriction to $(\mathrm{Rec}_{s}(\mathbf{T}_{\Sigma}(X)))_{s\in S}$ and $(\mathrm{Rec}_{s}(\mathbf{T}_{\Sigma}(Y)))_{s\in S}$. Then, for every $(w,s)\in S^{\star}\times S$ and every $\sigma\in \Sigma_{w,s}$, $f^{@}_{s}[\cdot]\circ \sigma^{\wp} = \sigma^{\wp}\circ f^{@}_{w}[\cdot]$, i.e., for every $(L_{i})_{i\in \bb{w}}\in \prod_{i\in \bb{w}}\mathrm{Rec}_{w_{i}}(\mathbf{T}_{\Sigma}(X))$, the sets
\begin{gather*}
\textstyle
(f^{@}_{s}[\cdot]\circ \sigma^{\wp})((L_{i})_{i\in \bb{w}}) = \{f^{@}_{s}(\sigma((P_{i})_{i\in n}))\mid (P_{i})_{i\in \bb{w}}\in \prod_{i\in \bb{w}}L_{i}\} \text{ and} \\
\textstyle
(\sigma^{\wp}\circ f^{@}_{w}[\cdot])((L_{i})_{i\in \bb{w}}) = \{\sigma((Q_{i})_{i\in n})\mid (Q_{i})_{i\in \bb{w}}\in \prod_{i\in \bb{w}}f^{@}_{w_{i}}[L_{i}]\}
\end{gather*}
are equal. Thus $(f^{@})^{\wp}$ is a homomorphism from the $\Sigma$-algebra $\mathbf{Rec}_{\boldsymbol{\cdot}}(\mathbf{T}_{\Sigma}(X))$ to the $\Sigma$-algebra  $\mathbf{Rec}_{\boldsymbol{\cdot}}(\mathbf{T}_{\Sigma}(Y))$. Therefore $\mathrm{Rec}_{\Sigma}$ is a functor from $\mathbf{Set}^{S}$ to $\mathbf{Alg}(\Sigma)$ and a subfunctor of $(\cdot)^{\wp}\circ \mathbf{T}_{\Sigma}$.

%Let $\mathrm{Rec}_{\Sigma}$ be the mapping that sends a set $X\in \boldsymbol{\mathcal{U}}$ to the $\Sigma$-algebra $\mathbf{Rec}(\mathbf{T}_{\Sigma}(X))$ and a mapping $f\colon X\mor Y$ to the mapping $f^{@}[\cdot]$ from $\mathrm{Rec}(\mathbf{T}_{\Sigma}(X))$ to $\mathrm{Rec}(\mathbf{T}_{\Sigma}(Y))$, where, we recall, $f^{@}$ is the unique $\Sigma$-homomorphism from $\mathbf{T}_{\Sigma}(X)$ to $\mathbf{T}_{\Sigma}(Y)$ such that $f^{@}\circ\eta_{X} = \eta_{Y}\circ f$. Then, for every $n\in \mathbb{N}$ and every $\sigma\in \Sigma_{n}$, $f^{@}[\cdot]\circ F_{\sigma}^{\wp} = G_{\sigma}^{\wp}\circ f^{@}[\cdot]^{n}$, i.e., for every $(L_{i})_{i\in n}\in \mathrm{Rec}(\mathbf{T}_{\Sigma}(X))^{n}$, the sets
%\begin{gather*}
%\textstyle
%(f^{@}[\cdot]\circ F_{\sigma}^{\wp})((L_{i})_{i\in n}) = \{f^{@}(F_{\sigma}((P_{i})_{i\in n})\mid (P_{i})_{i\in n}\in \prod_{i\in n}L_{i}\} \text{ and} \\
%\textstyle
%(G_{\sigma}^{\wp}\circ f^{@}[\cdot]^{n})((L_{i})_{i\in n}) = \{G_{\sigma}((Q_{i})_{i\in n})\mid (Q_{i})_{i\in n}\in \prod_{i\in n}f^{@}[L_{i}]\}
%\end{gather*}
%are equal. Therefore $f^{@}[\cdot]$ is a homomorphism from $\mathbf{Rec}(\mathbf{T}_{\Sigma}(X))$ to $\mathbf{Rec}(\mathbf{T}_{\Sigma}(Y))$. Thus $\mathrm{Rec}_{\Sigma}$ is a functor from $\mathbf{Set}$ to $\mathbf{Alg}(\Sigma)$ and a subfunctor of $(\cdot)^{\wp}\circ \mathbf{T}_{\Sigma}$.

Let $f$ be an $S$-sorted mapping from $X$ to $Y$, $\{\cdot\}_{X}^{\sharp}$ the canonical extension of the $S$-sorted mapping $\{\cdot\}_{X}$ from $X$ to $(\mathrm{Rec}_{s}(\mathbf{T}_{\Sigma}(X)))_{s\in S}$ that, for every $s\in S$, sends $x$ in $X_{s}$ to $\{x\}$ in $\mathrm{Rec}_{s}(\mathbf{T}_{\Sigma}(X))$, and $\{\cdot\}_{Y}^{\sharp}$ the canonical extension of the $S$-sorted  mapping $\{\cdot\}_{Y}$ from $Y$ to $(\mathrm{Rec}_{s}(\mathbf{T}_{\Sigma}(Y)))_{s\in S}$. Then the following diagram commutes
$$
\xymatrix@C=60pt{
\mathbf{T}_{\Sigma}(X)\ar[r]^-{\{\cdot\}_{X}^{\sharp}} \ar[d]_-{f^{@}}&
\mathbf{Rec}_{\boldsymbol{\cdot}}(\mathbf{T}_{\Sigma}(X))\ar[d]^-{(f^{@})^{\wp}}\\
\mathbf{T}_{\Sigma}(Y)\ar[r]_-{\{\cdot\}_{Y}^{\sharp}}&
\mathbf{Rec}_{\boldsymbol{\cdot}}(\mathbf{T}_{\Sigma}(Y))
}
$$
Therefore $(\{\cdot\}_{X}^{\sharp})_{X\in \boldsymbol{\mathcal{U}}}$ is a natural transformation from $\mathbf{T}_{\Sigma}$ to $\mathrm{Rec}_{\Sigma}$.

Let $X$ be an $S$-sorted set, $(w,s)\in S^{\star}\times S$, $\sigma\in \Sigma_{w,s}$, $i\in \bb{w}$, and $(L_{i})_{i\in \bb{w}}\in \prod_{i\in \bb{w}}\mathrm{Rec}_{w_{i}}(\mathbf{T}_{\Sigma}(X))$ such that $L_{i} = L'\cup L''$. Then
$$
\sigma^{\wp}((L_{i})_{i\in \bb{w}}) = \sigma^{\wp}(L_{0},\ldots,L',\ldots,L_{\bb{w}-1})\cup
\sigma^{\wp}(L_{0},\ldots,L'',\ldots,L_{\bb{w}-1}).
$$
Hence $\mathbf{Rec}_{\boldsymbol{\cdot}}(\mathbf{T}_{\Sigma}(X))$ is a many-sorted Boolean Algebra with operators, i.e., a $\Sigma$-algebra such that, for every $s\in S$, $\mathbf{Rec}_{s}(\mathbf{T}_{\Sigma}(X))$ is a Boolean algebra and the structures of $\Sigma$-algebra and of Boolean algebra are compatibles as stated above (for the concept of single-sorted Boolean algebra with operators see~\cite{{jt51}} and \cite{{jt52}}).

Finally, let $f\colon X\mor Y$ an $S$-sorted mapping. Then the homomorphism $(f^{@})^{\wp} =(f^{@}_{s}[\cdot])_{s\in S}$ is such that, for every $s\in S$, $f^{@}_{s}[\cdot]$ is a Boolean algebra homomorphism from $\mathbf{Rec}_{s}(\mathbf{T}_{\Sigma}(X))$ to $\mathbf{Rec}_{s}(\mathbf{T}_{\Sigma}(Y))$.
\end{proposition}

From the just stated proposition and Proposition~\ref{Rec SubdirectProd sRec} we obtain the following corollary.

\begin{corollary}
Let $X$ be an $S$-sorted set. Then the Boolean algebra $\mathbf{Rec}(\mathbf{T}_{\Sigma}(X))$ is a subdirect product of the family of Boolean algebras $(\mathbf{Rec}_{s}(\mathbf{T}_{\Sigma}(X)))_{s\in S}$. Moreover, if $f\colon X\mor Y$ is an $S$-sorted mapping, then $f^{@}[\cdot]$ and $\prod(f^{@})^{\wp} = \prod_{s\in S}f^{@}_{s}[\cdot]$ are compatible with the subdirect embeddings of $\mathbf{Rec}(\mathbf{T}_{\Sigma}(X))$ and $\mathbf{Rec}(\mathbf{T}_{\Sigma}(Y))$ into $\prod_{s\in S}\mathbf{Rec}_{s}(\mathbf{T}_{\Sigma}(X))$ and $\prod_{s\in S}\mathbf{Rec}_{s}(\mathbf{T}_{\Sigma}(Y))$, respectively.
\end{corollary}

\begin{remark}
From the results stated above it seems natural to consider categories of the type $\mathbf{BoolOp}^{S}$, where $\mathbf{BoolOp}$ is a category of Boolean algebras with operators, and to investigate a duality theory for them. Toward this end it is perhaps worth pointing out that, for a finite set of sorts $S$, since the discrete category determined by $S$ is loopless (it has no non-identity endomorphisms), $\mathbf{Pro}(\mathbf{Set}_{\mathrm{f}}^{S})$ is equivalent to $(\mathbf{Pro}(\mathbf{Set}_{\mathrm{f}}))^{S}$. But $(\mathbf{Pro}(\mathbf{Set}_{\mathrm{f}}))^{S}$ and $\mathbf{Stone}^{S}$ are equivalent, because  $\mathbf{Pro}(\mathbf{Set}_{\mathrm{f}})$ and $\mathbf{Stone}$ are equivalent. Thus $\mathbf{Pro}(\mathbf{Set}_{\mathrm{f}}^{S})$ and $\mathbf{Stone}^{S}$ are equivalent. Therefore $(\mathbf{Bool}^{S})^{\mathrm{op}}$ and $\mathbf{Stone}^{S}$ are equivalent.
\end{remark}

\subsection{Derivors and recognizability}

Thatcher in~\cite{Tha69}, on p.~132, wrote: ``Generally, the term `transformation' will mean any map from $\mathrm{T}_{\Sigma}$ [the free algebra $\mathrm{T}_{\Sigma}(\varnothing)$, \emph{we add}] into $\mathrm{T}_{\Omega}$ [the free algebra $\mathrm{T}_{\Omega}(\varnothing)$, \emph{we add}] where $\Sigma = (\Sigma,r)$ and $\Omega = (\Omega,s)$ are ranked alphabets. We will, in the sequel, define several types of transformations which are of particular interest. In formulating these definitions, there were three principal considerations or objectives: (1) It was intended that the transformations be, in some sense, natural in that they would fit the algebraic framework within which we are working; (2) We should be able to generalize the conventional concept of finite state mapping \ldots to the case for trees; and (3) It was hoped that the end result would be a unified approach taking into account various formulations of `transformation,' `transduction,' and `translation' which have appeared in the literature.''
%{%\small
%\begin{quotation}
%Generally, the term ``transformation'' will mean any map from $\mathrm{T}_{\Sigma}$ [the free algebra $\mathrm{T}_{\Sigma}(\varnothing)$, \emph{we add}] into $\mathrm{T}_{\Omega}$ [the free algebra $\mathrm{T}_{\Omega}(\varnothing)$, \emph{we add}] where $\Sigma = (\Sigma,r)$ are $\Omega = (\Omega,s)$ are ranked alphabets. We will, in the sequel, define several types of transformations which are of particular interest. In formulating these definitions, there were three principal considerations or objectives: (1) It was intended that the transformations be, in some sense, natural in that they would fit the algebraic framework within which we are working; (2) We should be able to generalize the conventional concept of finite state mapping \ldots to the case for trees; and (3) It was hoped that the end result would be a unified approach taking into account various formulations of ``transformation,'' ``transduction,'' and ``translation'' which have appeared in the literature.
%\end{quotation}
%}

%Although it is true that linear tree homomorphisms preserve recognizability of languages (relative to each sort), however, they, from a purely mathematical standpoint, are not convenient, because they are not the morphisms of any category.

Following Thatcher's first dictum, in this subsection, after defining the variety of Hall algebras, we introduce the notion of derivor between many-sorted signatures. This will allow us, using the homomorphisms between Hall algebras, to obtain $\mathbf{Sig}_{\mathfrak{d}}$, the category of many-sorted signatures and derivors---we recall that we were unable to show that hyperderivors are the morphisms of a category with objects ordered pairs $(\mathbf{\Sigma},X)$, where $\mathbf{\Sigma} = (S,\Sigma)$ is a many-sorted signature and $X$ an $S$-sorted set. Next, after defining a suitable contravariant functor from $\mathbf{Sig}_{\mathfrak{d}}$ to $\mathbf{Cat}$, the category of $\boldsymbol{\mathcal{U}}$-locally small categories and functors, we obtain, by means of the Grothendieck construction, $\mathbf{Alg}_{\mathfrak{d}}$, the category with objects the ordered pairs $((S,\Sigma),(A,F))$, with $(S,\Sigma)$ a many-sorted signature and $(A,F)$ a $\Sigma$-algebra, and morphisms from $((S,\Sigma),(A,F))$ to  $((T,\Lambda),(B,G))$ the ordered pairs $((\varphi,d),f)$, with $(\varphi,d)$ a derivor from $(S,\Sigma)$ to $(T,\Lambda)$ and $f$ a homomorphism of $\Sigma$-algebras from $(A,F)$ to $(B_{\varphi},G^{(\varphi,d)})$, a canonically derived algebra of $(B,G)$. Finally, after showing that every derivor is a hyperderivor, we state the counterparts of Propositions~\ref{PRecH} and~\ref{PRecLH} for suitable morphisms of $\mathbf{Alg}_{\mathfrak{d}}$.

Before defining the notion of Hall algebra, we recall that (1) a finitary specification is an ordered triple $(S,\Sigma,\mathcal{E})$, where $S$ is a set of sorts, $\Sigma$ an $S$-sorted signature, and $\mathcal{E}\subseteq\mathrm{Eq}_{V^{S}}(\Sigma)=
(\mathrm{T}_{\Sigma}(X)_{s}^{2})_{(X,s)\in \mathrm{Sub}_{\mathrm{f}}(V^{S})\times S}$, i.e., a set of finitary  $\Sigma$-equations, where $V^{S}$, the $S$-sorted set of the variables, is a fixed $S$-countably infinite $S$-sorted set; and that (2) a $\Sigma$-algebra $\mathbf{A}$ is a $(S,\Sigma,\mathcal{E})$-algebra if $\mathbf{A}\models^{\Sigma} \mathcal{E}$, i.e., if, for every $(X,s)\in \mathrm{Sub}_{\mathrm{f}}(V^{S})\times S$ and every $(P,Q)\in \mathcal{E}_{X,s}$, $\mathbf{A}$ is a model of $(P,Q)$, in symbols $\mathbf{A}\models^{\Sigma}_{X,s} (P,Q)$, which in turn means that, for every $a\in \mathrm{Hom}(X,A)$, $a^{\sharp}_{s}(P) = a^{\sharp}_{s}(Q)$.

\begin{definition}
Let $S$ be a set of sorts and $V^{\mathrm{H}_{S}}$ the $S^{\star}\times S$-sorted set of variables $(V_{w,s})_{(w,s)\in
S^{\star}\times S}$ where $V_{w,s} = \{\,v^{w,s}_{n}\mid n\in \mathbb{N}\,\}$, for every $(w,s)\in S^{\star}\times S$.  A
\emph{Hall algebra for} $S$ is a $\mathrm{H}_{S} =
(S^{\star}\times S,\Sigma^{\mathrm{H}_{S}},\mathcal{E}^{\mathrm{H}_{S}})$-algebra, where $\Sigma^{\mathrm{H}_{S}}$ is the $S^{\star}\times S$-sorted signature, i.e., the $(S^{\star}\times S)^{\star}\times (S^{\star}\times S)$-sorted set, defined as follows:
\begin{enumerate}
\item[$\mathrm{HS}_{1}$.] For every $w\in S^{\star}$ and $i\in\bb{w}$,
      $$
      \pi^{w}_{i}\colon\lambda\mor(w,w_{i}),
      $$
      where $\bb{w}$ is the \emph{length} of the word $w$ and $\lambda$ the
      \emph{empty word} in the underlying set of the free monoid on
      $S^{\star}\times S$.

\item[$\mathrm{HS}_{2}$.] For every $u$, $w\in S^{\star}$ and $s\in S$,
      $$
      \xi_{u,w,s}\colon ((w,s),(u,w_{0}),\ldots,(u,w_{\bb{w}-1}))\mor(u,s);
      $$
\end{enumerate}
while $\mathcal{E}^{\mathrm{H}_{S}}$ is the sub-$(S^{\star}\times
S)^{\star}\times (S^{\star}\times S)$-sorted set of
$\mathrm{Eq}(\Sigma^{\mathrm{H}_{S}})$, where
$$
\mathrm{Eq}(\Sigma^{\mathrm{H}_{S}}) =
(\mathrm{T}_{\Sigma^{\mathrm{H}_{S}}}(\vs{\ol{w}})_{(u,s)}^{2})_{(\ol{w},(u,s))\in
(S^{\star}\times S)^{\star}\times (S^{\star}\times S)},
$$
defined as follows:
\begin{enumerate}
\item[$\mathrm{H}_{1}$.] \emph{Projection}.
      For every $u$, $w\in S^{\star}$ and $i\in\bb{w}$, the equation

      $$
      \xi_{u,w,w_{i}}(\pi^{w}_{i},v^{u,w_{0}}_{0},\ldots,
           v^{u,w_{\bb{w}-1}}_{\bb{w}-1})=v^{u,w_{i}}_{i}
      $$
      of type
      $(((u,w_{0}),\ldots,(u,w_{\bb{w}-1})),(u,w_{i})).$
\item[$\mathrm{H}_{2}$.] \emph{Identity}.
      For every $u\in S^{\star}$ and $j\in \bb{u}$, the equation
      $$
      \xi_{u,u,u_{j}}(v^{u,u_{j}}_{j},\pi^{u}_{0},\ldots,\pi^{u}_{\bb{u}-1})=
           v^{u,u_{j}}_{j}
      $$
      of type
      $(((u,u_{j})),(u,u_{j})).$
\item[$\mathrm{H}_{3}$.] \emph{Associativity}.
      For every $u$, $v$, $w\in S^{\star}$ and $s\in S$, the
      equation
     \begin{align*}
      \xi_{u,v,s}(
      \xi_{v,w,s}(v^{w,s}_{0},v^{v,w_{0}}_{1},\ldots,
          v^{v,w_{\bb{w}-1}}_{\bb{w}}),
          v^{u,v_{0}}_{\bb{w}+1},
          \ldots,v^{u,v_{\bb{v}-1}}_{\bb{w}+\bb{v}}) = \\
      \begin{aligned}
       \xi_{u,w,s}(v^{w,s}_{0},
          &\xi_{u,v,w_{0}}(v^{v,w_{0}}_{1},v^{u,v_{0}}_{\bb{w}+1},
                   \ldots,v^{u,v_{\bb{v}-1}}_{\bb{w}+\bb{v}}),
          \ldots, \\
       &\xi_{u,v,w_{\bb{w}-1}}(v^{v,w_{\bb{w}-1}}_{\bb{w}},
                   v^{u,v_{0}}_{\bb{w}+1},
       \ldots,v^{u,v_{\bb{v}-1}}_{\bb{w}+\bb{v}}))
      \end{aligned}
     \end{align*}
     of type
     $(((w,s),(v,w_{0}),\ldots,(v,w_{\bb{w}-1}),
            (u,v_{0}),\ldots,(u,v_{\bb{v}-1})),(u,s)).$
\end{enumerate}

We call the formal constants $\pi^{w}_{i}$ \emph{projections}, and the formal operations $\xi_{u,w,s}$ \emph{substitution operators}. Furthermore, we will denote by $\mathbf{Alg}(\mathrm{H}_{S})$ the category of Hall algebras for $S$ and homomorphisms between Hall algebras.  Since $\mathbf{Alg}(\mathrm{H}_{S})$ is a variety, the forgetful functor $\mathrm{G}_{\mathrm{H}_{S}}$ from $\mathbf{Alg}(\mathrm{H}_{S})$ to $\mathbf{Set}^{S^{\star}\times S}$ has a left adjoint $\mathbf{T}_{\mathrm{H}_{S}}$, situation denoted by $\mathbf{T}_{\mathrm{H}_{S}}\dashv \mathrm{G}_{\mathrm{H}_{S}}$, or diagrammatically by
$$
\xymatrix@=5pc{ \mathbf{Alg}(\mathrm{H}_{S})
\ar@<1.5ex>[r]^{\mathrm{G}_{\mathrm{H}_{S}}} \ar@{}[r]|{\uadj} &
\mathbf{Set}^{S^{\star}\times S}
\ar@<1.5ex>[l]^{\mathbf{T}_{\mathrm{H}_{S}}} }
$$
which assigns to an $S^{\star}\times S$-sorted set $\Sigma$ the corresponding free Hall algebra
$\mathbf{T}_{\mathrm{H}_{S}}(\Sigma)$.
\end{definition}

\begin{remark}
From $\mathrm{H}_{3}$, for $w=\lambda$, the empty word on $S$, we get the invariance of constant functions axiom in~\cite{gm85}: For every $u$, $v\in S^{\star}$ and $s\in S$, we have the equation
$$
\xi_{u,v,s}(\xi_{v,\lambda,s}(v^{\lambda,s}_{0}),v^{u,v_{0}}_{1},
\ldots,v^{u,v_{\bb{v}-1}}_{\bb{v}})=
\xi_{u,\lambda,s}(v^{\lambda,s}_{0})
$$
of type $(((\lambda,s),(u,v_{0}),\ldots,(u,v_{\bb{v}-1})),(u,s))$.
\end{remark}

For every $S$-sorted set $A$, $\mathrm{Op}_{\mathrm{H}_{S}}(A) = (\mathrm{Hom}(A_{w},A_{s}))_{(w,s)\in S^{\star}\times S}$, the $S^{\star}\times S$-sorted set of operation for $A$, is naturally equipped with a structure of Hall algebra, as stated in the following proposition, if we realize the projections as the true projections and the substitution operators as the generalized composition of mappings.

\begin{proposition}
Let $A$ be an $S$-sorted set and $\mathbf{Op}_{\mathrm{H}_{S}}(A)$ the $\Sigma^{\mathrm{H}_{S}}$-algebra with underlying many-sorted set $\mathrm{Op}_{\mathrm{H}_{S}}(A)$ and algebraic structure defined as follows:
\begin{enumerate}
\item For every $w\in S^{\star}$ and $i\in\bb{w}$,
      $(\pi^{w}_{i})^{\mathbf{Op}_{\mathrm{H}_{S}}(A)} =
      \mathrm{pr}^{A}_{w,i}\colon A_{w}\mor A_{w_{i}}$.

\item For every $u,w\in S^{\star}$ and $s\in S$,
      $\xi_{u,w,s}^{\mathbf{Op}_{\mathrm{H}_{S}}(A)}$ is defined, for every $f\in
      A _{s}^{ A_{w} }$ and $g\in A^{A_{u}}_{w}$, as
      $\xi_{u,w,s}^{\mathbf{Op}_{\mathrm{H}_{S}}(A)}(f,g_{0},\ldots,g_{\bb{w}-1})
      = f\circ\langle g_{i}\rangle_{i\in\bb{w} }$, where $\langle
      g_{i}\rangle_{i\in\bb{w}}$ is the unique mapping from $A_{u}$ to
      $A_{w}$ such that, for every $i\in \bb{w}$, we have that
      $$
      \mathrm{pr}^{A}_{w,i}\circ\langle g_{i}\rangle_{i\in\bb{w}} = g_{i}.
      $$
\end{enumerate}
Then $\mathbf{Op}_{\mathrm{H}_{S}}(A)$ is a Hall algebra, the \emph{Hall algebra for} $(S,A)$.
\end{proposition}

\begin{remark}
The closed sets of the Hall algebra $\mathbf{Op}_{\mathrm{H}_{S}}(A)$ for $(S,A)$ are precisely the clones of (many-sorted) operations on the $S$-sorted set $A$. On the other hand, every $\Sigma$-algebra $\mathbf{A}$ has associated a Hall algebra. In fact, it suffices to consider $\mathbf{Op}_{\mathrm{H}_{S}}(A)$, denoted by $\mathbf{Op}_{\mathrm{H}_{S}}(\mathbf{A})$. Moreover, the finitary term operations on $\mathbf{A}$ and the finitary algebraic operations on $\mathbf{A}$ are subalgebras of the Hall algebra $\mathbf{Op}_{\mathrm{H}_{S}}(\mathbf{A})$.
\end{remark}

For every $S$-sorted signature $\Sigma$, $\mathrm{Ter}_{\mathrm{H}_{S}}(\Sigma) = (\mathrm{T}_{\Sigma}(\vs{w})_{s})_{(w,s)\in S^{\star}\times S}$ is also e\-quipped with a structure of Hall algebra that formalizes the concept of substitution as stated in the following proposition.

\begin{proposition}
Let $\Sigma$ be an $S$-sorted signature and $\mathbf{Ter}_{\mathrm{H}_{S}}(\Sigma)$ the $\Sigma^{\mathrm{H}_{S}}$-algebra
with underlying many-sorted set $\mathrm{Ter}_{\mathrm{H}_{S}}(\Sigma)$ and algebraic structure defined as follows:
\begin{enumerate}
\item For every $w\in S^{\star}$ and $i\in\bb{w}$,
      $(\pi^{w}_{i})^{\mathbf{Ter}_{\mathrm{H}_{S}}(\Sigma)}$
      is the image under $\eta_{\vs{w}, w_{i}}$ of
      the variable $v_{i}^{w_{i}}$, where $\eta_{\vs{w}} =
      (\eta_{\vs{w},s})_{s\in S}$ is the canonical embedding
      of $\vs{w}$ into $\mathrm{T}_{\Sigma}(\vs{w})$. Sometimes,
      to abbreviate, we will write $\pi^{w}_{i}$ instead of
      $(\pi^{w}_{i})^{\mathbf{Ter}_{\mathrm{H}_{S}}(\Sigma)}$.

\item For every $u,w\in S^{\star}$ and $s\in S$,
      $\xi_{u,w,s}^{\mathbf{Ter}_{\mathrm{H}_{S}}(\Sigma)}$
      is the mapping
      $$\xi_{u,w,s}^{\mathbf{Ter}_{\mathrm{H}_{S}}(\Sigma)}\nfunction
      {\mathrm{T}_{\Sigma}(\vs{w})_{s} \times
      \mathrm{T}_{\Sigma}(\vs{u})_{w_{0}} \times \cdots \times
      \mathrm{T}_{\Sigma}(\vs{u})_{w_{\bb{w}-1}}}
      {\mathrm{T}_{\Sigma}(\vs{u})_{s}}
      {(P,(Q_{i})_{i\in\bb{w}})}
      {\left(\left(\!\begin{smallmatrix}v^{w_{i}}_{i}\\Q_{i}\end{smallmatrix}\!\right)_{i\in \bb{w}}\right)^{\sharp}_{s}(P)}
     $$
     where, for $\left(\!\begin{smallmatrix}v^{w_{i}}_{i}\\Q_{i}\end{smallmatrix}\!\right)_{i\in \bb{w}}$, the $S$-sorted mapping from $\vs{w}$
     to $\mathrm{T}_{\Sigma}(\vs{u})$ canonically associated to the family
     $(Q_{i})_{i\in\bb{w}}$, $\left(\left(\!\begin{smallmatrix}v^{w_{i}}_{i}\\Q_{i}\end{smallmatrix}\!\right)_{i\in \bb{w}}\right)^{\sharp}$, also denoted by $\mathcal{Q}^{\sharp}$, is the
     unique homomorphism from $\mathbf{T}_{\Sigma}(\vs{w})$ into
     $\mathbf{T}_{\Sigma}(\vs{u})$ such that
     $$
     \left(\left(\!\begin{smallmatrix}v^{w_{i}}_{i}\\Q_{i}\end{smallmatrix}\!\right)_{i\in \bb{w}}\right)^{\sharp}\circ \eta_{\vs{w}} = \left(\!\begin{smallmatrix}v^{w_{i}}_{i}\\Q_{i}\end{smallmatrix}\!\right)_{i\in \bb{w}}.
     $$
     Sometimes, to abbreviate, we will write $\xi_{u,w,s}$ instead of
     $\xi_{u,w,s}^{\mathbf{Ter}_{\mathrm{H}_{S}}(\Sigma)}$.
\end{enumerate}

Then $\mathbf{Ter}_{\mathrm{H}_{S}}(\Sigma)$ is a Hall algebra, the \emph{Hall algebra for} $(S,\Sigma)$.
\end{proposition}

\begin{remark}
For every $\Sigma$-algebra $\mathbf{A}$ there exists a homomorphism from the Hall algebra $\mathbf{Ter}_{\mathrm{H}_{S}}(\Sigma)$ to the Hall algebra $\mathbf{Op}_{\mathrm{H}_{S}}(\mathbf{A})$ and its image is the Hall subalgebra of the finitary term operations on $\mathbf{A}$.
\end{remark}

\begin{proposition}
Let $\Sigma$ be an $S$-sorted signature. Then $\mathbf{Ter}_{\mathrm{H}_{S}}(\Sigma)^{\wp}$ is a Hall algebra.
\end{proposition}

%recall that the underlying $S^{\star}\times S$-sorted of $\mathbf{Ter}_{\mathrm{H}_{S}}(\Sigma)^{\wp}$ is $(\mathrm{Sub}(\mathrm{T}_{\Sigma}(\vs{w})_{s}))_{(w,s)\in S^{\star}\times S}$

Our next goal is to prove that, for every $S^{\star}\times S$-sorted set $\Sigma$, $\mathbf{T}_{\mathrm{H}_{S}}(\Sigma)$, the free Hall algebra on $\Sigma$, is isomorphic to $\mathbf{Ter}_{\mathrm{H}_{S}}(\Sigma)$. We remark that the existence of this isomorphism is interesting because it enables us to get a more tractable description of the terms in $\mathbf{T}_{\mathrm{H}_{S}}(\Sigma)$.
%We remark that the existence of this isomorphism is interesting because it enables us, on the one hand, to get a more tractable description of the terms in $\mathbf{T}_{\mathrm{H}_{S}}(\Sigma)$, and, on the other hand, as we will show afterwards, to state, for every $\Sigma$-algebra $\mathbf{A}$, taking into account the adjunction $\mathbf{T}_{\mathrm{H}_{S}}\dashv \mathrm{G}_{\mathrm{H}_{S}}$, the existence of a homomorphism of Hall algebras
%$\mathrm{Tr}^{\mathbf{A}}$ from $\mathbf{Ter}_{\mathrm{H}_{S}}(\Sigma)$ to $\mathbf{Op}_{\mathrm{H}_{S}}(\mathbf{A}) = \mathbf{Op}_{\mathrm{H}_{S}}(A)$ such that $\mathrm{Th}_{\Sigma}(\mathbf{A})$, the finitary $\Sigma$-equational theory determined by $\mathbf{A}$, is
%precisely $\mathrm{Ker}(\mathrm{Tr}^{\mathbf{A}})$, the kernel of the homomorphism $\mathrm{Tr}^{\mathbf{A}}$.

To attain the goal just stated we begin by defining, for a Hall algebra $\mathbf{A}$, an $S$-sorted signature $\Sigma$, an
$S^{\star}\times S$-mapping $f\colon\Sigma\mor A$, and a word $u\in S^{\star}$, the concept of derived $\Sigma$-algebra of
$\mathbf{A}$ for $(f,u)$, since it will be used afterwards in the proof of the isomorphism between
$\mathbf{T}_{\mathrm{H}_{S}}(\Sigma)$ and $\mathbf{Ter}_{\mathrm{H}_{S}}(\Sigma)$.

\begin{definition}
Let $\mathbf{A}$ be a Hall algebra and $\Sigma$ an $S$-sorted signature. Then, for every $f\colon\Sigma\mor A$ and $u\in
S^{\star}$, $\mathbf{A}^{f,u}$, the \emph{derived} $\Sigma$-\emph{algebra} \emph{of} $\mathbf{A}$ \emph{for} $(f,u)$,
is the $\Sigma$-algebra with underlying $S$-sorted set $A^{f,u} = (A_{u,s})_{s\in S}$ and algebraic structure $F^{f,u}$, defined, for every $(w,s)\in S^{\star}\times S$, as
$$
F^{f,u}_{w,s} \nfunction
{\Sigma_{w,s}}{\mathrm{Op}_{w}(A^{f,u})_{s}} {\sigma}{\nfunction
    {\prod_{i\in \bb{w}}A_{u,w_{i}}}
    {A_{u,s}}
    {(a_{0},\ldots,a_{\bb{w}-1})}
    {\xi_{u,w,s}^{\mathbf{A}}(f_{(w,s)}(\sigma),a_{0},\ldots,a_{\bb{w}-1})}
}
$$
where $\mathrm{Op}_{w}(A^{f,u})_{s} = A_{u,s}^{\prod_{i\in \bb{w}}A_{u,w_{i}}}$.

Furthermore, we will denote by $p^{u}$ the $S$-sorted mapping from $\vs{u}$ to $A^{f,u}$ defined, for every $s\in S$ and $i\in\bb{u}$, as $p^{u}_{s}(v^{s}_{i}) = (\pi^{u}_{i})^{\mathbf{A}}$, and by $(p^{u})^{\sharp}$ the unique homomorphism from
$\mathbf{T}_{\Sigma}(\vs{u})$ to $\mathbf{A}^{f,u}$ such that $(p^{u})^{\sharp}\circ \eta_{\vs{u}} = p^{u}$.
\end{definition}

\begin{remark}
For a $\Sigma$-algebra $\mathbf{B}=(B,G)$, we have that $G\colon\Sigma\mor\mathrm{Op}_{\mathrm{H}_{S}}(B)$ and
$\mathbf{B}\iso\mathbf{Op}_{\mathrm{H}_{S}}(B)^{G,\lambda}$, where $\lambda$ is the empty word on $S$.  Besides, for every $u\in S^{\star}$, we have that $\mathbf{B}^{B_{u}}$, the direct $B_{u}$-power of $\mathbf{B}$, is isomorphic to $\mathbf{Op}_{\mathrm{H}_{S}}(B)^{G,u}$.
\end{remark}

\begin{lemma}\label{L:aux}
Let $\Sigma$ be an $S$-sorted signature, $\mathbf{A}$ a Hall algebra, $f\colon\Sigma\mor A$ and $u\in S^{\star}$.  Then, for
every $(w,s)\in S^{\star}\times S$, $P\in \mathrm{T}_{\Sigma}(\vs{w})_{s}$ and $a\in \prod_{i\in \bb{w}}A_{u,w_{i}}$, we have that
$$
P^{\mathbf{A}^{f,u}}(a_{0},\ldots,a_{\bb{w}-1}) =
\xi_{u,w,s}^{\mathbf{A}}((p^{w})^{\sharp}_{s}(P),a_{0},\ldots,a_{\bb{w}-1}).
$$
\end{lemma}

\begin{proof}
By algebraic induction on the complexity of $P$.  If $P$ is a variable $v_{i}^{s}$, with $i\in\bb{w}$, then
\begin{align*}
    v_{i}^{s,\mathbf{A}^{f,u}}(a_{0},\ldots,a_{\bb{w}-1})\ &=
    a^{\sharp}_{w_{i}}(v^{s}_{i})\\
    &=
    a_{i} \\
    &=
    \xi_{u,w,s}^{\mathbf{A}}((\pi^{w}_{i})^{\mathbf{A}},a_{0},\ldots,a_{\bb{w}-1} )
    \quad\text{(by $\mathrm{H}_{1}$)} \\
    &=
    \xi_{u,w,s}^{\mathbf{A}}((p^{w})^{\sharp}_{s}(v^{s}_{i}),a_{0},\ldots,a_{\bb{w}-1} ).
\end{align*}
Let us assume that $P = \sigma(Q_{0},\ldots,Q_{\bb{x}-1})$, with $\sigma\colon x\mor s$ and that, for every $j\in\bb{x}$, $Q_{j}\in \mathrm{T}_{\Sigma}(\vs{w})_{x_{j}}$ fulfills the induction hypothesis. Then we have that
\begin{align*}
    &(\sigma(Q_{0},\ldots,Q_{\bb{x}-1}))^{\mathbf{A}^{f,u}}
      (a_{0},\ldots,a_{\bb{w}-1}) \\
    &=
    \sigma^{\mathbf{A}^{f,u}}
       (
       Q_{0}^{\mathbf{A}^{f,u}}(a_{0},\ldots,a_{\bb{w}-1}),
       \ldots,
       Q_{\bb{x}-1}^{\mathbf{A}^{f,u}}(a_{0},\ldots,a_{\bb{w}-1})
       )\\
    &=
    \xi_{u,x,s}^{\mathbf{A}}
       (
       f(\sigma),
       Q_{0}^{\mathbf{A}^{f,u}}(a_{0},\ldots,a_{\bb{w}-1}),
       \ldots,
       Q_{\bb{x}-1}^{\mathbf{A}^{f,u}}(a_{0},\ldots,a_{\bb{w}-1})
       )\\
   &=
   \begin{aligned}[t]
    \xi_{u,x,s}^{\mathbf{A}}
       (
       f(\sigma),
      &\xi_{u,w,x_{0}}^{\mathbf{A}}
            ((p^{w})^{\sharp}_{x_{0}}(Q_{0}),a_{0},\ldots,a_{\bb{w}-1}),
       \ldots, \\
      &\xi_{u,w,x_{\bb{x}-1}}^{\mathbf{A}}
            ((p^{w})^{\sharp}_{x_{\bb{x}-1}}(Q_{\bb{x}-1}),a_{0},\ldots,a_{\bb{w}-1})
       ) \quad\text{(by Ind. Hypothesis)}
  \end{aligned}  \\
    &=
    \xi_{u,w,s}^{\mathbf{A}}
       (
         \xi_{w,x,s}^{\mathbf{A}}
           (f(\sigma),
            (p^{w})^{\sharp}_{x_{0}}(Q_{0}),
            \ldots,
            (p^{w})^{\sharp}_{x_{\bb{x}-1}}(Q_{\bb{x}-1})
           ),
         a_{0},
         \ldots,
         a_{\bb{w}-1}
       ) \text{(by $\mathrm{H}_{3}$)}   \\
    &=
    \xi_{u,w,s}^{\mathbf{A}}
       (
       \sigma^{\mathbf{A}_{w}}
          (
          (p^{w})^{\sharp}_{x_{0}}(Q_{0}),
          \ldots,
          (p^{w})^{\sharp}_{x_{\bb{x}-1}}(Q_{\bb{x}-1})
          ),
       a_{0},
       \ldots,
       a_{\bb{w}-1}
       ) \\
    &=
    \xi_{u,w,s}^{\mathbf{A}}
       (
       (p^{w})^{\sharp}_{s}(\sigma,Q_{0},\ldots,Q_{\bb{x}-1}),
       a_{0},
       \ldots,
       a_{\bb{w}-1}
       ) \\
    &=
    \xi_{u,w,s}^{\mathbf{A}}
       (
       (p^{w})^{\sharp}_{s}(P),
       a_{0},
       \ldots,
       a_{\bb{w}-1}
       ).
       \qedhere
\end{align*}
\end{proof}

Next we prove that, for every $S^{\star}\times S$-sorted set $\Sigma$, the Hall algebra for $(S,\Sigma)$ is isomorphic to the free Hall algebra on $\Sigma$.

\begin{proposition}\label{iso:FrH-TerH}
Let $\Sigma$ be an $S$-sorted signature, i.e., an $S^{\star}\times S$-sorted set.  Then the Hall algebra $\mathbf{Ter}_{\mathrm{H}_{S}}(\Sigma)$ is isomorphic to $\mathbf{T}_{\mathrm{H}_{S}}(\Sigma)$.
\end{proposition}

\begin{proof}
It is enough to prove that $\mathbf{Ter}_{\mathrm{H}_{S}}(\Sigma)$ has the universal property of the free Hall algebra on $\Sigma$. Therefore we have to specify an $S^{\star}\times S$-sorted mapping $h$ from $\Sigma$ to $\mathrm{Ter}_{\mathrm{H}_{S}}(\Sigma)$ such that, for every Hall algebra $\mathbf{A}$ and $S^{\star}\times S$-sorted mapping $f$ from $\Sigma$ to $A$, there is a unique homomorphism $\widehat{f}$ from $\mathbf{Ter}_{\mathrm{H}_{S}}(\Sigma)$ to $\mathbf{A}$ such that
$\widehat{f}\circ h = f$.  Let $h$ be the $S^{\star}\times S$-sorted mapping defined, for every
$(w,s)\in S^{\star}\times S$, as%
$$
h_{w,s}\nfunction
{\Sigma_{w,s}}
{\mathrm{T}_{\Sigma}(\vs{w})_{s}}
{\sigma}
{\sigma\left(v^{w_{0}}_{0},\ldots,v^{w_{\bb{w}-1}}_{\bb{w}-1}\right)}
$$
Let $\mathbf{A}$ be a Hall algebra, $f\colon\Sigma\mor A$ an $S^{\star}\times S$-sorted mapping and $\widehat{f}$ the
$S^{\star}\times S$-sorted mapping from $\mathrm{Ter}_{\mathrm{H}_{S}}(\Sigma)$ to $A$ defined, for every $(w,s)\in
S^{\star}\times S$, as $\widehat{f}_{w,s} = (p^{w})^{\sharp}_{s}$, where, we recall, $(p^{w})^{\sharp}$ is the unique homomorphism from $\mathbf{T}_{\Sigma}(\vs{w})$ to $\mathbf{A}^{f,w}$ such that $(p^{w})^{\sharp}\circ \eta_{\vs{w}} = p^{w}$. Then $\widehat{f}$ is a homomorphism of Hall algebras, because, on the one hand, for $w\in S^{\star}$ and $i\in\bb{w}$ we have that%
\begin{align*}
    \widehat{f}_{w,w_{i}}((\pi^{w}_{i})^{\mathbf{Ter}_{\mathrm{H}_{S}}(\Sigma)})
    &=
    \widehat{f}_{w,w_{i}}(v^{w_{i}}_{i}) \\
    &=
    p^{w}_{w_{i}}(v^{w_{i}}_{i}) \\
    &=
    (\pi^{w}_{i})^{\mathbf{A}},
\end{align*}
and, on the other hand, for $P\in \mathrm{T}_{\Sigma}(\vs{w})_{s}$ and
$(Q_{i})_{i\in \bb{w}}\in \mathrm{T}_{\Sigma}(\vs{u})_{w}$ we have that %
\begin{align*}
    &\widehat{f}_{u,s}(
      \xi_{u,w,s}^{\mathbf{Ter}_{\mathrm{H}_{S}}(\Sigma)}
        (P,Q_{0},\ldots,Q_{\bb{w}-1})
                   ) \\
    &=
    (p^{u})^{\sharp}_{s}(
      \mathcal{Q}^{\sharp}_{s}(P)
                     ) \\
    &=
    ((p^{u})^{\sharp}\circ \mathcal{Q})^{\sharp}_{s} (P)
    \qquad(\text{because}\, (p^{u})^{\sharp}\circ\mathcal{Q}^{\sharp} =
    ((p^{u})^{\sharp}\circ \mathcal{Q})^{\sharp})
    \\
    &=
    P^{\mathbf{A}^{f,u}}
      ((p^{u})^{\sharp}_{w_{0}}(Q_{0}),\ldots,
        (p^{u})^{\sharp}_{w_{\bb{w}-1}}(Q_{\bb{w}-1})
      ) \\
    &=
    \xi_{u,w,s}^{\mathbf{A}}
      ((p^{w})^{\sharp}_{s}(P),
       (p^{u})^{\sharp}_{w_{0}}(Q_{0}),\ldots,
       (p^{u})^{\sharp}_{w_{\bb{w}-1}}(Q_{\bb{w}-1})
      )
    \qquad(\text{by}\, \mathrm{Lemma}~\ref{L:aux})
      \\
    &=
    \xi_{u,w,s}^{\mathbf{A}}
      ( \widehat{f}_{w,s}(P),
        \widehat{f}_{u,w_{0}}(Q_{0}),\ldots,
        \widehat{f}_{u,w_{\bb{w}-1}}(Q_{\bb{w}-1})
      ).
\end{align*}
Therefore the $S^{\star}\times S$-sorted mapping $\widehat{f}$ is a homomorphism. Furthermore, $\widehat{f}\circ h = f$, because, for every $w\in S^{\star}$, $s\in S$, and $\sigma\in \Sigma_{w,s}$, we have that%
\begin{align*}
    \widehat{f}_{w,s}(h_{w,s}(\sigma)) &=
    (p^{w})^{\sharp}_{s}(\sigma(v^{w_{0}}_{0},\ldots,v^{w_{\bb{w}-1}}_{\bb{w}-1})) \\
    &=
    \sigma^{\mathbf{A}_{w}}
      (p^{w}_{{w}_{0}}(v^{w_{0}}_{0}),
        \ldots,
       p^{w}_{{w}_{\bb{w}-1}}(v^{w_{\bb{w}-1}}_{\bb{w}-1})
      ) \\
    &=
    \xi_{w,w,s}^{\mathbf{A}}
      ( f_{(w,s)}(\sigma),
        (\pi^{w}_{0})^{\mathbf{A}},
        \ldots,
        (\pi^{w}_{\bb{w}-1})^{\mathbf{A}}
      ) \\
    &=
    f_{w,s}(\sigma) \quad\text{(by $\mathrm{H}_{2}$)}.
\end{align*}
It is obvious that $\widehat{f}$ is the unique homomorphism such that $\widehat{f}\circ h = f$.  Henceforth
$\mathbf{Ter}_{\mathrm{H}_{S}}(\Sigma)$ is isomorphic to $\mathbf{T}_{\mathrm{H}_{S}}(\Sigma)$.
\end{proof}

This isomorphism together with the adjunction $\mathbf{T}_{\mathrm{H}_{S}}\dashv\mathrm{G}_{\mathrm{H}_{S}}$ has as an immediate consequence that, for every $S$-sorted set $A$ and every $S$-sorted signature $\Sigma$, the sets $\mathrm{Hom}({\Sigma},\mathrm{Op}_{\mathrm{H}_{S}}(A))$, in the category
$\mathbf{Set}^{S^{\star}\times S}$, and $\mathrm{Hom}(\mathbf{Ter}_{\mathrm{H}_{S}}(\Sigma),\mathbf{Op}_{\mathrm{H}_{S}}(A))$, in the category $\mathbf{Alg}(\mathrm{H}_{S})$, are naturally isomorphic. Actually, (1) the mapping that assigns, for an $S$-sorted set $A$, to a structure of $\Sigma$-algebra $F$ on $A$ (i.e., an $S^{\star}\times S$-sorted mapping $F$ from $\Sigma$ to $\mathrm{Op}_{\mathrm{H}_{S}}(A)$) the homomorphism of Hall algebras $\mathrm{Tr}^{(A,F)} = (\mathrm{Tr}^{\vs{w},(A,F)}_{s})_{(w,s)\in S^{\star}\times S}$ from $\mathbf{Ter}_{\mathrm{H}_{S}}(\Sigma)$ to
$\mathbf{Op}_{\mathrm{H}_{S}}(A)$, where, for every $w \in S^{\star}$, the subfamily $\mathrm{Tr}^{\vs{w},(A,F)} =
(\mathrm{Tr}^{\vs{w},(A,F)}_{s})_{s\in S}$ of $\mathrm{Tr}^{(A,F)}$ is the unique homomorphism from
$\mathbf{T}_{\Sigma}(\vs{w})$ to $(A,F)^{A_{w}}$, the direct $A_{w}$-power of $(A,F)$, such that
$\mathrm{Tr}^{\vs{w},(A,F)}\circ \eta_{\vs{w}} = \mathrm{p}^{A}_{\vs{w}}$, where $\mathrm{p}^{A}_{\vs{w}}$ is the
$S$-sorted mapping from $\vs{w}$ to $A^{A_{w}}$ defined, for every $s\in S$ and $v^{s}_{i}\in(\vs{w})_{s}$, as
$\mathrm{p}^{A}_{\vs{w},s}(v^{s}_{i}) = \mathrm{pr}^{A}_{w,i}$; together with (2) the mapping that sends an homomorphism $h$ from $\mathbf{Ter}_{\mathrm{H}_{S}}(\Sigma)$ to $\mathbf{Op}_{\mathrm{H}_{S}}(A)$ to, essentially, the algebraic structure $\mathrm{G}_{\mathrm{H}_{S}}(h)\circ \eta_{\Sigma}$ on $A$, where $\eta_{\Sigma}$ is the canonical embedding of $\Sigma$ into $\mathrm{T}_{\mathrm{H}_{S}}(\Sigma)$, which are mutually inverse bijections, provide the mentioned natural isomorphism.

The derived operation symbols of an signature can be considered as the operation symbols of a new signature. The mappings that assign to operation symbols of a signature terms relative to another signature, together with mappings between the corresponding sets of sorts, form a new class of morphisms denominated \emph{derivors}.
%As Engelfriet, in \cite{je15}, p.~26, wrote: ``\ldots a tree homomorphism [for Engelfriet ``tree homomorphism'' $\equiv$ ``derivor'', we add] has the abilities of \emph{deleting}, \ldots, \emph{copying} \ldots and \emph{permuting} \ldots subtrees [``subtrees'' $\equiv$ ``subterms'', we add].''
%For a many-sorted signature $(S,\Sigma)$, in the sequel, $\Ter_{\mathrm{H}_{S}}(\Sigma)$ will stand for  $\Ter_{\mathrm{H}}(S,\Sigma)$.

\begin{definition}
Let $\mathbf{\Sigma} = (S,\Sigma)$ and $\mathbf{\Lambda} = (T,\Lambda)$ be many-sorted signatures. A \emph{derivor} from $\mathbf{\Sigma}$ to $\mathbf{\Lambda}$ is an ordered pair $\mathbf{d} = (\varphi,d)$, where $\varphi\colon S\mor T$ and  $d\colon\Sigma\mor\mathrm{Ter}_{\mathrm{H}_{T}}(\Lambda)_{\varphi^{\star}\times\varphi}$. Thus, for every $(w,s)\in S^{\star}\times S$, %
$$
d_{w,s}\colon \Sigma_{w,s}\mor\mathrm{Ter}_{\mathrm{H}_{T}}(\Lambda)_{\varphi^{\star}(w),\varphi(s)}=
    \mathrm{T}_{\Lambda}(\vs{\varphi^{\star}(w)})_{\varphi(s)}.
$$
\end{definition}

\begin{remark}
While derivors are not the most general type of morphism that might be considered between many-sorted signatures---for instance, one could consider, in addition to hyperderivors, polyderivors, see~\cite{CVST10}---, they are an important class of such morphisms. One reason for its relevance is its formal properties (see below), another that there are many mathematical examples of them which are of interest (see~\cite{CVST10}).
\end{remark}

%\begin{remark}
%For the set of sorts $T$, let $V^{T}$ denote a fixed $T$-sorted set with a countably infinity of variables in each component, say $V^{T}_{t} = \{v^{t}_{n}\mid n\in \mathbb{N}\}$ (which is isomorphic to $\{t\}\times\mathbb{N}$). Then, for every $w\in S^{\star}$, el $T$-sorted set $\vs{\varphi^{\star}(w)} = ((\vs{\varphi^{\star}(w)})_{t})_{t\in T}$ se define, para cada $t\in T$, como:
%$$
%(\vs{\varphi^{\star}(w)})_{t} = \{v^{t}_{i}\mid i\in \varphi^{\star}(w)^{-1}[\{t\}]\} = \{v^{t}_{0},\ldots,v^{t}_{\bb{\varphi^{\star}(w)}_{t}-1}\}.
%$$
%Thus, for every $t\in T$, $(\vs{\varphi^{\star}(w)})_{t}$ is an initial segment of $V^{T}_{t}$ and $\vs{\varphi^{\star}(w)}$ is a finite subset of $V^{T}$. Esto ˙ltimo se cumple porque $\mathrm{supp}_{T}(\vs{\varphi^{\star}(w)}) = \mathrm{Im}(\varphi^{\star}(w))$ y $\bb{\varphi^{\star}(w)} = \bb{w}$ es finito. Adem·s, $\mathrm{supp}_{T}(\vs{\varphi^{\star}(w)}) =\varphi[\mathrm{supp}_{S}(\vs{w})]$.
%\end{remark}

We next introduce the linear derivors, which, as we will see later on, are a re\-mark\-able subclass of the class of the derivors, especially in regards to recognizability, the same as linear hyperderivors are with respect to hyperderivors.

\begin{definition}
Let $\mathbf{d}\colon\mathbf{\Sigma}\mor\mathbf{\Lambda}$ be a derivor from $\mathbf{\Sigma}$ to $\mathbf{\Lambda}$. We will say that $\mathbf{d}$ is \emph{linear} if, for every $(w,s)\in S^{\star}\times S$, every $\sigma\in\Sigma_{w,s}$, and every $i\in \bb{w}$, no variable $v^{\varphi(w_{i})}_{i}$ appears more than once in $d_{w,s}(\sigma)$, i.e., $\bb{d_{w,s}(\sigma)}_{v^{\varphi(w_{i})}_{i}}\leq 1$.

%for every $(w,s)\in S^{\star}\times S$, every $\sigma\in \Sigma_{w,s}$, every $t\in T$, and every $i\in \varphi^{\star}(w)^{-1}[\{t\}]$, we have that $\mathrm{card}(\bb{d_{w,s}(\sigma)}_{v_{i}^{t}})\leq 1$.
\end{definition}

For every many-sorted signature $\mathbf{\Lambda}$, $\mathrm{Ter}_{\mathrm{H}_{T}}(\Lambda)$ is the underlying many-sorted set of $\mathbf{Ter}_{\mathrm{H}_{T}}(\Lambda)$, the Hall algebra for $T$. Since by,
Proposition~\ref{iso:FrH-TerH}, $\mathbf{Ter}_{\mathrm{H}_{T}}(\Lambda)$ is isomorphic to  $\mathbf{T}_{\mathrm{H}_{T}}(\Lambda)$, the derivors can be defined, alternative, but equivalently, as ordered pairs  $\mathbf{d} = (\varphi,d)$ with $d\colon\Sigma\mor \mathrm{T}_{\mathrm{H}_{T}}(\Lambda)$.

On the other hand, every mapping $\varphi\colon S\mor T$ determines a functor from the category  $\mathbf{Alg}(\mathrm{H}_{T})$ to the category $\mathbf{Alg}(\mathrm{H}_{S})$, so $\mathrm{Ter}_{\mathrm{H}_{T}}(\Lambda)_{\varphi^{\star}\times\varphi}$ is in its turn e\-quipped with a structure of Hall algebra for $S$, which will allow us, in particular, to define the composition of derivors. We next show the existence of such a functor by defining a morphism of algebraic presentations from  $(\Sigma^{\mathrm{H}_{S}},\mathcal{E}^{\mathrm{H}_{S}})$ to $(\Sigma^{\mathrm{H}_{T}},\mathcal{E}^{\mathrm{H}_{T}})$.

\begin{proposition}
Let $\varphi\colon S\mor T$ be a mapping between sets of sorts and $h^{\varphi}$ the $S^{\star}\times S$-sorted mapping from $\Sigma^{\mathrm{H}_{S}}$ to $\Sigma^{\mathrm{H}_{T}}_{\varphi^{\star}\times \varphi}$ defined as follows: %
\begin{enumerate}
\item For every $w\in S^{\star}$ and every $i\in\bb{w}$, $h^{\varphi}(\pi^{w}_{i}) = \pi^{\varphi^{\star}(w)}_{i}$.

\item  For every $u,\,w\in S^{\star}$ and every $s\in S$,
  $h^{\varphi}(\xi_{u,w,s}) = \xi_{\varphi^{\star}(u),\varphi^{\star}(w),\varphi(s)}$.
\end{enumerate}
Then $(\varphi^{\star}\times \varphi, h^{\varphi}) \colon(S^{\star}\times S,\Sigma^{\mathrm{H}_{S}},\mathcal{E}^{\mathrm{H}_{S}})\mor
(T^{\star}\times T,\Sigma^{\mathrm{H}_{T}},\mathcal{E}^{\mathrm{H}_{T}})$ is a morphism of algebraic presentations.
\end{proposition}

The morphisms of algebraic presentations determine functors in the opposite direction between the associated categories of algebras. Therefore, each mapping $\varphi\colon S\mor T$ between sets of sorts, determines a functor $(\varphi^{\star}\times \varphi, h^{\varphi})^{\ast}$ from $\mathbf{Alg}(\mathrm{H}_{T})$ to $\mathbf{Alg}(\mathrm{H}_{S})$, which transforms Hall algebras for $T$ into Hall algebras for $S$. The action of the functor on the free Hall algebra on a $T$-sorted signature $\Lambda$ is a Hall algebra for $S$, whose underlying $S^{\star}\times S$-sorted set is
$\mathrm{Ter}_{\mathrm{H}_{T}}(\Lambda)_{\varphi^{\star}\times \varphi}$.

If $\mathbf{d}\colon\mathbf{\Sigma}\mor\mathbf{\Lambda}$ is a derivor, then
$d\colon \Sigma\mor \mathrm{Ter}_{\mathrm{H}_{T}}(\Lambda)_{\varphi^{\star}\times\varphi}$ determines a homomorphism of Hall algebras $d^{\sharp}\colon\mathbf{Ter}_{\mathrm{H}_{S}}(\Sigma)\mor
\mathbf{Ter}_{\mathrm{H}_{T}}(\Lambda)_{\varphi^{\star}\times\varphi}$. Thus, for every $(w,s)\in S^{\star}\times S$, $d^{\sharp}_{w,s}$ sends terms in $\mathrm{T}_{\Sigma}(\vs{w})_{s}$ to terms in $\mathrm{T}_{\Lambda}(\vs{\varphi^{\star}(w)})_{\varphi(s)}$.

\begin{definition}
Let $\mathbf{d}\colon\mathbf{\Sigma}\mor\mathbf{\Lambda}$ and $\mathbf{d}\colon\mathbf{\Lambda}\mor\mathbf{\Omega}$ be  derivors. Then $\mathbf{e}\circ \mathbf{d} = (\psi,e)\circ(\varphi,d)$, the \emph{composition} of $\mathbf{d}$ and  $\mathbf{e}$, is the derivor $(\psi\circ\varphi,e^{\sharp}_{\varphi^{\star}\times
\varphi}\circ d)$, where $e^{\sharp}_{\varphi^{\star}\times\varphi}\circ d$ is obtained from
$$
\begin{aligned}
\xymatrix@C=40pt@R=30pt{
\Lambda
  \ar[r]^-{\eta_{\Lambda}}
  \ar[rd]_-{e} &
\mathrm{Ter}_{\mathrm{H}_{T}}(\Lambda) \ar[d]^{e^{\sharp}}
  \\
& \mathrm{Ter}_{\mathrm{H}_{U}}(\Omega)_{\psi^{\star}\times \psi}
}
\end{aligned}
\qquad \text{as} \qquad\;\;
\begin{aligned}
\xymatrix@C=40pt@R=30pt{
\mathrm{Ter}_{\mathrm{H}_{T}}(\Lambda)_{\varphi^{\star}\times \varphi}
\ar[d]^{e^{\sharp}_{\varphi^{\star}\times \varphi}} &
\Sigma \ar[l]_-{d} \\
{\mathrm{Ter}_{\mathrm{H}_{U}}(\Omega)_{\psi^{\star}\times\psi}}_{\varphi^{\star}\times\varphi}
}
\end{aligned}
$$
being $e^{\sharp}$ the canonical extension of $e$ to the free Hall algebra on $\Lambda$.

For every many-sorted signature $\mathbf{\Sigma} = (S,\Sigma)$, the \emph{identity} at $(S,\Sigma)$
is $(\id_{S},\eta_{\Sigma})$.
\end{definition}

The just stated definition allows us to form a category of many-sorted signatures whose morphisms are the derivors.

\begin{proposition}
The many-sorted signatures together with the derivors constitute a category, denoted by $\mathbf{Sig}_{\mathfrak{d}}$.
\end{proposition}

\begin{remark}
Let $\mathbf{Sig}$ be the category whose objects are the many-sorted signatures and whose morphisms from $\mathbf{\Sigma}$ to $\mathbf{\Lambda}$ are the ordered pairs $(\varphi,d)$, where $\varphi$ is a mapping from $S$ to $T$ and $d$ an $S^{\star}\times S$-sorted mapping from $\Sigma$ to $\Lambda_{\varphi^{\star}\times \varphi}$ (thus a mapping in $\mathbf{Sig}(S)$). Note that $\mathbf{Sig}$ is the category obtained by means of the Grothendieck construction for an  appropriate contravariant functor from $\mathbf{Set}$ to $\mathbf{Cat}$.
We next show that the category $\mathbf{Sig}_{\mathfrak{d}}$ can be obtained as the Kleisli category for a suitable monad. In fact, for every  set of sorts $S$, we have the adjoint situation $\mathbf{T}_{\mathrm{H}_{S}}\dashv\mathrm{G}_{\mathrm{H}_{S}}$
%$$
%\xymatrix@=5pc{ \mathbf{Alg}(\mathrm{H}_{S})
%\ar@<1.5ex>[r]^{\mathrm{G}_{\mathrm{H}_{S}}} \ar@{}[r]|{\uadj} &
%\mathbf{Set}^{S^{\star}\times S}
%\ar@<1.5ex>[l]^{\mathbf{T}_{\mathrm{H}_{S}}} }
%$$
and, thus, a monad on $\mathbf{Sig}(S)$ which we will denote by $\mathbb{T}_{\mathrm{H}_{S}} = (\mathrm{T}_{\mathrm{H}_{S}},\eta^{\mathrm{H}_{S}},\mu^{\mathrm{H}_{S}})$. Then the ordered triple $\boldsymbol{\mathfrak{d}} = (\mathfrak{d},\eta,\mu)$ in which (1) $\mathfrak{d}$ is the endofunctor at $\mathbf{Sig}$ that sends $(S,\Sigma)$ to $(S,\mathrm{T}_{\mathrm{H}_{S}}(\Sigma))$ and a morphism $(\varphi,d)$ from $(S,\Sigma)$ to  $(T,\Lambda)$ to the morphism $(\varphi,d^{\sharp})$ from $(S,\mathrm{T}_{\mathrm{H}_{S}}(\Sigma))$ to  $(T,\mathrm{T}_{\mathrm{H}_{T}}(\Lambda))$, (2) $\eta$ the natural transformation from $\mathrm{Id}_{\mathbf{Sig}}$ to  $\mathfrak{d}$ that sends $(S,\Sigma)$ to $\eta_{(S,\Sigma)} = (\mathrm{id},\eta^{\mathrm{H}_{S}}_{\Sigma})$, and (3) $\mu$ the natural transformation from $\mathfrak{d}\circ \mathfrak{d}$ to $\mathfrak{d}$ that sends $(S,\Sigma)$ to $\mu_{(S,\Sigma)} = (\mathrm{id},\mu^{\mathrm{H}_{S}}_{\Sigma})$, is a monad on $\mathbf{Sig}$ and $\mathbf{Kl}(\boldsymbol{\mathfrak{d}})$, the Kleisli category for $\boldsymbol{\mathfrak{d}}$, is isomorphic to $\mathbf{Sig}_{\mathfrak{d}}$. This shows, in our opinion, the mathematical naturalness of the notion of derivor. Moreover, by defining a suitable notion of transformation between derivors one can equip the category $\mathbf{Sig}_{\mathfrak{d}}$ with a structure of $2$-category (this was, in fact, already done in~\cite{CVST10} for polyderivors).

\end{remark}

\begin{remark}
Since $\mathbf{Sig}$ has coproducts, $\mathbf{Kl}(\boldsymbol{\mathfrak{d}})\cong \mathbf{Sig}_{\mathfrak{d}}$ has coproducts.
\end{remark}

%Ahora hay que demostrar que las signaturas junto con los linear derivors determinan una subcategorÌa $\mathbf{Sig}_{\mathfrak{ld}}$ de $\mathbf{Sig}_{\mathfrak{d}}$ (i.e., hay que verificar que las identidades son linear derivors and that the composite of two linear derivors is a linear derivor).

%El functor contravariante $\mathrm{Alg}\colon \mathbf{Sig}\mor \mathbf{Cat}$ se puede extender hasta un functor contravariante $\mathrm{Alg}_{\mathfrak{d}}\colon\mathbf{Sig}_{\mathfrak{d}}\mor \mathbf{Cat}$.

We next associate to every derivor $\mathbf{d}\colon\mathbf{\Sigma}\mor\mathbf{\Lambda}$ a functor from  $\mathbf{Alg}(\mathbf{\Lambda}) (= \mathbf{Alg}(\Lambda))$ to $\mathbf{Alg}(\mathbf{\Sigma}) (= \mathbf{Alg}(\Sigma))$.

\begin{proposition}
Let $\mathbf{d}\colon\mathbf{\Sigma}\mor\mathbf{\Lambda}$ be a morphism in
$\mathbf{Sig}_{\mathfrak{d}}$. Then $\mathrm{Alg}_{\mathfrak{d}}(\mathbf{d})$ is the functor from  $\mathbf{Alg}(\mathbf{\Lambda})$ to $\mathbf{Alg}(\mathbf{\Sigma})$ that sends $(B,G)$ to $(B_{\varphi},G^{\mathbf{d}})$ and a homomorphism $f$ from $(B,G)$ to $(B',G')$ to the homomorphism $f_{\varphi}$ from $(B_{\varphi},G^{\mathbf{d}})$ to  $(B'_{\varphi},{G'}^{\mathbf{d}})$,
%definido como %
%$$\xymatrixcolsep={65pt} %
%  \xymatrixrowsep={35pt} %
%\functor
%{\mathbf{Alg}(T,\Lambda)}{\mathrm{Alg}_{\mathfrak{d}}(\varphi,d)}{\mathbf{Alg}(S,\Sigma)}
%{(B,G)}{(B_{\varphi},G^{(\varphi,d)})}
%{f}{f_{\varphi}}
%{(B',G')}{(B'_{\varphi},{G'}^{(\varphi,d)})}
%$$ %
where, for every $\mathbf{\Lambda}$-algebra $(B,G)$,
$G^{\mathbf{d}} = G^{\sharp}_{\varphi^{\star}\times \varphi}\circ d$
is obtained from
% $$
% \xymatrix{
% & & \Sigma \ar[d]^{d} \\
% \Lambda \ar[r]^{\eta_{\Lambda}}
%         \ar[rd]_{G} &
% \Ter_{\Hall_{T}}(\Lambda) \ar[d]^{\ext{G}}
%             \ar@{}[r]|{\text{como}}
%             &
% \Ter_{\Hall_{T}}(\Lambda)_{\varphi^{\star}\times \varphi}
% \ar[d]^{\ext{G}_{\varphi^{\star}\times \varphi}} \\
% & \Op_{\Hall_{T}}(B)  &
% \Op_{\Hall_{T}}(B)_{\varphi^{\star}\times \varphi}=\Op_{\Hall_{T}}(B_{\varphi})
% }$$
$$
\begin{aligned}
\xymatrix@C=40pt@R=33pt{
\Lambda
  \ar[r]^-{\eta_{\Lambda}}
  \ar[rd]_-{G} &
\mathrm{Ter}_{\mathrm{H}_{T}}(\Lambda) \ar[d]^{G^{\sharp}}
  \\
& \mathrm{Op}_{\mathrm{H}_{T}}(B)
}
\end{aligned}
\qquad\quad \text{as} \quad
\begin{aligned}
\xymatrix@C=40pt@R=33pt{
\mathrm{Ter}_{\mathrm{H}_{T}}(\Lambda)_{\varphi^{\star}\times \varphi}
\ar[d]^{G^{\sharp}_{\varphi^{\star}\times \varphi}} &
\Sigma \ar[l]_-{d}\\
{\mathrm{Op}_{\mathrm{H}_{T}}(B)_{\varphi^{\star}\times\varphi}}=
\mathrm{Op}_{\mathrm{H}_{S}}(B_{\varphi})
}
\end{aligned}
$$
\end{proposition}

\begin{proof}
For every $\mathbf{\Lambda}$-algebra $(B,G)$, $G\colon\Lambda\mor\mathrm{Op}_{\mathrm{H}_{T}}(B)$, and since  $\mathbf{Op}_{\mathrm{H}_{T}}(B)$ is a Hall algebra, $G$  can be extended up to the free Hall algebra on $\Lambda$. Moreover, we have that $\mathrm{Op}_{\mathrm{H}_{T}}(B)_{\varphi^{\star}\times \varphi} = \mathrm{Op}_{\mathrm{H}_{S}}(B_{\varphi})$ since, for every $(w,s)\in S^{\star}\times S$ it holds that
\begin{align*}
    (\mathrm{Op}_{\mathrm{H}_{T}}(B)_{\varphi^{\star}\times\varphi})_{w,s} &=
    \mathrm{Op}_{\mathrm{H}_{T}}(B)_{\varphi^{\star}(w),\varphi(s)} \\
    &= B_{\varphi^{\star}(w)}\mor B_{\varphi(s)} \\
    &= B_{(\varphi(w_{0}),\ldots,\varphi(w_{\bb{w}-1}))}\mor B_{\varphi(s)} \\
    &= \textstyle\prod(B_{\varphi(w_{i})}\mid i\in\bb{w}) \mor B_{\varphi(s)} \\
    &= \textstyle\prod((B_{\varphi})_{w_{i}}\mid i\in\bb{w}) \mor (B_{\varphi})_{s} \\
    &= (B_{\varphi})_{w} \mor (B_{\varphi})_{s} \\
    &= \mathrm{Op}_{\mathrm{H}_{S}}(B_{\varphi})_{w,s},
\end{align*}
thus, so defined, $G^{(\varphi,d)}$ is  an algebraic structure on $B_{\varphi}$.

Let $f\colon(B,G)\mor(B',G')$ be a homomorphism of $\mathbf{\Lambda}$-algebras, $(w,s)\in S^{\star}\times S$, and $\sigma\in\Sigma_{w,s}$. Then  $f_{\varphi}\colon(B_{\varphi},G^{\mathbf{d}}) \mor (B'_{\varphi},{G'}^{\mathbf{d}})$ is a homomorphism of $\mathbf{\Sigma}$-algebras, because $G^{\mathbf{d}}(\sigma)$ is a term operation and, consequently, the following diagram
$$
\xymatrix@C=70pt{
{B_{\varphi}}_{w}
  \ar[r]^{G^{(\varphi,d)}(\sigma)}
  \ar[d]_{{f_{\varphi}}_{w}} &
{B_{\varphi}}_{s}
  \ar[d]^{{f_{\varphi}}_{s}}  \\
{B'_{\varphi}}_{w}
  \ar[r]_{{G'}^{(\varphi,d)}(\sigma)} &
{B'_{\varphi}}_{s}
}$$
commutes. Moreover, $(g\circ f)_{\varphi}=g_{\varphi}\circ f_{\varphi}$, so that  $\mathrm{Alg}_{\mathfrak{d}}(\mathbf{d})$ is a functor.
\end{proof}

From the definition of the functor $\mathrm{Alg}_{\mathfrak{d}}$, for every derivor $\mathbf{d}\colon\mathbf{\Sigma}\mor\mathbf{\Lambda}$, it is obvious that the following diagram %
$$
\xymatrix@C=10ex@R=8ex{
\mathbf{Alg}(\mathbf{\Sigma})
  \ar[r]^{\G_{\mathbf{\Sigma}}}&
\mathbf{Set}^{S}  \\
\mathbf{Alg}(\mathbf{\Lambda})
  \ar[u]^{\mathrm{Alg}_{\mathfrak{d}}(\mathbf{d})}
  \ar[r]_{\G_{\mathbf{\Lambda}}} &
\mathbf{Set}^{T}
  \ar[u]_{\Delta_{\varphi}}
}
$$
commutes.

The previous construction can be extended up to a contravariant functor from the category $\mathbf{Sig}_{\mathfrak{d}}$ to the category $\mathbf{Cat}$.

\begin{proposition}
From $\mathbf{Sig}_{\mathfrak{d}}$ to $\mathbf{Cat}$ there exists a contravariant functor,
denoted by $\mathrm{Alg}_{\mathfrak{d}}$, that sends $(S,\Sigma)$ to $\mathbf{Alg}(S,\Sigma)$ and a morphism  $(\varphi,d)$ from $(S,\Sigma)$ to $(T,\Lambda)$ to the functor $\mathrm{Alg}_{\mathfrak{d}}(\varphi,d)$ from  $\mathbf{Alg}(T,\Lambda)$ to $\mathbf{Alg}(S,\Sigma)$.
\end{proposition}

\begin{proof}
Given $(\varphi,d)\colon(S,\Sigma) \mor (T,\Lambda)$ and $(\psi,e)\colon (T,\Lambda)\mor (U,\Omega)$, we show that %
% $$
% \xymatrix{
% (S,\Sigma) \ar[r]^{(\varphi,d)} &
% (T,\Lambda) \ar[r]^{(\psi,e)} &
% (U,\Omega)
% }
% $$
$\mathrm{Alg}_{\mathfrak{d}}(\varphi,d)\circ\mathrm{Alg}_{\mathfrak{d}}(\psi,e)= \mathrm{Alg}_{\mathfrak{d}}((\psi,e)\circ(\varphi,d))$.

Let $(A,F)$ be a $(U,\Omega)$-algebra. Then ${A_{\psi}}_{\varphi}=A_{\psi\circ\varphi}$. Moreover, we have that %
\begin{align*}
{F^{(\psi,e)}}^{(\varphi,d)}
  &=
(F^{\sharp}_{\psi^{\star}\times\psi}\circ e)^{(\varphi,d)} \\
  &=
(F^{\sharp}_{\psi^{\star}\times\psi}\circ e)^{\sharp}_{\varphi^{\star}\times\varphi}\circ d \\
  &=
({F^{\sharp}_{\psi^{\star}\times\psi}}_{\varphi^{\star}\times\varphi}\circ
    e^{\sharp}_{\varphi^{\star}\times\varphi})\circ d \\
  &=
{F^{\sharp}_{\psi^{\star}\times\psi}}_{\varphi^{\star}\times\varphi}\circ
    (e^{\sharp}_{\varphi^{\star}\times\varphi}\circ d )\\
  &=
F^{\sharp}_{(\psi\circ\varphi)^{\sharp}\times(\psi\circ\varphi)}
    \circ (e^{\sharp}_{\varphi^{\star}\times\varphi}\circ d )\\
  &=
F^{((\psi\circ\varphi),e^{\sharp}_{\varphi^{\star}\times\varphi}\circ d)} \\
  &=
F^{(\psi,e)\circ(\varphi,d)}
\end{align*}
thus $({A_{\psi}}_{\varphi},{F^{(\psi,e)}}^{(\varphi,d)})= (A_{\psi\circ\varphi},F^{(\psi,e)\circ (\varphi,d)})$.  Finally, if $f$ is a homomorphism of $(U,\Omega)$-algebras, then ${f_{\psi}}_{\varphi} = f_{\psi\circ\varphi}$.
\end{proof}

\begin{definition}
The category $\mathbf{Alg}_{\mathfrak{d}}$ is $\int^{\mathbf{Sig}_{\mathfrak{d}}}\mathrm{Alg}_{\mathfrak{d}}$, i.e., the category obtained by means of the Grothendieck construction applied to the contravariant functor $\mathrm{Alg}_{\mathfrak{d}}$.
\end{definition}

The category $\mathbf{Alg}_{\mathfrak{d}}$ has as objects the ordered pairs $(\mathbf{\Sigma},\mathbf{A})$, with $\mathbf{\Sigma}$ a many-sorted signature and $\mathbf{A}$ a $\Sigma$-algebra, and as morphisms from $(\mathbf{\Sigma},\mathbf{A})$ to $(\mathbf{\Lambda},\mathbf{B})$ the ordered pairs $(\mathbf{d},f)$, with $\mathbf{d}$ a derivor from $\mathbf{\Sigma}$ to $\mathbf{\Lambda}$ and $f$ a homomorphism of $\Sigma$-algebras from $(A,F)$ to $(B_{\varphi},G^{(\varphi,d)})$.

\begin{remark}
There exists a pseudo-functor $\mathrm{Alg}_{\boldsymbol{\mathfrak{d}}}$ from the 2-category $\mathbf{Sig}_{\boldsymbol{\mathfrak{d}}}$ to the 2-category $\mathbf{Cat}$, contravariant in the 1-cells, i.e., the derivors, and covariant in the 2-cells, i.e., the transformations between derivors (this was already done in~\cite{CVST10} for polyderivors).
\end{remark}

\begin{proposition}
Let $\mathbf{d} = (\varphi,d)$ be a derivor (resp., a linear derivor) from $\mathbf{\Sigma}$ to $\mathbf{\Lambda}$, $X$ an $S$-sorted set, $Y$ a $T$-sorted set, and $f$ an $S$-sorted mapping from $X$ to $\mathrm{T}_{\Lambda}(Y)_{\varphi}$. Then $(\mathbf{d}^{Y},f) = ((\varphi,d^{Y}),f)$, where $d^{Y}$ is, for every $(w,s)\in S^{\star}\times S$, the composition of $d_{w,s}$, which is a mapping from $\Sigma_{w,s}$ to $\mathrm{T}_{\Lambda}(\downarrow\!\varphi^{\star}(w))_{\varphi(s)}$, and $(\mathrm{in}^{@}_{\downarrow\varphi^{\star}(w), Y\cup\downarrow\varphi^{\star}(w)})_{\varphi(s)}$, the underlying mapping of the component at $\varphi(s)$ of the canonical $\Lambda$-homomorphism $\mathrm{in}^{@}_{\downarrow\varphi^{\star}(w), Y\cup\downarrow\varphi^{\star}(w)}$ from $\mathbf{T}_{\Lambda}(\downarrow\!\varphi^{\star}(w))$ to $\mathbf{T}_{\Lambda}(Y\cup \downarrow\!\varphi^{\star}(w))$, is a hyperderivor (resp., a linear hyperderivor) from $(\mathbf{\Sigma},X)$ to $(\mathbf{\Lambda},Y)$.
\end{proposition}

In other words, some hyperderivors, but not all of them, are factorizable through a derivor.

\begin{proposition}
Let $\mathbf{d}$ be a derivor from $\mathbf{\Sigma}$ to $\mathbf{\Lambda}$, $X$ an $S$-sorted set, $Y$ a $T$-sorted set, $f$ an $S$-sorted mapping from $X$ to $\mathrm{T}_{\Lambda}(Y)_{\varphi}$, $f^{\sharp}$ the canonical homomorphism of $\Sigma$-algebras from $\mathbf{T}_{\Sigma}(X)$ to $\mathrm{Alg}_{\mathfrak{d}}(\mathbf{d})(\mathbf{T}_{\Lambda}(Y))$, $s\in S$, and $L\in\mathrm{Rec}_{\varphi(s)}(\mathbf{T}_{\Lambda}(Y))$. Then $\left(f^{\sharp}_{\varphi(s)}\right)^{-1}[L]\in\mathrm{Rec}_{s}(\mathbf{T}_{\Sigma}(X))$.
\end{proposition}

\begin{proof}
It follows from Proposition~\ref{PRecH}.
\end{proof}

\begin{assumption}
For the following proposition, as was the case with Proposition~\ref{PRecLH}, we will assume that $S$, $T$, $\Sigma$ and $X$ are finite.
\end{assumption}

\begin{proposition}
Let $\mathbf{d}$ be a linear derivor from $\mathbf{\Sigma}$ to $\mathbf{\Lambda}$, $X$ an $S$-sorted set, $Y$ a $T$-sorted set, $f$ an $S$-sorted mapping from $X$ to $\mathrm{T}_{\Lambda}(Y)_{\varphi}$, $f^{\sharp}$ the canonical homomorphism of $\Sigma$-algebras from $\mathbf{T}_{\Sigma}(X)$ to $\mathrm{Alg}_{\mathfrak{d}}(\mathbf{d})(\mathbf{T}_{\Lambda}(Y))$, $s\in S$, and $L\in \mathrm{T}_{\Sigma}(X)^{\wp}_{s}$. If $L\in\mathrm{Rec}_{s}(\mathbf{T}_{\Sigma}(X))$, then $f^{\sharp}_{s}[L]\in\mathrm{Rec}_{\varphi(s)}(\mathbf{T}_{\Lambda}(Y))$.
\end{proposition}

\begin{proof}
It follows from Proposition~\ref{PRecLH}.
\end{proof}

%%%%%%%%%%%%%%%%%%%%%%%%%%%%%%%%%%%

\vspace{0.2cm}
\noindent\textbf{Acknowledgements.}
The work of the second author was supported by the Mi\-nis\-te\-rio de Econom\'\i a y Competitividad, Spain, and the European Regional De\-vel\-op\-ment Fund, European Union (MTM2014-54707-C3-1-P).
% We would like to thank the reviewers for a careful reading and very helpful comments. The authors wish to express their appreciation to the referees for their careful reading of the manuscript and may suggested modifications of the text.

%$\boldsymbol{\cdot}$ $\Large\boldsymbol{\cdot}$

%%%%%%%%%%%%%%%%%%%%%%%%%%%%%%%%%%%%%%%%%%%%%%%

\end{document}